\documentclass[acmsmall,screen,nonacm]{acmart}
\renewcommand\footnotetextcopyrightpermission[1]{}

\makeatletter
\let\@authorsaddresses\@empty
\makeatother

\usepackage{savesym}
\savesymbol{Bbbk}

\usepackage{ottalt}
\usepackage{amsmath}
\usepackage{amssymb}
\usepackage{tikz}
\usepackage{tikz-cd}
\usepackage{proof}
\usepackage{hyperref}
\usepackage{stmaryrd}

\usepackage{caption}
\usepackage{subcaption}

\newcommand{\Bw}{$\text{BLDC}^{\omega}$}
\newcommand{\Gc}{\textsc{GCore}}
\newcommand{\Dt}{$\text{DDC}^{\top}$}

\title{Unifying Linearity and Dependency Analyses}
\author{Pritam Choudhury}
\affiliation{
 \institution{University of Pennsylvania}
  \city{Philadelphia}
  \country{USA}
}
\acmJournal{PACMPL}

\newcommand{\ottdrule}[4][]{{\displaystyle\frac{\begin{array}{l}#2\end{array}}{#3}\quad\ottdrulename{#4}}}

\newcommand{\ottpremise}[1]{ #1 \\}
\newenvironment{ottdefnblock}[3][]{ \framebox{\mbox{#2}} \quad #3 \\[0pt]}{}

\newcommand{\ottnt}[1]{\mathit{#1}}
\newcommand{\ottmv}[1]{\mathit{#1}}

\newcommand{\ottdrulename}[1]{\textsc{#1}}

\newcommand{\ottafterlastrule}{\\}

\newcommand{\ottdruleSTXXLamOmega}[1]{\ottdrule[#1]{%
\ottpremise{  \Gamma  ,   \ottmv{x}  :^{  \ottnt{q}  \cdot  \ottnt{r}  }  \ottnt{A}    \vdash  \ottnt{b}  :^{ \ottnt{q} }  \ottnt{B} }%
\ottpremise{  \ottnt{q}  =   \omega    \Rightarrow   \ottnt{r}  =   \omega   }%
\ottpremise{ \ottnt{q_{{\mathrm{0}}}}  \neq   0  }%
}{
  \ottnt{q_{{\mathrm{0}}}}  \cdot  \Gamma   \vdash   \lambda^{ \ottnt{r} }  \ottmv{x}  :  \ottnt{A}  .  \ottnt{b}   :^{  \ottnt{q_{{\mathrm{0}}}}  \cdot  \ottnt{q}  }   {}^{ \ottnt{r} }\!  \ottnt{A}  \to  \ottnt{B}  }{%
{\ottdrulename{ST\_LamOmega}}{}%
}}

\newcommand{\ottdruleSTXXLamOmegaZero}[1]{\ottdrule[#1]{%
\ottpremise{  \Gamma  ,   \ottmv{x}  :^{  \ottnt{q}  \cdot  \ottnt{r}  }  \ottnt{A}    \vdash  \ottnt{b}  :^{ \ottnt{q} }  \ottnt{B} }%
\ottpremise{  \ottnt{q}  =   \omega    \Rightarrow   \ottnt{r}  =   \omega   }%
}{
 \Gamma  \vdash   \lambda^{ \ottnt{r} }  \ottmv{x}  :  \ottnt{A}  .  \ottnt{b}   :^{ \ottnt{q} }   {}^{ \ottnt{r} }\!  \ottnt{A}  \to  \ottnt{B}  }{%
{\ottdrulename{ST\_LamOmega0}}{}%
}}

\newcommand{\ottdruleSTXXOmega}[1]{\ottdrule[#1]{%
\ottpremise{   \omega   \cdot  \Gamma   \vdash  \ottnt{a}  :^{ \ottnt{q} }  \ottnt{A} }%
\ottpremise{ \ottnt{q}  \neq   0  }%
}{
   \omega   \cdot  \Gamma   \vdash  \ottnt{a}  :^{  \omega  }  \ottnt{A} }{%
{\ottdrulename{ST\_Omega}}{}%
}}

\newcommand{\ottdrulePTSXXLamOmega}[1]{\ottdrule[#1]{%
\ottpremise{  \Gamma  ,   \ottmv{x}  :^{  \ottnt{q}  \cdot  \ottnt{r}  }  \ottnt{A}    \vdash  \ottnt{b}  :^{ \ottnt{q} }  \ottnt{B} }%
\ottpremise{ \Delta  \vdash_{0}   \Pi  \ottmv{x}  :^{ \ottnt{r} } \!  \ottnt{A}  .  \ottnt{B}   :   \ottmv{s}  }%
\ottpremise{  \lfloor  \Gamma  \rfloor   =  \Delta }%
\ottpremise{  \ottnt{q}  =   \omega    \Rightarrow   \ottnt{r}  =   \omega   }%
\ottpremise{ \ottnt{q_{{\mathrm{0}}}}  \neq   0  }%
}{
  \ottnt{q_{{\mathrm{0}}}}  \cdot  \Gamma   \vdash   \lambda^{ \ottnt{r} }  \ottmv{x}  :  \ottnt{A}  .  \ottnt{b}   :^{  \ottnt{q_{{\mathrm{0}}}}  \cdot  \ottnt{q}  }   \Pi  \ottmv{x}  :^{ \ottnt{r} } \!  \ottnt{A}  .  \ottnt{B}  }{%
{\ottdrulename{PTS\_LamOmega}}{}%
}}

  \renewottcommands[ott]



\hypersetup{
           breaklinks=true,   % splits links across lines
           colorlinks=true,   % displays links as colored text instead of blocks
           pdfusetitle=true}   % \title and \author values into pdf metadata

\begin{document}

\bibliographystyle{ACM-Reference-Format}
\citestyle{acmauthoryear} 

\begin{abstract}

Linearity and dependency analyses are key to several applications in computer science, especially, in resource management and information flow control. What connects these analyses is that both of them need to model at least two different worlds with constrained mutual interaction. To elaborate, a typical linearity analysis would model nonlinear and linear worlds with the constraint that derivations in the nonlinear world cannot make use of assumptions from the linear world; a typical dependency analysis would model low-security and high-security worlds with the constraint that information from the high-security world never leaks into the low-security world. Now, though linearity and dependency analyses address similar problems, these analyses are carried out by employing different methods. For linearity analysis, type systems employ the comonadic exponential modality from Girard's linear logic. For dependency analysis, type systems employ the monadic modality from Moggi's computational metalanguage. Owing to this methodical difference, a unification of the two analyses, though theoretically and practically desirable, is not straightforward. 

Fortunately, with recent advances in graded-context type systems, it has been realized that linearity and dependency analyses can be viewed through the same lens. However, existing graded-context type systems fall short of a unification of linearity and dependency analyses. The problem with existing graded-context type systems is that though their linearity analysis is general, their dependency analysis is limited, primarily because the graded modality they employ is comonadic and not monadic. In this paper, we address this limitation by systematically extending existing graded-context type systems so that the graded modality is both comonadic and monadic. This extension enables us to unify linearity analysis with a general dependency analysis. We present a unified Linear Dependency Calculus, LDC, which analyses linearity and dependency using the same mechanism in an arbitrary Pure Type System. We show that LDC is a general linear and dependency calculus by subsuming into it the standard calculi for the individual analyses.  

%In a type system, bounded linearity analysis guarantees fairness of resource usage whereas dependency analysis guarantees security of information flow. Both analyses have several applications in programming languages, for example, counting uses of variables in programs, erasing irrelevant parts of programs, etc. These analyses, however, use different methods to realize their objectives. Bounded linearity analysis fine-grains the exponential modality from Girard's linear logic to track resource usage whereas dependency analysis fine-grains the modal type constructor from Moggi's computational metalanguage to track information flow. These modalities are respectively comonadic and monadic in nature. In this paper, by using a graded modality that is simultaneously comonadic and monadic, we show that these analyses can be smoothly combined into a single system. Our system unifies several existing calculi analysing either bounded linearity or dependency. Our work also shows that bounded linearity analysis and dependency analysis are essentially the same analysis viewed from two different angles.

%Linearity analysis needs to model linear and non-linear worlds with the constraint that derivations in the non-linear world cannot depend upon assumptions from the linear world. Dependency analysis needs to model low and high security worlds with the constraint that information from the high-security world is not leaked to the low-security world.
\end{abstract}

\maketitle
\pagestyle{plain}

\section{Introduction} \label{intro}

Type systems formalize our intuition of correct programs: upon formalization, a correct program is one that is well-typed. Our intuition of correctness, however, varies depending upon application. For example: in a distributed system, correctness would entail absence of simultaneous write access to a file by multiple users; in a security system, correctness would entail absence of read access to secret files by public users. To formalize such and other similar notions of correctness, we employ linear type systems and dependency type systems. These type systems have a wide variety of applications. 

Linear type systems are good at \textit{managing resource usage}. They can reason about state and enable compiler optimizations like in-place update of memory locations \citep{linearchangeworld}. They can be employed for fine-grained memory management \citep{chirimar,linearregions}. Through session types, they can guarantee absence of deadlocks in distributed systems \citep{sessiontypes}.  Owing to their utility, linear types have found their way into several programming languages, like Haskell \citep{linearhaskell}, Granule \citep{orchard}, Idris 2 \citep{idris2}.

Dependency type systems are good at \textit{controlling information flow}. In the form of security type systems \citep{slam,volpano}, they guarantee that low-security outputs do not depend upon high-security inputs. In the form of binding-time type systems \citep{gomard,lambdacirc}, they guarantee that early-bound expressions do not depend upon late-bound ones. Dependency type systems are quite common in practical programming languages. For example, the metaprogramming language MetaOcaml \citep{meta} is based on the dependency type system of \citet{lambdacirc}. The Jif extension of Java, employed for ensuring secure information flow, is based on the dependency type system of \citet{jflow}. 

Though linear type systems and dependency type systems serve different purposes, they essentially address the same abstract problem. The problem is to model at least two different worlds that interact following given constraints. For example:  a typical linear type system  needs to model nonlinear and linear worlds with the constraint that derivations in the nonlinear world cannot make use of assumptions from the linear world; a typical dependency type system needs to model low-security and high-security worlds with the constraint that information from the high-security world never leaks into the low-security world. This fundamental similarity suggests that linearity and dependency analyses could be unified.

There are several benefits to such a unification. First, from a theoretical perspective, it would unify the standard calculi for linearity analysis \cite{benton,dill} with the standard calculi for dependency analysis \cite{dcc,igarashi}. Second, from a practical perspective, it would allow programmers to use the same type system for both linearity and dependency analyses. Presently, programmers use different type systems for this purpose. For example, Haskell programmers use the Linear Haskell extension \cite{linearhaskell} for linearity analysis and LIO library \cite{ifchaskell} for dependency analysis. Third, it would allow a combination of the two analyses. A combined analysis is more powerful than the individual analyses done separately because it would allow arbitrary combination of usage and flow constraints. For example, in a combined analysis, a piece of data may be simultaneously linear and private. 

A unification of linear and dependency type systems, though desirable, is not straightforward. This is so because these type systems employ different methods to enforce their respective constraints. Linear type systems \cite{abramsky,ill,dill,turner} employ the \textit{comonadic} exponential modality, $!$, taken from linear logic \cite{girard}, to manipulate nonlinear resources in a base linear world. On the other hand, dependency type systems \cite{dcc,slam,volpano} employ the \textit{monadic} modality, $T$, taken from computational metalanguage \cite{moggi}, to manipulate high-security values in a base low-security world. 

The modalities $!$ and $T$ behave differently. For example: In a dependency calculus, there is no non-constant function of type $  T \:   \mathbf{Bool}    \to   \mathbf{Bool}  $ because any such function would leak information. However, in a linear calculus, there exists a non-constant function, $  \emptyset   \vdash   \lambda  \ottmv{x}  .   \mathbf{derelict} \:   \ottmv{x}     : \:    ! \:   \mathbf{Bool}    \to   \mathbf{Bool}   $, since nonlinear resources can be used linearly. Next, in a dependency calculus, $  \ottmv{x}  :   \mathbf{Bool}    \vdash   \eta \:   \ottmv{x}    : \:   T \:   \mathbf{Bool}   $ is a valid typing judgment because information can flow from low-security world to high-security world. However, in a linear calculus, $  \ottmv{x}  :   \mathbf{Bool}    \vdash   ! \:   \ottmv{x}    : \:   ! \:   \mathbf{Bool}   $ is not a valid typing judgment because nonlinear resources cannot make use of linear assumptions. Owing to these differences, for a long time, linearity and dependency analyses have been carried out independent of one another.

However, with recent advances in graded-context type systems \cite{ghica,brunel,petricek,orchard,abel20}, it has been realized that the two analyses can be viewed through the same lens. Graded-context type systems are type systems parametrized over preordered semirings or similar abstract algebraic structures. These type systems can carry out a wide variety of usage analyses through different instantiations of the parametrizing structure. In other words, they are not limited to linear/nonlinear usage analysis but can also analyze other forms of usage, like no usage, affine usage, bounded usage, etc. As such, graded-context type systems are more general than linear type systems. Now, to understand how a graded-context type system might analyze dependency, we need to consider the similarity between the modalities $!$ and $T$, from the perspective of usage analysis. Assuming the base world to be linear, the monadic modality, $T$, may be understood as indicating no usage because terms from a `monad world' cannot be used outside that world. This understanding brings the modalities $!$ (indicating unrestricted usage) and $T$ (indicating no usage) under the same umbrella and gives us a perspective on the differences in their behavior mentioned above.

%Graded-context type systems have their basis in bounded linear logic \cite{bounded} that allows bounded reuse in lieu of just linear and nonlinear use, as is the case with linear logic.  This makes graded-context type systems more general than linear type systems because by grading the $!$-modality, they can reason about linear use, unrestricted use, no use, bounded use, etc.  For example, in a graded-context type system, unrestricted use and no use are expressed using $!_{ \omega }$ and $!_0$ modalities respectively. Note that the standard $!$ and $T$ modalities of linear and dependency calculi can be represented using $!_{ \omega }$ and $!_0$ modalities respectively. This representation brings the two modalities 

However, the problem with existing graded-context type systems is that though their linearity analysis is general, their dependency analysis is limited. There are several aspects of dependency analysis that these systems cannot capture. We discuss them in detail in the next section. The main reason behind this shortcoming is that graded-context type systems have been designed for analyzing coeffects, i.e. how programs depend upon their contexts. Coeffects include linearity (single usage), irrelevance (no usage), etc. Dependency, however, behaves more like an effect. To elaborate: low and high security computations may be seen as pure and effectful computations respectively. An effect like dependency is not well captured by existing graded-context type systems, which are basically coeffect calculi.

In this paper, we show that by systematically extending existing graded-context type systems, we can use them for a general dependency and linearity analysis. We design a calculus, LDC, that can simultaneously analyze, using the same mechanism, a coeffect like linearity and an effect like dependency. LDC is parametrized by an arbitrary Pure Type System and it subsumes standard calculi for linearity and dependency analyses. We show that linearity and dependency analyses in LDC are correct using a heap semantics.

In summary, we make the following contributions:
\begin{itemize}
%\item We show that bounded linearity analysis and dependency analysis can be smoothly combined into a single system.
\item We present a language, LDC, parametrized by an arbitrary pure type system, that analyzes linearity and dependency using the same mechanism.
\item We show that LDC subsumes the standard calculi for analyzing linearity and dependency, for example, Linear Nonlinear $\lambda$-calculus of \citet{benton}, DCC of \citet{dcc}, Sealing Calculus of  \citet{igarashi}, etc.
\item We show that correctness of both linearity and dependency analyses in LDC follow from the soundness theorem for the calculus.
\item We show that LDC can carry out a combined linearity and dependency analysis.
\end{itemize} 

\section{Challenges and Resolution} \label{sec:challenges}

\subsection{Dependency Analysis: Salient Aspects} \label{subsec:depsalient}

In the previous section, we discussed about dependency analysis with respect to low-security and high-security worlds. Such an analysis can be extended to an arbitrary (finite) number of worlds with dependency constraints among them. \citet{denning1} observed that dependency constraints upon worlds result in a lattice structure. Recall that a lattice $ \mathcal{L}  = (L, \sqsubseteq )$ is a partially-ordered set, where every pair of elements has a least upper bound, also called join and denoted by $\sqcup$, and a greatest lower bound, also called meet and denoted by $\sqcap$. If $L$ is finite, then it has a top element and a bottom element, denoted $\top$ and $\bot$ respectively, such that $\bot  \sqsubseteq  \ell  \sqsubseteq  \top$, for all $\ell \in  \mathcal{L} $. Now, to give an example of a dependency lattice, consider the following set of worlds: a low-security world $ \mathbf{L} $, two medium-security worlds $ \mathbf{M_1} $ and $ \mathbf{M_2} $ that do not share information with each other, and a high-security world $ \mathbf{H} $. These constraints may be modeled by a diamond lattice, $ \mathcal{L}_{\diamond} $, where $  \mathbf{L}   \sqsubseteq   \mathbf{M_1}   \sqsubseteq  \mathbf{H} $ and $  \mathbf{L}   \sqsubseteq   \mathbf{M_2}   \sqsubseteq  \mathbf{H} $, with the idea that any information flow that goes against this lattice order is illegal. More generally, for an arbitrary lattice, $ \mathcal{L} $, given levels $\ell_{{\mathrm{1}}}, \ell_{{\mathrm{2}}} \in  \mathcal{L} $, we have: $ \ell_{{\mathrm{1}}}  \sqsubseteq  \ell_{{\mathrm{2}}} $ if and only if information may flow from $\ell_{{\mathrm{1}}}$ to $\ell_{{\mathrm{2}}}$. 

Dependency type systems \cite{dcc,igarashi} are based on this lattice model of information flow. These type systems grade the monadic modality, $T$, of Moggi's computational metalanguage with labels drawn from an abstract dependency lattice, $ \mathcal{L} $. The idea behind the grading is that for any $\ell \in  \mathcal{L} $, $ T_{ \ell } \:  \ottnt{A} $ would denote $\ell$-secure terms of type $\ottnt{A}$. Figuratively, $ T_{ \ell } \:  \ottnt{A} $ represents the terms of type $\ottnt{A}$, but enclosed in $\ell$-secure boxes, which may be opened only with $\ell$-security clearance. For example, $ T_{  \mathbf{H}  } \:   \mathbf{Bool}  $ is the type of high-secure booleans, which may be observed only with high-security clearance. Now, a dependency type system may be said to be sound if it ensures that $\ell_{{\mathrm{1}}}$-secure inputs do not affect $\ell_{{\mathrm{2}}}$-secure outputs, whenever $\neg( \ell_{{\mathrm{1}}}  \sqsubseteq  \ell_{{\mathrm{2}}} )$. This principle may be rephrased as: whenever $\neg( \ell_{{\mathrm{1}}}  \sqsubseteq  \ell_{{\mathrm{2}}} )$, one cannot gain any information about inputs from $\ell_{{\mathrm{1}}}$ by observing outputs at $\ell_{{\mathrm{2}}}$. This is the well-known principle of \textit{noninterference}, from which dependency type systems draw their strength. To give a concrete example, a corollary of this principle would be that given security lattice $ \mathcal{L}_{\diamond} $, any function of type $   T_{  \mathbf{H}  } \:   \mathbf{Bool}     \to   T_{  \mathbf{L}  } \:   \mathbf{Bool}   $ is a constant function.  

Now, we consider the nature of the graded modality, $T_{\ell}$. While analyzing dependency, one should be able to move freely between the types $ T_{ \ell } \:   T_{ \ell } \:  \ottnt{A}  $ and $ T_{ \ell } \:  \ottnt{A} $ because enclosing a term in a cascade of $\ell$-secure boxes is the same as enclosing it in a single $\ell$-secure box. Moving from $ T_{ \ell } \:  \ottnt{A} $ to $ T_{ \ell } \:   T_{ \ell } \:  \ottnt{A}  $ is easy because one just needs to put an extra wrapper. But moving from $ T_{ \ell } \:   T_{ \ell } \:  \ottnt{A}  $ to $ T_{ \ell } \:  \ottnt{A} $ requires some consideration because it involves unwrapping. To handle the situation, one can invoke the properties of the modality. The modality, being monadic, supports the standard join operation: $   T_{ \ell_{{\mathrm{1}}} } \:   T_{ \ell_{{\mathrm{2}}} } \:  \ottnt{A}     \to   T_{  \ell_{{\mathrm{1}}}  \: \sqcup \:  \ell_{{\mathrm{2}}}  } \:  \ottnt{A}  $, where $\ell_{{\mathrm{1}}}, \ell_{{\mathrm{2}}} \in  \mathcal{L} $. Via this join operation, one can move from  $ T_{ \ell } \:   T_{ \ell } \:  \ottnt{A}  $ to $ T_{ \ell } \:  \ottnt{A} $. Now, if dependency labels correspond to effects, then the join operation corresponds to computing the union of these effects. Just as computing union is important in an effect calculus, having a join operation is important in a dependency calculus. Later in this section, we shall see that existing graded-context type systems cannot derive a join operation. This significantly limits dependency analysis in these type systems.

Before moving further, we want to point out another important aspect of dependency analysis that is sometimes ignored. It pertains to the treatment of functions that are wrapped under $T_{ \mathbf{H} }$. Consider the type: $ T_{  \mathbf{H}  } \:   (     T_{  \mathbf{H}  } \:   \mathbf{Bool}     \to   \mathbf{Bool}    )  $. What should the values of this type be? DCC \cite{dcc} would answer: just $ \eta_{  \mathbf{H}  } \:   (   \lambda  \ottmv{x}  .   \mathbf{true}    )  $ and $ \eta_{  \mathbf{H}  } \:   (   \lambda  \ottmv{x}  .   \mathbf{false}    )  $. But Sealing Calculus \cite{igarashi} would answer: $ \mathbf{seal}_{  \mathbf{H}  } \:   (   \lambda  \ottmv{x}  .   \mathbf{true}    )  $, $ \mathbf{seal}_{  \mathbf{H}  } \:   (   \lambda  \ottmv{x}  .   \mathbf{false}    )  $, $ \mathbf{seal}_{  \mathbf{H}  } \:   (   \lambda  \ottmv{x}  .   \mathbf{unseal}_{  \mathbf{H}  } \:   \ottmv{x}     )  $ and $ \mathbf{seal}_{  \mathbf{H}  } \:   (   \lambda  \ottmv{x}  .   \mathbf{not} \:   (   \mathbf{unseal}_{  \mathbf{H}  } \:   \ottmv{x}    )     )  $. This difference stems from the fact that in Sealing Calculus, the function $   T_{  \mathbf{H}  } \:   \mathbf{Bool}     \to   \mathbf{Bool}  $, if wrapped under $T_{ \mathbf{H} }$, may return a high-security output whereas in DCC, it must always return a constant output. In this regard, Sealing Calculus is more general than DCC because it doesn't restrict any function from returning high-security values, if the function itself is wrapped under the high-security label, $ \mathbf{H} $.  Note here that over terminating computations, Sealing Calculus subsumes DCC and is, in fact, more general than DCC, as we see above.

What makes Sealing Calculus more general than DCC is the form of its typing judgment. While DCC employs the traditional form of typing judgment, Sealing Calculus employs a labeled typing judgment of the form: 
\begin{equation} 
  \ottmv{x_{{\mathrm{1}}}}  :  \ottnt{A_{{\mathrm{1}}}}   ,   \ottmv{x_{{\mathrm{2}}}}  :  \ottnt{A_{{\mathrm{2}}}}  , \ldots , x_n : A_n \vdash b :^{\ell} B,
\label{sctyp}
\end{equation}
where $\ell \in \mathcal{L}$. The intuitive meaning of this judgment is that $b$ is an $\ell$-secure term of type $\ottnt{B}$, assuming $x_i$ has type $A_i$, for $i = 1, 2, \ldots, n$. The key advantage of a labeled typing judgment is that it facilitates smooth sealing and unsealing of secure values, as we see in the rules below:
\begin{center}
\drule[]{SC-Seal} \hspace*{20pt} \drule[]{SC-Unseal}
\end{center}
On the other hand, DCC needs to employ a nonstandard \textit{bind}-rule along with an auxiliary protection judgment in order to unseal secure values. The labeled typing judgment of Sealing Calculus also enables any function to return high-security outputs, provided the function itself is wrapped under a high-security label, as discussed above. Owing to these good properties, we use the Sealing Calculus as our model for dependency analysis. As an aside, we allow nonterminating computations in our language, even though Sealing Calculus does not. For a more elaborate comparison of Sealing Calculus and DCC, see \citet{gmcc}.  

\subsection{Graded-Context Type Systems: Salient Aspects} \label{gradedsalient}

Over the recent years, graded-context type systems \cite{ghica,brunel,petricek,effcoeff,mcbride,atkey,orchard,abel20,grad,moon} have been successfully employed for reasoning about resource usage in programs. Graded-context type systems have their roots in bounded linear logic \citep{bounded}, which adapts the `propositions as resources' doctrine of linear logic to characterize resource-bound computations. To achieve its goal, bounded linear logic  imposes wider distinctions on usage of resources, compared to linear logic. Graded-context type systems draw inspiration from bounded linear logic in widening the nature and scope of usage analysis, compared to traditional linear type systems \citep{benton,dill}. 

The power and flexibility of graded-context type systems stem from the fact that they are parametrized by an abstract preordered semiring or a similar structure that represents an algebra of resources. Recall that a semiring, $(Q,+,\cdot,0,1)$, is a set $Q$ with two binary operators, $+$ (addition) and $\cdot$ (multiplication), along with their respective identities, $0$ and $1$, such that addition is commutative and associative, multiplication is associative and distributive over addition, and $0$ is an annihilator for multiplication. A preordered semiring $ \mathcal{Q}  = (Q,+,\cdot,0,1, <: )$ is a semiring $(Q,+,\cdot,0,1)$ with a binary preorder relation, $ <: $, that respects the binary operations, meaning, if $ \ottnt{q_{{\mathrm{1}}}}  <:  \ottnt{q_{{\mathrm{2}}}} $, then $  \ottnt{q}  +  \ottnt{q_{{\mathrm{1}}}}   <:   \ottnt{q}  +  \ottnt{q_{{\mathrm{2}}}}  $ and $  \ottnt{q}  \cdot  \ottnt{q_{{\mathrm{1}}}}   <:   \ottnt{q}  \cdot  \ottnt{q_{{\mathrm{2}}}}  $ and $  \ottnt{q_{{\mathrm{1}}}}  \cdot  \ottnt{q}   <:   \ottnt{q_{{\mathrm{2}}}}  \cdot  \ottnt{q}  $, for all $q \in Q$. By varying the parameter $ \mathcal{Q} $, graded-context type systems can carry out a variety of analyses. Among these analyses, two specific ones that interest us are linearity and affinity. The preordered semirings employed for the two analyses, denoted $ \mathcal{Q}_{\text{Lin} } $ and $ \mathcal{Q}_{\text{Aff} } $ respectively, both have 3 elements: $0,1$ and $ \omega $, with $1 + 1 =   \omega   +   1   =   1   +   \omega   =   \omega   +   \omega   =  \omega $ and $  \omega   \cdot   \omega   =  \omega $. However, $ \mathcal{Q}_{\text{Lin} } $ and $ \mathcal{Q}_{\text{Aff} } $ are ordered differently, as shown in Figure \ref{linafford}. Owing to this difference in ordering, $1$ signifies linear usage in $ \mathcal{Q}_{\text{Lin} } $ but affine usage in $ \mathcal{Q}_{\text{Aff} } $. However, in both $ \mathcal{Q}_{\text{Lin} } $ and $ \mathcal{Q}_{\text{Aff} } $, $0$ signifies no usage while $ \omega $ signifies unrestricted usage.
\begin{figure}
\centering
\begin{subfigure}[b]{0.3\textwidth}
\centering
\begin{tikzpicture}[scale=0.3]
\node at (0,0) (lw) {$ \omega $};
\node at (2,2) (l1) {$1$};
\node at (-2,2) (l0) {$0$};
\draw (lw) -- (l1);
\draw (lw) -- (l0);
\end{tikzpicture}
\caption{Ordering in $ \mathcal{Q}_{\text{Lin} } $} 
\end{subfigure}
\begin{subfigure}[b]{0.3\textwidth}
\centering
\begin{tikzpicture}[scale=0.3]
\node at (0,0) (aw) {$ \omega $};
\node at (0,2) (a1) {$1$};
\node at (0,4) (a0) {$0$};
\draw (aw) -- (a1);
\draw (a1) -- (a0);
\end{tikzpicture} 
\caption{Ordering in $ \mathcal{Q}_{\text{Aff} } $}
\end{subfigure}
\caption{Ordering for tracking linear and affine use}
\label{linafford}
\end{figure}
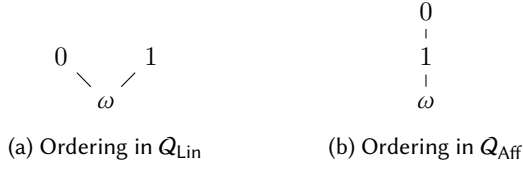

%For example: instantiating $ \mathcal{Q} $ to the set of natural numbers with discrete order, they can carry out an exact usage analysis; instantiating $ \mathcal{Q} $ to the set of natural numbers with descending natural order, they can carry out a bounded usage analysis; etc. 

Graded-context type systems use elements of the parametrizing preordered semiring, $ \mathcal{Q} $, to grade the $!$-modality. The idea behind the grading is that for any $\ottnt{q} \in  \mathcal{Q} $, $ \mspace{2mu} !_{ \ottnt{q} } \mspace{1mu}  \ottnt{A} $ would denote $\ottnt{q}$-usage terms of type $\ottnt{A}$. For example, when $ \mathcal{Q} $ is set to $ \mathcal{Q}_{\text{Lin} } $, the type $ \mspace{2mu} !_{  1  } \mspace{1mu}  \ottnt{A} $ represents terms of type $A$ that must be used exactly once. Again, when $ \mathcal{Q} $ is set to $ \mathcal{Q}_{\text{Aff} } $, the type $ \mspace{2mu} !_{  1  } \mspace{1mu}  \ottnt{A} $ represents terms that may be used at most once. However in both the cases, the type $ \mspace{2mu} !_{  \omega  } \mspace{1mu}  \ottnt{A} $ represents terms that may be used without restriction. Now, a graded-context type system may be said to be sound if it accounts usage correctly, thereby ensuring fairness of usage. Fairness of usage implies absence of arbitrary copying or discarding of resources. To give a concrete example, in a graded-context polymorphic type system, parametrized by $ \mathcal{Q}_{\text{Lin} } $, the types $\forall \alpha .    \mspace{2mu} !_{  1  } \mspace{1mu}   \alpha    \to   \mspace{2mu} !_{  1  } \mspace{1mu}   \alpha     \: \times \:   \mspace{2mu} !_{  1  } \mspace{1mu}   \alpha   $ and $\forall \alpha .   \mspace{2mu} !_{  1  } \mspace{1mu}   \alpha    \to   \mathbf{Unit}  $ should be uninhabited.

The key feature of graded-context type systems is that they grade contexts of typing judgments with elements of the parametrizing structure. A typical typing judgment in a graded-context type system, parametrized by a preordered semiring, $ \mathcal{Q} $, looks like: 
\begin{equation}
  \ottmv{x_{{\mathrm{1}}}}  :^{ \ottnt{q_{{\mathrm{1}}}} }  \ottnt{A_{{\mathrm{1}}}}   ,   \ottmv{x_{{\mathrm{2}}}}  :^{ \ottnt{q_{{\mathrm{2}}}} }  \ottnt{A_{{\mathrm{2}}}}  , \ldots , x_n :^{q_n} A_n \vdash b : B,
\label{gctyp}
\end{equation}
where $q_i \in \mathcal{Q}$. Note that a graded context, $\Gamma$, can be decomposed into a ungraded context, $ \lfloor  \Gamma  \rfloor $, and a vector of grades, $ \overline{ \Gamma } $. The intuitive meaning of the graded context $x_1 :^{q_1} A_1 , x_2 :^{q_2} A_2 , \ldots, x_n :^{q_n} A_n$ is the same as that of the standard context $x_1 : \, !_{q_1} A_1 , x_2 : \, !_{q_2} A_2 , \ldots, x_n : \, !_{q_n} A_n$. However, there is a very good reason behind employing graded contexts in place of standard ones: context operations can be easily defined on graded contexts by lifting the corresponding operations of the preordered semiring to the level of contexts. For example, for graded contexts $\Gamma_{{\mathrm{1}}}$ and $\Gamma_{{\mathrm{2}}}$, given $  \lfloor  \Gamma_{{\mathrm{1}}}  \rfloor   =   \lfloor  \Gamma_{{\mathrm{2}}}  \rfloor  $, the context addition operation, $ \Gamma_{{\mathrm{1}}}  +  \Gamma_{{\mathrm{2}}} $, is defined by pointwise addition of grades and the context multiplication operation, $q \cdot \Gamma$, is defined by pre-multiplying every grade in $\Gamma$ by $\ottnt{q}$. These context operations facilitate smooth accounting of resource usage in graded-context type systems. The introduction and elimination rules for $!_q$, shown below, illustrates how.
\begin{center}
\drule[]{GC-ExpIntro} \hspace*{20pt} \drule[]{GC-ExpElim}
\end{center}  

Next, we consider the nature of the graded modality, $!_q$. This graded modality is graded comonadic over $(Q,\cdot,1)$ \citep{fujii,cokatsumata}, meaning, one can derive the standard extract and fork functions for this modality:
\begin{align*}
& \mathbf{extract}  :=  \lambda  \ottmv{x}  .   \mathbf{let} \: !_{  1  } \:  \ottmv{y}  \: \mathbf{be} \:   \ottmv{x}   \: \mathbf{in} \:   \ottmv{y}    \; : \:    \mspace{2mu} !_{  1  } \mspace{1mu}  \ottnt{A}    \to  \ottnt{A} \\
& \mathbf{fork}^{\ottnt{q_{{\mathrm{1}}}},\ottnt{q_{{\mathrm{2}}}}}  :=  \lambda  \ottmv{x}  .   \mathbf{let} \: !_{  \ottnt{q_{{\mathrm{1}}}}  \cdot  \ottnt{q_{{\mathrm{2}}}}  } \:  \ottmv{y}  \: \mathbf{be} \:   \ottmv{x}   \: \mathbf{in} \:   \mspace{2mu} !_{ \ottnt{q_{{\mathrm{1}}}} } \mspace{1mu}   \mspace{2mu} !_{ \ottnt{q_{{\mathrm{2}}}} } \mspace{1mu}   \ottmv{y}      \; : \:     \mspace{2mu} !_{  \ottnt{q_{{\mathrm{1}}}}  \cdot  \ottnt{q_{{\mathrm{2}}}}  } \mspace{1mu}  \ottnt{A}    \to    \mspace{2mu} !_{ \ottnt{q_{{\mathrm{1}}}} } \mspace{1mu}   \mspace{2mu} !_{ \ottnt{q_{{\mathrm{2}}}} } \mspace{1mu}  \ottnt{A}    
\end{align*}
Here, $\ottnt{A}$ is an arbitrary type and $\ottnt{q_{{\mathrm{1}}}}, \ottnt{q_{{\mathrm{2}}}} \in Q$. The comonadic nature of this modality is essential to supporting usage analysis. To give an example, for the promotion $   \mspace{2mu} !_{  \omega  } \mspace{1mu}  \ottnt{A}    \to   \mspace{2mu} !_{  \omega  } \mspace{1mu}   \mspace{2mu} !_{  \omega  } \mspace{1mu}  \ottnt{A}   $, one needs $\mathbf{fork}^{ \omega , \omega }$. Though the fork function is necessary for usage analysis, the same is not true of the corresponding join function. As such, graded-context type systems do not derive a join function. In fact, these type systems \textit{cannot} derive a join function in general, meaning, there is no general function of type $   \mspace{2mu} !_{ \ottnt{q_{{\mathrm{1}}}} } \mspace{1mu}   \mspace{2mu} !_{ \ottnt{q_{{\mathrm{2}}}} } \mspace{1mu}  \ottnt{A}     \to   \mspace{2mu} !_{  \ottnt{q_{{\mathrm{1}}}}  \cdot  \ottnt{q_{{\mathrm{2}}}}  } \mspace{1mu}  \ottnt{A}  $ derivable in these type systems. Proposition \ref{nojoin} in Appendix \ref{app2} presents a model-theoretic proof of this claim. This shortcoming of graded-context type systems, while inconsequential for usage analysis, is limiting for dependency analysis, as we see next.

\subsection{Limitations of Dependency Analysis in Graded-Context Type Systems} \label{limitdep}

Existing graded-context type systems are limited in their analysis of dependency. Below, we elaborate why.
\begin{enumerate}
\item Graded-context type systems are parametrized by preordered semirings. Dependency analysis, however, is parametrized by lattices. Both preordered semirings and lattices are algebraic structures with two binary operators and a binary order relation. However, a crucial distinction between these structures is that while semirings need to ensure that one operator (multiplication) distributes over the other (addition), lattices do not need to ensure any such property. In fact, there are lattices where neither operator distributes over the other. The simplest of such lattices are $M_3$ and $N_5$ \citep{birkhoff}, both with 5 elements, ordered as shown in Figure \ref{nondistr}. 
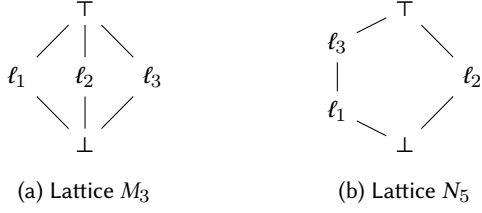
\begin{figure}
\centering
\begin{subfigure}[b]{0.3\textwidth}
\centering
\begin{tikzpicture}[scale=0.3]
\node at (0,0) (b) {$ \bot $};
\node at (3,3) (l3) {$\ell_{{\mathrm{3}}}$};
\node at (0,3) (l2) {$\ell_{{\mathrm{2}}}$};
\node at (-3,3) (l1) {$\ell_{{\mathrm{1}}}$};
\node at (0,6) (t) {$ \top $};
\draw (b) -- (l1);
\draw (b) -- (l2);
\draw (b) -- (l3);
\draw (l1) -- (t);
\draw (l2) -- (t);
\draw (l3) -- (t);
\end{tikzpicture}
\caption{Lattice $M_3$} 
\end{subfigure}
\begin{subfigure}[b]{0.3\textwidth}
\centering
\begin{tikzpicture}[scale=0.3]
\node at (0,0) (b) {$ \bot $};
\node at (3,3) (l2) {$\ell_{{\mathrm{2}}}$};
\node at (-3,1.5) (l1) {$\ell_{{\mathrm{1}}}$};
\node at (-3,4.5) (l3) {$\ell_{{\mathrm{3}}}$};
\node at (0,6) (t) {$ \top $};
\draw (b) -- (l1);
\draw (b) -- (l2);
\draw (l1) -- (l3);
\draw (l3) -- (t);
\draw (l2) -- (t);
\end{tikzpicture} 
\caption{Lattice $N_5$}
\end{subfigure}
\caption{Examples of non-distributive lattices}
\label{nondistr}
\end{figure} 
In $M_3$, join and meet do not distribute over one another: 
\begin{align*}
&   (   \ell_{{\mathrm{1}}}  \: \sqcup \:  \ell_{{\mathrm{2}}}   )   \: \sqcap \:   (   \ell_{{\mathrm{3}}}  \: \sqcup \:  \ell_{{\mathrm{2}}}   )   =  \top  \neq \ell_{{\mathrm{2}}} =   (   \ell_{{\mathrm{1}}}  \: \sqcap \:  \ell_{{\mathrm{3}}}   )   \: \sqcup \:  \ell_{{\mathrm{2}}}  \\
&   (   \ell_{{\mathrm{1}}}  \: \sqcap \:  \ell_{{\mathrm{2}}}   )   \: \sqcup \:   (   \ell_{{\mathrm{3}}}  \: \sqcap \:  \ell_{{\mathrm{2}}}   )   =  \bot  \neq \ell_{{\mathrm{2}}} =   (   \ell_{{\mathrm{1}}}  \: \sqcup \:  \ell_{{\mathrm{3}}}   )   \: \sqcap \:  \ell_{{\mathrm{2}}} .
\end{align*}
The same is true of $N_5$: 
\begin{align*}
&   (   \ell_{{\mathrm{2}}}  \: \sqcup \:  \ell_{{\mathrm{1}}}   )   \: \sqcap \:   (   \ell_{{\mathrm{3}}}  \: \sqcup \:  \ell_{{\mathrm{1}}}   )   = \ell_{{\mathrm{3}}} \neq \ell_{{\mathrm{1}}} =   (   \ell_{{\mathrm{2}}}  \: \sqcap \:  \ell_{{\mathrm{3}}}   )   \: \sqcup \:  \ell_{{\mathrm{1}}}  \\
&   (   \ell_{{\mathrm{1}}}  \: \sqcap \:  \ell_{{\mathrm{3}}}   )   \: \sqcup \:   (   \ell_{{\mathrm{2}}}  \: \sqcap \:  \ell_{{\mathrm{3}}}   )   = \ell_{{\mathrm{1}}} \neq \ell_{{\mathrm{3}}} =   (   \ell_{{\mathrm{1}}}  \: \sqcup \:  \ell_{{\mathrm{2}}}   )   \: \sqcap \:  \ell_{{\mathrm{3}}} .
\end{align*} Thus, an arbitrary lattice can not be viewed as a preordered semiring. Distributive lattices, i.e. lattices where join and meet distribute over one another, however, may be viewed as preordered semirings where multiplication, addition and order are given by join, meet and the lattice order respectively. So, existing graded-context type systems could \textit{potentially} carry out dependency analysis over distributive lattices. However, we view this as a limitation since we see no principled justification in restricting dependency analysis to distributive lattices only. 

\item Dependency analysis needs a join operator, as discussed in Section \ref{subsec:depsalient}. However, graded-context type systems cannot derive a join operator, as discussed in Section \ref{gradedsalient}. An ad hoc solution to this problem might be to add an explicit join operator to the type system, as follows:
\begin{center}
\drule[]{GC-Join}
\end{center}
%for example, an operator $\mathbf{join}$ of type $   \mspace{2mu} !_{ \ottnt{q_{{\mathrm{1}}}} } \mspace{1mu}   \mspace{2mu} !_{ \ottnt{q_{{\mathrm{2}}}} } \mspace{1mu}  \ottnt{A}     \to   \mspace{2mu} !_{  \ottnt{q_{{\mathrm{1}}}}  \cdot  \ottnt{q_{{\mathrm{2}}}}  } \mspace{1mu}  \ottnt{A}  $ for all $\ottnt{A}$ and $\ottnt{q_{{\mathrm{1}}}}, \ottnt{q_{{\mathrm{2}}}} \in Q$. 
However, such a solution is not very satisfactory owing to the following reasons:
\begin{itemize}
\item This $\mathbf{join}$ operator raises problems in operational semantics, especially with call-by-name reduction. To understand why, consider the question: how should $ \mathbf{join}^{ \ottnt{q_{{\mathrm{1}}}} , \ottnt{q_{{\mathrm{2}}}} } \:  \ottnt{b} $ reduce in a call-by-name calculus?  To reduce terms headed by $\mathbf{join}$, one might come up with the following rules:

\begin{center}
\drule[]{Step-JoinLeft} \hspace*{15pt} \drule[]{Step-JoinBeta}
\vspace*{2pt}
\end{center} 

Then, to reduce $ \mathbf{join}^{ \ottnt{q_{{\mathrm{1}}}} , \ottnt{q_{{\mathrm{2}}}} } \:  \ottnt{b} $, one would first reduce $\ottnt{b}$ to a value and thereafter attempt to apply \rref{Step-JoinBeta}. But notice that such an attempt might not be successful even for well-typed terms. This is so because in a call-by-name calculus, $ \mspace{2mu} !_{ \ottnt{q} } \mspace{1mu}  \ottnt{a'} $ is a value, irrespective of whether $\ottnt{a'}$ itself is a value or not. As such, the reduction of $ \mathbf{join}^{ \ottnt{q_{{\mathrm{1}}}} , \ottnt{q_{{\mathrm{2}}}} } \:  \ottnt{b} $ might just stop at a term like $ \mathbf{join}^{ \ottnt{q_{{\mathrm{1}}}} , \ottnt{q_{{\mathrm{2}}}} } \:   \mspace{2mu} !_{ \ottnt{q_{{\mathrm{1}}}} } \mspace{1mu}  \ottnt{b'}  $, where $\ottnt{b'}$ is not headed by $!_{q_2}$. In such a case, the term $ \mathbf{join}^{ \ottnt{q_{{\mathrm{1}}}} , \ottnt{q_{{\mathrm{2}}}} } \:   \mspace{2mu} !_{ \ottnt{q_{{\mathrm{1}}}} } \mspace{1mu}  \ottnt{b'}  $ cannot step via \rref{Step-JoinBeta}. Given this situation, to maintain type soundness, one would be forced to declare terms like $ \mathbf{join}^{ \ottnt{q_{{\mathrm{1}}}} , \ottnt{q_{{\mathrm{2}}}} } \:   \mspace{2mu} !_{ \ottnt{q_{{\mathrm{1}}}} } \mspace{1mu}  \ottnt{b'}  $ as values, thereby allowing unprincipled values within the calculus.    

%If $v$ is a value, should $ \mathbf{join} \:  \ottnt{v} $ be also a value? Then, we would have more values like $ \mathbf{join} \:   \mspace{2mu} !_{ \ottnt{q_{{\mathrm{1}}}} } \mspace{1mu}   \mspace{2mu} !_{ \ottnt{q_{{\mathrm{2}}}} } \mspace{1mu}   \mathbf{true}    $. Alternatively, if $ \mathbf{join} \:  \ottnt{v} $ is not a value, how should it step? One might answer $ \mathbf{join} \:   \mspace{2mu} !_{ \ottnt{q_{{\mathrm{1}}}} } \mspace{1mu}   \mspace{2mu} !_{ \ottnt{q_{{\mathrm{2}}}} } \mspace{1mu}  \ottnt{v}   $ should step to $ \mspace{2mu} !_{  \ottnt{q_{{\mathrm{1}}}}  \cdot  \ottnt{q_{{\mathrm{2}}}}  } \mspace{1mu}  \ottnt{v} $. While such a stepping rule would make sense in a call-by-value calculus, it would not in a call-by-name calculus because there $ \mspace{2mu} !_{ \ottnt{q} } \mspace{1mu}  \ottnt{a} $ is a value, irrespective of whether $\ottnt{a}$ itself is a value or not. 

\item Further, the addition of this $\mathbf{join}$ operator breaks the symmetry of the type system since \rref{GC-Join} is neither an introduction rule nor an elimination rule for $!$. Owing to this break in symmetry, there is no principled way to ensure that $\mathbf{join}$ is the inverse of $\mathbf{fork}$, which is already derivable in the type system, as shown in Section \ref{gradedsalient}. Thus, this $\mathbf{join}$ operator does not bode well for the equational theory of the type system as well.

\end{itemize}

In light of these drawbacks, we conclude that adding an ad hoc $\mathbf{join}$ operator to enable dependency analysis in a graded-context type system is not the right approach. So, we avoid this approach.

\item Next, as discussed in Section \ref{subsec:depsalient}, we want a dependency calculus where any function can return high-security values, whenever the function itself is wrapped under a high-security label. Existing graded-context type systems do not allow this. For example, in a graded-context type system parametrized by the lattice, $  \mathbf{L}   \sqsubseteq   \mathbf{H}  $, the type $ \mspace{2mu} !_{  \mathbf{H}  } \mspace{1mu}   (     \mspace{2mu} !_{  \mathbf{H}  } \mspace{1mu}   \mathbf{Bool}     \to   \mathbf{Bool}    )  $ contains only two distinct terms, $ \mspace{2mu} !_{  \mathbf{H}  } \mspace{1mu}   (   \lambda  \ottmv{x}  .   \mathbf{true}    )  $ and $ \mspace{2mu} !_{  \mathbf{H}  } \mspace{1mu}   (   \lambda  \ottmv{x}  .   \mathbf{false}    )  $. Put differently, in a graded-context type system, there are no terms corresponding to $ \mathbf{seal}_{  \mathbf{H}  } \:   (   \lambda  \ottmv{x}  .   \mathbf{unseal}_{  \mathbf{H}  } \:   \ottmv{x}     )  $ and $ \mathbf{seal}_{  \mathbf{H}  } \:   (   \lambda  \ottmv{x}  .   \mathbf{not} \:   (   \mathbf{unseal}_{  \mathbf{H}  } \:   \ottmv{x}    )     )  $ from the Sealing Calculus.
\end{enumerate} 

Hence, we see that there are several impediments to dependency analysis in existing graded-context type systems. This motivates us to look for other solutions for unifying linearity and dependency analyses.

\subsection{Towards Resolution} \label{subsec:res}

Recent work by \citet{ddc} points towards a possible way of unifying linearity and dependency analyses. \citet{ddc} present $\text{DDC}^{\top}$, a type system for general dependency analysis in Pure Type Systems. $\text{DDC}^{\top}$, though similar to existing graded-context type systems, avoids all their limitations listed above. \Dt{} is parametrized by an arbitrary lattice and subsumes the Sealing Calculus. Further, \Dt{} can analyze dependencies in a dependent setting, for example, run-time irrelevance in dependently-typed programs. The key difference between $\text{DDC}^{\top}$ and graded-context type systems is in the form of their typing judgments. $\text{DDC}^{\top}$, in addition to grading contexts of typing judgments, also puts labels on the judgments, as shown below: 
\begin{equation}
  \ottmv{x_{{\mathrm{1}}}}  :^{ \ell_{{\mathrm{1}}} }  \ottnt{A_{{\mathrm{1}}}}   ,   \ottmv{x_{{\mathrm{2}}}}  :^{ \ell_{{\mathrm{2}}} }  \ottnt{A_{{\mathrm{2}}}}  , \ldots , x_n :^{\ell_n} A_n \vdash b :^{\ell} B,
\label{ddctyp}
\end{equation}
where $\ell_i, \ell$ are elements of the parametrizing lattice. The label to the right of the turnstile in (\ref{ddctyp}) denotes the observer's level. In graded-context type systems, the observer's level is fixed at $1$. The added flexibility of varying the observer's level enables $\text{DDC}^{\top}$ carry out a general dependency analysis.

However, $\text{DDC}^{\top}$ cannot carry out linearity analysis. So the problem of unifying linearity and dependency analyses still remains open. Nevertheless, we can take inspiration from $\text{DDC}^{\top}$ and see what happens when we allow graded-context type systems to vary the observer's level via typing judgments of the form:
\begin{equation} \label{typform}
  \ottmv{x_{{\mathrm{1}}}}  :^{ \ottnt{q_{{\mathrm{1}}}} }  \ottnt{A_{{\mathrm{1}}}}   ,   \ottmv{x_{{\mathrm{2}}}}  :^{ \ottnt{q_{{\mathrm{2}}}} }  \ottnt{A_{{\mathrm{2}}}}  , \ldots , x_n :^{q_n} A_n \vdash b :^q B,
\end{equation} 
where $q_i, q \in \mathcal{Q}$. This typing judgment is essentially a fusion of the typing judgments of the Sealing Calculus and existing graded-context type systems, shown in (\ref{sctyp}) and (\ref{gctyp}) respectively. This seems to be a good start. However, a known roadblock awaits us along this way. \citet{atkey} showed that a type system that uses the above judgment form and is parametrized by an \textit{arbitrary} semiring does not admit substitution. This roadblock is discouraging but we found that it is not a dead end. We found that though substitution is inadmissible over some semirings, it is in fact admissible over several preordered semirings that interest us. In particular, substitution is admissible over the standard preordered semirings employed for tracking linearity and affinity, viz., $ \mathcal{Q}_{\text{Lin} } $ and $ \mathcal{Q}_{\text{Aff} } $ respectively. 

Thus, we finally have a way to unify linearity and dependency analyses. In our Linear Dependency Calculus, LDC, we use typing judgments of the form shown in (\ref{typform}). For linearity and other usage analyses, we parametrize LDC over certain preordered semirings that we describe as we go along. For dependency analysis, we parametrize LDC over an \textit{arbitrary} lattice. For a combined linearity and dependency analysis, we parametrize LDC over the cartesian product of the structures used for the individual analyses.

Now, LDC can analyze linearity and dependency in an arbitrary pure type system. However, some of the key ideas of the calculus are best explained in a simply-typed setting. So, we first present the simply-typed version of LDC and thereafter generalize it to its pure type system version.

\section{Linearity and Dependency Analyses in Simple Type Systems}

\subsection{Type System for Linearity Analysis} \label{subsec:linsimple}

First, we shall analyze exact usage and bounded usage. We shall add unrestricted usage, referred to by $ \omega $, to our calculus in Section \ref{sec:unrestricted}. Exact usage can be analyzed by $ \mathbb{N}_{=} $, the semiring of natural numbers with discrete order. Bounded usage can be analyzed by $ \mathbb{N}_{\geq} $, the semiring of natural numbers with descending natural order. The ordering in $ \mathbb{N}_{\geq} $ looks like: $\ldots <:   4   <:   3   <:   2   <:   1   <: 0$. The reason behind the difference in ordering is that in bounded usage analysis, resources may be discarded. But note that resources are never copied in either of these analyses. %Note that the semirings $ \mathbb{N}_{=} $ and $ \mathbb{N}_{\geq} $ satisfy a special property: they are left multiplicatively cancellable \cite{golan}, meaning if $q \neq 0$ and $  \ottnt{q}  \cdot  \ottnt{r_{{\mathrm{1}}}}   <:   \ottnt{q}  \cdot  \ottnt{r_{{\mathrm{2}}}}  $, then $ \ottnt{r_{{\mathrm{1}}}}  <:  \ottnt{r_{{\mathrm{2}}}} $. This property is useful in some of our proofs. 

Next, we present LDC parametrized over these two preordered semirings. Whenever we need precision, we refer to LDC parametrized over an algebraic structure, $\mathcal{AS}$, as LDC($\mathcal{AS}$). The algebraic structure, $\mathcal{AS}$, may be either a preordered semiring or a lattice or their cartesian product (in case of combined analysis).  

Now, let $ \mathcal{Q}_{\mathbb{N} } $ vary over $\{  \mathbb{N}_{=} ,  \mathbb{N}_{\geq} \}$. The grammar of LDC($ \mathcal{Q}_{\mathbb{N} } $) appears in Figure \ref{GrammarSimple}. The calculus has function types, product types, sum types and a $ \mathbf{Unit} $ type. Product types correspond to multiplicative conjunction of linear logic; sum type corresponds to additive disjunction. Observe that the function type and the product type are annotated with grades. Types $ {}^{ \ottnt{r} }\!  \ottnt{A}  \to  \ottnt{B} $ and $ {}^{ \ottnt{r} }\!  \ottnt{A_{{\mathrm{1}}}}  \: \times \:  \ottnt{A_{{\mathrm{2}}}} $ are essentially $  (   \mspace{2mu} !_{ \ottnt{r} } \mspace{1mu}  \ottnt{A}   )   \to  \ottnt{B} $ and $  (   \mspace{2mu} !_{ \ottnt{r} } \mspace{1mu}  \ottnt{A_{{\mathrm{1}}}}   )   \: \times \:  \ottnt{A_{{\mathrm{2}}}} $ respectively. In lieu of $ {}^{ \ottnt{r} }\!  \ottnt{A}  \to  \ottnt{B} $ and $ {}^{ \ottnt{r} }\!  \ottnt{A_{{\mathrm{1}}}}  \: \times \:  \ottnt{A_{{\mathrm{2}}}} $, we could have used unannotated function and product types, i.e., $ \ottnt{A}  \to  \ottnt{B} $ and $ \ottnt{A_{{\mathrm{1}}}}  \: \times \:  \ottnt{A_{{\mathrm{2}}}} $ respectively, along with a graded exponential modal type, $ \mspace{2mu} !_{ \ottnt{r} } \mspace{1mu}  \ottnt{A} $. However, we chose to present this way because it generalizes easily to the Pure Type System setting. For terms, we have the introduction and the elimination forms corresponding to these types. The terms are also annotated with grades that track resources used by them. Assumptions in the context also appear along with their allowed usages. 

\begin{figure}
\begin{align*}
& \text{types}, A, B, C && ::=  \mathbf{Unit}  \: | \:  {}^{ \ottnt{r} }\!  \ottnt{A}  \to  \ottnt{B}  \: | \:  {}^{ \ottnt{r} }\!  \ottnt{A_{{\mathrm{1}}}}  \: \times \:  \ottnt{A_{{\mathrm{2}}}}  \: | \:  \ottnt{A_{{\mathrm{1}}}}  +  \ottnt{A_{{\mathrm{2}}}}  \\
& \text{terms}, a, b, c && ::= x \: | \:  \lambda^{ \ottnt{r} }  \ottmv{x}  :  \ottnt{A}  .  \ottnt{b}  \: | \:  \ottnt{b}  \:  \ottnt{a} ^{ \ottnt{r} }  \: \\
                   & && | \:  \mathbf{unit}  \: | \:  \mathbf{let}_{ \ottnt{q_{{\mathrm{0}}}} } \: \mathbf{unit} \: \mathbf{be} \:  \ottnt{a}  \: \mathbf{in} \:  \ottnt{b}  \: | \:  (  \ottnt{a_{{\mathrm{1}}}} ^{ \ottnt{r} } ,  \ottnt{a_{{\mathrm{2}}}}  )  \: | \:                      \mathbf{let}_{ \ottnt{q_{{\mathrm{0}}}} } \: (  \ottmv{x} ^{ \ottnt{r} } ,  \ottmv{y}  ) \: \mathbf{be} \:  \ottnt{a}  \: \mathbf{in} \:  \ottnt{b}  \\
                   & && | \:  \mathbf{inj}_1 \:  \ottnt{a_{{\mathrm{1}}}}  \: | \:  \mathbf{inj}_2 \:  \ottnt{a_{{\mathrm{2}}}}  \: | \:  \mathbf{case}_{ \ottnt{q_{{\mathrm{0}}}} } \:  \ottnt{a}  \: \mathbf{of} \:  \ottmv{x_{{\mathrm{1}}}}  .  \ottnt{b_{{\mathrm{1}}}}  \: ; \:  \ottmv{x_{{\mathrm{2}}}}  .  \ottnt{b_{{\mathrm{2}}}}  \\ 
& \text{contexts}, \Gamma && ::=  \emptyset  \: | \:  \Gamma  ,   \ottmv{x}  :^{ \ottnt{q} }  \ottnt{A}  
\end{align*}
\caption{Grammar of LDC($ \mathcal{Q}_{\mathbb{N} } $) (Simple Version)}
\label{GrammarSimple}
\end{figure}  

Next, we look at the type system. The type system appears in Figure \ref{TypeSystemSimple}. There are a few things to note:
\begin{itemize}
\item A graded context $\Gamma$ is essentially a combination of the underlying ungraded context, denoted $\Delta :=  \lfloor  \Gamma  \rfloor $, and the associated vector of grades, denoted $ \overline{ \Gamma } $.
\item For contexts $\Gamma_{{\mathrm{1}}}$ and $\Gamma_{{\mathrm{2}}}$ such that $  \lfloor  \Gamma_{{\mathrm{1}}}  \rfloor   =   \lfloor  \Gamma_{{\mathrm{2}}}  \rfloor   = \Delta$, we define $\Gamma :=  \Gamma_{{\mathrm{1}}}  +  \Gamma_{{\mathrm{2}}} $ as: $  \lfloor  \Gamma  \rfloor   =  \Delta $ and $ \overline{ \Gamma }  =  \overline{ \Gamma_{{\mathrm{1}}} }  +  \overline{ \Gamma_{{\mathrm{2}}} } $ (pointwise vector addition).
\item For context $\Gamma_{{\mathrm{0}}}$ and grade $\ottnt{q}$, we define $\Gamma :=  \ottnt{q}  \cdot  \Gamma_{{\mathrm{0}}} $ as: $ \lfloor  \Gamma  \rfloor  =  \lfloor  \Gamma_{{\mathrm{0}}}  \rfloor $ and $ \overline{ \Gamma }  = \ottnt{q} \cdot  \overline{ \Gamma_{{\mathrm{0}}} } $ (scalar multiplication).
\item For contexts $\Gamma$ and $\Gamma'$ such that $  \lfloor  \Gamma  \rfloor   =   \lfloor  \Gamma'  \rfloor  $, we say $ \Gamma  <:  \Gamma' $ if and only if $ \overline{ \Gamma }   <:   \overline{ \Gamma' } $ (pointwise order).
\item For any context $\Gamma$, we implicitly assume that no two assumptions assign types to the same variable.
\end{itemize}

\begin{figure}
\drules[ST]{$ \Gamma  \vdash  \ottnt{a}  :^{ \ottnt{q} }  \ottnt{A} $}{Simple Version}{Var,Lam,App,Pair,LetPair,Unit,LetUnit,InjOne,InjTwo,Case,SubL,SubR}
\caption{Typing rules for LDC($ \mathcal{Q}_{\mathbb{N} } $)}
\label{TypeSystemSimple}
\end{figure} 

The typing judgment $  \ottmv{x_{{\mathrm{1}}}}  :^{ \ottnt{q_{{\mathrm{1}}}} }  \ottnt{A_{{\mathrm{1}}}}   ,   \ottmv{x_{{\mathrm{2}}}}  :^{ \ottnt{q_{{\mathrm{2}}}} }  \ottnt{A_{{\mathrm{2}}}}  , \ldots, x_n :^{q_n} A_n \vdash a :^q A$ may be intuitively understood as: one can produce $\ottnt{q}$ copies of $\ottnt{a}$ of type $\ottnt{A}$, using $q_i$ copies of $x_i$ of type $A_i$, where $i = 1, 2, \ldots, n$. With this understanding, let us look at the typing rules.

Most of the rules are as expected. A point to note is that the elimination \rref{ST-LetPair,ST-LetUnit,ST-Case} have a side condition $ \ottnt{q_{{\mathrm{0}}}}  <:   1  $. The reason behind this condition is that for reducing any of these elimination forms, we first need to reduce the term $\ottnt{a}$, implying that in any case, we would need the resources for reducing at least one copy of $\ottnt{a}$. We may set $\ottnt{q_{{\mathrm{0}}}}$ to $1$ but allowing any $ \ottnt{q_{{\mathrm{0}}}}  <:   1  $  makes these rules more flexible. For example, owing to this flexibility, in \rref{ST-LetPair}, $\ottnt{b}$ may use $ \ottnt{q}  \cdot  \ottnt{q_{{\mathrm{0}}}} $ copies of $\ottmv{y}$ in lieu of just $\ottnt{q}$ copies of $\ottmv{y}$. Another point to note is how the \rref{ST-SubL,ST-SubR} help discard resources in LDC($ \mathbb{N}_{\geq} $). %The \rref{ST-SubL} discards away some resources and uses the rest to produce $\ottnt{q}$ copies of $\ottnt{a}$. The \rref{ST-SubR} produces $\ottnt{q}$ copies of $\ottnt{a}$ and then discards some of them, retaining only $\ottnt{q'}$ copies.  

Now, let us look at a few examples of derivable and non-derivable terms in LDC($ \mathbb{N}_{\geq} $). (For the sake of simplicity, we omit the domain types in the $\lambda$s.)
\begin{align*}
&   \emptyset   \vdash   \lambda^{  1  }  \ottmv{x}  .   \ottmv{x}    :^{  1  }   {}^{  1  }\!  \ottnt{A}  \to  \ottnt{A}   \text{ but }  \emptyset   \nvdash   \lambda^{  0  }  \ottmv{x}  .   \ottmv{x}    :^{  1  }   {}^{  0  }\!  \ottnt{A}  \to  \ottnt{A}   \\
&   \emptyset   \vdash   \lambda^{  1  }  \ottmv{x}  .   \ottmv{x}    :^{  2  }   {}^{  1  }\!  \ottnt{A}  \to  \ottnt{A}   \text{ but }   \emptyset   \nvdash   \lambda^{  1  }  \ottmv{x}  .   (   \ottmv{x}  ^{  2  } ,   \mathbf{unit}   )    :^{  1  }   {}^{  1  }\!  \ottnt{A}  \to   {}^{  2  }\!  \ottnt{A}  \: \times \:   \mathbf{Unit}    \\
&   \emptyset   \vdash   \lambda^{  1  }  \ottmv{x}  .   \mathbf{let}_{  1  } \: (  \ottmv{x_{{\mathrm{1}}}} ^{  1  } ,  \ottmv{x_{{\mathrm{2}}}}  ) \: \mathbf{be} \:   \ottmv{x}   \: \mathbf{in} \:   \ottmv{x_{{\mathrm{1}}}}     :^{  1  }    {}^{  1  }\!   (   {}^{  1  }\!  \ottnt{A}  \: \times \:  \ottnt{A}   )   \to  \ottnt{A}    \text{ but }   \emptyset   \nvdash   \lambda^{  1  }  \ottmv{x}  .   (   \ottmv{x}  ^{  1  } ,   \ottmv{x}   )    :^{  1  }   {}^{  1  }\!  \ottnt{A}  \to   {}^{  1  }\!  \ottnt{A}  \: \times \:  \ottnt{A}   
\end{align*}

On a closer look, we find that the non-derivable terms above can not use resources fairly. Consider the types of these terms: $ {}^{  0  }\!  \ottnt{A}  \to  \ottnt{A} $, $ {}^{  1  }\!  \ottnt{A}  \to   {}^{  2  }\!  \ottnt{A}  \: \times \:   \mathbf{Unit}   $ and $ {}^{  1  }\!  \ottnt{A}  \to   {}^{  1  }\!  \ottnt{A}  \: \times \:  \ottnt{A}  $. Such types may be inhabited only if resources can be copied. Since we disallow copying, they are essentially uninhabited. However, since LDC($ \mathbb{N}_{\geq} $) allows discarding of resources, the type $ {}^{  1  }\!   (   {}^{  1  }\!  \ottnt{A}  \: \times \:  \ottnt{A}   )   \to  \ottnt{A} $ is, in fact, inhabited. 

Note that in order to produce $0$ copies of any term, we do not need any resources. So in the $0$ world, any annotated type is inhabited, provided its unannotated counterpart is inhabited. However, resources do not have any meaning in the $0$ world. In other words, the judgment $ \Gamma  \vdash  \ottnt{a}  :^{  0  }  \ottnt{A} $ conveys no more information than its corresponding standard $\lambda$-calculus counterpart. 

Next, we look at the operational semantics and metatheory of LDC($ \mathcal{Q}_{\mathbb{N} } $).

\subsection{Metatheory of Linearity Analysis} \label{MetaSimpleBounded}

First, we consider some syntactic properties. The calculus satisfies the multiplication lemma stated below. This lemma says that if we need $\Gamma$ resources to produce $\ottnt{q}$ copies of $\ottnt{a}$, then we would need $ \ottnt{r_{{\mathrm{0}}}}  \cdot  \Gamma $ resources to produce $ \ottnt{r_{{\mathrm{0}}}}  \cdot  \ottnt{q} $ copies of $\ottnt{a}$.
\begin{lemma}[Multiplication]\label{BLSMult}
If $ \Gamma  \vdash  \ottnt{a}  :^{ \ottnt{q} }  \ottnt{A} $, then $  \ottnt{r_{{\mathrm{0}}}}  \cdot  \Gamma   \vdash  \ottnt{a}  :^{  \ottnt{r_{{\mathrm{0}}}}  \cdot  \ottnt{q}  }  \ottnt{A} $.
\end{lemma}

Note that we don't provide proofs of the lemmas and the theorems in the main body of the paper but the proofs are available in the appendices.
 
%As a corollary of the above lemma, if $ \Gamma  \vdash  \ottnt{a}  :^{ \ottnt{q} }  \ottnt{A} $, then $   0   \cdot  \Gamma   \vdash  \ottnt{a}  :^{  0  }  \ottnt{A} $. Now, for any $\Gamma_{{\mathrm{0}}}$ satisfying $  \lfloor  \Gamma_{{\mathrm{0}}}  \rfloor   =   \lfloor  \Gamma  \rfloor  $, we have, $ \Gamma_{{\mathrm{0}}}  <:    0   \cdot  \Gamma  $, and as such, $ \Gamma_{{\mathrm{0}}}  \vdash  \ottnt{a}  :^{  0  }  \ottnt{A} $ by \rref{ST-SubL}. This shows why resources don't have any meaning in the $0$ world.
 
The calculus also satisfies a factorization lemma, stated below. This lemma says that if a context $\Gamma$ produces $\ottnt{q}$ copies of $\ottnt{a}$, then $\Gamma$ can be split into $\ottnt{q}$ parts whereby each part produces $1$ copy of $\ottnt{a}$. We need the precondition, $q \neq 0$, since resources don't have any meaning in the $0$ world but they are meaningful in all other worlds, in particular, the $1$ world.
\begin{lemma}[Factorization]\label{BLSFact}
If $ \Gamma  \vdash  \ottnt{a}  :^{ \ottnt{q} }  \ottnt{A} $ and $q \neq 0$, then there exists $\Gamma'$ such that $ \Gamma'  \vdash  \ottnt{a}  :^{  1  }  \ottnt{A} $ and $ \Gamma  <:   \ottnt{q}  \cdot  \Gamma'  $.
\end{lemma} 

Using the above two lemmas, we can prove a splitting lemma, stated below. This lemma says that if we have the resources, $\Gamma$, to produce $ \ottnt{q_{{\mathrm{1}}}}  +  \ottnt{q_{{\mathrm{2}}}} $ copies of $\ottnt{a}$, then we can split $\Gamma$ into two parts, $\Gamma_{{\mathrm{1}}}$ and $\Gamma_{{\mathrm{2}}}$, such that $\Gamma_{{\mathrm{1}}}$ and $\Gamma_{{\mathrm{2}}}$ can produce $\ottnt{q_{{\mathrm{1}}}}$ and $\ottnt{q_{{\mathrm{2}}}}$ copies of $\ottnt{a}$ respectively. \citet{atkey} showed that the splitting lemma does not hold in a type system that is similar to ours, when parametrized over certain semirings. However, the splitting lemma holds for the preordered semirings we use for parametrizing LDC, i.e., $ \mathbb{N}_{=} $ and $ \mathbb{N}_{\geq} $.
  
\begin{lemma}[Splitting] \label{BLSSplit}
If $ \Gamma  \vdash  \ottnt{a}  :^{  \ottnt{q_{{\mathrm{1}}}}  +  \ottnt{q_{{\mathrm{2}}}}  }  \ottnt{A} $, then there exists $\Gamma_{{\mathrm{1}}}$ and $\Gamma_{{\mathrm{2}}}$ such that $ \Gamma_{{\mathrm{1}}}  \vdash  \ottnt{a}  :^{ \ottnt{q_{{\mathrm{1}}}} }  \ottnt{A} $ and $ \Gamma_{{\mathrm{2}}}  \vdash  \ottnt{a}  :^{ \ottnt{q_{{\mathrm{2}}}} }  \ottnt{A} $ and $ \Gamma  =   \Gamma_{{\mathrm{1}}}  +  \Gamma_{{\mathrm{2}}}  $. 
\end{lemma}

Next, we look at weakening and substitution. For weakening, we add an extra assumption, whose allowed usage is $0$. For substitution, we need to ensure that the allowed usage of the variable matches the number of available copies of the substitute. Note that after substitution, the term needs the combined resources. %This lemma says that we can always add, to the context, a new assumption with an arbitrary allowed usage. The resources corresponding to such an assumption are discarded during the derivation of the term.  
\begin{lemma}[Weakening] \label{BLSWeak}
If $  \Gamma_{{\mathrm{1}}}  ,  \Gamma_{{\mathrm{2}}}   \vdash  \ottnt{a}  :^{ \ottnt{q} }  \ottnt{A} $, then $    \Gamma_{{\mathrm{1}}}  ,   \ottmv{z}  :^{  0  }  \ottnt{C}     ,  \Gamma_{{\mathrm{2}}}   \vdash  \ottnt{a}  :^{ \ottnt{q} }  \ottnt{A} $.
\end{lemma}
%The substitution lemma for the system appears below. 
\begin{lemma}[Substitution] \label{BLSSubst}
If $    \Gamma_{{\mathrm{1}}}  ,   \ottmv{z}  :^{ \ottnt{r_{{\mathrm{0}}}} }  \ottnt{C}     ,  \Gamma_{{\mathrm{2}}}   \vdash  \ottnt{a}  :^{ \ottnt{q} }  \ottnt{A} $ and $ \Gamma  \vdash  \ottnt{c}  :^{ \ottnt{r_{{\mathrm{0}}}} }  \ottnt{C} $ and $  \lfloor  \Gamma_{{\mathrm{1}}}  \rfloor   =   \lfloor  \Gamma  \rfloor  $, then $    \Gamma_{{\mathrm{1}}}  +  \Gamma    ,  \Gamma_{{\mathrm{2}}}   \vdash   \ottnt{a}  \{  \ottnt{c}  /  \ottmv{z}  \}   :^{ \ottnt{q} }  \ottnt{A} $. 
\end{lemma}

Now, we consider the operational semantics of the language. Our operational semantics is a call-by-name small-step semantics. All the step rules are standard other than the $\beta$-rules that appear in Figure \ref{SmallStep}. These rules ensure that the grade in the introduction form matches with that in the elimination form.% These rules require that the quantity in the introduction form matches with that in the elimination form. This makes sense because: a function that uses its argument at most $\ottnt{r}$ times would require at most $\ottnt{r}$ copies of its argument; similarly, a term that uses at most $\ottnt{r}$ copies of a pattern variable would require at most $\ottnt{r}$ copies of its substitute.
\begin{figure}
\drules[Step]{$ \vdash  \ottnt{a}  \leadsto  \ottnt{a'} $}{Excerpt}{AppBeta,LetPairBeta}
\caption{Small-step reduction for LDC($ \mathcal{Q}_{\mathbb{N} } $)}
\label{SmallStep}
\end{figure}
With this operational semantics, our language enjoys the standard type soundness property.
\begin{theorem}[Preservation] \label{SimplePreserve}
If $ \Gamma  \vdash  \ottnt{a}  :^{ \ottnt{q} }  \ottnt{A} $ and $ \vdash  \ottnt{a}  \leadsto  \ottnt{a'} $, then $ \Gamma  \vdash  \ottnt{a'}  :^{ \ottnt{q} }  \ottnt{A} $.
\end{theorem}
\begin{theorem}[Progress] \label{BLSProg}
If $  \emptyset   \vdash  \ottnt{a}  :^{ \ottnt{q} }  \ottnt{A} $, then either $\ottnt{a}$ is a value or there exists $\ottnt{a'}$ such that $ \vdash  \ottnt{a}  \leadsto  \ottnt{a'} $.
\end{theorem}
%Note that in Theorem \ref{SimplePreserve}, $\ottnt{q}$ copies of $\ottnt{a}$ reduce to $\ottnt{q}$ copies of $\ottnt{a'}$. Further, both $\ottnt{a}$ and $\ottnt{a'}$ are type-checked in the same context $\Gamma$. At first sight, it appears that the reduct $\ottnt{a'}$ should require less resources. However, observe that as $\ottnt{a}$ reduces to $\ottnt{a'}$, the resource requirements of the term do not change. 

Now, consider the fact that in standard small-step semantics, we don't substitute the free variables of the redex. So such a semantics cannot really model resource usage of programs. Environment-based semantics, where free variables of terms get substituted with values from the environment, are more amenable to modeling resource usage by programs. In Section \ref{heapsimple}, we present an environment-based semantics for this calculus and show that the type system accounts usage correctly. But for now, we move on to dependency analysis.

\subsection{Type System for Dependency Analysis}  

Let $ \mathcal{L}  = (L,\sqcap,\sqcup, \top, \bot, \sqsubseteq )$ be an arbitrary lattice. (Technically, a lattice does not necessarily have top and bottom elements, but we can always add them.) We use $\ell, \ottnt{m}$ to denote elements of $L$. Now, interpreting $+,\cdot,0,1$ and $ <: $ as $\sqcap,\sqcup,\top,\bot$ and $ \sqsubseteq $ respectively, and using $q$ and $r$ for the elements of $ \mathcal{L} $, we have a dependency calculus in the type system presented in Figure \ref{TypeSystemSimple}. In Figure \ref{TypeSystemSimpleD}, we present a few selected rules from this type system with just changed notation.
\begin{figure}[h]
\drules[ST]{$ \Gamma  \vdash  \ottnt{a}  :^{ \ell }  \ottnt{A} $}{Selected rules}{VarD,LamD,AppD,PairD,LetPairD,SubLD,SubRD}
\caption{Typing rules of LDC($ \mathcal{L} $)}
\label{TypeSystemSimpleD}
\end{figure} 

The type system is now parametrized by $ \mathcal{L} $ in lieu of $ \mathcal{Q}_{\mathbb{N} } $. We define $ \ell  \sqcup  \Gamma $, $ \Gamma_{{\mathrm{1}}}  \sqcap  \Gamma_{{\mathrm{2}}} $ and $ \Gamma   \sqsubseteq   \Gamma' $ in the same way as their counterparts $ \ottnt{q}  \cdot  \Gamma $, $ \Gamma_{{\mathrm{1}}}  +  \Gamma_{{\mathrm{2}}} $ and $ \Gamma  <:  \Gamma' $ respectively. The typing judgment $  \ottmv{x_{{\mathrm{1}}}}  :^{ \ell_{{\mathrm{1}}} }  \ottnt{A_{{\mathrm{1}}}}   ,   \ottmv{x_{{\mathrm{2}}}}  :^{ \ell_{{\mathrm{2}}} }  \ottnt{A_{{\mathrm{2}}}}   , \ldots , x_n :^{\ell_n} A_n \vdash b :^{\ell} B$ may be read as: $\ottnt{b}$ is observable at $\ell$, assuming $x_i$ is observable at $\ell_i$, for $i = 1,2,\ldots,n$. With this reading, let us consider some of the typing rules. 

The \rref{ST-VarD} is as expected. The \rref{ST-LamD} is interesting. The type $ {}^{ \ottnt{m} }\!  \ottnt{A}  \to  \ottnt{B} $ contains functions that take arguments which are $\ottnt{m}$-secure, i.e. observable only at $\ottnt{m}$ and higher levels. The body of such a function, therefore, needs to be checked setting the observability level of the argument at least as high as $\ottnt{m}$. 

Conversely in \rref{ST-AppD}, the argument $\ottnt{a}$ needs to be observable at $\ottnt{m}$ or higher levels only. Note that the context in the conclusion judgment is formed by taking pointwise meet of the contexts checking the function and the argument. This ensures that no subterm that is observable in the premises becomes unobservable in the conclusion. 

The \rref{ST-WPairD} shows how to embox a secure term in a potentially insecure world. If $\ottnt{a_{{\mathrm{1}}}}$ is $\ottnt{m}$-secure, then we can release it at level $\ell$, but only after putting it in a box that may be opened at $\ottnt{m}$ and higher levels. This is similar to how the modal type $ T_{ \ottnt{m} } \:  \ottnt{A} $ of Sealing Calculus, essentially $ {}^{ \ottnt{m} }\!  \ottnt{A}  \: \times \:   \mathbf{Unit}  $, protects information, as we saw in \rref{SC-Seal}.

Conversely, the \rref{ST-LetPairD} ensures that an $\ottnt{m}$-secure box may be opened only at $\ottnt{m}$ and higher levels. Note that in this rule, the side-condition $ \ell_{{\mathrm{0}}}  \sqsubseteq   \bot  $ forces $\ell_{{\mathrm{0}}}$ to be $\bot$ because $ \bot $ is the bottom element. So \rref{ST-LetPairD} may be simplified with $\ell_{{\mathrm{0}}}$ set to $\bot$. However, we present the rule in this way to emphasize the similarity between LDC($ \mathcal{Q}_{\mathbb{N} } $) and LDC($ \mathcal{L} $).

%The \rref{ST-SubLD} says that if $\ottnt{a}$ is visible at $\ell$ under the constraints of $\Gamma'$, then $\ottnt{a}$ remains visible at $\ell$ under the relaxed constraints of $\Gamma$. The \rref{ST-SubRD} says that if $\ottnt{a}$ is visible at $\ell$, then $\ottnt{a}$ remains visible at all levels that are at least as secure as $\ell$.
%So, the \rref{ST-LetPairD} may be simplified to:
%\begin{center}
%$\ottdruleSTXXLetPairDTwo{}$
%\end{center}
%In \rref{ST-LetPair}, the quantity $\ottnt{q_{{\mathrm{0}}}}$ increases flexibility; however, we don't need this flexibility here because $\_\sqcup\_$ is idempotent. 
 
Next, we consider a few examples of derivable and non-derivable terms in LDC($ \mathcal{L}_{\diamond} $), where $ \mathcal{L}_{\diamond} $ is the diamond lattice (see Section \ref{subsec:depsalient}). 

Let $ \mathbf{Bool}  :=   \mathbf{Unit}   +   \mathbf{Unit}  $ and let $ \mathbf{true}  :=  \mathbf{inj}_1 \:   \mathbf{unit}  $ and $ \mathbf{false}  :=  \mathbf{inj}_2 \:   \mathbf{unit}  $. Then, %\\ and $  \mathbf{if} \:  \ottnt{b}  \: \mathbf{then} \:   \ottnt{c_{{\mathrm{1}}}}  \{   \mathbf{unit}   /  \ottmv{x_{{\mathrm{1}}}}  \}   \: \mathbf{else} \:  \ottnt{c_{{\mathrm{2}}}}   \{   \mathbf{unit}   /  \ottmv{x_{{\mathrm{2}}}}  \}  :=  \mathbf{case} \:  \ottnt{b}  \: \mathbf{of} \:  \ottmv{x_{{\mathrm{1}}}}  .  \ottnt{c_{{\mathrm{1}}}}  \: ; \:  \ottmv{x_{{\mathrm{2}}}}  .  \ottnt{c_{{\mathrm{2}}}} $. Then,
\begin{align*}
&   \emptyset   \vdash   \lambda^{  \mathbf{L}  }  \ottmv{x}  .   \ottmv{x}    :^{  \mathbf{H}  }   {}^{  \mathbf{L}  }   \mathbf{Bool}   \to   \mathbf{Bool}    \text{ but }   \emptyset   \nvdash   \lambda^{  \mathbf{H}  }  \ottmv{x}  .   \ottmv{x}    :^{  \mathbf{L}  }   {}^{  \mathbf{H}  }   \mathbf{Bool}   \to   \mathbf{Bool}    \\
&   \emptyset   \vdash   \lambda^{  \mathbf{H}  }  \ottmv{x}  .   \eta_{  \mathbf{M_1}  } \:   \eta_{  \mathbf{M_2}  } \:   \ottmv{x}      :^{  \mathbf{L}  }   {}^{  \mathbf{H}  }   \mathbf{Bool}   \to   T_{  \mathbf{M_1}  } \:   T_{  \mathbf{M_2}  } \:   \mathbf{Bool}      \text{ but }   \ottmv{x}  :^{  \mathbf{H}  }   \mathbf{Bool}    \nvdash   \ottmv{x}   :^{  \mathbf{M_1}  }   \mathbf{Bool}   
\end{align*}

Here, $ T_{ \ottnt{m} } \:  \ottnt{A}  :=  {}^{ \ottnt{m} }\!  \ottnt{A}  \: \times \:   \mathbf{Unit}  $ and $ \eta_{ \ell } \:  \ottnt{a}  :=  (  \ottnt{a} ^{ \ell } ,   \mathbf{unit}   ) $. On a closer look, we find that the non-derivable terms violate the dependency constraints modeled by $ \mathcal{L}_{\diamond} $. The first term transfers information from $ \mathbf{H} $ to $ \mathbf{L} $ while the second one does so from $ \mathbf{H} $ to $ \mathbf{M_1} $. The derivable terms, on the other hand, respect the dependency constraints modeled by $ \mathcal{L}_{\diamond} $. The first term transfers information from $ \mathbf{L} $ to $ \mathbf{H} $ while the second one emboxes $ \mathbf{H} $ information in an $ \mathbf{M_2} $ box nested within an $ \mathbf{M_1} $ box. 

Below, we present the terms that witness the standard join and fork operations in LDC($ \mathcal{L} $). (Note that any erased annotation is assumed to be $ \bot $.) Recall that existing graded-context type systems cannot derive such a join operator. LDC($ \mathcal{L} $), though similar to these type systems, can do so owing to its flexible typing judgment, which allows one to vary the observer's level. % $  T_{ \ell_{{\mathrm{1}}} } \:   T_{ \ell_{{\mathrm{2}}} } \:  \ottnt{A}    \Rightarrow   T_{  \ell_{{\mathrm{1}}}  \: \sqcup \:  \ell_{{\mathrm{2}}}  } \:  \ottnt{A}  $ and $  T_{  \ell_{{\mathrm{1}}}  \: \sqcup \:  \ell_{{\mathrm{2}}}  } \:  \ottnt{A}   \Rightarrow   T_{ \ell_{{\mathrm{1}}} } \:   T_{ \ell_{{\mathrm{2}}} } \:  \ottnt{A}   $ for any $\ell_{{\mathrm{1}}}, \ell_{{\mathrm{2}}} \in  \mathcal{L} $. As we pointed out earlier, these implications correspond to effect and coeffect aspects of dependency analysis respectively. We have,
\begin{align*}
&   \emptyset   \vdash  \ottnt{c_{{\mathrm{1}}}}  : \:     T_{ \ell_{{\mathrm{1}}} } \:   T_{ \ell_{{\mathrm{2}}} } \:  \ottnt{A}     \to   T_{  \ell_{{\mathrm{1}}}  \: \sqcup \:  \ell_{{\mathrm{2}}}  } \:  \ottnt{A}    \text{ and }   \emptyset   \vdash  \ottnt{c_{{\mathrm{2}}}}  : \:     T_{  \ell_{{\mathrm{1}}}  \: \sqcup \:  \ell_{{\mathrm{2}}}  } \:  \ottnt{A}    \to   T_{ \ell_{{\mathrm{1}}} } \:   T_{ \ell_{{\mathrm{2}}} } \:  \ottnt{A}     \text{ where,}\\
& \ottnt{c_{{\mathrm{1}}}} :=  \lambda  \ottmv{x}  .   \eta_{  \ell_{{\mathrm{1}}}  \: \sqcup \:  \ell_{{\mathrm{2}}}  } \:   (   \mathbf{let} \: (  \ottmv{y} ^{ \ell_{{\mathrm{2}}} } , \_  ) \: \mathbf{be} \:   (   \mathbf{let} \: (  \ottmv{z} ^{ \ell_{{\mathrm{1}}} } , \_  ) \: \mathbf{be} \:   \ottmv{x}   \: \mathbf{in} \:   \ottmv{z}    )   \: \mathbf{in} \:   \ottmv{y}    )    \\
& \ottnt{c_{{\mathrm{2}}}} :=  \lambda  \ottmv{x}  .   \eta_{ \ell_{{\mathrm{1}}} } \:   \eta_{ \ell_{{\mathrm{2}}} } \:   (   \mathbf{let} \: (  \ottmv{y} ^{  \ell_{{\mathrm{1}}}  \: \sqcup \:  \ell_{{\mathrm{2}}}  } , \_  ) \: \mathbf{be} \:   \ottmv{x}   \: \mathbf{in} \:   \ottmv{y}    )    
\end{align*}
The derivations of these terms appear in Appendix \ref{ap:der}.

Next, we look at the metatheory of LDC($ \mathcal{L} $).

\subsection{Metatheory of Dependency Analysis} \label{MetaSimpleDependency}

We consider the dependency counterparts of the properties of LDC($ \mathcal{Q}_{\mathbb{N} } $), presented in Section \ref{MetaSimpleBounded}. Most of these properties are true of LDC($ \mathcal{L} $). LDC($ \mathcal{L} $) satisfies the multiplication lemma. The lemma says that we can always simultaneously upgrade the context and the level at which the derived term is observed. 
\begin{lemma}[Multiplication] \label{DSMult}
If $ \Gamma  \vdash  \ottnt{a}  :^{ \ell }  \ottnt{A} $, then $  \ottnt{m_{{\mathrm{0}}}}  \sqcup  \Gamma   \vdash  \ottnt{a}  :^{  \ottnt{m_{{\mathrm{0}}}}  \: \sqcup \:  \ell  }  \ottnt{A} $.
\end{lemma}
The splitting lemma is also true. Since $\sqcap$ is idempotent, it follows directly from \rref{ST-SubRD}.
\begin{lemma}[Splitting] \label{DSSplit}
If $ \Gamma  \vdash  \ottnt{a}  :^{  \ell_{{\mathrm{1}}}  \: \sqcap \:  \ell_{{\mathrm{2}}}  }  \ottnt{A} $, then there exists $\Gamma_{{\mathrm{1}}}$ and $\Gamma_{{\mathrm{2}}}$ such that $ \Gamma_{{\mathrm{1}}}  \vdash  \ottnt{a}  :^{ \ell_{{\mathrm{1}}} }  \ottnt{A} $ and $ \Gamma_{{\mathrm{2}}}  \vdash  \ottnt{a}  :^{ \ell_{{\mathrm{2}}} }  \ottnt{A} $ and $ \Gamma  =   \Gamma_{{\mathrm{1}}}  \sqcap  \Gamma_{{\mathrm{2}}}  $.
\end{lemma}
The factorization lemma, however, is not true here because if true, it would allow secure information to leak. To see how, consider the following LDC($ \mathcal{L}_{\diamond} $) judgment: $  \emptyset   \vdash   \lambda  \ottmv{x}  .   \mathbf{let} \: (  \ottmv{y} ^{  \mathbf{H}  } , \_  ) \: \mathbf{be} \:   \ottmv{x}   \: \mathbf{in} \:   \ottmv{y}     :^{  \mathbf{H}  }    T_{  \mathbf{H}  } \:   \mathbf{Bool}    \to   \mathbf{Bool}   $. If the factorization lemma were true, from this judgment, we would have: $  \emptyset   \vdash   \lambda  \ottmv{x}  .   \mathbf{let} \: (  \ottmv{y} ^{  \mathbf{H}  } , \_  ) \: \mathbf{be} \:   \ottmv{x}   \: \mathbf{in} \:   \ottmv{y}     :^{  \mathbf{L}  }    T_{  \mathbf{H}  } \:   \mathbf{Bool}    \to   \mathbf{Bool}   $, a non-constant function from $ T_{  \mathbf{H}  } \:   \mathbf{Bool}  $ to $ \mathbf{Bool} $ in an $ \mathbf{L} $-secure world, representing a leak of secure  information. Note that in LDC($ \mathcal{Q}_{\mathbb{N} } $), the factorization lemma is required for proving the splitting lemma which, in turn, is necessary to show substitution. In LDC($ \mathcal{L} $), the splitting lemma holds trivially and does not need any factorization lemma.  

Next, we consider weakening and substitution. The weakening lemma adds an extra assumption at the highest security level to the context. The substitution lemma, on the other hand, substitutes an assumption held at $\ottnt{m_{{\mathrm{0}}}}$ with a term derived at $\ottnt{m_{{\mathrm{0}}}}$. %The weakening lemma says that we can always add an extra assumption at an arbitrary security level to the context.
\begin{lemma}[Weakening] \label{DSWeak}
If $  \Gamma_{{\mathrm{1}}}  ,  \Gamma_{{\mathrm{2}}}   \vdash  \ottnt{a}  :^{ \ell }  \ottnt{A} $, then $    \Gamma_{{\mathrm{1}}}  ,   \ottmv{z}  :^{  \top  }  \ottnt{C}     ,  \Gamma_{{\mathrm{2}}}   \vdash  \ottnt{a}  :^{ \ell }  \ottnt{A} $.
\end{lemma}

\begin{lemma}[Substitution] \label{DSSubst}
If $    \Gamma_{{\mathrm{1}}}  ,   \ottmv{z}  :^{ \ottnt{m_{{\mathrm{0}}}} }  \ottnt{C}     ,  \Gamma_{{\mathrm{2}}}   \vdash  \ottnt{a}  :^{ \ell }  \ottnt{A} $ and $ \Gamma  \vdash  \ottnt{c}  :^{ \ottnt{m_{{\mathrm{0}}}} }  \ottnt{C} $ and $  \lfloor  \Gamma_{{\mathrm{1}}}  \rfloor   =   \lfloor  \Gamma  \rfloor  $, then $    \Gamma_{{\mathrm{1}}}  \sqcap  \Gamma    ,  \Gamma_{{\mathrm{2}}}   \vdash   \ottnt{a}  \{  \ottnt{c}  /  \ottmv{z}  \}   :^{ \ell }  \ottnt{A} $. 
\end{lemma}
Now, LDC($ \mathcal{L} $) can be given the same operational semantics as LDC($ \mathcal{Q}_{\mathbb{N} } $). With respect to this operational semantics, LDC($ \mathcal{L} $) enjoys the standard type soundness property. %The preservation theorem says that if a term is visible at $\ell$ under the assumptions in $\Gamma$, then it remains visible at $\ell$, under the same assumptions, as it reduces.
\begin{theorem}[Preservation]\label{DSPreserve}
If $ \Gamma  \vdash  \ottnt{a}  :^{ \ell }  \ottnt{A} $ and $ \vdash  \ottnt{a}  \leadsto  \ottnt{a'} $, then $ \Gamma  \vdash  \ottnt{a'}  :^{ \ell }  \ottnt{A} $.
\end{theorem}
\begin{theorem}[Progress]\label{DSProgress}
If $  \emptyset   \vdash  \ottnt{a}  :^{ \ell }  \ottnt{A} $, then either $\ottnt{a}$ is a value or there exists $\ottnt{a'}$ such that $ \vdash  \ottnt{a}  \leadsto  \ottnt{a'} $.
\end{theorem}
Type soundness is not strong enough to show that LDC($ \mathcal{L} $) analyses dependencies correctly. For that, we need to show the calculus passes the noninterference test, which ensures that variations in $\ell_{{\mathrm{1}}}$-inputs do not affect $\ell_{{\mathrm{2}}}$-outputs, whenever $\neg( \ell_{{\mathrm{1}}}  \sqsubseteq  \ell_{{\mathrm{2}}} )$. 
In Section \ref{heapsimple}, we shall prove that LDC($ \mathcal{L} $) passes the noninterference test. But before that, we discuss the relation between LDC($ \mathcal{L} $) and Sealing Calculus.

\subsection{Sealing Calculus and LDC} \label{DCCandBLDC}

%The simply-typed version of BLDC contains only terminating computations (the general PTS-version of the calculus includes nontermination). DCC, on the other hand, allows nontermination. While comparing BLDC to DCC here, we restrict ourselves to terminating computations. 

%We have already discussed that while DCC focuses more on effect aspect of dependency analysis, BLDC focuses on both effect and coeffect aspects. We may see BLDC as an extension of DCC with coeffect analysis. In fact, we can embed (the terminating fragment of) DCC into BLDC. 

The Sealing Calculus \cite{igarashi} embeds into LDC($ \mathcal{L} $). We don't provide details of the embedding here because LDC($ \mathcal{L} $) is the same as SDC of \citet{ddc}. \citet{ddc} show that the Sealing Calculus embeds into SDC. Note however that SDC can only be parametrized over lattices and not over semirings that help track linearity.

A technical difference between LDC($ \mathcal{L} $) and SDC is that SDC does not have security annotations on function and product types. However, SDC has modal types, $ T_{ \ell } \:  \ottnt{A} $. Annotated function and product types of LDC($ \mathcal{L} $), $ {}^{ \ottnt{m} }\!  \ottnt{A}  \to  \ottnt{B} $ and $ {}^{ \ottnt{m} }\!  \ottnt{A}  \: \times \:  \ottnt{B} $, correspond to types $   T_{ \ottnt{m} } \:  \ottnt{A}    \to  \ottnt{B} $ and $   T_{ \ottnt{m} } \:  \ottnt{A}    \: \times \:  \ottnt{B} $ respectively in SDC. Conversely, the modal type of SDC, $ T_{ \ottnt{m} } \:  \ottnt{A} $, corresponds to $ {}^{ \ottnt{m} }\!  \ottnt{A}  \: \times \:   \mathbf{Unit}  $ in LDC($ \mathcal{L} $). 

Given that Sealing Calculus is a general dependency calculus that embeds into LDC($ \mathcal{L} $), we conclude that LDC($ \mathcal{L} $) is also a general dependency calculus. Additionally, LDC($ \mathcal{Q}_{\mathbb{N} } $) can analyze resource usage. Thus, LDC can be used for both usage and dependency analyses by parametrizing the calculus over appropriate structures. In the next section, we prove correctness of these analyses in LDC through a heap semantics for the calculus.

\section{Heap Semantics for LDC} \label{heapsimple}

LDC models resource usage and information flow. Standard operational semantics cannot enforce constraints on usage and flow. However, an environment-based semantics, for example a heap semantics, can do so. As such, in this section, we present a weighted-heap-based semantics for LDC. We shall use the same heap semantics for analyzing both resource usage and information flow in LDC, just as we used the same standard small-step semantics for both the analyses. 

Heap-based semantics have been used for analyzing resource usage in literature \cite{turner,grad,uniqueness}. Our heap semantics follows \citet{grad}. The difference between our heap semantics and that of \citet{grad} is that we use it for both usage and flow analyses while they use it for usage analysis only. Other than this difference, both the semantics are essentially the same.

Heap semantics shows how a term reduces in a heap that assigns values to the free variables of the term. Heaps are ordered lists of variable-term pairs, where the terms may be seen as the definitions of the corresponding variables. To every variable-term pair in a heap, we assign a weight, which may be either a $q \in  \mathcal{Q}_{\mathbb{N} } $ or an $\ell \in  \mathcal{L} $. A heap where every variable-term pair has a weight associated with it is referred to as a weighted heap. We assume that our weighted heaps satisfy the following two properties: uniqueness, meaning, a variable is not defined twice; and acyclicity, meaning, definition of a variable does not refer to itself or to other variables appearing subsequently in the heap. Next, we model reductions in terms of interactions between terms and  weighted heaps that define the free variables of terms. 

\subsection{Reduction Relation}

The heap-based reduction rules appear in Figure \ref{HeapReduction}. There are a few things to note with regard to this reduction relation, $ [  \ottnt{H}  ]  \ottnt{a}  \Longrightarrow^{ \ottnt{q} }_{ \ottnt{S} } [  \ottnt{H'}  ]  \ottnt{a'} $:
\begin{itemize}
\item $\ottnt{S}$ here denotes a support set \cite{pitts} of variables that must be avoided while choosing fresh names.
\item From a resource usage perspective, the judgment above may be read as: $\ottnt{q}$ copies of $\ottnt{a}$ use resources from heap $\ottnt{H}$ to produce $\ottnt{q}$ copies of $\ottnt{a'}$, with $\ottnt{H'}$ being the left-over resources. Regarding a heap as a memory, usage of resources corresponds to memory look-up during reduction.
\item From an information flow perspective, the judgment may be read as: under the security constraints of $\ottnt{H}$, the term $\ottnt{a}$ steps to $\ottnt{a'}$ at security level $\ottnt{q}$ with $\ottnt{H'}$ being the updated security constraints. Regarding a heap as a memory, security labels on assignments correspond to access permissions on data while the label on the judgment corresponds to the security clearance of the user.  
\end{itemize} 

\begin{figure}
Heap, $H ::=  \emptyset  \: | \:  \ottnt{H}  ,   \ottmv{x}  \overset{ \ottnt{q} }{\mapsto}  \ottnt{a}  $ \\
\drules[HeapStep]{$ [  \ottnt{H}  ]  \ottnt{a}  \Longrightarrow^{ \ottnt{q} }_{ \ottnt{S} } [  \ottnt{H'}  ]  \ottnt{a'} $}{Selected Heap Step Rules}{Var,Discard,AppL,AppBeta,LetPairBeta,CaseOneBeta}
\caption{Heap Semantics for LDC (Excerpt)}
\label{HeapReduction}
\end{figure}

Now, we consider some of the reduction rules, presented in Figure \ref{HeapReduction}. The most interesting of the rules is \rref{HeapStep-Var}. From a resource usage perspective, this rule may be read as: to step $\ottnt{q}$ copies of $\ottmv{x}$, we need to look-up the value of $\ottmv{x}$ for $\ottnt{q}$ times, thereby using up $\ottnt{q}$ resources. From an information flow perspective, this rule may be read as: the data pointed to by $\ottmv{x}$, held at security level $ \ottnt{r}  \: \sqcap \:  \ottnt{q} $, is observable to $\ottnt{q}$ since $  \ottnt{r}  \: \sqcap \:  \ottnt{q}   \sqsubseteq  \ottnt{q} $. Note that since $  \ottnt{r}  \: \sqcap \:  \ottnt{q}   \sqsubseteq  \ottnt{r} $, the security level of the assignment cannot go down in the updated heap. However, it can always remain the same because $  (   \ottnt{r}  \: \sqcap \:  \ottnt{q}   )   \: \sqcap \:  \ottnt{q}  =  \ottnt{r}  \: \sqcap \:  \ottnt{q} $. The \rref{HeapStep-Var} includes the precondition $ \ottnt{q}  \neq   0  $ because at world $0$ or $\top$, usage and flow constraints are not meaningful.

The \rref{HeapStep-Discard} is also interesting. From a resource usage perspective, it enables discarding of resources, whenever permitted. From an information flow perspective, it corresponds to information remaining visible to an observer as the observer's security clearance goes up.

The \rref{HeapStep-AppL} shows how $ \ottnt{b}  \:  \ottnt{a} ^{ \ottnt{r} } $ steps, as $\ottnt{b}$ steps. Note that the support set in the premise contains the free variables of the argument to ensure that they are avoided while choosing fresh names. We omit the left rules corresponding to other elimination forms, which are as expected.

In the beta \rref{HeapStep-AppBeta,HeapStep-LetPairBeta,HeapStep-Case1Beta}, we add new assignments to the heap avoiding variable capture. The weight at which a new assignment is added is decided by the annotation on the term and the label on the judgment. For example, in \rref{HeapStep-AppBeta}, the new assumption is added at $ \ottnt{q}  \cdot  \ottnt{r} $ or $ \ottnt{q}  \: \sqcup \:  \ottnt{r} $, since the annotation on the application is $\ottnt{r}$ and the label on the judgment is $\ottnt{q}$.  

%We omit \rref{HeapStep-Proj2Beta,HeapStep-Case2Beta} which are similar to \rref{HeapStep-Proj1Beta,HeapStep-Case1Beta} respectively.

With these rules, let us look at some reductions that go through and some that don't.
\begin{align*}
 [   \ottmv{x}  \overset{  1  }{\mapsto}   \mathbf{true}    ]   \ottmv{x}   \Longrightarrow^{  1  }_{ \ottnt{S} } [   \ottmv{x}  \overset{  0  }{\mapsto}   \mathbf{true}    ]   \mathbf{true}   & \text{ but not }  [   \ottmv{x}  \overset{  0  }{\mapsto}   \mathbf{true}    ]   \ottmv{x}   \Longrightarrow^{  1  }_{ \ottnt{S} } [   \ottmv{x}  \overset{  0  }{\mapsto}   \mathbf{true}    ]   \mathbf{true}   \\
 [   \ottmv{x}  \overset{  2  }{\mapsto}   \mathbf{true}    ]   \ottmv{x}   \Longrightarrow^{  2  }_{ \ottnt{S} } [   \ottmv{x}  \overset{  0  }{\mapsto}   \mathbf{true}    ]   \mathbf{true}   & \text{ but not }  [   \ottmv{x}  \overset{  1  }{\mapsto}   \mathbf{true}    ]   \ottmv{x}   \Longrightarrow^{  2  }_{ \ottnt{S} } [   \ottmv{x}  \overset{  0  }{\mapsto}   \mathbf{true}    ]   \mathbf{true}   \\
 [   \ottmv{x}  \overset{  \mathbf{L}  }{\mapsto}   \mathbf{true}    ]   \ottmv{x}   \Longrightarrow^{  \mathbf{H}  }_{ \ottnt{S} } [   \ottmv{x}  \overset{  \mathbf{L}  }{\mapsto}   \mathbf{true}    ]   \mathbf{true}   & \text{ but not }  [   \ottmv{x}  \overset{  \mathbf{H}  }{\mapsto}   \mathbf{true}    ]   \ottmv{x}   \Longrightarrow^{  \mathbf{L}  }_{ \ottnt{S} } [   \ottmv{x}  \overset{  \mathbf{H}  }{\mapsto}   \mathbf{true}    ]   \mathbf{true}  
\end{align*}

\subsection{Ensuring Fair Usage and Secure Flow}

Looking at the rules in Figure \ref{HeapReduction}, we observe that they enforce fairness of resource usage. The only rule that allows usage of resources is \rref{HeapStep-Var}. This rule ensures that a look-up goes through only when the environment can provide adequate resources. It also takes away the necessary resources from the environment after a successful look-up. The rules in Figure \ref{HeapReduction} also ensure security of information flow. The only rule that allows information to flow from heap to program is again \rref{HeapStep-Var}. This rule ensures that information can flow through only when the user has the necessary permission.

The following two lemma formalize the arguments presented above. The first lemma says that a definition that is not available at some point during reduction does not become available at a later point. The second lemma says that an unavailable definition does not play any role in reduction. Note that in case of information flow, the constraint $\neg(\exists q_0, r = q + q_0)$ is equivalent to $\neg( \ottnt{r}  \sqsubseteq  \ottnt{q} )$. (Here, $\lvert \ottnt{H} \rvert$ denotes the length of $\ottnt{H}$.) 
\begin{lemma}[Unchanged] \label{HeapUnchanged}
If  $ [     \ottnt{H_{{\mathrm{1}}}}  ,   \ottmv{x}  \overset{ \ottnt{r} }{\mapsto}  \ottnt{a}     ,  \ottnt{H_{{\mathrm{2}}}}   ]  \ottnt{c}  \Longrightarrow^{ \ottnt{q} }_{ \ottnt{S} } [     \ottnt{H'_{{\mathrm{1}}}}  ,   \ottmv{x}  \overset{ \ottnt{r'} }{\mapsto}  \ottnt{a}     ,  \ottnt{H'_{{\mathrm{2}}}}   ]  \ottnt{c'} $ (where $\lvert H_1 \rvert = \lvert H'_1 \rvert$) and $\neg(\exists q_0, r = q + q_0)$, then $\ottnt{r'} = r$. 
\end{lemma}    
\begin{lemma}[Irrelevant] \label{irrel}
If $ [     \ottnt{H_{{\mathrm{1}}}}  ,   \ottmv{x}  \overset{ \ottnt{r} }{\mapsto}  \ottnt{a}     ,  \ottnt{H_{{\mathrm{2}}}}   ]  \ottnt{c}  \Longrightarrow^{ \ottnt{q} }_{  \ottnt{S}  \, \cup \,   \textit{fv} \:  \ottnt{b}   } [     \ottnt{H'_{{\mathrm{1}}}}  ,   \ottmv{x}  \overset{ \ottnt{r'} }{\mapsto}  \ottnt{a}     ,  \ottnt{H'_{{\mathrm{2}}}}   ]  \ottnt{c'} $ (where $\lvert \ottnt{H_{{\mathrm{1}}}} \rvert = \lvert \ottnt{H'_{{\mathrm{1}}}} \rvert$) and $\neg(\exists q_0, r = q + q_0)$, then $ [     \ottnt{H_{{\mathrm{1}}}}  ,   \ottmv{x}  \overset{ \ottnt{r} }{\mapsto}  \ottnt{b}     ,  \ottnt{H_{{\mathrm{2}}}}   ]  \ottnt{c}  \Longrightarrow^{ \ottnt{q} }_{  \ottnt{S}  \, \cup \,   \textit{fv} \:  \ottnt{a}   } [     \ottnt{H'_{{\mathrm{1}}}}  ,   \ottmv{x}  \overset{ \ottnt{r'} }{\mapsto}  \ottnt{b}     ,  \ottnt{H'_{{\mathrm{2}}}}   ]  \ottnt{c'} $.
\end{lemma}

These lemmas, in conjunction with the soundness theorem we present next, shall guarantee fairness of resource usage and security of information flow in LDC.

\subsection{Soundness With Respect To Heap Semantics} \label{heapsimplesound}

The key idea behind usage analysis through heap semantics is that, if a heap contains the right amount of resources to evaluate some number of copies of a term, as judged by the type system, then the evaluation of that many number of copies of the term in that heap does not get stuck. Since the heap-based reduction rules enforce fairness of resource usage, this would mean that the type system accounts resource usage correctly. 

The key idea behind dependency analysis through heap semantics is similar. If a heap sets the right access permissions for a user, as judged by the type system, then the evaluation, in that heap, of any program observable to that user does not get stuck. Since the reduction rules enforce security of information flow, this would mean that the type system allows only secure flows. 

The compatibility relation, $ \ottnt{H}  \models  \Gamma $, between a heap $\ottnt{H}$ and a context $\Gamma$, formalizes the idea that the heap $\ottnt{H}$ contains the right amount of resources or has set the right access permissions for evaluating any term type-checked in context $\Gamma$. The compatibility relation \cite{grad} is defined below:

\drules[HeapCompat]{$ \ottnt{H}  \models  \Gamma $}{Compatibility}{Empty,Cons}

The soundness theorem stated next says that if a heap $\ottnt{H}$ is compatible with a context $\Gamma$, then the evaluation, starting with heap $\ottnt{H}$, of a term type-checked in context $\Gamma$ does not get stuck.

\begin{theorem}[Soundness]\label{heapsound}
If $ \ottnt{H}  \models  \Gamma $ and $ \Gamma  \vdash  \ottnt{a}  :^{ \ottnt{q} }  \ottnt{A} $ and $q \neq 0$, then either $\ottnt{a}$ is a value or there exists $\ottnt{H'}, \Gamma', \ottnt{a'}$ such that:
\begin{itemize}
\item $ [  \ottnt{H}  ]  \ottnt{a}  \Longrightarrow^{ \ottnt{q} }_{ \ottnt{S} } [  \ottnt{H'}  ]  \ottnt{a'} $
\item $ \ottnt{H'}  \models  \Gamma' $
\item $ \Gamma'  \vdash  \ottnt{a'}  :^{ \ottnt{q} }  \ottnt{A} $
\end{itemize}
\end{theorem}

Note here that the typing context gets updated with every step of reduction, unlike standard type preservation theorem. This is so because the term, as it reduces, needs less resources from the context. However, the updated context is always compatible with the updated heap. Further, note that the labels on the typing judgments and the stepping judgment are the same. From a resource usage perspective, this match-up corresponds to an invariance of the number of copies of the term during reduction. From an information flow perspective, this match-up corresponds to an invariance of the level of the observer during reduction.

Below, we present some corollaries of this soundness theorem.

\begin{corollary}[No Usage] \label{nonint}
In LDC($ \mathcal{Q}_{\mathbb{N} } $): Let $  \emptyset   \vdash  \ottnt{f}  :^{  1  }   {}^{  0  }\!  \ottnt{A}  \to  \ottnt{A}  $. Then, for any $  \emptyset   \vdash  \ottnt{a_{{\mathrm{1}}}}  :^{  0  }  \ottnt{A} $ and $  \emptyset   \vdash  \ottnt{a_{{\mathrm{2}}}}  :^{  0  }  \ottnt{A} $, the terms $ \ottnt{f}  \:  \ottnt{a_{{\mathrm{1}}}} ^{  0  } $ and $ \ottnt{f}  \:  \ottnt{a_{{\mathrm{2}}}} ^{  0  } $ have the same operational behavior, i.e., either both the terms diverge or both reduce to the same value. 
\end{corollary}

The above corollary also holds in LDC($ \mathcal{L} $) with $ 0 $ and $ 1 $ replaced by $ \mathbf{H} $ and $ \mathbf{L} $ respectively. In LDC($ \mathcal{L} $), this corollary shows \textit{non-interference} of high-security inputs in low-security outputs.

\begin{corollary}[Noninterference] \label{noninterference}
In LDC($ \mathcal{L} $): Let $  \emptyset   \vdash  \ottnt{f}  :^{  \mathbf{L}  }   {}^{  \mathbf{H}  }\!  \ottnt{A}  \to  \ottnt{A}  $. Then, for any $  \emptyset   \vdash  \ottnt{a_{{\mathrm{1}}}}  :^{  \mathbf{H}  }  \ottnt{A} $ and $  \emptyset   \vdash  \ottnt{a_{{\mathrm{2}}}}  :^{  \mathbf{H}  }  \ottnt{A} $, the terms $ \ottnt{f}  \:  \ottnt{a_{{\mathrm{1}}}} ^{  \mathbf{H}  } $ and $ \ottnt{f}  \:  \ottnt{a_{{\mathrm{2}}}} ^{  \mathbf{H}  } $ have the same operational behavior. 
\end{corollary}

The soundness theorem may also be employed to reason about other usages, like affine usage in LDC($ \mathbb{N}_{\geq} $).

\begin{corollary}[Affine Usage] \label{single}
In LDC($ \mathbb{N}_{\geq} $): Let $  \emptyset   \vdash  \ottnt{f}  :^{  1  }   {}^{  1  }\!  \ottnt{A}  \to  \ottnt{A}  $. Then, for any $  \emptyset   \vdash  \ottnt{a}  :^{  1  }  \ottnt{A} $, the term $ \ottnt{f}  \:  \ottnt{a} ^{  1  } $ uses $\ottnt{a}$ at most once during reduction.
\end{corollary}

%\subsubsection*{Example: No use or Noninterference} 
%\hfill $\square$

%\subsubsection*{Example: High input and Low output}

%The above example may be paraphrased in terms of information flow as: if $  \emptyset   \vdash  \ottnt{f}  :^{  \mathbf{L}  }   {}^{  \mathbf{H}  }\!  \ottnt{A}  \to  \ottnt{A}  $, then for any $  \emptyset   \vdash  \ottnt{a_{{\mathrm{1}}}}  :^{  \mathbf{H}  }  \ottnt{A} $ and $  \emptyset   \vdash  \ottnt{a_{{\mathrm{2}}}}  :^{  \mathbf{H}  }  \ottnt{A} $, the terms $ \ottnt{f}  \:  \ottnt{a_{{\mathrm{1}}}} ^{  \mathbf{H}  } $ and $ \ottnt{f}  \:  \ottnt{a_{{\mathrm{2}}}} ^{  \mathbf{H}  } $ have the same operational behavior. This follows by virtue of the same argument presented above.

%\subsubsection*{Example: Single use}

%\hfill $\square$

Now that we have seen the syntax and semantics of simply-typed version of LDC, we move on to its Pure Type System (PTS) version.

\section{Linearity and Dependency Analyses in Pure Type Systems}

A Pure Type System (PTS) is characterized by a tuple, $( \mathcal{S} ,  \mathcal{A} ,  \mathcal{R} )$, where $ \mathcal{S} $ is a set of sorts, $ \mathcal{A} $ is a set of axioms and $ \mathcal{R} $ is a ternary relation between sorts \cite{pts}. Many type systems like simply-typed $\lambda$-calculus, System F, System $\text{F}{\omega}$, Calculus of Constructions, Type-in-Type, etc. may be seen as PTSs. Note that a PTS need not be normalizing, for example, Type-in-Type allows nonterminating computations. We parametrize LDC over an abstract PTS so that it may be instantiated to particular PTSs as required. 

\subsection{Simple Version Vs PTS Version} \label{simpdepcomp}

The PTS version of LDC is similar to its simply-typed version. As far as types and terms are concerned, we just need to add to them the sorts in $ \mathcal{S} $ and generalize $ {}^{ \ottnt{r} }\!  \ottnt{A}  \to  \ottnt{B} $ and $ {}^{ \ottnt{r} }\!  \ottnt{A}  \: \times \:  \ottnt{B} $ to $ \Pi  \ottmv{x}  :^{ \ottnt{r} } \!  \ottnt{A}  .  \ottnt{B} $ and $ \Sigma  \ottmv{x}  :^{ \ottnt{r} } \!  \ottnt{A}  .  \ottnt{B} $ respectively. For resource usage and information flow analyses, we use the same parametrizing structures, i.e. $ \mathcal{Q}_{\mathbb{N} } $ and $ \mathcal{L} $ respectively. However, there is an important distinction between these two versions. In the PTS version, we need to extend our analyses from terms to both types and terms.

The key idea behind our extension is that usage and flow analyses for types and terms can be carried out separately. This idea is inspired by recent developments in graded-context dependent type systems \cite{mcbride,atkey,grad,moon}. \citet{mcbride} first noted that linearity and dependent types can be smoothly combined by distinguishing between `contemplative use' of resources in types and `consumptive use' of resources in terms. \citet{atkey} carried this work forward with the calculus \textsc{QTT}, where types live in a resource-agnostic world and terms live in a resource-aware world. \citet{grad} presented an alternative system, \textsc{GraD}, where both types and terms live in a resource-aware world but resources used by a type are zeroed-out while calculating resources used by terms of that type. \citet{moon} presented yet another alternative system, \textsc{Grtt}, where resources used by types are not zeroed-out but tracked simultaneously along with the resources used by terms. In its analysis, \textsc{GraD} is more uniform than \textsc{QTT} and much simpler than \textsc{Grtt}. So LDC analyzes usage in types \`{a} la \textsc{GraD}.

We now look at the type system of the calculus.

\subsection{Type System of LDC}

\begin{figure}
\drules[PTS]{$ \Gamma  \vdash  \ottnt{a}  :^{ \ottnt{q} }  \ottnt{A} $}{PTS version}{Axiom,Var,Weak,Pi,Lam,App,Conv,Sigma,Pair,LetPair,Sum,InjOne,Case,SubL,SubR}
\caption{Type System for LDC (Excerpt)}
\label{pts}
\end{figure}

\begin{figure} 
\drules[Eq]{$ \ottnt{A}  =_{\beta}  \ottnt{B} $}{Definitional Equality}{PiCong,LamCong,AppCong}%,AppBeta}
\caption{Equality Rules for LDC (Excerpt)}
\label{equality}
\end{figure}

The typing and equality rules appear in Figures \ref{pts} and \ref{equality} respectively. There are a few things to note:
\begin{itemize}
\item We use these rules for both resource usage analysis and information flow analysis (with appropriate interpretations of $+, \cdot, 0, 1$ and $ <: $).
\item The judgment $ \Delta  \vdash_{0}  \ottnt{a}  :  \ottnt{A} $ is shorthand for the judgment $ \Gamma  \vdash  \ottnt{a}  :^{  0  }  \ottnt{A} $ where $  \lfloor  \Gamma  \rfloor   =  \Delta $ and $ \overline{ \Gamma } $ is a $\mathbf{0}$ vector. Note that this judgment is essentially the standard typing judgment $ \Delta  \vdash  \ottnt{a}  : \:  \ottnt{A} $ because in world $0$, neither resource usage nor information flow constraints are meaningful.

\item We track usage and flow in terms and types separately. The \rref{PTS-Var} illustrates this principle nicely. The type $\ottnt{A}$ may use some resources or be observable at some low security level. But while type-checking a term of type $\ottnt{A}$, we zero-out the requirements of $\ottnt{A}$ or set $\ottnt{A}$ to the highest security level. This principle also applies to several other rules, for example, \rref{PTS-Lam}, \rref{PTS-Pair}, etc.

\item The \rref{PTS-Pi} shows how we track usage and flow in types. This rule brings out an important aspect of our analysis: not only do we separate the analysis in types and terms but also we allow a term and its type to treat the same bound variable differently. %The decoupling allows the bound variable to be used differently in the body of a type compared to the body of a term having that type. 
The annotation on the type, $\ottnt{r}$ in this case, shows how the bound variable is used in the body of a term having that type. This annotation is \textit{not} related to how the bound variable is used in the body of the type itself. 

Let us consider an example: the polymorphic identity type, $ \Pi  \ottmv{x}  :^{  0  } \!   \ottmv{s}   .   \Pi  \ottmv{y}  :^{  1  } \!   \ottmv{x}   .   \ottmv{x}   $, uses the bound variable $\ottmv{x}$ in its body but a function having this type (e.g. polymorphic identity function $ \lambda^{  0  }  \ottmv{x}  .   \lambda^{  1  }  \ottmv{y}  .   \ottmv{y}   $) can not use the bound variable $\ottmv{x}$ in its body.

\item We use $\beta$-equivalence for equality in \rref{PTS-Conv}. It is a congruent, equivalence relation closed under $\beta$-reduction of terms. Some of the equality rules appear in Figure \ref{equality}. They are mostly standard. However, the congruence \rref{Eq-PiCong,Eq-LamCong,Eq-AppCong} need to check that the grade annotations on the terms being equated match up. 
%\item We omit some rules, ex. \rref{PTS-WSigma,PTS-SSigma} which are similar to \rref{PTS-Pi}. 
\end{itemize}
Next, we look at the metatheory of the calculus.

\subsection{Metatheory of LDC}

The PTS version of LDC satisfies the PTS analogues of all the lemmas and theorems satisfied by the simply-typed version, presented in Sections \ref{MetaSimpleBounded} and \ref{MetaSimpleDependency}. The PTS version also enjoys the same standard call-by-name semantics as the simply-typed version. Further, the PTS version is type-sound with respect to this semantics. 

Next, we state the PTS analogues of some of the crucial lemmas and theorems presented in Sections \ref{MetaSimpleBounded} and \ref{MetaSimpleDependency}.

\iffalse
\begin{lemma}[Weakening]\label{DWeak}
If $  \Gamma_{{\mathrm{1}}}  ,  \Gamma_{{\mathrm{2}}}   \vdash  \ottnt{a}  :^{ \ottnt{q} }  \ottnt{A} $ and $ \Delta_{{\mathrm{1}}}  \vdash_{0}  \ottnt{C}  :   \ottmv{s}  $ and $  \lfloor  \Gamma_{{\mathrm{1}}}  \rfloor   =  \Delta_{{\mathrm{1}}} $, then $    \Gamma_{{\mathrm{1}}}  ,   \ottmv{z}  :^{  0  }  \ottnt{C}     ,  \Gamma_{{\mathrm{2}}}   \vdash  \ottnt{a}  :^{ \ottnt{q} }  \ottnt{A} $.
\end{lemma}
\fi

\begin{lemma}[Substitution]\label{DSubst}
If $    \Gamma_{{\mathrm{1}}}  ,   \ottmv{z}  :^{ \ottnt{r_{{\mathrm{0}}}} }  \ottnt{C}     ,  \Gamma_{{\mathrm{2}}}   \vdash  \ottnt{a}  :^{ \ottnt{q} }  \ottnt{A} $ and $ \Gamma  \vdash  \ottnt{c}  :^{ \ottnt{r_{{\mathrm{0}}}} }  \ottnt{C} $ and $  \lfloor  \Gamma_{{\mathrm{1}}}  \rfloor   =   \lfloor  \Gamma  \rfloor  $, then $     \Gamma_{{\mathrm{1}}}  +  \Gamma    ,  \Gamma_{{\mathrm{2}}}   \{  \ottnt{c}  /  \ottmv{z}  \}   \vdash   \ottnt{a}  \{  \ottnt{c}  /  \ottmv{z}  \}   :^{ \ottnt{q} }   \ottnt{A}  \{  \ottnt{c}  /  \ottmv{z}  \}  $. 
\end{lemma} 

\begin{theorem}[Preservation] \label{preservationDep}
If $ \Gamma  \vdash  \ottnt{a}  :^{ \ottnt{q} }  \ottnt{A} $ and $ \vdash  \ottnt{a}  \leadsto  \ottnt{a'} $, then $ \Gamma  \vdash  \ottnt{a'}  :^{ \ottnt{q} }  \ottnt{A} $.
\end{theorem}
\begin{theorem}[Progress] \label{progressDep}
If $  \emptyset   \vdash  \ottnt{a}  :^{ \ottnt{q} }  \ottnt{A} $, then either $\ottnt{a}$ is a value or there exists $\ottnt{a'}$ such that $ \vdash  \ottnt{a}  \leadsto  \ottnt{a'} $.
\end{theorem}

Next, we consider heap semantics for LDC.

\subsection{Heap Semantics for LDC}
 
The PTS version of LDC enjoys the same heap reduction relation as its simply-typed counterpart. However, the PTS version presents a challenge with regard to well-typedness of terms during reduction. In the simply-typed version, we could delay substitutions in a term by loading them into the heap without being concerned about how it might affect the type of that term. In the PTS version, delayed substitutions may cause the term to `lag behind' the type. We consider an example from \citet{grad} that illustrates this point. 

The polymorphic identity function, $ \lambda^{  0  }  \ottmv{x}  :   \ottmv{s}   .   \lambda^{  1  }  \ottmv{y}  :   \ottmv{x}   .   \ottmv{y}   $, has type $ \Pi  \ottmv{x}  :^{  0  } \!   \ottmv{s}   .   \Pi  \ottmv{y}  :^{  1  } \!   \ottmv{x}   .   \ottmv{x}   $. Instantiating the function at $ \mathbf{Unit} $, we get $  (   \lambda^{  0  }  \ottmv{x}  :   \ottmv{s}   .   \lambda^{  1  }  \ottmv{y}  :   \ottmv{x}   .   \ottmv{y}     )   \:   \mathbf{Unit}  ^{  0  } $ of type $ \Pi  \ottmv{y}  :^{  1  } \!   \mathbf{Unit}   .   \mathbf{Unit}  $. Now, $ [   \emptyset   ]    (   \lambda^{  0  }  \ottmv{x}  :   \ottmv{s}   .   \lambda^{  1  }  \ottmv{y}  :   \ottmv{x}   .   \ottmv{y}     )   \:   \mathbf{Unit}  ^{  0  }   \Longrightarrow^{  1  }_{ \ottnt{S} } [   \ottmv{x}  \overset{  0  }{\mapsto}   \mathbf{Unit}    ]   \lambda^{  1  }  \ottmv{y}  :   \ottmv{x}   .   \ottmv{y}   $. Unless we look at the definition in the heap, we have no reason to believe that $ \lambda^{  1  }  \ottmv{y}  :   \ottmv{x}   .   \ottmv{y}  $ has type $ \Pi  \ottmv{y}  :^{  1  } \!   \mathbf{Unit}   .   \mathbf{Unit}  $. The delayed substitution $ \ottmv{x}  \overset{  0  }{\mapsto}   \mathbf{Unit}  $ causes the term $ \lambda^{  1  }  \ottmv{y}  :   \ottmv{x}   .   \ottmv{y}  $ to lag behind the type $ \Pi  \ottmv{y}  :^{  1  } \!   \mathbf{Unit}   .   \mathbf{Unit}  $. Note that this challenge arises due to delayed substitution only and is independent of usage and flow analyses. 

To overcome this challenge, \citet{grad} use the following strategy. First, they extend their type system with contexts that allow definitions. Then, they show that the extended calculus is sound with respect to heap semantics. Thereafter, they prove the original calculus equivalent to the extended calculus. Using this equivalence, they conclude that the original calculus is also sound with respect to heap semantics. We use the same strategy for LDC. Owing to space constraints, we omit the details in the main body of the paper. For details, the interested reader may please refer to Appendix \ref{heapPTS}. 

LDC is sound with respect to heap semantics. Note the statement below is the same as its simply-typed counterpart.
\begin{theorem}[Soundness] \label{heapDep}
If $ \ottnt{H}  \models  \Gamma $ and $ \Gamma  \vdash  \ottnt{a}  :^{ \ottnt{q} }  \ottnt{A} $ and $q \neq 0$, then either $\ottnt{a}$ is a value or there exists $\ottnt{H'}, \Gamma', \ottnt{a'}$ such that $ [  \ottnt{H}  ]  \ottnt{a}  \Longrightarrow^{ \ottnt{q} }_{ \ottnt{S} } [  \ottnt{H'}  ]  \ottnt{a'} $ and $ \ottnt{H'}  \models  \Gamma' $ and $ \Gamma'  \vdash  \ottnt{a'}  :^{ \ottnt{q} }  \ottnt{A} $. 
\end{theorem}  

The corollaries of the soundness theorem presented in Section \ref{heapsimplesound} also hold for the PTS version. Below, we present some corollaries that are related to polymorphic types.

\begin{corollary} \label{Null}
In LDC($ \mathcal{Q}_{\mathbb{N} } $): If $  \emptyset   \vdash  \ottnt{f}  :^{  1  }   \Pi  \ottmv{x}  :^{  0  } \!   \ottmv{s}   .   \ottmv{x}   $ and $  \emptyset   \vdash_{0}  \ottnt{A}  :   \ottmv{s}  $, then $ \ottnt{f}  \:  \ottnt{A} ^{  0  } $ must diverge.
\end{corollary}

\begin{corollary} \label{Identity}
In LDC($ \mathcal{Q}_{\mathbb{N} } $): In a strongly normalizing PTS, if $  \emptyset   \vdash  \ottnt{f}  :^{  1  }   \Pi  \ottmv{x}  :^{  0  } \!   \ottmv{s}   .   \Pi  \ottmv{y}  :^{  1  } \!   \ottmv{x}   .   \ottmv{x}    $ and $  \emptyset   \vdash_{0}  \ottnt{A}  :   \ottmv{s}  $ and $  \emptyset   \vdash  \ottnt{a}  :^{  1  }  \ottnt{A} $, then $    \ottnt{f}  \:  \ottnt{A} ^{  0  }    \:  \ottnt{a} ^{  1  }   =_{\beta}  \ottnt{a} $.
\end{corollary}

\section{Adding Unrestricted Usage} \label{sec:unrestricted}

Till now, we used LDC for analyzing exact usage, bounded usage and dependency. In this section, we use the calculus for analyzing unrestricted usage as well. Unrestricted usage, referred to by $\omega$, is different from exact and bounded usage, referred to by $n \in \mathbb{N}$, in two ways:
\begin{itemize}
\item $\omega$ is an additive annihilator, meaning $\omega + q = q + \omega =  \omega$, for all $q \in \mathbb{N} \cup \{ \omega \}$. \\ No $n \in \mathbb{N}$ is an additive annihilator.
\item $\omega$ is a multiplicative annihilator (almost) as well, meaning $  \omega   \cdot  \ottnt{q}  =  \ottnt{q}  \cdot   \omega   =   \omega $ for $q \in (\mathbb{N} - \{ 0 \}) \cup \{ \omega \}$. No $n \in \mathbb{N} - \{ 0 \}$ is a multiplicative annihilator.
\end{itemize}
To accommodate this behavior of $ \omega $, we need to make a change to our type system. But before we make this change, let us fix our preordered semirings: %We use the following preordered semirings:
\begin{itemize}
\item $ \mathbb{N}_{=}^{\omega} $, that contains $ \omega $ and the preordered semiring $ \mathbb{N}_{=} $, with $  \omega   <:  \ottnt{q} $ for all $q$
\item $ \mathbb{N}_{\geq}^{\omega} $, that contains $ \omega $ and the preordered semiring $ \mathbb{N}_{\geq} $, with $  \omega   <:  \ottnt{q} $ for all $q$
\item $ \mathcal{Q}_{\text{Lin} } $ and $ \mathcal{Q}_{\text{Aff} } $, described in Section \ref{gradedsalient}
\end{itemize}
 %\Bw{} uses the semiring $ \mathbb{N}^{\omega} = (\mathbb{N} \cup \{  \omega  \}, + , 0 , \cdot , 1 ,  <:  )$ to track resource usage. Here, $+$ and $\cdot$ on $\mathbb{N}$ are defined as before; $+$ and $\times$ on $ \omega $ are defined as above; $ <: $ is defined as $ \omega   <:  \ldots  <:    4   <:   3    <:    2   <:   1    <:  0$. Note that if one wishes to track only affine and unbounded reuse, one can use the subsemiring $(\{ 0 , 1 ,  \omega  \}, + , 0 , \cdot , 1 ,  <:  )$ and the same analysis would carry over.

We use $ \mathcal{Q}_{\mathbb{N} }^{\omega} $ to denote an arbitrary member of the above set of semirings. Next, we discuss the change necessary as we move from LDC($ \mathcal{Q}_{\mathbb{N} } $) to LDC($ \mathcal{Q}_{\mathbb{N} }^{\omega} $). 

\subsection{The Problem and its Solution}
   
When unrestricted usage is allowed, the type systems in Figures \ref{TypeSystemSimple} and \ref{pts} cannot enforce fairness of resource usage. Consider the following `unfair' derivation allowed by the simple type system:
$$\infer[\textsc{ST-SubR}]{  \emptyset   \vdash   \lambda^{  1  }  \ottmv{x}  :  \ottnt{A}  .   (   \ottmv{x}  ^{  1  } ,   \ottmv{x}   )    :^{  1  }   {}^{  1  }\!  \ottnt{A}  \to   {}^{  1  }\!  \ottnt{A}  \: \times \:  \ottnt{A}   }
                      {\infer[\textsc{ST-Lam}]{  \emptyset   \vdash   \lambda^{  1  }  \ottmv{x}  :  \ottnt{A}  .   (   \ottmv{x}  ^{  1  } ,   \ottmv{x}   )    :^{  \omega  }   {}^{  1  }\!  \ottnt{A}  \to   {}^{  1  }\!  \ottnt{A}  \: \times \:  \ottnt{A}   }
                             {\infer[\textsc{ST-Pair}]{  \ottmv{x}  :^{  \omega  }  \ottnt{A}   \vdash   (   \ottmv{x}  ^{  1  } ,   \ottmv{x}   )   :^{  \omega  }   {}^{  1  }\!  \ottnt{A}  \: \times \:  \ottnt{A}  }
                                        {\infer[\textsc{ST-Var}]{  \ottmv{x}  :^{  \omega  }  \ottnt{A}   \vdash   \ottmv{x}   :^{  \omega  }  \ottnt{A} }{} 
                                       & \infer[\textsc{ST-Var}]{  \ottmv{x}  :^{  \omega  }  \ottnt{A}   \vdash   \ottmv{x}   :^{  \omega  }  \ottnt{A} }{} }}}
$$
The judgment $  \emptyset   \vdash   \lambda^{  1  }  \ottmv{x}  :  \ottnt{A}  .   (   \ottmv{x}  ^{  1  } ,   \ottmv{x}   )    :^{  1  }   {}^{  1  }\!  \ottnt{A}  \to   {}^{  1  }\!  \ottnt{A}  \: \times \:  \ottnt{A}   $ is unfair because it allows copying of resources. Carefully observing the derivation, we find that the unfairness arises when $ \omega $ `tricks' the \textsc{ST-Lam} rule into believing that the term uses $\ottmv{x}$ once. 

This unfairness leads to a failure in type soundness. To see how, consider the term: $  \ottmv{y}  :^{  1  }  \ottnt{A}   \vdash    (   \lambda^{  1  }  \ottmv{x}  :  \ottnt{A}  .   (   \ottmv{x}  ^{  1  } ,   \ottmv{x}   )    )   \:   \ottmv{y}  ^{  1  }   :^{  1  }   {}^{  1  }\!  \ottnt{A}  \: \times \:  \ottnt{A}  $ that type-checks via the above derivation and \rref{ST-App}. This term steps to:
$ \vdash    (   \lambda^{  1  }  \ottmv{x}  :  \ottnt{A}  .   (   \ottmv{x}  ^{  1  } ,   \ottmv{x}   )    )   \:   \ottmv{y}  ^{  1  }   \leadsto   (   \ottmv{y}  ^{  1  } ,   \ottmv{y}   )  $. But then, we have unsoundness because: $  \ottmv{y}  :^{  1  }  \ottnt{A}   \nvdash   (   \ottmv{y}  ^{  1  } ,   \ottmv{y}   )   :^{  1  }   {}^{  1  }\!  \ottnt{A}  \: \times \:  \ottnt{A}  $. Therefore, to ensure type soundness, we need to modify \rref{ST-Lam} and \rref{PTS-Lam}.

We modify these rules as follows:
\begin{center}
$\ottdruleSTXXLamOmega{}$ \hspace*{15pt} $\ottdrulePTSXXLamOmega{}$
\end{center}
There are several points to note regarding the above rules:
\begin{itemize}
\item \Rref{ST-LamOmega,PTS-LamOmega} are generalizations of \rref{ST-Lam,PTS-Lam} respectively, meaning, when the grades are restricted to natural numbers, replacing \rref{ST-Lam,PTS-Lam} with \rref{ST-LamOmega,PTS-LamOmega} has no effect on the type system.
\item These rules impose the constraint: $  \ottnt{q}  =   \omega    \Rightarrow   \ottnt{r}  =   \omega   $. This way $ \omega $ won't be able to `trick' the Lambda-rule into believing that functions use their arguments less than what they actually do. In particular, with these modified rules, the above unfair derivation won't go through.
\item The constraint $  \ottnt{q}  =   \omega    \Rightarrow   \ottnt{r}  =   \omega   $, while required for blocking unfair derivations, also blocks some fair ones like the one below:
$$ \infer[]{  \emptyset   \vdash   \lambda^{  1  }  \ottmv{x}  :  \ottnt{A}  .   \ottmv{x}    :^{  \omega  }   {}^{  1  }\!  \ottnt{A}  \to  \ottnt{A}  }{  \ottmv{x}  :^{  \omega  }  \ottnt{A}   \vdash   \ottmv{x}   :^{  \omega  }  \ottnt{A} } $$
To allow such derivations, while still imposing this constraint, \rref{ST-LamOmega,PTS-LamOmega} multiply the conclusion judgment by $\ottnt{q_{{\mathrm{0}}}}$. This multiplication helps these rules allow the above derivation as:
$$ \infer[]{  \emptyset   \vdash   \lambda^{  1  }  \ottmv{x}  :  \ottnt{A}  .   \ottmv{x}    :^{  \omega  }   {}^{  1  }\!  \ottnt{A}  \to  \ottnt{A}  }{  \ottmv{x}  :^{  1  }  \ottnt{A}   \vdash   \ottmv{x}   :^{  1  }  \ottnt{A} } $$
\item The side condition $ \ottnt{q_{{\mathrm{0}}}}  \neq   0  $ makes sure that a meaningful judgment is not turned into a meaningless one. Recall that judgments in $ 0 $ world are meaningless, as far as linearity and dependency analyses are concerned.
\iffalse
\item The \rref{ST-LamOmega,PTS-LamOmega} may seem specific to usage analysis because the constraint $  \ottnt{q}  =   \omega    \Rightarrow   \ottnt{r}  =   \omega   $ does not have an analogue in dependency analysis. This, however, is a minor issue because we can rephrase this constraint as: $   (    \ottnt{q}  +  \ottnt{q}   =  \ottnt{q}   )   \wedge   (    \ottnt{q}  \cdot  \ottnt{r}   =  \ottnt{q}   )    \Rightarrow   \ottnt{r}  <:  \ottnt{q}  $.\\
With respect to usage analysis, the above two constraints are equivalent. \\
Wrt dependency analysis, the latter constraint is a tautology because $   \ottnt{q}  \: \sqcup \:  \ottnt{r}   =  \ottnt{q}   \Rightarrow   \ottnt{r}  \sqsubseteq  \ottnt{q}  $. So \rref{ST-LamOmega,PTS-LamOmega} make sense from the perspective of dependency analysis as well. Further, when viewed as rules analysing dependency, \rref{ST-LamOmega,PTS-LamOmega} are equivalent to \rref{ST-Lam,PTS-Lam} respectively.
\fi
\item \Rref{ST-LamOmega} can also be equivalently replaced with the following two simpler rules (and similarly for \rref{PTS-LamOmega}):\\
\begin{center}
$\ottdruleSTXXLamOmegaZero{} \hspace*{10pt} \ottdruleSTXXOmega{}$ 
\end{center}
\end{itemize}

Replacing \rref{ST-Lam,PTS-Lam} with \rref{ST-LamOmega,PTS-LamOmega} is the only modification that we need to make to the type systems presented in Figures \ref{TypeSystemSimple} and \ref{pts} in order to enable them track unrestricted usage.

With this modification to the type system, LDC($ \mathcal{Q}_{\mathbb{N} }^{\omega} $) satisfies all the lemmas and theorems satisfied by LDC($ \mathcal{Q}_{\mathbb{N} } $). In particular, LDC($ \mathcal{Q}_{\mathbb{N} }^{\omega} $) satisfies type soundness (Theorems \ref{preservationDep} and \ref{progressDep}) and heap soundness (Theorem \ref{heapDep}). We state this property as a theorem below.

\begin{theorem} \label{Bwsound}
LDC($ \mathcal{Q}_{\mathbb{N} }^{\omega} $) satisfies type soundness and heap soundness.
\end{theorem} 

Thus, LDC is a general linear dependency calculus. For tracking linearity, we can use any of the $ \mathcal{Q}_{\mathbb{N} }^{\omega} $s. For tracking dependency, we can use any lattice. For tracking linearity and dependency simultaneously, we can use the cartesian product of these structures. Below, we take up some examples that illustrate combined linearity and dependency analysis in LDC. 

\subsection{Combined Linearity and Dependency Analysis}

First, note that even though preordered semirings and lattices have irreconcilably different sets of axioms, as described in Section \ref{limitdep}, they are still algebraic structures with two binary operators and a binary order relation. Therefore, one can define the cartesian product of these two structures. Concretely, given a preordered semiring $ \mathcal{Q}  = (Q,+,\cdot,0,1, <: )$ and a lattice $ \mathcal{L}  = (L,\sqcap,\sqcup,\top,\bot, \sqsubseteq )$, one can define $ \mathcal{Q}  \times  \mathcal{L} $, the cartesian product of $ \mathcal{Q} $ and $ \mathcal{L} $, as the set $Q \times L$, together with: 
\begin{itemize}
\item two constants, $0$ and $1$, defined as: $0 \triangleq (0,\top)$ and $1 \triangleq (1,\bot)$;
\item a binary operator, $+$, defined as: $(q_1, \ell_1) \, \mathbf{+} \, (q_2, \ell_2) \triangleq ( \ottnt{q_{{\mathrm{1}}}}  +  \ottnt{q_{{\mathrm{2}}}} ,  \ell_{{\mathrm{1}}}  \: \sqcap \:  \ell_{{\mathrm{2}}} )$;
\item another binary operator, $\mathbf{\cdot}$, defined as: $(q_1, \ell_1) \, \mathbf{\cdot} \, (q_2, \ell_2) \triangleq ( \ottnt{q_{{\mathrm{1}}}}  \cdot  \ottnt{q_{{\mathrm{2}}}} ,  \ell_{{\mathrm{1}}}  \: \sqcup \:  \ell_{{\mathrm{2}}} )$; and
\item a binary order relation, $\mathbf{ <: }$, defined as: $(q_1, \ell_1) \, \mathbf{ <: } \, (q_2 , \ell_2) \triangleq  \ottnt{q_{{\mathrm{1}}}}  <:  \ottnt{q_{{\mathrm{2}}}}  \text{ and }  \ell_{{\mathrm{1}}}  \sqsubseteq  \ell_{{\mathrm{2}}} $.
\end{itemize}

Now, given any $ \mathcal{Q}_{\mathbb{N} }^{\omega} $ and $ \mathcal{L} $, we can parametrize LDC over $ \mathcal{Q}_{\mathbb{N} }^{\omega}  \times  \mathcal{L} $, using the above interpretation of $+, \cdot, 0, 1$ and $ <: $. Such parametrization helps us combine linearity and dependency analyses. To illustrate, let $ \mathcal{Q}  =  \mathcal{Q}_{\text{Lin} } $ (or $ \mathcal{Q}_{\text{Aff} } $) and $ \mathcal{L}  =   \mathbf{L}   \sqsubseteq   \mathbf{M}    \sqsubseteq   \mathbf{H} $. Then, in LDC($ \mathcal{Q}  \times  \mathcal{L} $) (extended with $\mathbf{Int}$ type), we have:
\begin{align*}
&   \ottmv{x}  :^{  (   \omega   ,   \mathbf{L}   )  }   \mathbf{Int}    \vdash    \ottmv{x}   +   \ottmv{x}    :^{  (   1   ,   \mathbf{L}   )  }   \mathbf{Int}   \\ 
&    \ottmv{x}  :^{  (   \omega   ,   \mathbf{L}   )  }   \mathbf{Int}    ,   \ottmv{y}  :^{  (   1   ,   \mathbf{M}   )  }   \mathbf{Bool}     \vdash   \mathbf{if} \:   \ottmv{y}   \: \mathbf{then} \:    \ottmv{x}   +   \ottmv{x}    \: \mathbf{else} \:   \ottmv{x}    :^{  (   1   ,   \mathbf{M}   )  }   \mathbf{Int}  .
\end{align*}
Observe that if we change the assumption $ \ottmv{x}  :^{  (   \omega   ,   \mathbf{L}   )  }   \mathbf{Int}  $ to $ \ottmv{x}  :^{  (   1   ,   \mathbf{L}   )  }   \mathbf{Int}  $ or to $ \ottmv{x}  :^{  (   \omega   ,   \mathbf{H}   )  }   \mathbf{Int}  $ in either judgment, then that judgment would no longer hold. Similarly, if we change the assumption $ \ottmv{y}  :^{  (   1   ,   \mathbf{M}   )  }   \mathbf{Bool}  $ to $ \ottmv{y}  :^{  (   1   ,   \mathbf{H}   )  }   \mathbf{Bool}  $ in the second judgment, then too it would no longer hold.

\section{LDC vs. Standard Linear and Dependency Calculi}

\subsection{Comparison with Simply-Typed Calculi}

In section \ref{DCCandBLDC}, we discussed how Sealing Calculus, a standard dependency calculus, embeds into LDC($ \mathcal{L} $). Now, we compare LDC with a standard linear calculus, the Linear Nonlinear (LNL) $\lambda$-calculus of \citet{benton}. The LNL $\lambda$-calculus tracks just linear and unrestricted usage and is simply-typed. So we compare it with simply-typed LDC($ \mathcal{Q}_{\text{Lin} } $). Below, we present a meaning-preserving translation from LNL $\lambda$-calculus to LDC($ \mathcal{Q}_{\text{Lin} } $). 

The LNL calculus employs two forms of contexts and two forms of typing judgments. The two forms of contexts, linear and nonlinear, correspond to assumptions at grades $1$ and $ \omega $ respectively in LDC($ \mathcal{Q}_{\text{Lin} } $). The two forms of judgments, linear and nonlinear, correspond to derivations in worlds $ 1 $ and $ \omega $ respectively in LDC($ \mathcal{Q}_{\text{Lin} } $). In LNL $\lambda$-calculus, linear and nonlinear contexts are denoted by $\Gamma$ and $\Theta$ respectively while linear and nonlinear judgments are written as $ \Theta  ;  \Gamma  \vdash_{\mathcal{L} }   \ottmv{e}   :  \ottnt{A} $ and $ \Theta  \vdash_{\mathcal{C} }   \ottmv{t}   :  \ottnt{X} $ respectively. The calculus contains standard intuitionistic types and linear types; it also contains two type constructors, $F$ and $G$, via which the linear and the nonlinear worlds interact. The calculus uses $\ottnt{A}, \ottnt{B}$ for linear types; $\ottnt{X}, \ottnt{Y}$ for nonlinear types; $\ottnt{a}, \ottnt{b}$ for linear variables; $\ottmv{x}, \ottmv{y}$ for nonlinear variables; $\ottmv{e}, \ottnt{f}$ for linear terms; and $\ottmv{s}, \ottmv{t}$ for nonlinear terms.

%with a standard dependency calculus, DCC. However, we could not compare BLDC with standard linear calculi \cite[etc.]{ill,benton,dill} because BLDC does not allow unbounded reuse. Now that we allow unbounded reuse in \Bw{}, we can compare it with standard linear calculi. For comparison, we choose the Linear Nonlinear (LNL) $\lambda$-calculus of \cite{benton} because its presentation is similar to that of \Bw{}. Since the LNL calculus is simply-typed, we shall compare it with the simply-typed version of \Bw{}. Next, we shall show that the LNL calculus can be translated into \Bw{} while preserving meaning.    

%\subsection{LNL $\lambda$-calculus and \Bw{}}

We present the translation function from LNL $\lambda$-calculus \cite{benton} to LDC($ \mathcal{Q}_{\text{Lin} } $) in Figure \ref{termT}. This translation preserves typing and meaning.
\begin{figure}
\begin{align*}
&  \overline{  1  } =   \mathbf{Unit}   & &  \overline{  \ottnt{X}  \: \times \:  \ottnt{Y}  } =   {}^{  \omega  }\!    \overline{ \ottnt{X} }    \: \times \:   \overline{ \ottnt{Y} }    & &  \overline{  I  } =   \mathbf{Unit}   & &  \overline{  \ottnt{A}  \otimes  \ottnt{B}  } =   {}^{  1  }\!    \overline{ \ottnt{A} }    \: \times \:   \overline{ \ottnt{B} }   \\
&  \overline{  \ottnt{X}  \to  \ottnt{Y}  } =   {}^{  \omega  }\!    \overline{ \ottnt{X} }    \to   \overline{ \ottnt{Y} }    & &  \overline{  G \:  \ottnt{A}  } =   \overline{ \ottnt{A} }   & &  \overline{  \ottnt{A}  \multimap  \ottnt{B}  } =   {}^{  1  }\!    \overline{ \ottnt{A} }    \to   \overline{ \ottnt{B} }    & &  \overline{  F \:  \ottnt{X}  } =   {}^{  \omega  }\!    \overline{ \ottnt{X} }    \: \times \:   \mathbf{Unit}   
\end{align*}
\begin{align*}
&  \overline{  \ottmv{x}  } =   \ottmv{x}   & &  \overline{ \ottnt{a} } =  \ottnt{a}  \\
&  \overline{  ()  } =   \mathbf{unit}   & &  \overline{  \ast  } =   \mathbf{unit}   \\
&  \overline{  (   \ottmv{s}   ,   \ottmv{t}   )  } =   (   \overline{  \ottmv{s}  }  ^{  \omega  } ,   \overline{  \ottmv{t}  }   )   & &  \overline{   \ottmv{e}   \otimes  \ottnt{f}  } =   (   \overline{  \ottmv{e}  }  ^{  1  } ,   \overline{ \ottnt{f} }   )   \\
&  \overline{  \mathbf{fst}(  \ottmv{s}  )  } =   \mathbf{let}_{  1  } \: (  \ottmv{x} ^{  \omega  } ,  \ottmv{y}  ) \: \mathbf{be} \:   \overline{  \ottmv{s}  }   \: \mathbf{in} \:   \ottmv{x}    \hspace*{4pt} (x, y \; \text{fresh}) & &  \overline{  \mathbf{let} \:  \ottnt{a}  \otimes  \ottnt{b}  \: \mathbf{be} \:   \ottmv{e}   \: \mathbf{in} \:  \ottnt{f}  } =   \mathbf{let}_1 \: (  \ottnt{a} ^{  1  } ,  \ottnt{b}  ) \: \mathbf{be} \:    \overline{  \ottmv{e}  }    \: \mathbf{in} \:   \overline{ \ottnt{f} }    \\
&  \overline{  \mathbf{snd}(  \ottmv{s}  )  } =   \mathbf{let}_{  1  } \: (  \ottmv{x} ^{  \omega  } ,  \ottmv{y}  ) \: \mathbf{be} \:   \overline{  \ottmv{s}  }   \: \mathbf{in} \:   \ottmv{y}    \hspace*{4pt} (x, y \; \text{fresh}) & &  \overline{  \mathbf{let} \: \ast \: \mathbf{be} \:   \ottmv{e}   \: \mathbf{in} \:  \ottnt{f}  } =   \mathbf{let}_{  1  } \: \mathbf{unit} \: \mathbf{be} \:   \overline{  \ottmv{e}  }   \: \mathbf{in} \:   \overline{ \ottnt{f} }    \\
&  \overline{  \lambda  \ottmv{x}  :  \ottnt{X}  .   \ottmv{s}   } =   \lambda^{  \omega  }  \ottmv{x}  :   \overline{ \ottnt{X} }   .   \overline{  \ottmv{s}  }    & &  \overline{  \lambda  \ottnt{a}  :  \ottnt{A}  .   \ottmv{e}   } =   \lambda^{1}  \ottnt{a}  :   \overline{ \ottnt{A} }   .   \overline{  \ottmv{e}  }    \\
&  \overline{   \ottmv{s}   \:   \ottmv{t}   } =     \overline{  \ottmv{s}  }    \:    \overline{  \ottmv{t}  }   ^{  \omega  }   & &  \overline{   \ottmv{e}   \:  \ottnt{f}  } =     \overline{  \ottmv{e}  }    \:    \overline{ \ottnt{f} }   ^{  1  }   \\
&  \overline{  G \:   \ottmv{e}   } =   \overline{  \ottmv{e}  }   & &  \overline{  \mathbf{derelict} \:   \ottmv{s}   } =   \overline{  \ottmv{s}  }   \\
&  \overline{  F \:   \ottmv{s}   } =   (   \overline{  \ottmv{s}  }  ^{  \omega  } ,   \mathbf{unit}   )   & &  \overline{  \mathbf{let} \: F  \ottmv{x}  \: \mathbf{be} \:   \ottmv{e}   \: \mathbf{in} \:  \ottnt{f}  } =   \mathbf{let}_{  1  } \: (  \ottmv{x} ^{  \omega  } ,  \ottmv{y}  ) \: \mathbf{be} \:   \overline{  \ottmv{e}  }   \: \mathbf{in} \:   \mathbf{let}_{  1  } \: \mathbf{unit} \: \mathbf{be} \:   \ottmv{y}   \: \mathbf{in} \:   \overline{ \ottnt{f} }     \\
& & & \hspace*{180pt} (y \; \text{fresh})
\end{align*}
\caption{Type and term translation from LNL $\lambda$-calculus to LDC($ \mathcal{Q}_{\text{Lin} } $)}
\label{termT}
\end{figure}

%. And, for a context $\Delta$, we define $ \Delta ^{ \ottnt{q} } $ as the graded context that has $\Delta$ as the underlying context and whose assumptions are all graded at $\ottnt{q}$. %The next theorem shows that this translation preserves typing.
\begin{theorem} \label{LNLTyping}
The translation from LNL $\lambda$-calculus to LDC($ \mathcal{Q}_{\text{Lin} } $), shown in Figure \ref{termT}, is sound:
\begin{itemize}
\item If $ \Theta  ;  \Gamma  \vdash_{\mathcal{L} }   \ottmv{e}   :  \ottnt{A} $, then $     \overline{ \Theta }  ^{  \omega  }   ,   \overline{ \Gamma }   ^{  1  }   \vdash   \overline{  \ottmv{e}  }   :^{  1  }   \overline{ \ottnt{A} }  $. 
\item If $ \Theta  \vdash_{\mathcal{C} }   \ottmv{t}   :  \ottnt{X} $, then $   \overline{ \Theta }  ^{  \omega  }   \vdash   \overline{  \ottmv{t}  }   :^{  \omega  }   \overline{ \ottnt{X} }  $.
\item If $  \ottmv{e}   =_{\beta}  \ottnt{f} $, then $  \overline{  \ottmv{e}  }   =_{\beta}   \overline{ \ottnt{f} }  $. If $  \ottmv{s}   =_{\beta}   \ottmv{t}  $ then $  \overline{  \ottmv{s}  }   =_{\beta}   \overline{  \ottmv{t}  }  $.
\end{itemize}
\end{theorem} 

%Now, Then, we can show that this relation is preserved under translation.
Here, $  \overline{ \Gamma }  ^{  1  } $ and $  \overline{ \Theta }  ^{  \omega  } $ denote $\Gamma$ and $\Theta$, with the types translated, and assumptions held at grades $1$ and $ \omega $ respectively. Further, $\_=_{\beta}\_$ denotes the beta equivalence relation on the terms of LNL $\lambda$-calculus \cite{benton}. 

The soundness theorem of LDC, in conjunction with the above meaning-preserving translation, shows that LDC($ \mathcal{Q}_{\text{Lin} } $) is no less expressive than the LNL $\lambda$-calculus. In fact, LDC($ \mathcal{Q}_{\text{Lin} } $) is more expressive than the LNL $\lambda$-calculus because the latter does not model $0$-usage. Owing to this reason, a translation in the other direction from LDC($ \mathcal{Q}_{\text{Lin} } $) to LNL $\lambda$-calculus fails! 

\subsection{Comparison with Dependently-Typed Calculi}

Next, we compare LDC with standard dependently-typed linear and dependency calculi. As discussed in Section \ref{simpdepcomp}, there are several calculi for linearity analysis in dependent type systems. Among these calculi, LDC is closest to \textsc{GraD} of \citet{grad}. So we compare LDC with \textsc{GraD}. 

\textsc{GraD} is a general coeffect calculus parametrized by an arbitrary partially-ordered semiring. LDC($ \mathcal{Q}_{\mathbb{N} }^{\omega} $), on the other hand, is a linearity calculus parametrized by specific preordered semirings, i.e. $ \mathcal{Q}_{\mathbb{N} }^{\omega} $s. When compared over these semirings, we can show that LDC subsumes \textsc{GraD}.

\begin{theorem}\label{GraD}
With $ \mathcal{Q}_{\mathbb{N} }^{\omega} $ as the parametrizing structure, if $ \Gamma  \vdash  \ottnt{a}  : \:  \ottnt{A} $ in \textsc{GraD}, then $ \Gamma  \vdash  \ottnt{a}  :^{  1  }  \ottnt{A} $ in LDC. Further, if $ \vdash  \ottnt{a}  \leadsto  \ottnt{a'} $ in \textsc{GraD}, then $ \vdash  \ottnt{a}  \leadsto  \ottnt{a'} $ in LDC.
\end{theorem}

The above theorem is not surprising because we followed \textsc{GraD} while designing LDC. 

Next, we compare LDC with standard dependent dependency calculi. In literature, there are only a few calculi \cite{prost,color,ddc} on this topic. Among them, the calculus \Dt{} \citep{ddc} inspired the design of LDC, as discussed in Section \ref{subsec:res}. Owing to this reason, LDC is similar to \Dt{}, as far as dependency analysis is concerned. Formally, we can show that when parametrized over arbitrary lattices, LDC subsumes $\text{DDC}^{\top}$.

\begin{theorem} \label{DDCT}
With $\mathcal{L}$ as the parametrizing structure, if $ \Gamma  \vdash  \ottnt{a}  :^{ \ell }  \ottnt{A} $ in $\text{DDC}^{\top}$, then $ \Gamma  \vdash  \ottnt{a}  :^{ \ell }  \ottnt{A} $ in LDC. Further, if $ \vdash  \ottnt{a}  \leadsto  \ottnt{a'} $ in $\text{DDC}^{\top}$, then $ \vdash  \ottnt{a}  \leadsto  \ottnt{a'} $ in LDC.
\end{theorem}

%The theorems in this section show that LDC is a general linear dependency calculus.

\subsection{Comparison with Other Related Calculi}

In this paper, we presented a unified perspective on usage and dependency analyses. There is a precedent to our presentation: \citet{linearmonad} observed that models of linear logic also provide models of Moggi's computational metalanguage. Now, broadly speaking, usage analysis may be seen as fine-grained linear logic in action and dependency analysis as fine-grained computational metalanguage in action. Therefore, our unified perspective on usage and dependency analyses may be seen as a generalization of the observation made by \citet{linearmonad}.

Next, linearity and dependency have been traditionally analyzed as a coeffect and an effect respectively. There is existing work \cite{effcoeff} in literature on combining effects and coeffects. \citet{effcoeff} present a calculus that employs distinct graded modalities for analyzing coeffects and effects, which are then combined by allowing the modalities to interact via graded distributive laws. In contrast, LDC employs the same graded modality for analyzing both linearity and dependency but draws the grades from different algebraic structures during the two analyses. Further, the graded modality employed in LDC is both monadic and comonadic whereas the graded modalities employed in \citet{effcoeff} are either monadic or comonadic but not necessarily both. The main reason behind these differences is that while the motivation of \citet{effcoeff} is a general calculus for combining coeffects and effects, our motivation is a specialized calculus for combining linearity and dependency. Owing to this specialized nature, our calculus does not need multiple modalities or graded distributive laws. Finally, note that the calculus presented in \citet{effcoeff} is simply-typed whereas LDC allows dependent types. 

%However, the key difference between existing work and our approach is that while generally uses different modalities for representing effects and coeffects and then combines them either via distributive laws as in \cite{effcoeff} or via potentials as in \cite{rajani}. We, on the other hand, use the same modality for representing effects and coeffects. While such a treatment of effects and coeffects may not be possible in general, we show that it is possible for interactions like resource usage and information flow.

\section{Conclusion}

We have shown that linearity and dependency analyses can be systematically unified and combined into a single calculus. We presented, LDC, a general calculus for combined linearity and dependency analysis in pure type systems. We showed that linearity and dependency analyses in LDC are sound using a heap semantics. We also showed that LDC subsumes standard calculi for linearity and dependency analyses. In this paper, we focused on the syntactic properties of LDC. In a future work, we plan to explore the semantic properties of the calculus. In particular, we want to find out how semantic models of LDC compare with the categorical models of linear and dependency type systems.

\newpage

\bibliography{Biblio}

\newpage
\appendix

%----------------------------------------------------------------------------------
\section{Join Not Derivable in Graded-Context Type Systems} \label{app2}
%----------------------------------------------------------------------------------

\begin{proposition} \label{nojoin}
Graded-context type systems \cite{ghica,brunel,petricek,atkey,orchard,abel20,grad} cannot derive a monadic join operator.
\end{proposition}

\begin{proof}
Graded-context type systems mentioned in the proposition vary slightly in their design. However, all of them contain a core graded calculus. Below, we first present this Core Graded Calculus, \Gc{}, and thereafter show that \Gc{} cannot derive a monadic join operator.

\Gc{} is parametrized by an arbitrary preordered semiring $ \mathcal{Q}  = (Q, +, \cdot, 0, 1  <: )$. If parametrized by $ \mathcal{Q} $, the calculus is referred to as \Gc($ \mathcal{Q} $). The grammar and typing rules of \Gc($ \mathcal{Q} $) appear in Figures \ref{GCGrammar} and \ref{GCTyping} respectively. The operations on contexts are defined as in Section \ref{subsec:linsimple}. 

\begin{figure}
\begin{align*}
\text{grades}, q \in Q & ::= 0 \: | \: 1 \: | \:  \ottnt{q_{{\mathrm{1}}}}  +  \ottnt{q_{{\mathrm{2}}}}  \: | \:  \ottnt{q_{{\mathrm{1}}}}  \cdot  \ottnt{q_{{\mathrm{2}}}}  \: | \: \ldots \\
\text{types}, A, B & ::=  \mathbf{Bool}  \: | \:  \ottnt{A}  \multimap  \ottnt{B}  \: | \:  \mspace{2mu} !_{ \ottnt{q} } \mspace{1mu}  \ottnt{A}  \\
\text{terms}, a, b & ::= x \: | \:  \lambda  \ottmv{x}  :  \ottnt{A}  .  \ottnt{b}  \: | \:  \ottnt{b}  \:  \ottnt{a}  \: | \:  \mspace{2mu} !_{ \ottnt{q} } \mspace{1mu}  \ottnt{a}  \: | \:  \mathbf{let} \: !_{ \ottnt{q} } \:  \ottmv{x}  \: \mathbf{be} \:  \ottnt{a}  \: \mathbf{in} \:  \ottnt{b}  \\
                   & | \:  \mathbf{true}  \: | \:  \mathbf{false}  \: | \:  \mathbf{if} \:  \ottnt{b}  \: \mathbf{then} \:  \ottnt{a_{{\mathrm{1}}}}  \: \mathbf{else} \:  \ottnt{a_{{\mathrm{2}}}}  \\
\text{contexts}, \Gamma & ::=  \emptyset  \: | \:  \Gamma  ,   \ottmv{x}  :^{ \ottnt{q} }  \ottnt{A}  
\end{align*}
\caption{Grammar of GC($ \mathcal{Q} $)}
\label{GCGrammar}
\end{figure} 

\begin{figure}[h]
\drules[GC]{$ \Gamma  \vdash  \ottnt{a}  :  \ottnt{A} $}{Typing Rules}{Var,Lam,App,ExpIntro,ExpElim,True,False,If,Sub}
\caption{Typing rules for GC($ \mathcal{Q} $)}
\label{GCTyping}
\end{figure}

Now, towards contradiction, assume that \Gc($ \mathcal{Q} $) can derive a monadic join operator. Then, for any $\ottnt{A}$, there exists a closed non-constant function of type $   \mspace{2mu} !_{ \ottnt{q_{{\mathrm{1}}}} } \mspace{1mu}   \mspace{2mu} !_{ \ottnt{q_{{\mathrm{2}}}} } \mspace{1mu}  \ottnt{A}     \to   \mspace{2mu} !_{  \ottnt{q_{{\mathrm{1}}}}  \cdot  \ottnt{q_{{\mathrm{2}}}}  } \mspace{1mu}  \ottnt{A}  $ for all $\ottnt{q_{{\mathrm{1}}}} , \ottnt{q_{{\mathrm{2}}}} \in  \mathcal{Q} $. By model-theoretic arguments, we shall show that for some $ \mathcal{Q} $, $\ottnt{A}$, $\ottnt{q_{{\mathrm{1}}}}$ and $\ottnt{q_{{\mathrm{2}}}}$, any function having such a type must be constant.

Fix $ \mathcal{Q} $ to be the following preordered semiring. The underlying set is $\{ 0, 1, k_1, k_2 \}$ and the semiring operations and preorder are defined as follows:
\begin{align*}
 \ottnt{q_{{\mathrm{1}}}}  +  \ottnt{q_{{\mathrm{2}}}}  & \triangleq k_2 \text{ if } \ottnt{q_{{\mathrm{1}}}}, \ottnt{q_{{\mathrm{2}}}} \notin \{ 0 \}\\
 \ottnt{q_{{\mathrm{1}}}}  \cdot  \ottnt{q_{{\mathrm{2}}}}  & \triangleq k_2 \text{ if } \ottnt{q_{{\mathrm{1}}}}, \ottnt{q_{{\mathrm{2}}}} \notin \{ 0 , 1 \}\\
 \ottnt{q_{{\mathrm{1}}}}  <:  \ottnt{q_{{\mathrm{2}}}}  & \triangleq \ottnt{q_{{\mathrm{1}}}} = \ottnt{q_{{\mathrm{2}}}}  
\end{align*}
Check that the above definitions satisfy the axioms for preordered semirings. Now, note that $ \ottnt{k_{{\mathrm{1}}}}  \cdot  \ottnt{k_{{\mathrm{1}}}}  = \ottnt{k_{{\mathrm{2}}}}$. Below, we shall show that there exists no non-constant function of type $   \mspace{2mu} !_{ \ottnt{k_{{\mathrm{1}}}} } \mspace{1mu}   \mspace{2mu} !_{ \ottnt{k_{{\mathrm{1}}}} } \mspace{1mu}   \mathbf{Bool}      \to   \mspace{2mu} !_{ \ottnt{k_{{\mathrm{2}}}} } \mspace{1mu}   \mathbf{Bool}   $ in \Gc($ \mathcal{Q} $).

Towards contradiction, suppose such a function, $  \emptyset   \vdash  \ottnt{f}  : \:     \mspace{2mu} !_{ \ottnt{k_{{\mathrm{1}}}} } \mspace{1mu}   \mspace{2mu} !_{ \ottnt{k_{{\mathrm{1}}}} } \mspace{1mu}   \mathbf{Bool}      \to   \mspace{2mu} !_{ \ottnt{k_{{\mathrm{2}}}} } \mspace{1mu}   \mathbf{Bool}    $ exists. Then, $ \ottnt{f}  \:   (   \mspace{2mu} !_{ \ottnt{k_{{\mathrm{1}}}} } \mspace{1mu}   \mspace{2mu} !_{ \ottnt{k_{{\mathrm{1}}}} } \mspace{1mu}   \mathbf{true}     )  $ and $ \ottnt{f}  \:   (   \mspace{2mu} !_{ \ottnt{k_{{\mathrm{1}}}} } \mspace{1mu}   \mspace{2mu} !_{ \ottnt{k_{{\mathrm{1}}}} } \mspace{1mu}   \mathbf{false}     )  $ would reduce to different values (assume a call-by-value reduction). Without loss of generality, say $  \ottnt{f}  \:   (   \mspace{2mu} !_{ \ottnt{k_{{\mathrm{1}}}} } \mspace{1mu}   \mspace{2mu} !_{ \ottnt{k_{{\mathrm{1}}}} } \mspace{1mu}   \mathbf{true}     )    \leadsto^{\ast} \:    \mspace{2mu} !_{ \ottnt{k_{{\mathrm{2}}}} } \mspace{1mu}   \mathbf{true}   $ and $  \ottnt{f}  \:   (   \mspace{2mu} !_{ \ottnt{k_{{\mathrm{1}}}} } \mspace{1mu}   \mspace{2mu} !_{ \ottnt{k_{{\mathrm{1}}}} } \mspace{1mu}   \mathbf{false}     )    \leadsto^{\ast} \:    \mspace{2mu} !_{ \ottnt{k_{{\mathrm{2}}}} } \mspace{1mu}   \mathbf{false}   $, where $\leadsto^{\ast}$ is the multistep call-by-value reduction relation. Now, for any sound interpretation, $\llbracket \_ \rrbracket_{\mathcal{M}}$, of \Gc($ \mathcal{Q} $) in a model $\mathcal{M}$, we should have: \[ \llbracket    \ottnt{f}  \:   (   \mspace{2mu} !_{ \ottnt{k_{{\mathrm{1}}}} } \mspace{1mu}   \mspace{2mu} !_{ \ottnt{k_{{\mathrm{1}}}} } \mspace{1mu}   \mathbf{true}     )     \rrbracket_{\mathcal{M} }  =  \llbracket    \mspace{2mu} !_{ \ottnt{k_{{\mathrm{2}}}} } \mspace{1mu}   \mathbf{true}     \rrbracket_{\mathcal{M} } \text{ and } \llbracket    \ottnt{f}  \:   (   \mspace{2mu} !_{ \ottnt{k_{{\mathrm{1}}}} } \mspace{1mu}   \mspace{2mu} !_{ \ottnt{k_{{\mathrm{1}}}} } \mspace{1mu}   \mathbf{false}     )     \rrbracket_{\mathcal{M} }  =  \llbracket    \mspace{2mu} !_{ \ottnt{k_{{\mathrm{2}}}} } \mspace{1mu}   \mathbf{false}     \rrbracket_{\mathcal{M} } \] We construct a sound model of \Gc($ \mathcal{Q} $) where these equalities lead to a contradiction.

Let $\mathbf{Set}$ be the category of sets and functions. Note that $\mathbf{Set}$ is a symmetric monoidal category with the monoidal product given by cartesian product. Now, any $ \mathcal{Q} $-graded linear exponential comonad on $\mathbf{Set}$ provides a sound interpretation of \Gc($ \mathcal{Q} $) \cite{cokatsumata}. We define $\mathbf{!}$, a $ \mathcal{Q} $-graded linear exponential comonad on $\mathbf{Set}$, as follows:
\begin{align*}
\mathbf{!}(0) = \mathbf{!}(k_1) & = \mathbf{\ast} \\
\mathbf{!}(1) = \mathbf{!}(k_2) & = \mathbf{Id}
\end{align*}
Here, $\mathbf{Id}$ is the identity functor and $\mathbf{\ast}$ is the functor that maps every object to the terminal object. The morphisms associated with $\mathbf{!}$ are as expected.

Now, interpreting using $\mathbf{!}$, we have: $$ \llbracket    \ottnt{f}  \:   (   \mspace{2mu} !_{ \ottnt{k_{{\mathrm{1}}}} } \mspace{1mu}   \mspace{2mu} !_{ \ottnt{k_{{\mathrm{1}}}} } \mspace{1mu}   \mathbf{true}     )     \rrbracket_{(\mathbf{Set},\mathbf{!}) }  =  \llbracket    \ottnt{f}  \:   (   \mspace{2mu} !_{ \ottnt{k_{{\mathrm{1}}}} } \mspace{1mu}   \mspace{2mu} !_{ \ottnt{k_{{\mathrm{1}}}} } \mspace{1mu}   \mathbf{false}     )     \rrbracket_{(\mathbf{Set},\mathbf{!}) }  = \text{app} \: \circ \langle  \llbracket  \ottnt{f}  \rrbracket_{(\mathbf{Set},\mathbf{!}) }  , \langle \rangle \rangle $$
But $ \llbracket    \mspace{2mu} !_{ \ottnt{k_{{\mathrm{2}}}} } \mspace{1mu}   \mathbf{true}     \rrbracket_{(\mathbf{Set},\mathbf{!}) }  \neq  \llbracket    \mspace{2mu} !_{ \ottnt{k_{{\mathrm{2}}}} } \mspace{1mu}   \mathbf{false}     \rrbracket_{(\mathbf{Set},\mathbf{!}) } $. A contradiction.

\end{proof}

%---------------------------------------------------------------------------------
\section{Linearity Analysis in Simply-Typed LDC}
%---------------------------------------------------------------------------------

\begin{lemma}[Multiplication (Lemma \ref{BLSMult})] \label{BLSMultP}
If $ \Gamma  \vdash  \ottnt{a}  :^{ \ottnt{q} }  \ottnt{A} $, then $  \ottnt{r_{{\mathrm{0}}}}  \cdot  \Gamma   \vdash  \ottnt{a}  :^{  \ottnt{r_{{\mathrm{0}}}}  \cdot  \ottnt{q}  }  \ottnt{A} $.
\end{lemma}

\begin{proof}

By induction on $ \Gamma  \vdash  \ottnt{a}  :^{ \ottnt{q} }  \ottnt{A} $.

\begin{itemize}

\item \Rref{ST-Var}. Have: $       0   \cdot  \Gamma_{{\mathrm{1}}}    ,   \ottmv{x}  :^{ \ottnt{q} }  \ottnt{A}     ,    0   \cdot  \Gamma_{{\mathrm{2}}}    \vdash   \ottmv{x}   :^{ \ottnt{q} }  \ottnt{A} $.  \\
Need to show: $       0   \cdot  \Gamma_{{\mathrm{1}}}    ,   \ottmv{x}  :^{  \ottnt{r_{{\mathrm{0}}}}  \cdot  \ottnt{q}  }  \ottnt{A}     ,    0   \cdot  \Gamma_{{\mathrm{2}}}    \vdash   \ottmv{x}   :^{  \ottnt{r_{{\mathrm{0}}}}  \cdot  \ottnt{q}  }  \ottnt{A} $. \\
This case  follows by \rref{ST-Var}.
 
\item \Rref{ST-Lam}. Have: $ \Gamma  \vdash   \lambda^{ \ottnt{r} }  \ottmv{x}  :  \ottnt{A}  .  \ottnt{b}   :^{ \ottnt{q} }   {}^{ \ottnt{r} }\!  \ottnt{A}  \to  \ottnt{B}  $ where $  \Gamma  ,   \ottmv{x}  :^{  \ottnt{q}  \cdot  \ottnt{r}  }  \ottnt{A}    \vdash  \ottnt{b}  :^{ \ottnt{q} }  \ottnt{B} $. \\
Need to show: $  \ottnt{r_{{\mathrm{0}}}}  \cdot  \Gamma   \vdash   \lambda^{ \ottnt{r} }  \ottmv{x}  :  \ottnt{A}  .  \ottnt{b}   :^{   \ottnt{r_{{\mathrm{0}}}}  \cdot  \ottnt{q}   }   {}^{ \ottnt{r} }\!  \ottnt{A}  \to  \ottnt{B}  $. \\
By IH, $   \ottnt{r_{{\mathrm{0}}}}  \cdot  \Gamma   ,   \ottmv{x}  :^{  \ottnt{r_{{\mathrm{0}}}}  \cdot   (   \ottnt{q}  \cdot  \ottnt{r}   )   }  \ottnt{A}    \vdash  \ottnt{b}  :^{  \ottnt{r_{{\mathrm{0}}}}  \cdot  \ottnt{q}  }  \ottnt{B} $. \\
By associativity of multiplication, $ \ottnt{r_{{\mathrm{0}}}}  \cdot   (   \ottnt{q}  \cdot  \ottnt{r}   )   =   (   \ottnt{r_{{\mathrm{0}}}}  \cdot  \ottnt{q}   )   \cdot  \ottnt{r} $.\\ This case, then, follows by \rref{ST-Lam}.

\item \Rref{ST-App}. Have: $  \Gamma_{{\mathrm{1}}}  +  \Gamma_{{\mathrm{2}}}   \vdash   \ottnt{b}  \:  \ottnt{a} ^{ \ottnt{r} }   :^{ \ottnt{q} }  \ottnt{B} $ where $ \Gamma_{{\mathrm{1}}}  \vdash  \ottnt{b}  :^{ \ottnt{q} }   {}^{ \ottnt{r} }\!  \ottnt{A}  \to  \ottnt{B}  $ and $ \Gamma_{{\mathrm{2}}}  \vdash  \ottnt{a}  :^{  \ottnt{q}  \cdot  \ottnt{r}  }  \ottnt{A} $.\\
Need to show: $  \ottnt{r_{{\mathrm{0}}}}  \cdot   (   \Gamma_{{\mathrm{1}}}  +  \Gamma_{{\mathrm{2}}}   )    \vdash   \ottnt{b}  \:  \ottnt{a} ^{ \ottnt{r} }   :^{  \ottnt{r_{{\mathrm{0}}}}  \cdot  \ottnt{q}  }  \ottnt{B} $.\\
By IH, $  \ottnt{r_{{\mathrm{0}}}}  \cdot  \Gamma_{{\mathrm{1}}}   \vdash  \ottnt{b}  :^{  \ottnt{r_{{\mathrm{0}}}}  \cdot  \ottnt{q}  }   {}^{ \ottnt{r} }\!  \ottnt{A}  \to  \ottnt{B}  $ and $  \ottnt{r_{{\mathrm{0}}}}  \cdot  \Gamma_{{\mathrm{2}}}   \vdash  \ottnt{a}  :^{  \ottnt{r_{{\mathrm{0}}}}  \cdot   (   \ottnt{q}  \cdot  \ottnt{r}   )   }  \ottnt{A} $. \\
This case, then, follows by \rref{ST-App}, using associativity and distributivity properties.

\item \Rref{ST-Pair}. Have: $  \Gamma_{{\mathrm{1}}}  +  \Gamma_{{\mathrm{2}}}   \vdash   (  \ottnt{a_{{\mathrm{1}}}} ^{ \ottnt{r} } ,  \ottnt{a_{{\mathrm{2}}}}  )   :^{ \ottnt{q} }   {}^{ \ottnt{r} }\!  \ottnt{A_{{\mathrm{1}}}}  \: \times \:  \ottnt{A_{{\mathrm{2}}}}  $ where $ \Gamma_{{\mathrm{1}}}  \vdash  \ottnt{a_{{\mathrm{1}}}}  :^{  \ottnt{q}  \cdot  \ottnt{r}  }  \ottnt{A_{{\mathrm{1}}}} $ and $ \Gamma_{{\mathrm{2}}}  \vdash  \ottnt{a_{{\mathrm{2}}}}  :^{ \ottnt{q} }  \ottnt{A_{{\mathrm{2}}}} $.\\
Need to show: $  \ottnt{r_{{\mathrm{0}}}}  \cdot   (   \Gamma_{{\mathrm{1}}}  +  \Gamma_{{\mathrm{2}}}   )    \vdash   (  \ottnt{a_{{\mathrm{1}}}} ^{ \ottnt{r} } ,  \ottnt{a_{{\mathrm{2}}}}  )   :^{  \ottnt{r_{{\mathrm{0}}}}  \cdot  \ottnt{q}  }   {}^{ \ottnt{r} }\!  \ottnt{A_{{\mathrm{1}}}}  \: \times \:  \ottnt{A_{{\mathrm{2}}}}  $.\\
By IH, $  \ottnt{r_{{\mathrm{0}}}}  \cdot  \Gamma_{{\mathrm{1}}}   \vdash  \ottnt{a_{{\mathrm{1}}}}  :^{  \ottnt{r_{{\mathrm{0}}}}  \cdot   (   \ottnt{q}  \cdot  \ottnt{r}   )   }  \ottnt{A_{{\mathrm{1}}}} $ and $  \ottnt{r_{{\mathrm{0}}}}  \cdot  \Gamma_{{\mathrm{2}}}   \vdash  \ottnt{a_{{\mathrm{2}}}}  :^{  \ottnt{r_{{\mathrm{0}}}}  \cdot  \ottnt{q}  }  \ottnt{A_{{\mathrm{2}}}} $.\\ 
This case, then, follows by \rref{ST-Pair}, using associativity and distributivity properties.

\item \Rref{ST-LetPair}. Have: $  \Gamma_{{\mathrm{1}}}  +  \Gamma_{{\mathrm{2}}}   \vdash   \mathbf{let}_{ \ottnt{q_{{\mathrm{0}}}} } \: (  \ottmv{x} ^{ \ottnt{r} } ,  \ottmv{y}  ) \: \mathbf{be} \:  \ottnt{a}  \: \mathbf{in} \:  \ottnt{b}   :^{ \ottnt{q} }  \ottnt{B} $ where $ \Gamma_{{\mathrm{1}}}  \vdash  \ottnt{a}  :^{  \ottnt{q}  \cdot  \ottnt{q_{{\mathrm{0}}}}  }   {}^{ \ottnt{r} }\!  \ottnt{A_{{\mathrm{1}}}}  \: \times \:  \ottnt{A_{{\mathrm{2}}}}  $ and $    \Gamma_{{\mathrm{2}}}  ,   \ottmv{x}  :^{  \ottnt{q}  \cdot    \ottnt{q_{{\mathrm{0}}}}  \cdot  \ottnt{r}    }  \ottnt{A_{{\mathrm{1}}}}     ,   \ottmv{y}  :^{  \ottnt{q}  \cdot  \ottnt{q_{{\mathrm{0}}}}  }  \ottnt{A_{{\mathrm{2}}}}    \vdash  \ottnt{b}  :^{ \ottnt{q} }  \ottnt{B} $.\\
Need to show: $  \ottnt{r_{{\mathrm{0}}}}  \cdot   (   \Gamma_{{\mathrm{1}}}  +  \Gamma_{{\mathrm{2}}}   )    \vdash   \mathbf{let}_{ \ottnt{q_{{\mathrm{0}}}} } \: (  \ottmv{x} ^{ \ottnt{r} } ,  \ottmv{y}  ) \: \mathbf{be} \:  \ottnt{a}  \: \mathbf{in} \:  \ottnt{b}   :^{  \ottnt{r_{{\mathrm{0}}}}  \cdot  \ottnt{q}  }  \ottnt{B} $.\\
By IH, $  \ottnt{r_{{\mathrm{0}}}}  \cdot  \Gamma_{{\mathrm{1}}}   \vdash  \ottnt{a}  :^{  \ottnt{r_{{\mathrm{0}}}}  \cdot   (   \ottnt{q}  \cdot  \ottnt{q_{{\mathrm{0}}}}   )   }   {}^{ \ottnt{r} }\!  \ottnt{A_{{\mathrm{1}}}}  \: \times \:  \ottnt{A_{{\mathrm{2}}}}  $ and $   \ottnt{r_{{\mathrm{0}}}}  \cdot  \Gamma_{{\mathrm{2}}}   ,     \ottmv{x}  :^{  \ottnt{r_{{\mathrm{0}}}}  \cdot   (   \ottnt{q}  \cdot    \ottnt{q_{{\mathrm{0}}}}  \cdot  \ottnt{r}     )   }  \ottnt{A_{{\mathrm{1}}}}   ,   \ottmv{y}  :^{  \ottnt{r_{{\mathrm{0}}}}  \cdot   (   \ottnt{q}  \cdot  \ottnt{q_{{\mathrm{0}}}}   )   }  \ottnt{A_{{\mathrm{2}}}}      \vdash  \ottnt{b}  :^{  \ottnt{r_{{\mathrm{0}}}}  \cdot  \ottnt{q}  }  \ottnt{B} $.\\
This case, then, follows by \rref{ST-LetPair}, using associativity and distributivity properties.

\item \Rref{ST-Unit}. Have: $   0   \cdot  \Gamma   \vdash   \mathbf{unit}   :^{ \ottnt{q} }   \mathbf{Unit}  $.\\
Need to show: $  \ottnt{r_{{\mathrm{0}}}}  \cdot   (    0   \cdot  \Gamma   )    \vdash   \mathbf{unit}   :^{  \ottnt{r_{{\mathrm{0}}}}  \cdot  \ottnt{q}  }   \mathbf{Unit}  $.\\
This case follows by \rref{ST-Unit}, using associativity and identity properties.

\item \Rref{ST-LetUnit}. Have: $  \Gamma_{{\mathrm{1}}}  +  \Gamma_{{\mathrm{2}}}   \vdash   \mathbf{let}_{ \ottnt{q_{{\mathrm{0}}}} } \: \mathbf{unit} \: \mathbf{be} \:  \ottnt{a}  \: \mathbf{in} \:  \ottnt{b}   :^{ \ottnt{q} }  \ottnt{B} $ where $ \Gamma_{{\mathrm{1}}}  \vdash  \ottnt{a}  :^{  \ottnt{q}  \cdot  \ottnt{q_{{\mathrm{0}}}}  }   \mathbf{Unit}  $ and $ \Gamma_{{\mathrm{2}}}  \vdash  \ottnt{b}  :^{ \ottnt{q} }  \ottnt{B} $.\\
Need to show: $  \ottnt{r_{{\mathrm{0}}}}  \cdot   (   \Gamma_{{\mathrm{1}}}  +  \Gamma_{{\mathrm{2}}}   )    \vdash   \mathbf{let}_{ \ottnt{q_{{\mathrm{0}}}} } \: \mathbf{unit} \: \mathbf{be} \:  \ottnt{a}  \: \mathbf{in} \:  \ottnt{b}   :^{  \ottnt{r_{{\mathrm{0}}}}  \cdot  \ottnt{q}  }  \ottnt{B} $.\\
By IH, $  \ottnt{r_{{\mathrm{0}}}}  \cdot  \Gamma_{{\mathrm{1}}}   \vdash  \ottnt{a}  :^{  \ottnt{r_{{\mathrm{0}}}}  \cdot   (   \ottnt{q}  \cdot  \ottnt{q_{{\mathrm{0}}}}   )   }   \mathbf{Unit}  $ and $  \ottnt{r_{{\mathrm{0}}}}  \cdot  \Gamma_{{\mathrm{2}}}   \vdash  \ottnt{b}  :^{  \ottnt{r_{{\mathrm{0}}}}  \cdot  \ottnt{q}  }  \ottnt{B} $.\\
This case, then, follows by \rref{ST-LetUnit}, using associativity and distributivity properties.

\item \Rref{ST-Inj1}. Have $ \Gamma  \vdash   \mathbf{inj}_1 \:  \ottnt{a_{{\mathrm{1}}}}   :^{ \ottnt{q} }   \ottnt{A_{{\mathrm{1}}}}  +  \ottnt{A_{{\mathrm{2}}}}  $ where $ \Gamma  \vdash  \ottnt{a_{{\mathrm{1}}}}  :^{ \ottnt{q} }  \ottnt{A_{{\mathrm{1}}}} $.\\
Need to show: $  \ottnt{r_{{\mathrm{0}}}}  \cdot  \Gamma   \vdash   \mathbf{inj}_1 \:  \ottnt{a_{{\mathrm{1}}}}   :^{  \ottnt{r_{{\mathrm{0}}}}  \cdot  \ottnt{q}  }   \ottnt{A_{{\mathrm{1}}}}  +  \ottnt{A_{{\mathrm{2}}}}  $.\\
By IH, $  \ottnt{r_{{\mathrm{0}}}}  \cdot  \Gamma   \vdash  \ottnt{a_{{\mathrm{1}}}}  :^{  \ottnt{r_{{\mathrm{0}}}}  \cdot  \ottnt{q}  }  \ottnt{A_{{\mathrm{1}}}} $.\\
This case, then, follows by \rref{ST-Inj1}.

\item \Rref{ST-Inj2}. Similar to \rref{ST-Inj1}.

\item \Rref{ST-Case}. Have: $  \Gamma_{{\mathrm{1}}}  +  \Gamma_{{\mathrm{2}}}   \vdash   \mathbf{case}_{ \ottnt{q_{{\mathrm{0}}}} } \:  \ottnt{a}  \: \mathbf{of} \:  \ottmv{x_{{\mathrm{1}}}}  .  \ottnt{b_{{\mathrm{1}}}}  \: ; \:  \ottmv{x_{{\mathrm{2}}}}  .  \ottnt{b_{{\mathrm{2}}}}   :^{ \ottnt{q} }  \ottnt{B} $ where $ \Gamma_{{\mathrm{1}}}  \vdash  \ottnt{a}  :^{  \ottnt{q}  \cdot  \ottnt{q_{{\mathrm{0}}}}  }   \ottnt{A_{{\mathrm{1}}}}  +  \ottnt{A_{{\mathrm{2}}}}  $ and $  \Gamma_{{\mathrm{2}}}  ,   \ottmv{x_{{\mathrm{1}}}}  :^{  \ottnt{q}  \cdot  \ottnt{q_{{\mathrm{0}}}}  }  \ottnt{A_{{\mathrm{1}}}}    \vdash  \ottnt{b_{{\mathrm{1}}}}  :^{ \ottnt{q} }  \ottnt{B} $ and $  \Gamma_{{\mathrm{2}}}  ,   \ottmv{x_{{\mathrm{2}}}}  :^{  \ottnt{q}  \cdot  \ottnt{q_{{\mathrm{0}}}}  }  \ottnt{A_{{\mathrm{2}}}}    \vdash  \ottnt{b_{{\mathrm{2}}}}  :^{ \ottnt{q} }  \ottnt{B} $.\\
Need to show: $  \ottnt{r_{{\mathrm{0}}}}  \cdot   (   \Gamma_{{\mathrm{1}}}  +  \Gamma_{{\mathrm{2}}}   )    \vdash   \mathbf{case}_{ \ottnt{q_{{\mathrm{0}}}} } \:  \ottnt{a}  \: \mathbf{of} \:  \ottmv{x_{{\mathrm{1}}}}  .  \ottnt{b_{{\mathrm{1}}}}  \: ; \:  \ottmv{x_{{\mathrm{2}}}}  .  \ottnt{b_{{\mathrm{2}}}}   :^{  \ottnt{r_{{\mathrm{0}}}}  \cdot  \ottnt{q}  }  \ottnt{B} $.\\
By IH, $  \ottnt{r_{{\mathrm{0}}}}  \cdot  \Gamma_{{\mathrm{1}}}   \vdash  \ottnt{a}  :^{  \ottnt{r_{{\mathrm{0}}}}  \cdot   (   \ottnt{q}  \cdot  \ottnt{q_{{\mathrm{0}}}}   )   }   \ottnt{A_{{\mathrm{1}}}}  +  \ottnt{A_{{\mathrm{2}}}}  $ and $   \ottnt{r_{{\mathrm{0}}}}  \cdot  \Gamma_{{\mathrm{2}}}   ,   \ottmv{x_{{\mathrm{1}}}}  :^{  \ottnt{r_{{\mathrm{0}}}}  \cdot   (   \ottnt{q}  \cdot  \ottnt{q_{{\mathrm{0}}}}   )   }  \ottnt{A_{{\mathrm{1}}}}    \vdash  \ottnt{b_{{\mathrm{1}}}}  :^{  \ottnt{r_{{\mathrm{0}}}}  \cdot  \ottnt{q}  }  \ottnt{B} $ and $   \ottnt{r_{{\mathrm{0}}}}  \cdot  \Gamma_{{\mathrm{2}}}   ,   \ottmv{x_{{\mathrm{2}}}}  :^{  \ottnt{r_{{\mathrm{0}}}}  \cdot   (   \ottnt{q}  \cdot  \ottnt{q_{{\mathrm{0}}}}   )   }  \ottnt{A_{{\mathrm{2}}}}    \vdash  \ottnt{b_{{\mathrm{2}}}}  :^{  \ottnt{r_{{\mathrm{0}}}}  \cdot  \ottnt{q}  }  \ottnt{B} $.\\
This case, then, follows by \rref{ST-Case}, using associativity and distributivity properties.

\item \Rref{ST-SubL}. Have: $ \Gamma  \vdash  \ottnt{a}  :^{ \ottnt{q} }  \ottnt{A} $ where $ \Gamma'  \vdash  \ottnt{a}  :^{ \ottnt{q} }  \ottnt{A} $ and $ \Gamma  <:  \Gamma' $. \\
Need to show: $  \ottnt{r_{{\mathrm{0}}}}  \cdot  \Gamma   \vdash  \ottnt{a}  :^{  \ottnt{r_{{\mathrm{0}}}}  \cdot  \ottnt{q}  }  \ottnt{A} $.\\
By IH, $  \ottnt{r_{{\mathrm{0}}}}  \cdot  \Gamma'   \vdash  \ottnt{a}  :^{  \ottnt{r_{{\mathrm{0}}}}  \cdot  \ottnt{q}  }  \ottnt{A} $.\\
Next, $  \ottnt{r_{{\mathrm{0}}}}  \cdot  \Gamma   <:   \ottnt{r_{{\mathrm{0}}}}  \cdot  \Gamma'  $, since $ \Gamma  <:  \Gamma' $.\\
This case, then, follows by \rref{ST-SubL}.

\item \Rref{ST-SubR}. Have: $ \Gamma  \vdash  \ottnt{a}  :^{ \ottnt{q'} }  \ottnt{A} $ where $ \Gamma  \vdash  \ottnt{a}  :^{ \ottnt{q} }  \ottnt{A} $ and $ \ottnt{q}  <:  \ottnt{q'} $.\\
Need to show: $  \ottnt{r_{{\mathrm{0}}}}  \cdot  \Gamma   \vdash  \ottnt{a}  :^{  \ottnt{r_{{\mathrm{0}}}}  \cdot  \ottnt{q'}  }  \ottnt{A} $.\\
By IH, $  \ottnt{r_{{\mathrm{0}}}}  \cdot  \Gamma   \vdash  \ottnt{a}  :^{  \ottnt{r_{{\mathrm{0}}}}  \cdot  \ottnt{q}  }  \ottnt{A} $.\\
Next, $  \ottnt{r_{{\mathrm{0}}}}  \cdot  \ottnt{q}   <:   \ottnt{r_{{\mathrm{0}}}}  \cdot  \ottnt{q'}  $, since $ \ottnt{q}  <:  \ottnt{q'} $.\\
This case, then, follows by \rref{ST-SubR}.

\end{itemize}
\end{proof}

%------------------------------------------------------------------------------------------

\begin{lemma}[Factorization (Lemma \ref{BLSFact})] \label{BLSFactP}
If $ \Gamma  \vdash  \ottnt{a}  :^{ \ottnt{q} }  \ottnt{A} $ and $q \neq 0$, then there exists $\Gamma'$ such that $ \Gamma'  \vdash  \ottnt{a}  :^{  1  }  \ottnt{A} $ and $ \Gamma  <:   \ottnt{q}  \cdot  \Gamma'  $.
\end{lemma}

\begin{proof}

By induction on $ \Gamma  \vdash  \ottnt{a}  :^{ \ottnt{q} }  \ottnt{A} $.

\begin{itemize}

\item \Rref{ST-Var}. Have: $       0   \cdot  \Gamma_{{\mathrm{1}}}    ,   \ottmv{x}  :^{ \ottnt{q} }  \ottnt{A}     ,    0   \cdot  \Gamma_{{\mathrm{2}}}    \vdash   \ottmv{x}   :^{ \ottnt{q} }  \ottnt{A} $.\\
Need to show: $\exists \Gamma'$ such that $ \Gamma'  \vdash   \ottmv{x}   :^{  1  }  \ottnt{A} $ and $       0   \cdot  \Gamma_{{\mathrm{1}}}    ,   \ottmv{x}  :^{ \ottnt{q} }  \ottnt{A}     ,    0   \cdot  \Gamma_{{\mathrm{2}}}    <:   \ottnt{q}  \cdot  \Gamma'  $.\\
This case follows by setting $\Gamma' :=       0   \cdot  \Gamma_{{\mathrm{1}}}    ,   \ottmv{x}  :^{  1  }  \ottnt{A}     ,    0   \cdot  \Gamma_{{\mathrm{2}}}  $.

\item \Rref{ST-Lam}. Have: $ \Gamma  \vdash   \lambda^{ \ottnt{r} }  \ottmv{x}  :  \ottnt{A}  .  \ottnt{b}   :^{ \ottnt{q} }   {}^{ \ottnt{r} }\!  \ottnt{A}  \to  \ottnt{B}  $ where $  \Gamma  ,   \ottmv{x}  :^{  \ottnt{q}  \cdot  \ottnt{r}  }  \ottnt{A}    \vdash  \ottnt{b}  :^{ \ottnt{q} }  \ottnt{B} $.\\
Need to show: $\exists \Gamma'$ such that $ \Gamma'  \vdash   \lambda^{ \ottnt{r} }  \ottmv{x}  :  \ottnt{A}  .  \ottnt{b}   :^{  1  }   {}^{ \ottnt{r} }\!  \ottnt{A}  \to  \ottnt{B}  $ and $ \Gamma  <:   \ottnt{q}  \cdot  \Gamma'  $. \\
By IH, $\exists \Gamma'_{{\mathrm{1}}}$ and $\ottnt{r'_{{\mathrm{1}}}}$ such that $  \Gamma'_{{\mathrm{1}}}  ,   \ottmv{x}  :^{ \ottnt{r'_{{\mathrm{1}}}} }  \ottnt{A}    \vdash  \ottnt{b}  :^{  1  }  \ottnt{B} $ and $ \Gamma  <:   \ottnt{q}  \cdot  \Gamma'_{{\mathrm{1}}}  $ and $  \ottnt{q}  \cdot  \ottnt{r}   <:   \ottnt{q}  \cdot  \ottnt{r'_{{\mathrm{1}}}}  $.\\
Since $  \ottnt{q}  \cdot  \ottnt{r}   <:   \ottnt{q}  \cdot  \ottnt{r'_{{\mathrm{1}}}}  $ and $q \neq 0$, therefore $ \ottnt{r}  <:  \ottnt{r'_{{\mathrm{1}}}} $.\\
By \rref{ST-SubL}, $  \Gamma'_{{\mathrm{1}}}  ,   \ottmv{x}  :^{ \ottnt{r} }  \ottnt{A}    \vdash  \ottnt{b}  :^{  1  }  \ottnt{B} $.\\
This case, then, follows by \rref{ST-Lam} by setting $\Gamma' := \Gamma'_{{\mathrm{1}}}$.

\item \Rref{ST-App}. Have: $  \Gamma_{{\mathrm{1}}}  +  \Gamma_{{\mathrm{2}}}   \vdash   \ottnt{b}  \:  \ottnt{a} ^{ \ottnt{r} }   :^{ \ottnt{q} }  \ottnt{B} $ where $ \Gamma_{{\mathrm{1}}}  \vdash  \ottnt{b}  :^{ \ottnt{q} }   {}^{ \ottnt{r} }\!  \ottnt{A}  \to  \ottnt{B}  $ and $ \Gamma_{{\mathrm{2}}}  \vdash  \ottnt{a}  :^{  \ottnt{q}  \cdot  \ottnt{r}  }  \ottnt{A} $.\\
Need to show: $\exists \Gamma'$ such that $ \Gamma'  \vdash   \ottnt{b}  \:  \ottnt{a} ^{ \ottnt{r} }   :^{  1  }  \ottnt{B} $ and $  \Gamma_{{\mathrm{1}}}  +  \Gamma_{{\mathrm{2}}}   <:   \ottnt{q}  \cdot  \Gamma'  $.\\
There are two cases to consider.

\begin{itemize}
\item $r \neq 0$. Since $q \neq 0$, therefore $ \ottnt{q}  \cdot  \ottnt{r}  \neq 0$.\\
By IH, $\exists \Gamma'_{{\mathrm{1}}}, \Gamma'_{{\mathrm{2}}}$ such that $ \Gamma'_{{\mathrm{1}}}  \vdash  \ottnt{b}  :^{  1  }   {}^{ \ottnt{r} }\!  \ottnt{A}  \to  \ottnt{B}  $ and $ \Gamma'_{{\mathrm{2}}}  \vdash  \ottnt{a}  :^{  1  }  \ottnt{A} $ and $ \Gamma_{{\mathrm{1}}}  <:   \ottnt{q}  \cdot  \Gamma'_{{\mathrm{1}}}  $ and $ \Gamma_{{\mathrm{2}}}  <:    (   \ottnt{q}  \cdot  \ottnt{r}   )   \cdot  \Gamma'_{{\mathrm{2}}}  $.\\
By Lemma \ref{BLSMultP}, $  \ottnt{r}  \cdot  \Gamma'_{{\mathrm{2}}}   \vdash  \ottnt{a}  :^{ \ottnt{r} }  \ottnt{A} $. \\
By \rref{ST-App}, $  \Gamma'_{{\mathrm{1}}}  +   \ottnt{r}  \cdot  \Gamma'_{{\mathrm{2}}}    \vdash   \ottnt{b}  \:  \ottnt{a} ^{ \ottnt{r} }   :^{  1  }  \ottnt{B} $.\\
This case, then, follows by setting $\Gamma' :=  \Gamma'_{{\mathrm{1}}}  +   \ottnt{r}  \cdot  \Gamma'_{{\mathrm{2}}}  $.

\item $r = 0$. Then, $ \Gamma_{{\mathrm{2}}}  \vdash  \ottnt{a}  :^{  0  }  \ottnt{A} $. By Lemma \ref{BLSMultP}, $   0   \cdot  \Gamma_{{\mathrm{2}}}   \vdash  \ottnt{a}  :^{  0  }  \ottnt{A} $.\\
By IH, $\exists \Gamma'_{{\mathrm{1}}}$ such that $ \Gamma'_{{\mathrm{1}}}  \vdash  \ottnt{b}  :^{  1  }   {}^{  0  }\!  \ottnt{A}  \to  \ottnt{B}  $ and $ \Gamma_{{\mathrm{1}}}  <:   \ottnt{q}  \cdot  \Gamma'_{{\mathrm{1}}}  $.\\
By \rref{ST-App}, $ \Gamma'_{{\mathrm{1}}}  \vdash   \ottnt{b}  \:  \ottnt{a} ^{  0  }   :^{  1  }  \ottnt{B} $.\\
Further, $  \Gamma_{{\mathrm{1}}}  +  \Gamma_{{\mathrm{2}}}   <:  \Gamma_{{\mathrm{1}}} $ ($\because$ $  \Gamma_{{\mathrm{1}}}  +  \Gamma_{{\mathrm{0}}}   <:  \Gamma_{{\mathrm{1}}} $ for any $\Gamma_{{\mathrm{0}}}$ in $ \mathbb{N}_{\geq} $ and $ \overline{ \Gamma_{{\mathrm{2}}} }  = \overline{0}$ in $ \mathbb{N}_{=} $).\\
This case, then, follows by setting $\Gamma' := \Gamma'_{{\mathrm{1}}}$.
\end{itemize}

\item \Rref{ST-Pair}. Have: $  \Gamma_{{\mathrm{1}}}  +  \Gamma_{{\mathrm{2}}}   \vdash   (  \ottnt{a_{{\mathrm{1}}}} ^{ \ottnt{r} } ,  \ottnt{a_{{\mathrm{2}}}}  )   :^{ \ottnt{q} }   {}^{ \ottnt{r} }\!  \ottnt{A_{{\mathrm{1}}}}  \: \times \:  \ottnt{A_{{\mathrm{2}}}}  $ where $ \Gamma_{{\mathrm{1}}}  \vdash  \ottnt{a_{{\mathrm{1}}}}  :^{  \ottnt{q}  \cdot  \ottnt{r}  }  \ottnt{A_{{\mathrm{1}}}} $ and $ \Gamma_{{\mathrm{2}}}  \vdash  \ottnt{a_{{\mathrm{2}}}}  :^{ \ottnt{q} }  \ottnt{A_{{\mathrm{2}}}} $.\\
Need to show: $\exists \Gamma'$ such that $ \Gamma'  \vdash   (  \ottnt{a_{{\mathrm{1}}}} ^{ \ottnt{r} } ,  \ottnt{a_{{\mathrm{2}}}}  )   :^{  1  }   {}^{ \ottnt{r} }\!  \ottnt{A_{{\mathrm{1}}}}  \: \times \:  \ottnt{A_{{\mathrm{2}}}}  $ and $  \Gamma_{{\mathrm{1}}}  +  \Gamma_{{\mathrm{2}}}   <:   \ottnt{q}  \cdot  \Gamma'  $.\\
There are two cases to consider.

\begin{itemize}
\item $r \neq 0$. Since $q \neq 0$, therefore $ \ottnt{q}  \cdot  \ottnt{r}  \neq 0$.\\
By IH, $\exists \Gamma'_{{\mathrm{1}}}, \Gamma'_{{\mathrm{2}}}$ such that $ \Gamma'_{{\mathrm{1}}}  \vdash  \ottnt{a_{{\mathrm{1}}}}  :^{  1  }  \ottnt{A_{{\mathrm{1}}}} $ and $ \Gamma'_{{\mathrm{2}}}  \vdash  \ottnt{a_{{\mathrm{2}}}}  :^{  1  }  \ottnt{A_{{\mathrm{2}}}} $ and $ \Gamma_{{\mathrm{1}}}  <:    (   \ottnt{q}  \cdot  \ottnt{r}   )   \cdot  \Gamma'_{{\mathrm{1}}}  $ and $ \Gamma_{{\mathrm{2}}}  <:   \ottnt{q}  \cdot  \Gamma'_{{\mathrm{2}}}  $.\\
By Lemma \ref{BLSMultP}, $  \ottnt{r}  \cdot  \Gamma'_{{\mathrm{1}}}   \vdash  \ottnt{a_{{\mathrm{1}}}}  :^{ \ottnt{r} }  \ottnt{A_{{\mathrm{1}}}} $.\\
By \rref{ST-Pair}, $   \ottnt{r}  \cdot  \Gamma'_{{\mathrm{1}}}   +  \Gamma'_{{\mathrm{2}}}   \vdash   (  \ottnt{a_{{\mathrm{1}}}} ^{ \ottnt{r} } ,  \ottnt{a_{{\mathrm{2}}}}  )   :^{  1  }   {}^{ \ottnt{r} }\!  \ottnt{A_{{\mathrm{1}}}}  \: \times \:  \ottnt{A_{{\mathrm{2}}}}  $.\\
This case, then, follows by setting $\Gamma' :=   \ottnt{r}  \cdot  \Gamma'_{{\mathrm{1}}}   +  \Gamma'_{{\mathrm{2}}} $.

\item $r = 0$. Then, $ \Gamma_{{\mathrm{1}}}  \vdash  \ottnt{a_{{\mathrm{1}}}}  :^{  0  }  \ottnt{A_{{\mathrm{1}}}} $. By Lemma \ref{BLSMultP}, $   0   \cdot  \Gamma_{{\mathrm{1}}}   \vdash  \ottnt{a_{{\mathrm{1}}}}  :^{  0  }  \ottnt{A_{{\mathrm{1}}}} $.\\
By IH, $\exists \Gamma'_{{\mathrm{2}}}$ such that $ \Gamma'_{{\mathrm{2}}}  \vdash  \ottnt{a_{{\mathrm{2}}}}  :^{  1  }  \ottnt{A_{{\mathrm{2}}}} $ and $ \Gamma_{{\mathrm{2}}}  <:   \ottnt{q}  \cdot  \Gamma'_{{\mathrm{2}}}  $.\\
By \rref{ST-Pair}, $ \Gamma'_{{\mathrm{2}}}  \vdash   (  \ottnt{a_{{\mathrm{1}}}} ^{  0  } ,  \ottnt{a_{{\mathrm{2}}}}  )   :^{  1  }   {}^{  0  }\!  \ottnt{A_{{\mathrm{1}}}}  \: \times \:  \ottnt{A_{{\mathrm{2}}}}  $.\\
Further, $  \Gamma_{{\mathrm{1}}}  +  \Gamma_{{\mathrm{2}}}   <:  \Gamma_{{\mathrm{2}}} $ ($\because$ $  \Gamma_{{\mathrm{0}}}  +  \Gamma_{{\mathrm{2}}}   <:  \Gamma_{{\mathrm{2}}} $ for any $\Gamma_{{\mathrm{0}}}$ in $ \mathbb{N}_{\geq} $ and $ \overline{ \Gamma_{{\mathrm{1}}} }  = \overline{0}$ in $ \mathbb{N}_{=} $).\\
This case, then, follows by setting $\Gamma' := \Gamma'_{{\mathrm{2}}}$.
\end{itemize} 

\item \Rref{ST-LetPair}. Have: $  \Gamma_{{\mathrm{1}}}  +  \Gamma_{{\mathrm{2}}}   \vdash   \mathbf{let}_{ \ottnt{q_{{\mathrm{0}}}} } \: (  \ottmv{x} ^{ \ottnt{r} } ,  \ottmv{y}  ) \: \mathbf{be} \:  \ottnt{a}  \: \mathbf{in} \:  \ottnt{b}   :^{ \ottnt{q} }  \ottnt{B} $ where $ \Gamma_{{\mathrm{1}}}  \vdash  \ottnt{a}  :^{  \ottnt{q}  \cdot  \ottnt{q_{{\mathrm{0}}}}  }   {}^{ \ottnt{r} }\!  \ottnt{A_{{\mathrm{1}}}}  \: \times \:  \ottnt{A_{{\mathrm{2}}}}  $ and $    \Gamma_{{\mathrm{2}}}  ,   \ottmv{x}  :^{  \ottnt{q}  \cdot    \ottnt{q_{{\mathrm{0}}}}  \cdot  \ottnt{r}    }  \ottnt{A_{{\mathrm{1}}}}     ,   \ottmv{y}  :^{  \ottnt{q}  \cdot  \ottnt{q_{{\mathrm{0}}}}  }  \ottnt{A_{{\mathrm{2}}}}    \vdash  \ottnt{b}  :^{ \ottnt{q} }  \ottnt{B} $ and $ \ottnt{q_{{\mathrm{0}}}}  <:   1  $.\\
Need to show: $\exists \Gamma'$ such that $ \Gamma'  \vdash   \mathbf{let}_{ \ottnt{q_{{\mathrm{0}}}} } \: (  \ottmv{x} ^{ \ottnt{r} } ,  \ottmv{y}  ) \: \mathbf{be} \:  \ottnt{a}  \: \mathbf{in} \:  \ottnt{b}   :^{  1  }  \ottnt{B} $ and $  \Gamma_{{\mathrm{1}}}  +  \Gamma_{{\mathrm{2}}}   <:   \ottnt{q}  \cdot  \Gamma'  $.\\ 
Since $q \neq 0$ and $ \ottnt{q_{{\mathrm{0}}}}  <:   1  $, therefore $ \ottnt{q}  \cdot  \ottnt{q_{{\mathrm{0}}}}  \neq 0$.\\
By IH, $\exists \Gamma'_{{\mathrm{1}}},\Gamma'_{{\mathrm{2}}}$ and $\ottnt{r'},\ottnt{q'}$ such that $ \Gamma'_{{\mathrm{1}}}  \vdash  \ottnt{a}  :^{  1  }   {}^{ \ottnt{r} }\!  \ottnt{A_{{\mathrm{1}}}}  \: \times \:  \ottnt{A_{{\mathrm{2}}}}  $ and $    \Gamma'_{{\mathrm{2}}}  ,   \ottmv{x}  :^{ \ottnt{r'} }  \ottnt{A_{{\mathrm{1}}}}     ,   \ottmv{y}  :^{ \ottnt{q'} }  \ottnt{A_{{\mathrm{2}}}}    \vdash  \ottnt{b}  :^{  1  }  \ottnt{B} $ and $ \Gamma_{{\mathrm{1}}}  <:    (   \ottnt{q}  \cdot  \ottnt{q_{{\mathrm{0}}}}   )   \cdot  \Gamma'_{{\mathrm{1}}}  $ and $ \Gamma_{{\mathrm{2}}}  <:   \ottnt{q}  \cdot  \Gamma'_{{\mathrm{2}}}  $ and $  \ottnt{q}  \cdot    \ottnt{q_{{\mathrm{0}}}}  \cdot  \ottnt{r}     <:   \ottnt{q}  \cdot  \ottnt{r'}  $ and $  \ottnt{q}  \cdot  \ottnt{q_{{\mathrm{0}}}}   <:   \ottnt{q}  \cdot  \ottnt{q'}  $.\\ 
Since $q \neq 0$, therefore $  \ottnt{q_{{\mathrm{0}}}}  \cdot  \ottnt{r}   <:  \ottnt{r'} $ and $ \ottnt{q_{{\mathrm{0}}}}  <:  \ottnt{q'} $.\\
By \rref{ST-SubL}, $    \Gamma'_{{\mathrm{2}}}  ,   \ottmv{x}  :^{  \ottnt{q_{{\mathrm{0}}}}  \cdot  \ottnt{r}  }  \ottnt{A_{{\mathrm{1}}}}     ,   \ottmv{y}  :^{ \ottnt{q_{{\mathrm{0}}}} }  \ottnt{A_{{\mathrm{2}}}}    \vdash  \ottnt{b}  :^{  1  }  \ottnt{B} $.\\
Again, by Lemma \ref{BLSMultP}, $  \ottnt{q_{{\mathrm{0}}}}  \cdot  \Gamma'_{{\mathrm{1}}}   \vdash  \ottnt{a}  :^{ \ottnt{q_{{\mathrm{0}}}} }   {}^{ \ottnt{r} }\!  \ottnt{A_{{\mathrm{1}}}}  \: \times \:  \ottnt{A_{{\mathrm{2}}}}  $.\\
By \rref{ST-LetPair}, $   \ottnt{q_{{\mathrm{0}}}}  \cdot  \Gamma'_{{\mathrm{1}}}   +  \Gamma'_{{\mathrm{2}}}   \vdash   \mathbf{let}_{ \ottnt{q_{{\mathrm{0}}}} } \: (  \ottmv{x} ^{ \ottnt{r} } ,  \ottmv{y}  ) \: \mathbf{be} \:  \ottnt{a}  \: \mathbf{in} \:  \ottnt{b}   :^{  1  }  \ottnt{B} $.\\
This case, then, follows by setting $\Gamma' :=   \ottnt{q_{{\mathrm{0}}}}  \cdot  \Gamma'_{{\mathrm{1}}}   +  \Gamma'_{{\mathrm{2}}} $.

\item \Rref{ST-Unit}. Have: $   0   \cdot  \Gamma   \vdash   \mathbf{unit}   :^{ \ottnt{q} }   \mathbf{Unit}  $.\\
Need to show: $\exists \Gamma'$ such that $ \Gamma'  \vdash   \mathbf{unit}   :^{  1  }   \mathbf{Unit}  $ and $   0   \cdot  \Gamma   <:   \ottnt{q}  \cdot  \Gamma'  $.\\
This case follows by setting $\Gamma' :=   0   \cdot  \Gamma $.

\item \Rref{ST-LetUnit}. Have: $  \Gamma_{{\mathrm{1}}}  +  \Gamma_{{\mathrm{2}}}   \vdash   \mathbf{let}_{ \ottnt{q_{{\mathrm{0}}}} } \: \mathbf{unit} \: \mathbf{be} \:  \ottnt{a}  \: \mathbf{in} \:  \ottnt{b}   :^{ \ottnt{q} }  \ottnt{B} $ where $ \Gamma_{{\mathrm{1}}}  \vdash  \ottnt{a}  :^{  \ottnt{q}  \cdot  \ottnt{q_{{\mathrm{0}}}}  }   \mathbf{Unit}  $ and $ \Gamma_{{\mathrm{2}}}  \vdash  \ottnt{b}  :^{ \ottnt{q} }  \ottnt{B} $ and $ \ottnt{q_{{\mathrm{0}}}}  <:   1  $.\\
Need to show: $\exists \Gamma'$ such that $ \Gamma'  \vdash   \mathbf{let}_{ \ottnt{q_{{\mathrm{0}}}} } \: \mathbf{unit} \: \mathbf{be} \:  \ottnt{a}  \: \mathbf{in} \:  \ottnt{b}   :^{  1  }  \ottnt{B} $ and $  \Gamma_{{\mathrm{1}}}  +  \Gamma_{{\mathrm{2}}}   <:   \ottnt{q}  \cdot  \Gamma'  $.\\
By IH, $\exists \Gamma'_{{\mathrm{1}}}, \Gamma'_{{\mathrm{2}}}$ such that $ \Gamma'_{{\mathrm{1}}}  \vdash  \ottnt{a}  :^{  1  }   \mathbf{Unit}  $ and $ \Gamma'_{{\mathrm{2}}}  \vdash  \ottnt{b}  :^{  1  }  \ottnt{B} $ and $ \Gamma_{{\mathrm{1}}}  <:    (   \ottnt{q}  \cdot  \ottnt{q_{{\mathrm{0}}}}   )   \cdot  \Gamma'_{{\mathrm{1}}}  $ and $ \Gamma_{{\mathrm{2}}}  <:   \ottnt{q}  \cdot  \Gamma'_{{\mathrm{2}}}  $.\\
By Lemma \ref{BLSMultP}, $  \ottnt{q_{{\mathrm{0}}}}  \cdot  \Gamma'_{{\mathrm{1}}}   \vdash  \ottnt{a}  :^{ \ottnt{q_{{\mathrm{0}}}} }   \mathbf{Unit}  $. \\
By \rref{ST-LetUnit}, $    \ottnt{q_{{\mathrm{0}}}}  \cdot  \Gamma'_{{\mathrm{1}}}    +  \Gamma'_{{\mathrm{2}}}   \vdash   \mathbf{let}_{ \ottnt{q_{{\mathrm{0}}}} } \: \mathbf{unit} \: \mathbf{be} \:  \ottnt{a}  \: \mathbf{in} \:  \ottnt{b}   :^{  1  }  \ottnt{B} $.\\
This case follows by setting $\Gamma' :=   \ottnt{q_{{\mathrm{0}}}}  \cdot  \Gamma'_{{\mathrm{1}}}   +  \Gamma'_{{\mathrm{2}}} $.

\item \Rref{ST-Inj1}. Have $ \Gamma  \vdash   \mathbf{inj}_1 \:  \ottnt{a_{{\mathrm{1}}}}   :^{ \ottnt{q} }   \ottnt{A_{{\mathrm{1}}}}  +  \ottnt{A_{{\mathrm{2}}}}  $ where $ \Gamma  \vdash  \ottnt{a_{{\mathrm{1}}}}  :^{ \ottnt{q} }  \ottnt{A_{{\mathrm{1}}}} $.\\
Need to show: $\exists \Gamma'$ such that $ \Gamma'  \vdash   \mathbf{inj}_1 \:  \ottnt{a_{{\mathrm{1}}}}   :^{  1  }   \ottnt{A_{{\mathrm{1}}}}  +  \ottnt{A_{{\mathrm{2}}}}  $ and $ \Gamma  <:   \ottnt{q}  \cdot  \Gamma'  $.\\
By IH, $\exists \Gamma'_{{\mathrm{1}}}$ such that $ \Gamma'_{{\mathrm{1}}}  \vdash  \ottnt{a_{{\mathrm{1}}}}  :^{  1  }  \ottnt{A_{{\mathrm{1}}}} $ and $ \Gamma  <:   \ottnt{q}  \cdot  \Gamma'_{{\mathrm{1}}}  $.\\
This case follows by setting $\Gamma' := \Gamma'_{{\mathrm{1}}}$.

\item \Rref{ST-Inj2}. Similar to \rref{ST-Inj1}.

\item \Rref{ST-Case}. Have: $  \Gamma_{{\mathrm{1}}}  +  \Gamma_{{\mathrm{2}}}   \vdash   \mathbf{case}_{ \ottnt{q_{{\mathrm{0}}}} } \:  \ottnt{a}  \: \mathbf{of} \:  \ottmv{x_{{\mathrm{1}}}}  .  \ottnt{b_{{\mathrm{1}}}}  \: ; \:  \ottmv{x_{{\mathrm{2}}}}  .  \ottnt{b_{{\mathrm{2}}}}   :^{ \ottnt{q} }  \ottnt{B} $ where $ \Gamma_{{\mathrm{1}}}  \vdash  \ottnt{a}  :^{  \ottnt{q}  \cdot  \ottnt{q_{{\mathrm{0}}}}  }   \ottnt{A_{{\mathrm{1}}}}  +  \ottnt{A_{{\mathrm{2}}}}  $ and $  \Gamma_{{\mathrm{2}}}  ,   \ottmv{x_{{\mathrm{1}}}}  :^{  \ottnt{q}  \cdot  \ottnt{q_{{\mathrm{0}}}}  }  \ottnt{A_{{\mathrm{1}}}}    \vdash  \ottnt{b_{{\mathrm{1}}}}  :^{ \ottnt{q} }  \ottnt{B} $ and $  \Gamma_{{\mathrm{2}}}  ,   \ottmv{x_{{\mathrm{2}}}}  :^{  \ottnt{q}  \cdot  \ottnt{q_{{\mathrm{0}}}}  }  \ottnt{A_{{\mathrm{2}}}}    \vdash  \ottnt{b_{{\mathrm{2}}}}  :^{ \ottnt{q} }  \ottnt{B} $ and $ \ottnt{q_{{\mathrm{0}}}}  <:   1  $.\\
Need to show: $\exists \Gamma'$ such that $ \Gamma'  \vdash   \mathbf{case}_{ \ottnt{q_{{\mathrm{0}}}} } \:  \ottnt{a}  \: \mathbf{of} \:  \ottmv{x_{{\mathrm{1}}}}  .  \ottnt{b_{{\mathrm{1}}}}  \: ; \:  \ottmv{x_{{\mathrm{2}}}}  .  \ottnt{b_{{\mathrm{2}}}}   :^{  1  }  \ottnt{B} $ and $  \Gamma_{{\mathrm{1}}}  +  \Gamma_{{\mathrm{2}}}   <:   \ottnt{q}  \cdot  \Gamma'  $.\\
By IH, $\exists \Gamma'_{{\mathrm{1}}}, \Gamma'_{{\mathrm{21}}}, \Gamma'_{{\mathrm{22}}}$ and $\ottnt{q'_{{\mathrm{1}}}}, \ottnt{q'_{{\mathrm{2}}}}$ such that $ \Gamma'_{{\mathrm{1}}}  \vdash  \ottnt{a}  :^{  1  }   \ottnt{A_{{\mathrm{1}}}}  +  \ottnt{A_{{\mathrm{2}}}}  $ and $  \Gamma'_{{\mathrm{21}}}  ,   \ottmv{x_{{\mathrm{1}}}}  :^{ \ottnt{q'_{{\mathrm{1}}}} }  \ottnt{A_{{\mathrm{1}}}}    \vdash  \ottnt{b_{{\mathrm{1}}}}  :^{  1  }  \ottnt{B} $ and $  \Gamma'_{{\mathrm{22}}}  ,   \ottmv{x_{{\mathrm{2}}}}  :^{ \ottnt{q'_{{\mathrm{2}}}} }  \ottnt{A_{{\mathrm{2}}}}    \vdash  \ottnt{b_{{\mathrm{2}}}}  :^{  1  }  \ottnt{B} $ and $ \Gamma_{{\mathrm{1}}}  <:    (   \ottnt{q}  \cdot  \ottnt{q_{{\mathrm{0}}}}   )   \cdot  \Gamma'_{{\mathrm{1}}}  $ and $ \Gamma_{{\mathrm{2}}}  <:   \ottnt{q}  \cdot  \Gamma'_{{\mathrm{21}}}  $ and $ \Gamma_{{\mathrm{2}}}  <:   \ottnt{q}  \cdot  \Gamma'_{{\mathrm{22}}}  $ and $  \ottnt{q}  \cdot  \ottnt{q_{{\mathrm{0}}}}   <:   \ottnt{q}  \cdot  \ottnt{q'_{{\mathrm{1}}}}  $ and $  \ottnt{q}  \cdot  \ottnt{q_{{\mathrm{0}}}}   <:   \ottnt{q}  \cdot  \ottnt{q'_{{\mathrm{2}}}}  $.\\
Since $q \neq 0$, therefore $ \ottnt{q_{{\mathrm{0}}}}  <:  \ottnt{q'_{{\mathrm{1}}}} $ and $ \ottnt{q_{{\mathrm{0}}}}  <:  \ottnt{q'_{{\mathrm{2}}}} $.\\
The remainder of the proof for this case is different for $ \mathbb{N}_{=} $ and $ \mathbb{N}_{\geq} $.
\begin{itemize}
\item In case of $ \mathbb{N}_{=} $, since $ <: $ is discrete, $\Gamma'_{{\mathrm{21}}} = \Gamma'_{{\mathrm{22}}} = \Gamma'_{{\mathrm{2}}}$ (say).\\
Then, $ \Gamma_{{\mathrm{2}}}  <:   \ottnt{q}  \cdot  \Gamma'_{{\mathrm{2}}}  $. Further, by \rref{ST-SubL}, $  \Gamma'_{{\mathrm{2}}}  ,   \ottmv{x_{{\mathrm{1}}}}  :^{ \ottnt{q_{{\mathrm{0}}}} }  \ottnt{A_{{\mathrm{1}}}}    \vdash  \ottnt{b_{{\mathrm{1}}}}  :^{  1  }  \ottnt{B} $ and $  \Gamma'_{{\mathrm{2}}}  ,   \ottmv{x_{{\mathrm{2}}}}  :^{ \ottnt{q_{{\mathrm{0}}}} }  \ottnt{A_{{\mathrm{2}}}}    \vdash  \ottnt{b_{{\mathrm{2}}}}  :^{  1  }  \ottnt{B} $.\\
By Lemma \ref{BLSMultP}, $  \ottnt{q_{{\mathrm{0}}}}  \cdot  \Gamma'_{{\mathrm{1}}}   \vdash  \ottnt{a}  :^{ \ottnt{q_{{\mathrm{0}}}} }   \ottnt{A_{{\mathrm{1}}}}  +  \ottnt{A_{{\mathrm{2}}}}  $. \\
This case, then, follows by setting $\Gamma' :=   \ottnt{q_{{\mathrm{0}}}}  \cdot  \Gamma'_{{\mathrm{1}}}   +  \Gamma'_{{\mathrm{2}}} $.
\item In case of $ \mathbb{N}_{\geq} $, $\Gamma'_{{\mathrm{21}}}$ may not be equal to $\Gamma'_{{\mathrm{22}}}$. So we need to employ the following construction. For quantities $\ottnt{q_{{\mathrm{1}}}}, \ottnt{q_{{\mathrm{2}}}}$ define an operator: $ \ottnt{q_{{\mathrm{1}}}}  ;  \ottnt{q_{{\mathrm{2}}}}  = \ottnt{q_{{\mathrm{1}}}}$ if $ \ottnt{q_{{\mathrm{1}}}}  <:  \ottnt{q_{{\mathrm{2}}}} $ and $\ottnt{q_{{\mathrm{2}}}}$ otherwise. 
Extending the operator to contexts $\Gamma_{{\mathrm{1}}}$ and $\Gamma_{{\mathrm{2}}}$ where $ \lfloor  \Gamma_{{\mathrm{1}}}  \rfloor  =  \lfloor  \Gamma_{{\mathrm{2}}}  \rfloor $, define $\Gamma_{{\mathrm{0}}} := \Gamma_{{\mathrm{1}}} ; \Gamma_{{\mathrm{2}}}$ as $ \lfloor  \Gamma_{{\mathrm{0}}}  \rfloor  =  \lfloor  \Gamma_{{\mathrm{1}}}  \rfloor $ and $ \overline{ \Gamma_{{\mathrm{0}}} }  =  \overline{ \Gamma_{{\mathrm{1}}} }  ;  \overline{ \Gamma_{{\mathrm{2}}} } $ (defined pointwise).\\
Now, let $\Gamma'_{{\mathrm{2}}} := \Gamma'_{{\mathrm{21}}} ; \Gamma'_{{\mathrm{22}}}$. Then, $ \Gamma_{{\mathrm{2}}}  <:   \ottnt{q}  \cdot  \Gamma'_{{\mathrm{2}}}  $.\\ 
Again, by \rref{ST-SubL}, $  \Gamma'_{{\mathrm{2}}}  ,   \ottmv{x_{{\mathrm{1}}}}  :^{ \ottnt{q_{{\mathrm{0}}}} }  \ottnt{A_{{\mathrm{1}}}}    \vdash  \ottnt{b_{{\mathrm{1}}}}  :^{  1  }  \ottnt{B} $ and $  \Gamma'_{{\mathrm{2}}}  ,   \ottmv{x_{{\mathrm{2}}}}  :^{ \ottnt{q_{{\mathrm{0}}}} }  \ottnt{A_{{\mathrm{2}}}}    \vdash  \ottnt{b_{{\mathrm{2}}}}  :^{  1  }  \ottnt{B} $.\\
By Lemma \ref{BLSMultP}, $  \ottnt{q_{{\mathrm{0}}}}  \cdot  \Gamma'_{{\mathrm{1}}}   \vdash  \ottnt{a}  :^{ \ottnt{q_{{\mathrm{0}}}} }   \ottnt{A_{{\mathrm{1}}}}  +  \ottnt{A_{{\mathrm{2}}}}  $. \\
This case, then, follows by setting $\Gamma' :=   \ottnt{q_{{\mathrm{0}}}}  \cdot  \Gamma'_{{\mathrm{1}}}   +  \Gamma'_{{\mathrm{2}}} $.   
\end{itemize} 

\item \Rref{ST-SubL}. Have: $ \Gamma  \vdash  \ottnt{a}  :^{ \ottnt{q} }  \ottnt{A} $ where $ \Gamma_{{\mathrm{1}}}  \vdash  \ottnt{a}  :^{ \ottnt{q} }  \ottnt{A} $ and $ \Gamma  <:  \Gamma_{{\mathrm{1}}} $. \\
Need to show: $\exists \Gamma'$ such that $ \Gamma'  \vdash  \ottnt{a}  :^{  1  }  \ottnt{A} $ and $ \Gamma  <:   \ottnt{q}  \cdot  \Gamma'  $.\\
By IH, $\exists \Gamma'_{{\mathrm{1}}}$ such that $ \Gamma'_{{\mathrm{1}}}  \vdash  \ottnt{a}  :^{  1  }  \ottnt{A} $ and $ \Gamma_{{\mathrm{1}}}  <:   \ottnt{q}  \cdot  \Gamma'_{{\mathrm{1}}}  $.\\
This case, then, follows by setting $\Gamma' := \Gamma'_{{\mathrm{1}}}$.

\item \Rref{ST-SubR}. Have: $ \Gamma  \vdash  \ottnt{a}  :^{ \ottnt{q'} }  \ottnt{A} $ where $ \Gamma  \vdash  \ottnt{a}  :^{ \ottnt{q} }  \ottnt{A} $ and $ \ottnt{q}  <:  \ottnt{q'} $.\\
Need to show: $\exists \Gamma'$ such that $ \Gamma'  \vdash  \ottnt{a}  :^{  1  }  \ottnt{A} $ and $ \Gamma  <:   \ottnt{q'}  \cdot  \Gamma'  $.\\
Since $\ottnt{q'} \neq 0$, therefore $\ottnt{q} \neq 0$.\\
By IH, $\exists \Gamma'_{{\mathrm{1}}}$ such that $ \Gamma'_{{\mathrm{1}}}  \vdash  \ottnt{a}  :^{  1  }  \ottnt{A} $ and $ \Gamma  <:   \ottnt{q}  \cdot  \Gamma'_{{\mathrm{1}}}  $.\\
Now, since $ \ottnt{q}  <:  \ottnt{q'} $, therefore $  \ottnt{q}  \cdot  \Gamma'_{{\mathrm{1}}}   <:   \ottnt{q'}  \cdot  \Gamma'_{{\mathrm{1}}}  $.\\
This case, then, follows by setting $\Gamma' := \Gamma'_{{\mathrm{1}}}$.

\end{itemize}
\end{proof}

%--------------------------------------------------------------------------------------------

\begin{lemma}[Splitting (Lemma \ref{BLSSplit})] \label{BLSSplitP}
If $ \Gamma  \vdash  \ottnt{a}  :^{  \ottnt{q_{{\mathrm{1}}}}  +  \ottnt{q_{{\mathrm{2}}}}  }  \ottnt{A} $, then there exists $\Gamma_{{\mathrm{1}}}$ and $\Gamma_{{\mathrm{2}}}$ such that $ \Gamma_{{\mathrm{1}}}  \vdash  \ottnt{a}  :^{ \ottnt{q_{{\mathrm{1}}}} }  \ottnt{A} $ and $ \Gamma_{{\mathrm{2}}}  \vdash  \ottnt{a}  :^{ \ottnt{q_{{\mathrm{2}}}} }  \ottnt{A} $ and $ \Gamma  =   \Gamma_{{\mathrm{1}}}  +  \Gamma_{{\mathrm{2}}}  $. 
\end{lemma}

\begin{proof}
If $ \ottnt{q_{{\mathrm{1}}}}  +  \ottnt{q_{{\mathrm{2}}}}  = 0$, then $\Gamma_{{\mathrm{1}}} :=   0   \cdot  \Gamma $ and $\Gamma_{{\mathrm{2}}} := \Gamma$.\\
Otherwise, by Lemma \ref{BLSFactP}, $\exists \Gamma'$ such that $ \Gamma'  \vdash  \ottnt{a}  :^{  1  }  \ottnt{A} $ and $ \Gamma  <:    (   \ottnt{q_{{\mathrm{1}}}}  +  \ottnt{q_{{\mathrm{2}}}}   )   \cdot  \Gamma'  $.\\
Then, $\Gamma =    (   \ottnt{q_{{\mathrm{1}}}}  +  \ottnt{q_{{\mathrm{2}}}}   )   \cdot  \Gamma'   +  \Gamma_{{\mathrm{0}}} $ for some $\Gamma_{{\mathrm{0}}}$ (in case of $ \mathbb{N}_{=} $, we have, $ \overline{ \Gamma_{{\mathrm{0}}} }  = \overline{0}$).\\
Now, by Lemma \ref{BLSMultP}, $  \ottnt{q_{{\mathrm{1}}}}  \cdot  \Gamma'   \vdash  \ottnt{a}  :^{ \ottnt{q_{{\mathrm{1}}}} }  \ottnt{A} $ and $  \ottnt{q_{{\mathrm{2}}}}  \cdot  \Gamma'   \vdash  \ottnt{a}  :^{ \ottnt{q_{{\mathrm{2}}}} }  \ottnt{A} $.\\
By \rref{ST-SubL}, $   \ottnt{q_{{\mathrm{2}}}}  \cdot  \Gamma'   +  \Gamma_{{\mathrm{0}}}   \vdash  \ottnt{a}  :^{ \ottnt{q_{{\mathrm{2}}}} }  \ottnt{A} $.\\
The lemma follows by setting $\Gamma_{{\mathrm{1}}} :=  \ottnt{q_{{\mathrm{1}}}}  \cdot  \Gamma' $ and $\Gamma_{{\mathrm{2}}} :=   \ottnt{q_{{\mathrm{2}}}}  \cdot  \Gamma'   +  \Gamma_{{\mathrm{0}}} $.
\end{proof}

%-------------------------------------------------------------------------------------------

\begin{lemma}[Weakening (Lemma \ref{BLSWeak})] \label{BLSWeakP}
If $  \Gamma_{{\mathrm{1}}}  ,  \Gamma_{{\mathrm{2}}}   \vdash  \ottnt{a}  :^{ \ottnt{q} }  \ottnt{A} $, then $    \Gamma_{{\mathrm{1}}}  ,   \ottmv{z}  :^{  0  }  \ottnt{C}     ,  \Gamma_{{\mathrm{2}}}   \vdash  \ottnt{a}  :^{ \ottnt{q} }  \ottnt{A} $.
\end{lemma}

\begin{proof}
By induction on $  \Gamma_{{\mathrm{1}}}  ,  \Gamma_{{\mathrm{2}}}   \vdash  \ottnt{a}  :^{ \ottnt{q} }  \ottnt{A} $.
\end{proof}

%--------------------------------------------------------------------------------------------

\begin{lemma}[Substitution (Lemma \ref{BLSSubst})] \label{BLSSubstP}
If $    \Gamma_{{\mathrm{1}}}  ,   \ottmv{z}  :^{ \ottnt{r_{{\mathrm{0}}}} }  \ottnt{C}     ,  \Gamma_{{\mathrm{2}}}   \vdash  \ottnt{a}  :^{ \ottnt{q} }  \ottnt{A} $ and $ \Gamma  \vdash  \ottnt{c}  :^{ \ottnt{r_{{\mathrm{0}}}} }  \ottnt{C} $ and $  \lfloor  \Gamma_{{\mathrm{1}}}  \rfloor   =   \lfloor  \Gamma  \rfloor  $, then $    \Gamma_{{\mathrm{1}}}  +  \Gamma    ,  \Gamma_{{\mathrm{2}}}   \vdash   \ottnt{a}  \{  \ottnt{c}  /  \ottmv{z}  \}   :^{ \ottnt{q} }  \ottnt{A} $. 
\end{lemma}

\begin{proof}
By induction on $    \Gamma_{{\mathrm{1}}}  ,   \ottmv{z}  :^{ \ottnt{r_{{\mathrm{0}}}} }  \ottnt{C}     ,  \Gamma_{{\mathrm{2}}}   \vdash  \ottnt{a}  :^{ \ottnt{q} }  \ottnt{A} $.
\begin{itemize}

\item \Rref{ST-Var}. There are three cases to consider.
\begin{itemize}
\item $          0   \cdot  \Gamma_{{\mathrm{11}}}   ,   \ottmv{z}  :^{  0  }  \ottnt{C}     ,    0   \cdot  \Gamma_{{\mathrm{12}}}     ,   \ottmv{x}  :^{ \ottnt{q} }  \ottnt{A}     ,    0   \cdot  \Gamma_{{\mathrm{2}}}    \vdash   \ottmv{x}   :^{ \ottnt{q} }  \ottnt{A} $. Also, $ \Gamma  \vdash  \ottnt{c}  :^{  0  }  \ottnt{C} $ where $  \lfloor  \Gamma_{{\mathrm{11}}}  \rfloor   =   \lfloor  \Gamma  \rfloor  $.\\
Need to show: $      \Gamma  ,    0   \cdot  \Gamma_{{\mathrm{12}}}     ,   \ottmv{x}  :^{ \ottnt{q} }  \ottnt{A}     ,    0   \cdot  \Gamma_{{\mathrm{2}}}    \vdash   \ottmv{x}   :^{ \ottnt{q} }  \ottnt{A} $.\\
In case of $ \mathbb{N}_{=} $, we have, $ \overline{ \Gamma }  = \overline{0}$. This case, then, follows by \rref{ST-Var}.\\ In case of $ \mathbb{N}_{\geq} $, this case follows by \rref{ST-Var} and \rref{ST-SubL}.

\item $      0   \cdot  \Gamma_{{\mathrm{1}}}   ,   \ottmv{x}  :^{ \ottnt{q} }  \ottnt{A}     ,    0   \cdot  \Gamma_{{\mathrm{2}}}    \vdash   \ottmv{x}   :^{ \ottnt{q} }  \ottnt{A} $. Also, $ \Gamma  \vdash  \ottnt{a}  :^{ \ottnt{q} }  \ottnt{A} $ where $  \lfloor  \Gamma_{{\mathrm{1}}}  \rfloor   =   \lfloor  \Gamma  \rfloor  $.\\
Need to show: $  \Gamma  ,    0   \cdot  \Gamma_{{\mathrm{2}}}    \vdash  \ottnt{a}  :^{ \ottnt{q} }  \ottnt{A} $.\\
Follows by lemma \ref{BLSWeakP}.

\item $      0   \cdot  \Gamma_{{\mathrm{1}}}   ,   \ottmv{x}  :^{ \ottnt{q} }  \ottnt{A}     ,        0   \cdot  \Gamma_{{\mathrm{21}}}   ,   \ottmv{z}  :^{  0  }  \ottnt{C}     ,    0   \cdot  \Gamma_{{\mathrm{22}}}      \vdash   \ottmv{x}   :^{ \ottnt{q} }  \ottnt{A} $. Also, $    \Gamma_{{\mathrm{31}}}  ,   \ottmv{x}  :^{ \ottnt{r} }  \ottnt{A}     ,  \Gamma_{{\mathrm{32}}}   \vdash  \ottnt{c}  :^{  0  }  \ottnt{C} $ where $  \lfloor  \Gamma_{{\mathrm{31}}}  \rfloor   =   \lfloor  \Gamma_{{\mathrm{1}}}  \rfloor  $ and $  \lfloor  \Gamma_{{\mathrm{32}}}  \rfloor   =   \lfloor  \Gamma_{{\mathrm{21}}}  \rfloor  $.\\
Need to show: $    \Gamma_{{\mathrm{31}}}  ,   \ottmv{x}  :^{  (   \ottnt{q}  +  \ottnt{r}   )  }  \ottnt{A}     ,    \Gamma_{{\mathrm{32}}}  ,    0   \cdot  \Gamma_{{\mathrm{22}}}      \vdash   \ottmv{x}   :^{ \ottnt{q} }  \ottnt{A} $.\\
In case of $ \mathbb{N}_{=} $, we have, $ \overline{ \Gamma_{{\mathrm{31}}} }  = \overline{0}$ and $ \overline{ \Gamma_{{\mathrm{32}}} }  = \overline{0}$ and $r = 0$. This case, then, follows by \rref{ST-Var}.\\
In case of $ \mathbb{N}_{\geq} $, this case follows by \rref{ST-Var} and \rref{ST-SubL}. 

\end{itemize}

\item \Rref{ST-Lam}. Have: $    \Gamma_{{\mathrm{1}}}  ,   \ottmv{z}  :^{ \ottnt{r_{{\mathrm{0}}}} }  \ottnt{C}     ,  \Gamma_{{\mathrm{2}}}   \vdash   \lambda^{ \ottnt{r} }  \ottmv{x}  :  \ottnt{A}  .  \ottnt{b}   :^{ \ottnt{q} }   {}^{ \ottnt{r} }\!  \ottnt{A}  \to  \ottnt{B}  $ where $    \Gamma_{{\mathrm{1}}}  ,   \ottmv{z}  :^{ \ottnt{r_{{\mathrm{0}}}} }  \ottnt{C}     ,    \Gamma_{{\mathrm{2}}}  ,   \ottmv{x}  :^{  \ottnt{q}  \cdot  \ottnt{r}  }  \ottnt{A}      \vdash  \ottnt{b}  :^{ \ottnt{q} }  \ottnt{B} $. Also, $ \Gamma  \vdash  \ottnt{c}  :^{ \ottnt{r_{{\mathrm{0}}}} }  \ottnt{C} $ where $  \lfloor  \Gamma  \rfloor   =   \lfloor  \Gamma_{{\mathrm{1}}}  \rfloor  $.\\
Need to show: $    \Gamma_{{\mathrm{1}}}  +  \Gamma    ,  \Gamma_{{\mathrm{2}}}   \vdash    \lambda^{ \ottnt{r} }  \ottmv{x}  :  \ottnt{A}  .  \ottnt{b}   \{  \ottnt{c}  /  \ottmv{z}  \}   :^{ \ottnt{q} }   {}^{ \ottnt{r} }\!  \ottnt{A}  \to  \ottnt{B}  $.\\
Follows by IH and \rref{ST-Lam}. 

\item \Rref{ST-App}. Have: $      \Gamma_{{\mathrm{11}}}  +  \Gamma_{{\mathrm{12}}}    ,   \ottmv{z}  :^{  \ottnt{r_{{\mathrm{01}}}}  +  \ottnt{r_{{\mathrm{02}}}}  }  \ottnt{C}     ,    \Gamma_{{\mathrm{21}}}  +  \Gamma_{{\mathrm{22}}}     \vdash   \ottnt{b}  \:  \ottnt{a} ^{ \ottnt{r} }   :^{ \ottnt{q} }  \ottnt{B} $ where $    \Gamma_{{\mathrm{11}}}  ,   \ottmv{z}  :^{ \ottnt{r_{{\mathrm{01}}}} }  \ottnt{C}     ,  \Gamma_{{\mathrm{21}}}   \vdash  \ottnt{b}  :^{ \ottnt{q} }   {}^{ \ottnt{r} }\!  \ottnt{A}  \to  \ottnt{B}  $ and $    \Gamma_{{\mathrm{12}}}  ,   \ottmv{z}  :^{ \ottnt{r_{{\mathrm{02}}}} }  \ottnt{C}     ,  \Gamma_{{\mathrm{22}}}   \vdash  \ottnt{a}  :^{  \ottnt{q}  \cdot  \ottnt{r}  }  \ottnt{A} $. Also, $ \Gamma  \vdash  \ottnt{c}  :^{  \ottnt{r_{{\mathrm{01}}}}  +  \ottnt{r_{{\mathrm{02}}}}  }  \ottnt{C} $ where $  \lfloor  \Gamma  \rfloor   =   \lfloor  \Gamma_{{\mathrm{11}}}  \rfloor  $.\\
Need to show: $      \Gamma_{{\mathrm{11}}}  +  \Gamma_{{\mathrm{12}}}    +  \Gamma    ,    \Gamma_{{\mathrm{21}}}  +  \Gamma_{{\mathrm{22}}}     \vdash    \ottnt{b}  \{  \ottnt{c}  /  \ottmv{z}  \}   \:   \ottnt{a}  \{  \ottnt{c}  /  \ottmv{z}  \}  ^{ \ottnt{r} }   :^{ \ottnt{q} }  \ottnt{B} $.\\
By lemma \ref{BLSSplitP}, $\exists \Gamma_{{\mathrm{31}}}, \Gamma_{{\mathrm{32}}}$ such that $ \Gamma_{{\mathrm{31}}}  \vdash  \ottnt{c}  :^{ \ottnt{r_{{\mathrm{01}}}} }  \ottnt{C} $ and $ \Gamma_{{\mathrm{32}}}  \vdash  \ottnt{c}  :^{ \ottnt{r_{{\mathrm{02}}}} }  \ottnt{C} $ and $\Gamma =  \Gamma_{{\mathrm{31}}}  +  \Gamma_{{\mathrm{32}}} $.\\
By IH, $    \Gamma_{{\mathrm{11}}}  +  \Gamma_{{\mathrm{31}}}    ,  \Gamma_{{\mathrm{21}}}   \vdash   \ottnt{b}  \{  \ottnt{c}  /  \ottmv{z}  \}   :^{ \ottnt{q} }   {}^{ \ottnt{r} }\!  \ottnt{A}  \to  \ottnt{B}  $ and $    \Gamma_{{\mathrm{12}}}  +  \Gamma_{{\mathrm{32}}}    ,  \Gamma_{{\mathrm{22}}}   \vdash   \ottnt{a}  \{  \ottnt{c}  /  \ottmv{z}  \}   :^{  \ottnt{q}  \cdot  \ottnt{r}  }  \ottnt{A} $.\\
This case, then, follows by \rref{ST-App}.

\item \Rref{ST-Pair}. Have: $      \Gamma_{{\mathrm{11}}}  +  \Gamma_{{\mathrm{12}}}    ,   \ottmv{z}  :^{  \ottnt{r_{{\mathrm{01}}}}  +  \ottnt{r_{{\mathrm{02}}}}  }  \ottnt{C}     ,    \Gamma_{{\mathrm{21}}}  +  \Gamma_{{\mathrm{22}}}     \vdash   (  \ottnt{a_{{\mathrm{1}}}} ^{ \ottnt{r} } ,  \ottnt{a_{{\mathrm{2}}}}  )   :^{ \ottnt{q} }   {}^{ \ottnt{r} }\!  \ottnt{A_{{\mathrm{1}}}}  \: \times \:  \ottnt{A_{{\mathrm{2}}}}  $ where $    \Gamma_{{\mathrm{11}}}  ,   \ottmv{z}  :^{ \ottnt{r_{{\mathrm{01}}}} }  \ottnt{C}     ,  \Gamma_{{\mathrm{21}}}   \vdash  \ottnt{a_{{\mathrm{1}}}}  :^{  \ottnt{q}  \cdot  \ottnt{r}  }  \ottnt{A_{{\mathrm{1}}}} $ and $    \Gamma_{{\mathrm{12}}}  ,   \ottmv{z}  :^{ \ottnt{r_{{\mathrm{02}}}} }  \ottnt{C}     ,  \Gamma_{{\mathrm{22}}}   \vdash  \ottnt{a_{{\mathrm{2}}}}  :^{ \ottnt{q} }  \ottnt{A_{{\mathrm{2}}}} $. Also, $ \Gamma  \vdash  \ottnt{c}  :^{  \ottnt{r_{{\mathrm{01}}}}  +  \ottnt{r_{{\mathrm{02}}}}  }  \ottnt{C} $ where $  \lfloor  \Gamma  \rfloor   =   \lfloor  \Gamma_{{\mathrm{11}}}  \rfloor  $.\\
Need to show: $      \Gamma_{{\mathrm{11}}}  +  \Gamma_{{\mathrm{12}}}    +  \Gamma    ,    \Gamma_{{\mathrm{21}}}  +  \Gamma_{{\mathrm{22}}}     \vdash   (   \ottnt{a_{{\mathrm{1}}}}  \{  \ottnt{c}  /  \ottmv{z}  \}  ^{ \ottnt{r} } ,   \ottnt{a_{{\mathrm{2}}}}  \{  \ottnt{c}  /  \ottmv{z}  \}   )   :^{ \ottnt{q} }   {}^{ \ottnt{r} }\!  \ottnt{A_{{\mathrm{1}}}}  \: \times \:  \ottnt{A_{{\mathrm{2}}}}  $.\\
By lemma \ref{BLSSplitP}, $\exists \Gamma_{{\mathrm{31}}}, \Gamma_{{\mathrm{32}}}$ such that $ \Gamma_{{\mathrm{31}}}  \vdash  \ottnt{c}  :^{ \ottnt{r_{{\mathrm{01}}}} }  \ottnt{C} $ and $ \Gamma_{{\mathrm{32}}}  \vdash  \ottnt{c}  :^{ \ottnt{r_{{\mathrm{02}}}} }  \ottnt{C} $ and $\Gamma =  \Gamma_{{\mathrm{31}}}  +  \Gamma_{{\mathrm{32}}} $.\\
By IH, $    \Gamma_{{\mathrm{11}}}  +  \Gamma_{{\mathrm{31}}}    ,  \Gamma_{{\mathrm{21}}}   \vdash   \ottnt{a_{{\mathrm{1}}}}  \{  \ottnt{c}  /  \ottmv{z}  \}   :^{  \ottnt{q}  \cdot  \ottnt{r}  }  \ottnt{A_{{\mathrm{1}}}} $ and $    \Gamma_{{\mathrm{12}}}  +  \Gamma_{{\mathrm{32}}}    ,  \Gamma_{{\mathrm{22}}}   \vdash   \ottnt{a_{{\mathrm{2}}}}  \{  \ottnt{c}  /  \ottmv{z}  \}   :^{ \ottnt{q} }  \ottnt{A_{{\mathrm{2}}}} $.\\
This case, then, follows by \rref{ST-Pair}.

\item \Rref{ST-LetPair}. Have: $      \Gamma_{{\mathrm{11}}}  +  \Gamma_{{\mathrm{12}}}    ,   \ottmv{z}  :^{  \ottnt{r_{{\mathrm{01}}}}  +  \ottnt{r_{{\mathrm{02}}}}  }  \ottnt{C}     ,    \Gamma_{{\mathrm{21}}}  +  \Gamma_{{\mathrm{22}}}     \vdash   \mathbf{let}_{ \ottnt{q_{{\mathrm{0}}}} } \: (  \ottmv{x} ^{ \ottnt{r} } ,  \ottmv{y}  ) \: \mathbf{be} \:  \ottnt{a}  \: \mathbf{in} \:  \ottnt{b}   :^{ \ottnt{q} }  \ottnt{B} $ where $    \Gamma_{{\mathrm{11}}}  ,   \ottmv{z}  :^{ \ottnt{r_{{\mathrm{01}}}} }  \ottnt{C}     ,  \Gamma_{{\mathrm{21}}}   \vdash  \ottnt{a}  :^{  \ottnt{q}  \cdot  \ottnt{q_{{\mathrm{0}}}}  }   {}^{ \ottnt{r} }\!  \ottnt{A_{{\mathrm{1}}}}  \: \times \:  \ottnt{A_{{\mathrm{2}}}}  $ and $    \Gamma_{{\mathrm{12}}}  ,   \ottmv{z}  :^{ \ottnt{r_{{\mathrm{02}}}} }  \ottnt{C}     ,    \Gamma_{{\mathrm{22}}}  ,     \ottmv{x}  :^{  \ottnt{q}  \cdot   \ottnt{q_{{\mathrm{0}}}}  \cdot  \ottnt{r}   }  \ottnt{A_{{\mathrm{1}}}}   ,   \ottmv{y}  :^{  \ottnt{q}  \cdot  \ottnt{q_{{\mathrm{0}}}}  }  \ottnt{A_{{\mathrm{2}}}}        \vdash  \ottnt{b}  :^{ \ottnt{q} }  \ottnt{B} $. Also, $ \Gamma  \vdash  \ottnt{c}  :^{  \ottnt{r_{{\mathrm{01}}}}  +  \ottnt{r_{{\mathrm{02}}}}  }  \ottnt{C} $ where $  \lfloor  \Gamma  \rfloor   =   \lfloor  \Gamma_{{\mathrm{11}}}  \rfloor  $.\\
Need to show: $      \Gamma_{{\mathrm{11}}}  +  \Gamma_{{\mathrm{12}}}    +  \Gamma    ,    \Gamma_{{\mathrm{21}}}  +  \Gamma_{{\mathrm{22}}}     \vdash    \mathbf{let}_{ \ottnt{q_{{\mathrm{0}}}} } \: (  \ottmv{x} ^{ \ottnt{r} } ,  \ottmv{y}  ) \: \mathbf{be} \:   \ottnt{a}  \{  \ottnt{c}  /  \ottmv{z}  \}   \: \mathbf{in} \:  \ottnt{b}   \{  \ottnt{c}  /  \ottmv{z}  \}   :^{ \ottnt{q} }  \ottnt{B} $.\\
By lemma \ref{BLSSplitP}, $\exists \Gamma_{{\mathrm{31}}}, \Gamma_{{\mathrm{32}}}$ such that $ \Gamma_{{\mathrm{31}}}  \vdash  \ottnt{c}  :^{ \ottnt{r_{{\mathrm{01}}}} }  \ottnt{C} $ and $ \Gamma_{{\mathrm{32}}}  \vdash  \ottnt{c}  :^{ \ottnt{r_{{\mathrm{02}}}} }  \ottnt{C} $ and $\Gamma =  \Gamma_{{\mathrm{31}}}  +  \Gamma_{{\mathrm{32}}} $.\\
By IH, $    \Gamma_{{\mathrm{11}}}  +  \Gamma_{{\mathrm{31}}}    ,  \Gamma_{{\mathrm{21}}}   \vdash   \ottnt{a}  \{  \ottnt{c}  /  \ottmv{z}  \}   :^{  \ottnt{q}  \cdot  \ottnt{q_{{\mathrm{0}}}}  }   {}^{ \ottnt{r} }\!  \ottnt{A_{{\mathrm{1}}}}  \: \times \:  \ottnt{A_{{\mathrm{2}}}}  $ and $    \Gamma_{{\mathrm{12}}}  +  \Gamma_{{\mathrm{32}}}    ,    \Gamma_{{\mathrm{22}}}  ,     \ottmv{x}  :^{    \ottnt{q}  \cdot  \ottnt{q_{{\mathrm{0}}}}    \cdot  \ottnt{r}  }  \ottnt{A_{{\mathrm{1}}}}   ,   \ottmv{y}  :^{  \ottnt{q}  \cdot  \ottnt{q_{{\mathrm{0}}}}  }  \ottnt{A_{{\mathrm{2}}}}        \vdash   \ottnt{b}  \{  \ottnt{c}  /  \ottmv{z}  \}   :^{ \ottnt{q} }  \ottnt{B} $.\\
This case, then, follows by \rref{ST-LetPair}.

\item \Rref{ST-Unit}. Have: $      0   \cdot  \Gamma_{{\mathrm{1}}}   ,   \ottmv{z}  :^{  0  }  \ottnt{C}     ,    0   \cdot  \Gamma_{{\mathrm{2}}}    \vdash   \mathbf{unit}   :^{ \ottnt{q} }   \mathbf{Unit}  $. Also $ \Gamma  \vdash  \ottnt{c}  :^{  0  }  \ottnt{C} $ where $ \lfloor  \Gamma  \rfloor  =  \lfloor  \Gamma_{{\mathrm{1}}}  \rfloor $.\\
Need to show: $  \Gamma  ,    0   \cdot  \Gamma_{{\mathrm{2}}}    \vdash   \mathbf{unit}   :^{ \ottnt{q} }   \mathbf{Unit}  $.\\
In case of $ \mathbb{N}_{=} $, we have, $ \overline{ \Gamma }  = \overline{0}$. This case, then, follows by \rref{ST-Unit}.\\ 
In case of $ \mathbb{N}_{\geq} $, this case follows by \rref{ST-Unit} and \rref{ST-SubL}.

\item \Rref{ST-LetUnit}. Have: $      \Gamma_{{\mathrm{11}}}  +  \Gamma_{{\mathrm{12}}}    ,   \ottmv{z}  :^{  \ottnt{r_{{\mathrm{01}}}}  +  \ottnt{r_{{\mathrm{02}}}}  }  \ottnt{C}     ,    \Gamma_{{\mathrm{21}}}  +  \Gamma_{{\mathrm{22}}}     \vdash   \mathbf{let}_{ \ottnt{q_{{\mathrm{0}}}} } \: \mathbf{unit} \: \mathbf{be} \:  \ottnt{a}  \: \mathbf{in} \:  \ottnt{b}   :^{ \ottnt{q} }  \ottnt{B} $ where $    \Gamma_{{\mathrm{11}}}  ,   \ottmv{z}  :^{ \ottnt{r_{{\mathrm{01}}}} }  \ottnt{C}     ,  \Gamma_{{\mathrm{21}}}   \vdash  \ottnt{a}  :^{  \ottnt{q}  \cdot  \ottnt{q_{{\mathrm{0}}}}  }   \mathbf{Unit}  $ and $    \Gamma_{{\mathrm{12}}}  ,   \ottmv{z}  :^{ \ottnt{r_{{\mathrm{02}}}} }  \ottnt{C}     ,  \Gamma_{{\mathrm{22}}}   \vdash  \ottnt{b}  :^{ \ottnt{q} }  \ottnt{B} $. Also, $ \Gamma  \vdash  \ottnt{c}  :^{  \ottnt{r_{{\mathrm{01}}}}  +  \ottnt{r_{{\mathrm{02}}}}  }  \ottnt{C} $ where $  \lfloor  \Gamma  \rfloor   =   \lfloor  \Gamma_{{\mathrm{11}}}  \rfloor  $.\\
Need to show: $      \Gamma_{{\mathrm{11}}}  +  \Gamma_{{\mathrm{12}}}    +  \Gamma    ,    \Gamma_{{\mathrm{21}}}  +  \Gamma_{{\mathrm{22}}}     \vdash    \mathbf{let}_{ \ottnt{q_{{\mathrm{0}}}} } \: \mathbf{unit} \: \mathbf{be} \:   \ottnt{a}  \{  \ottnt{c}  /  \ottmv{z}  \}   \: \mathbf{in} \:  \ottnt{b}   \{  \ottnt{c}  /  \ottmv{z}  \}   :^{ \ottnt{q} }  \ottnt{B} $.\\
By lemma \ref{BLSSplitP}, $\exists \Gamma_{{\mathrm{31}}}, \Gamma_{{\mathrm{32}}}$ such that $ \Gamma_{{\mathrm{31}}}  \vdash  \ottnt{c}  :^{ \ottnt{r_{{\mathrm{01}}}} }  \ottnt{C} $ and $ \Gamma_{{\mathrm{32}}}  \vdash  \ottnt{c}  :^{ \ottnt{r_{{\mathrm{02}}}} }  \ottnt{C} $ and $\Gamma =  \Gamma_{{\mathrm{31}}}  +  \Gamma_{{\mathrm{32}}} $.\\
By IH, $    \Gamma_{{\mathrm{11}}}  +  \Gamma_{{\mathrm{31}}}    ,  \Gamma_{{\mathrm{21}}}   \vdash   \ottnt{a}  \{  \ottnt{c}  /  \ottmv{z}  \}   :^{  \ottnt{q}  \cdot  \ottnt{q_{{\mathrm{0}}}}  }   \mathbf{Unit}  $ and $    \Gamma_{{\mathrm{12}}}  +  \Gamma_{{\mathrm{32}}}    ,  \Gamma_{{\mathrm{22}}}   \vdash   \ottnt{b}  \{  \ottnt{c}  /  \ottmv{z}  \}   :^{ \ottnt{q} }  \ottnt{B} $.\\
This case, then, follows by \rref{ST-LetUnit}. 

\item \Rref{ST-Inj1}. Have: $    \Gamma_{{\mathrm{1}}}  ,   \ottmv{z}  :^{ \ottnt{r_{{\mathrm{0}}}} }  \ottnt{C}     ,  \Gamma_{{\mathrm{2}}}   \vdash   \mathbf{inj}_1 \:  \ottnt{a_{{\mathrm{1}}}}   :^{ \ottnt{q} }   \ottnt{A_{{\mathrm{1}}}}  +  \ottnt{A_{{\mathrm{2}}}}  $ where $    \Gamma_{{\mathrm{1}}}  ,   \ottmv{z}  :^{ \ottnt{r_{{\mathrm{0}}}} }  \ottnt{C}     ,  \Gamma_{{\mathrm{2}}}   \vdash  \ottnt{a_{{\mathrm{1}}}}  :^{ \ottnt{q} }  \ottnt{A_{{\mathrm{1}}}} $. Also, $ \Gamma  \vdash  \ottnt{c}  :^{ \ottnt{r_{{\mathrm{0}}}} }  \ottnt{C} $ where $  \lfloor  \Gamma  \rfloor   =   \lfloor  \Gamma_{{\mathrm{1}}}  \rfloor  $.\\
Need to show: $    \Gamma_{{\mathrm{1}}}  +  \Gamma    ,  \Gamma_{{\mathrm{2}}}   \vdash    \mathbf{inj}_1 \:  \ottnt{a_{{\mathrm{1}}}}   \{  \ottnt{c}  /  \ottmv{z}  \}   :^{ \ottnt{q} }   \ottnt{A_{{\mathrm{1}}}}  +  \ottnt{A_{{\mathrm{2}}}}  $.\\
By IH, $    \Gamma_{{\mathrm{1}}}  +  \Gamma    ,  \Gamma_{{\mathrm{2}}}   \vdash   \ottnt{a_{{\mathrm{1}}}}  \{  \ottnt{c}  /  \ottmv{z}  \}   :^{ \ottnt{q} }  \ottnt{A_{{\mathrm{1}}}} $.\\
This case, then, follows by \rref{ST-Inj1}.

\item \Rref{ST-Inj2}. Similar to \rref{ST-Inj1}.

\item \Rref{ST-Case}. Have: $      \Gamma_{{\mathrm{11}}}  +  \Gamma_{{\mathrm{12}}}    ,   \ottmv{z}  :^{  \ottnt{r_{{\mathrm{01}}}}  +  \ottnt{r_{{\mathrm{02}}}}  }  \ottnt{C}     ,    \Gamma_{{\mathrm{21}}}  +  \Gamma_{{\mathrm{22}}}     \vdash   \mathbf{case}_{ \ottnt{q_{{\mathrm{0}}}} } \:  \ottnt{a}  \: \mathbf{of} \:  \ottmv{x_{{\mathrm{1}}}}  .  \ottnt{b_{{\mathrm{1}}}}  \: ; \:  \ottmv{x_{{\mathrm{2}}}}  .  \ottnt{b_{{\mathrm{2}}}}   :^{ \ottnt{q} }  \ottnt{B} $ where $    \Gamma_{{\mathrm{11}}}  ,   \ottmv{z}  :^{ \ottnt{r_{{\mathrm{01}}}} }  \ottnt{C}     ,  \Gamma_{{\mathrm{21}}}   \vdash  \ottnt{a}  :^{  \ottnt{q}  \cdot  \ottnt{q_{{\mathrm{0}}}}  }   \ottnt{A_{{\mathrm{1}}}}  +  \ottnt{A_{{\mathrm{2}}}}  $ and $    \Gamma_{{\mathrm{12}}}  ,   \ottmv{z}  :^{ \ottnt{r_{{\mathrm{02}}}} }  \ottnt{C}     ,    \Gamma_{{\mathrm{22}}}  ,   \ottmv{x_{{\mathrm{1}}}}  :^{  \ottnt{q}  \cdot  \ottnt{q_{{\mathrm{0}}}}  }  \ottnt{A_{{\mathrm{1}}}}      \vdash  \ottnt{b_{{\mathrm{1}}}}  :^{ \ottnt{q} }  \ottnt{B} $ and $    \Gamma_{{\mathrm{12}}}  ,   \ottmv{z}  :^{ \ottnt{r_{{\mathrm{02}}}} }  \ottnt{C}     ,    \Gamma_{{\mathrm{22}}}  ,   \ottmv{x_{{\mathrm{2}}}}  :^{  \ottnt{q}  \cdot  \ottnt{q_{{\mathrm{0}}}}  }  \ottnt{A_{{\mathrm{2}}}}      \vdash  \ottnt{b_{{\mathrm{2}}}}  :^{ \ottnt{q} }  \ottnt{B} $. Also, $ \Gamma  \vdash  \ottnt{c}  :^{  \ottnt{r_{{\mathrm{01}}}}  +  \ottnt{r_{{\mathrm{02}}}}  }  \ottnt{C} $ where $  \lfloor  \Gamma  \rfloor   =   \lfloor  \Gamma_{{\mathrm{11}}}  \rfloor  $.\\
Need to show: $      \Gamma_{{\mathrm{11}}}  +  \Gamma_{{\mathrm{12}}}    +  \Gamma    ,    \Gamma_{{\mathrm{21}}}  +  \Gamma_{{\mathrm{22}}}     \vdash    \mathbf{case}_{ \ottnt{q_{{\mathrm{0}}}} } \:   \ottnt{a}  \{  \ottnt{c}  /  \ottmv{z}  \}   \: \mathbf{of} \:  \ottmv{x_{{\mathrm{1}}}}  .   \ottnt{b_{{\mathrm{1}}}}  \{  \ottnt{c}  /  \ottmv{z}  \}   \: ; \:  \ottmv{x_{{\mathrm{2}}}}  .  \ottnt{b_{{\mathrm{2}}}}   \{  \ottnt{c}  /  \ottmv{z}  \}   :^{ \ottnt{q} }  \ottnt{B} $.\\
By lemma \ref{BLSSplitP}, $\exists \Gamma_{{\mathrm{31}}}, \Gamma_{{\mathrm{32}}}$ such that $ \Gamma_{{\mathrm{31}}}  \vdash  \ottnt{c}  :^{ \ottnt{r_{{\mathrm{01}}}} }  \ottnt{C} $ and $ \Gamma_{{\mathrm{32}}}  \vdash  \ottnt{c}  :^{ \ottnt{r_{{\mathrm{02}}}} }  \ottnt{C} $ and $\Gamma =  \Gamma_{{\mathrm{31}}}  +  \Gamma_{{\mathrm{32}}} $.\\
By IH, $    \Gamma_{{\mathrm{11}}}  +  \Gamma_{{\mathrm{31}}}    ,  \Gamma_{{\mathrm{21}}}   \vdash   \ottnt{a}  \{  \ottnt{c}  /  \ottmv{z}  \}   :^{  \ottnt{q}  \cdot  \ottnt{q_{{\mathrm{0}}}}  }   \ottnt{A_{{\mathrm{1}}}}  +  \ottnt{A_{{\mathrm{2}}}}  $ and $    \Gamma_{{\mathrm{12}}}  +  \Gamma_{{\mathrm{32}}}    ,    \Gamma_{{\mathrm{22}}}  ,   \ottmv{x_{{\mathrm{1}}}}  :^{  \ottnt{q}  \cdot  \ottnt{q_{{\mathrm{0}}}}  }  \ottnt{A_{{\mathrm{1}}}}      \vdash   \ottnt{b_{{\mathrm{1}}}}  \{  \ottnt{c}  /  \ottmv{z}  \}   :^{ \ottnt{q} }  \ottnt{B} $ and $    \Gamma_{{\mathrm{12}}}  +  \Gamma_{{\mathrm{32}}}    ,    \Gamma_{{\mathrm{22}}}  ,   \ottmv{x_{{\mathrm{2}}}}  :^{  \ottnt{q}  \cdot  \ottnt{q_{{\mathrm{0}}}}  }  \ottnt{A_{{\mathrm{2}}}}      \vdash   \ottnt{b_{{\mathrm{2}}}}  \{  \ottnt{c}  /  \ottmv{z}  \}   :^{ \ottnt{q} }  \ottnt{B} $.\\
This case, then, follows by \rref{ST-Case}.

\item \Rref{ST-SubL}. Have: $    \Gamma_{{\mathrm{1}}}  ,   \ottmv{z}  :^{ \ottnt{r_{{\mathrm{0}}}} }  \ottnt{C}     ,  \Gamma_{{\mathrm{2}}}   \vdash  \ottnt{a}  :^{ \ottnt{q} }  \ottnt{A} $ where $    \Gamma'_{{\mathrm{1}}}  ,   \ottmv{z}  :^{ \ottnt{r'_{{\mathrm{0}}}} }  \ottnt{C}     ,  \Gamma'_{{\mathrm{2}}}   \vdash  \ottnt{a}  :^{ \ottnt{q} }  \ottnt{A} $ where $ \Gamma_{{\mathrm{1}}}  <:  \Gamma'_{{\mathrm{1}}} $ and $ \ottnt{r_{{\mathrm{0}}}}  <:  \ottnt{r'_{{\mathrm{0}}}} $ and $ \Gamma_{{\mathrm{2}}}  <:  \Gamma'_{{\mathrm{2}}} $. Also, $ \Gamma  \vdash  \ottnt{c}  :^{ \ottnt{r_{{\mathrm{0}}}} }  \ottnt{C} $ where $  \lfloor  \Gamma  \rfloor   =   \lfloor  \Gamma_{{\mathrm{1}}}  \rfloor  $. \\
Need to show: $    \Gamma_{{\mathrm{1}}}  +  \Gamma    ,  \Gamma_{{\mathrm{2}}}   \vdash   \ottnt{a}  \{  \ottnt{c}  /  \ottmv{z}  \}   :^{ \ottnt{q} }  \ottnt{A} $.\\
Since $ \ottnt{r_{{\mathrm{0}}}}  <:  \ottnt{r'_{{\mathrm{0}}}} $, by \rref{ST-SubR}, $ \Gamma  \vdash  \ottnt{c}  :^{ \ottnt{r'_{{\mathrm{0}}}} }  \ottnt{C} $.\\
By IH, $    \Gamma'_{{\mathrm{1}}}  +  \Gamma    ,  \Gamma'_{{\mathrm{2}}}   \vdash   \ottnt{a}  \{  \ottnt{c}  /  \ottmv{z}  \}   :^{ \ottnt{q} }  \ottnt{A} $.\\
This case, then, follows by \rref{ST-SubL}. 

\item \Rref{ST-SubR}. Have: $    \Gamma_{{\mathrm{1}}}  ,   \ottmv{z}  :^{ \ottnt{r_{{\mathrm{0}}}} }  \ottnt{C}     ,  \Gamma_{{\mathrm{2}}}   \vdash  \ottnt{a}  :^{ \ottnt{q'} }  \ottnt{A} $ where $    \Gamma_{{\mathrm{1}}}  ,   \ottmv{z}  :^{ \ottnt{r_{{\mathrm{0}}}} }  \ottnt{C}     ,  \Gamma_{{\mathrm{2}}}   \vdash  \ottnt{a}  :^{ \ottnt{q} }  \ottnt{A} $ and $ \ottnt{q}  <:  \ottnt{q'} $. Also, $ \Gamma  \vdash  \ottnt{c}  :^{ \ottnt{r_{{\mathrm{0}}}} }  \ottnt{C} $ where $  \lfloor  \Gamma  \rfloor   =   \lfloor  \Gamma_{{\mathrm{1}}}  \rfloor  $.\\
Need to show: $    \Gamma_{{\mathrm{1}}}  +  \Gamma    ,  \Gamma_{{\mathrm{2}}}   \vdash   \ottnt{a}  \{  \ottnt{c}  /  \ottmv{z}  \}   :^{ \ottnt{q'} }  \ottnt{A} $.\\
By IH, $    \Gamma_{{\mathrm{1}}}  +  \Gamma    ,  \Gamma_{{\mathrm{2}}}   \vdash   \ottnt{a}  \{  \ottnt{c}  /  \ottmv{z}  \}   :^{ \ottnt{q} }  \ottnt{A} $.\\
This case, then, follows by \rref{ST-SubR}.

\end{itemize}
\end{proof}

%------------------------------------------------------------------------------------------

\begin{theorem}[Preservation (Theorem \ref{SimplePreserve})] \label{SimplePreserveP}
If $ \Gamma  \vdash  \ottnt{a}  :^{ \ottnt{q} }  \ottnt{A} $ and $ \vdash  \ottnt{a}  \leadsto  \ottnt{a'} $, then $ \Gamma  \vdash  \ottnt{a'}  :^{ \ottnt{q} }  \ottnt{A} $.
\end{theorem}

\begin{proof}
By induction on $ \Gamma  \vdash  \ottnt{a}  :^{ \ottnt{q} }  \ottnt{A} $ and inversion on $ \vdash  \ottnt{a}  \leadsto  \ottnt{a'} $.

\begin{itemize}
\item \Rref{ST-App}. Have: $  \Gamma_{{\mathrm{1}}}  +  \Gamma_{{\mathrm{2}}}   \vdash   \ottnt{b}  \:  \ottnt{a} ^{ \ottnt{r} }   :^{ \ottnt{q} }  \ottnt{B} $ where $ \Gamma_{{\mathrm{1}}}  \vdash  \ottnt{b}  :^{ \ottnt{q} }   {}^{ \ottnt{r} }\!  \ottnt{A}  \to  \ottnt{B}  $ and $ \Gamma_{{\mathrm{2}}}  \vdash  \ottnt{a}  :^{  \ottnt{q}  \cdot  \ottnt{r}  }  \ottnt{A} $. \\ Let $ \vdash   \ottnt{b}  \:  \ottnt{a} ^{ \ottnt{r} }   \leadsto  \ottnt{c} $. By inversion:

\begin{itemize}
\item $ \vdash   \ottnt{b}  \:  \ottnt{a} ^{ \ottnt{r} }   \leadsto   \ottnt{b'}  \:  \ottnt{a} ^{ \ottnt{r} }  $, when $ \vdash  \ottnt{b}  \leadsto  \ottnt{b'} $. \\
Need to show: $  \Gamma_{{\mathrm{1}}}  +  \Gamma_{{\mathrm{2}}}   \vdash   \ottnt{b'}  \:  \ottnt{a} ^{ \ottnt{r} }   :^{ \ottnt{q} }  \ottnt{B} $.\\
Follows by IH and \rref{ST-App}.

\item $\ottnt{b} =  \lambda^{ \ottnt{r} }  \ottmv{x}  :  \ottnt{A'}  .  \ottnt{b'} $ and $ \vdash   \ottnt{b}  \:  \ottnt{a} ^{ \ottnt{r} }   \leadsto   \ottnt{b'}  \{  \ottnt{a}  /  \ottmv{x}  \}  $.\\
Need to show: $  \Gamma_{{\mathrm{1}}}  +  \Gamma_{{\mathrm{2}}}   \vdash   \ottnt{b'}  \{  \ottnt{a}  /  \ottmv{x}  \}   :^{ \ottnt{q} }  \ottnt{B} $.\\
By inversion on $ \Gamma_{{\mathrm{1}}}  \vdash   \lambda^{ \ottnt{r} }  \ottmv{x}  :  \ottnt{A'}  .  \ottnt{b'}   :^{ \ottnt{q} }   {}^{ \ottnt{r} }\!  \ottnt{A}  \to  \ottnt{B}  $, we get $\ottnt{A'} = \ottnt{A}$ and $  \Gamma_{{\mathrm{1}}}  ,   \ottmv{x}  :^{  \ottnt{q_{{\mathrm{0}}}}  \cdot  \ottnt{r}  }  \ottnt{A}    \vdash  \ottnt{b'}  :^{ \ottnt{q_{{\mathrm{0}}}} }  \ottnt{B} $ for some $ \ottnt{q_{{\mathrm{0}}}}  <:  \ottnt{q} $.\\
Now, there are two cases to consider.
\begin{itemize}
\item $\ottnt{q_{{\mathrm{0}}}} = 0$. Since $ \ottnt{q_{{\mathrm{0}}}}  <:  \ottnt{q} $, so $\ottnt{q} = 0$.\\
Then, by the substitution lemma, $  \Gamma_{{\mathrm{1}}}  +  \Gamma_{{\mathrm{2}}}   \vdash   \ottnt{b'}  \{  \ottnt{a}  /  \ottmv{x}  \}   :^{  0  }  \ottnt{B} $. 
\item $\ottnt{q_{{\mathrm{0}}}} \neq 0$. By lemma \ref{BLSFactP}, $\exists \Gamma'_{{\mathrm{1}}}$ and $\ottnt{r'}$ such that $  \Gamma'_{{\mathrm{1}}}  ,   \ottmv{x}  :^{ \ottnt{r'} }  \ottnt{A}    \vdash  \ottnt{b'}  :^{  1  }  \ottnt{B} $ and $ \Gamma_{{\mathrm{1}}}  <:   \ottnt{q_{{\mathrm{0}}}}  \cdot  \Gamma'_{{\mathrm{1}}}  $ and $  \ottnt{q_{{\mathrm{0}}}}  \cdot  \ottnt{r}   <:   \ottnt{q_{{\mathrm{0}}}}  \cdot  \ottnt{r'}  $.\\
Since $\ottnt{q_{{\mathrm{0}}}} \neq 0$, therefore $ \ottnt{r}  <:  \ottnt{r'} $. Hence, by \rref{ST-SubL}, $  \Gamma'_{{\mathrm{1}}}  ,   \ottmv{x}  :^{ \ottnt{r} }  \ottnt{A}    \vdash  \ottnt{b'}  :^{  1  }  \ottnt{B} $.\\
Now, by lemma \ref{BLSMultP}, $   \ottnt{q}  \cdot  \Gamma'_{{\mathrm{1}}}   ,   \ottmv{x}  :^{  \ottnt{q}  \cdot  \ottnt{r}  }  \ottnt{A}    \vdash  \ottnt{b'}  :^{ \ottnt{q} }  \ottnt{B} $. By \rref{ST-SubL}, $  \Gamma_{{\mathrm{1}}}  ,   \ottmv{x}  :^{  \ottnt{q}  \cdot  \ottnt{r}  }  \ottnt{A}    \vdash  \ottnt{b'}  :^{ \ottnt{q} }  \ottnt{B} $.\\
This case, then, follows by the substitution lemma.
\end{itemize}

\end{itemize} 

\item \Rref{ST-LetPair}. Have: $  \Gamma_{{\mathrm{1}}}  +  \Gamma_{{\mathrm{2}}}   \vdash   \mathbf{let}_{ \ottnt{q_{{\mathrm{0}}}} } \: (  \ottmv{x} ^{ \ottnt{r} } ,  \ottmv{y}  ) \: \mathbf{be} \:  \ottnt{a}  \: \mathbf{in} \:  \ottnt{b}   :^{ \ottnt{q} }  \ottnt{B} $ where $ \Gamma_{{\mathrm{1}}}  \vdash  \ottnt{a}  :^{  \ottnt{q}  \cdot  \ottnt{q_{{\mathrm{0}}}}  }   {}^{ \ottnt{r} }\!  \ottnt{A_{{\mathrm{1}}}}  \: \times \:  \ottnt{A_{{\mathrm{2}}}}  $ and $    \Gamma_{{\mathrm{2}}}  ,   \ottmv{x}  :^{    \ottnt{q}  \cdot  \ottnt{q_{{\mathrm{0}}}}    \cdot  \ottnt{r}  }  \ottnt{A_{{\mathrm{1}}}}     ,   \ottmv{y}  :^{  \ottnt{q}  \cdot  \ottnt{q_{{\mathrm{0}}}}  }  \ottnt{A_{{\mathrm{2}}}}    \vdash  \ottnt{b}  :^{ \ottnt{q} }  \ottnt{B} $. \\ Let $ \vdash   \mathbf{let}_{ \ottnt{q_{{\mathrm{0}}}} } \: (  \ottmv{x} ^{ \ottnt{r} } ,  \ottmv{y}  ) \: \mathbf{be} \:  \ottnt{a}  \: \mathbf{in} \:  \ottnt{b}   \leadsto  \ottnt{c} $. By inversion:

\begin{itemize}
\item $ \vdash   \mathbf{let}_{ \ottnt{q_{{\mathrm{0}}}} } \: (  \ottmv{x} ^{ \ottnt{r} } ,  \ottmv{y}  ) \: \mathbf{be} \:  \ottnt{a}  \: \mathbf{in} \:  \ottnt{b}   \leadsto   \mathbf{let}_{ \ottnt{q_{{\mathrm{0}}}} } \: (  \ottmv{x} ^{ \ottnt{r} } ,  \ottmv{y}  ) \: \mathbf{be} \:  \ottnt{a'}  \: \mathbf{in} \:  \ottnt{b}  $, when $ \vdash  \ottnt{a}  \leadsto  \ottnt{a'} $.\\
Need to show: $  \Gamma_{{\mathrm{1}}}  +  \Gamma_{{\mathrm{2}}}   \vdash   \mathbf{let}_{ \ottnt{q_{{\mathrm{0}}}} } \: (  \ottmv{x} ^{ \ottnt{r} } ,  \ottmv{y}  ) \: \mathbf{be} \:  \ottnt{a'}  \: \mathbf{in} \:  \ottnt{b}   :^{ \ottnt{q} }  \ottnt{B} $.\\
Follows by IH and \rref{ST-LetPair}.

\item $ \vdash   \mathbf{let}_{ \ottnt{q_{{\mathrm{0}}}} } \: (  \ottmv{x} ^{ \ottnt{r} } ,  \ottmv{y}  ) \: \mathbf{be} \:   (  \ottnt{a_{{\mathrm{1}}}} ^{ \ottnt{r} } ,  \ottnt{a_{{\mathrm{2}}}}  )   \: \mathbf{in} \:  \ottnt{b}   \leadsto    \ottnt{b}  \{  \ottnt{a_{{\mathrm{1}}}}  /  \ottmv{x}  \}   \{  \ottnt{a_{{\mathrm{2}}}}  /  \ottmv{y}  \}  $.\\
Need to show: $  \Gamma_{{\mathrm{1}}}  +  \Gamma_{{\mathrm{2}}}   \vdash    \ottnt{b}  \{  \ottnt{a_{{\mathrm{1}}}}  /  \ottmv{x}  \}   \{  \ottnt{a_{{\mathrm{2}}}}  /  \ottmv{y}  \}   :^{ \ottnt{q} }  \ottnt{B} $.\\
By inversion on $ \Gamma_{{\mathrm{1}}}  \vdash   (  \ottnt{a_{{\mathrm{1}}}} ^{ \ottnt{r} } ,  \ottnt{a_{{\mathrm{2}}}}  )   :^{  \ottnt{q}  \cdot  \ottnt{q_{{\mathrm{0}}}}  }   {}^{ \ottnt{r} }\!  \ottnt{A_{{\mathrm{1}}}}  \: \times \:  \ottnt{A_{{\mathrm{2}}}}  $, we have:\\
$\exists \Gamma_{{\mathrm{11}}}, \Gamma_{{\mathrm{12}}}$ such that $ \Gamma_{{\mathrm{11}}}  \vdash  \ottnt{a_{{\mathrm{1}}}}  :^{    \ottnt{q}  \cdot  \ottnt{q_{{\mathrm{0}}}}    \cdot  \ottnt{r}  }  \ottnt{A_{{\mathrm{1}}}} $ and $ \Gamma_{{\mathrm{12}}}  \vdash  \ottnt{a_{{\mathrm{2}}}}  :^{  \ottnt{q}  \cdot  \ottnt{q_{{\mathrm{0}}}}  }  \ottnt{A_{{\mathrm{2}}}} $ and $\Gamma_{{\mathrm{1}}} =  \Gamma_{{\mathrm{11}}}  +  \Gamma_{{\mathrm{12}}} $.\\
This case, then, follows by applying the substitution lemma twice.
\end{itemize}

\item \Rref{ST-LetUnit}. Have: $  \Gamma_{{\mathrm{1}}}  +  \Gamma_{{\mathrm{2}}}   \vdash   \mathbf{let}_{ \ottnt{q_{{\mathrm{0}}}} } \: \mathbf{unit} \: \mathbf{be} \:  \ottnt{a}  \: \mathbf{in} \:  \ottnt{b}   :^{ \ottnt{q} }  \ottnt{B} $ where $ \Gamma_{{\mathrm{1}}}  \vdash  \ottnt{a}  :^{  \ottnt{q}  \cdot  \ottnt{q_{{\mathrm{0}}}}  }   \mathbf{Unit}  $ and $ \Gamma_{{\mathrm{2}}}  \vdash  \ottnt{b}  :^{ \ottnt{q} }  \ottnt{B} $.\\
Let $ \vdash   \mathbf{let}_{ \ottnt{q_{{\mathrm{0}}}} } \: \mathbf{unit} \: \mathbf{be} \:  \ottnt{a}  \: \mathbf{in} \:  \ottnt{b}   \leadsto  \ottnt{c} $. By inversion:

\begin{itemize}
\item $ \vdash   \mathbf{let}_{ \ottnt{q_{{\mathrm{0}}}} } \: \mathbf{unit} \: \mathbf{be} \:  \ottnt{a}  \: \mathbf{in} \:  \ottnt{b}   \leadsto   \mathbf{let}_{ \ottnt{q_{{\mathrm{0}}}} } \: \mathbf{unit} \: \mathbf{be} \:  \ottnt{a'}  \: \mathbf{in} \:  \ottnt{b}  $, when $ \vdash  \ottnt{a}  \leadsto  \ottnt{a'} $.\\
Need to show: $  \Gamma_{{\mathrm{1}}}  +  \Gamma_{{\mathrm{2}}}   \vdash   \mathbf{let}_{ \ottnt{q_{{\mathrm{0}}}} } \: \mathbf{unit} \: \mathbf{be} \:  \ottnt{a'}  \: \mathbf{in} \:  \ottnt{b}   :^{ \ottnt{q} }  \ottnt{B} $.\\
Follows by IH and \rref{ST-LetUnit}.

\item $ \vdash   \mathbf{let}_{ \ottnt{q_{{\mathrm{0}}}} } \: \mathbf{unit} \: \mathbf{be} \:   \mathbf{unit}   \: \mathbf{in} \:  \ottnt{b}   \leadsto  \ottnt{b} $.\\
Need to show: $  \Gamma_{{\mathrm{1}}}  +  \Gamma_{{\mathrm{2}}}   \vdash  \ottnt{b}  :^{ \ottnt{q} }  \ottnt{B} $.\\
In case of $ \mathbb{N}_{=} $, we have, $ \overline{ \Gamma_{{\mathrm{1}}} }  = \overline{0}$. This case, then, follows directly from the premise, $ \Gamma_{{\mathrm{2}}}  \vdash  \ottnt{b}  :^{ \ottnt{q} }  \ottnt{B} $.\\
In case of $ \mathbb{N}_{\geq} $, this case follows by \rref{ST-SubL}.
\end{itemize}

\item \Rref{ST-Case}. Have: $  \Gamma_{{\mathrm{1}}}  +  \Gamma_{{\mathrm{2}}}   \vdash   \mathbf{case}_{ \ottnt{q_{{\mathrm{0}}}} } \:  \ottnt{a}  \: \mathbf{of} \:  \ottmv{x_{{\mathrm{1}}}}  .  \ottnt{b_{{\mathrm{1}}}}  \: ; \:  \ottmv{x_{{\mathrm{2}}}}  .  \ottnt{b_{{\mathrm{2}}}}   :^{ \ottnt{q} }  \ottnt{B} $ where $ \Gamma_{{\mathrm{1}}}  \vdash  \ottnt{a}  :^{  \ottnt{q}  \cdot  \ottnt{q_{{\mathrm{0}}}}  }   \ottnt{A_{{\mathrm{1}}}}  +  \ottnt{A_{{\mathrm{2}}}}  $ and $  \Gamma_{{\mathrm{2}}}  ,   \ottmv{x_{{\mathrm{1}}}}  :^{  \ottnt{q}  \cdot  \ottnt{q_{{\mathrm{0}}}}  }  \ottnt{A_{{\mathrm{1}}}}    \vdash  \ottnt{b_{{\mathrm{1}}}}  :^{ \ottnt{q} }  \ottnt{B} $ and $  \Gamma_{{\mathrm{2}}}  ,   \ottmv{x_{{\mathrm{2}}}}  :^{  \ottnt{q}  \cdot  \ottnt{q_{{\mathrm{0}}}}  }  \ottnt{A_{{\mathrm{2}}}}    \vdash  \ottnt{b_{{\mathrm{2}}}}  :^{ \ottnt{q} }  \ottnt{B} $.\\ Let $ \vdash   \mathbf{case}_{ \ottnt{q_{{\mathrm{0}}}} } \:  \ottnt{a}  \: \mathbf{of} \:  \ottmv{x_{{\mathrm{1}}}}  .  \ottnt{b_{{\mathrm{1}}}}  \: ; \:  \ottmv{x_{{\mathrm{2}}}}  .  \ottnt{b_{{\mathrm{2}}}}   \leadsto  \ottnt{c} $. By inversion:

\begin{itemize}
\item $ \vdash   \mathbf{case}_{ \ottnt{q_{{\mathrm{0}}}} } \:  \ottnt{a}  \: \mathbf{of} \:  \ottmv{x_{{\mathrm{1}}}}  .  \ottnt{b_{{\mathrm{1}}}}  \: ; \:  \ottmv{x_{{\mathrm{2}}}}  .  \ottnt{b_{{\mathrm{2}}}}   \leadsto   \mathbf{case}_{ \ottnt{q_{{\mathrm{0}}}} } \:  \ottnt{a'}  \: \mathbf{of} \:  \ottmv{x_{{\mathrm{1}}}}  .  \ottnt{b_{{\mathrm{1}}}}  \: ; \:  \ottmv{x_{{\mathrm{2}}}}  .  \ottnt{b_{{\mathrm{2}}}}  $, when $ \vdash  \ottnt{a}  \leadsto  \ottnt{a'} $.\\
Need to show: $  \Gamma_{{\mathrm{1}}}  +  \Gamma_{{\mathrm{2}}}   \vdash   \mathbf{case}_{ \ottnt{q_{{\mathrm{0}}}} } \:  \ottnt{a'}  \: \mathbf{of} \:  \ottmv{x_{{\mathrm{1}}}}  .  \ottnt{b_{{\mathrm{1}}}}  \: ; \:  \ottmv{x_{{\mathrm{2}}}}  .  \ottnt{b_{{\mathrm{2}}}}   :^{ \ottnt{q} }  \ottnt{B} $.\\
Follows by IH and \rref{ST-Case}.

\item $ \vdash   \mathbf{case}_{ \ottnt{q_{{\mathrm{0}}}} } \:   (   \mathbf{inj}_1 \:  \ottnt{a_{{\mathrm{1}}}}   )   \: \mathbf{of} \:  \ottmv{x_{{\mathrm{1}}}}  .  \ottnt{b_{{\mathrm{1}}}}  \: ; \:  \ottmv{x_{{\mathrm{2}}}}  .  \ottnt{b_{{\mathrm{2}}}}   \leadsto   \ottnt{b_{{\mathrm{1}}}}  \{  \ottnt{a_{{\mathrm{1}}}}  /  \ottmv{x_{{\mathrm{1}}}}  \}  $.\\
Need to show $  \Gamma_{{\mathrm{1}}}  +  \Gamma_{{\mathrm{2}}}   \vdash   \ottnt{b_{{\mathrm{1}}}}  \{  \ottnt{a_{{\mathrm{1}}}}  /  \ottmv{x_{{\mathrm{1}}}}  \}   :^{ \ottnt{q} }  \ottnt{B} $.\\
By inversion on $ \Gamma_{{\mathrm{1}}}  \vdash   \mathbf{inj}_1 \:  \ottnt{a_{{\mathrm{1}}}}   :^{  \ottnt{q}  \cdot  \ottnt{q_{{\mathrm{0}}}}  }   \ottnt{A_{{\mathrm{1}}}}  +  \ottnt{A_{{\mathrm{2}}}}  $, we have, $ \Gamma_{{\mathrm{1}}}  \vdash  \ottnt{a_{{\mathrm{1}}}}  :^{  \ottnt{q}  \cdot  \ottnt{q_{{\mathrm{0}}}}  }  \ottnt{A_{{\mathrm{1}}}} $.\\
This case, then, follows by applying the substitution lemma.

\item $ \vdash   \mathbf{case}_{ \ottnt{q_{{\mathrm{0}}}} } \:   (   \mathbf{inj}_2 \:  \ottnt{a_{{\mathrm{2}}}}   )   \: \mathbf{of} \:  \ottmv{x_{{\mathrm{1}}}}  .  \ottnt{b_{{\mathrm{1}}}}  \: ; \:  \ottmv{x_{{\mathrm{2}}}}  .  \ottnt{b_{{\mathrm{2}}}}   \leadsto   \ottnt{b_{{\mathrm{2}}}}  \{  \ottnt{a_{{\mathrm{2}}}}  /  \ottmv{x_{{\mathrm{2}}}}  \}  $.\\
Similar to the previous case.

\end{itemize}
\item \Rref{ST-SubL,ST-SubR}. Follows by IH.
\end{itemize}
\end{proof}

%--------------------------------------------------------------------------------------------

\begin{theorem}[Progress (Theorem \ref{BLSProg})] \label{BLSProgP}
If $  \emptyset   \vdash  \ottnt{a}  :^{ \ottnt{q} }  \ottnt{A} $, then either $\ottnt{a}$ is a value or there exists $\ottnt{a'}$ such that $ \vdash  \ottnt{a}  \leadsto  \ottnt{a'} $.
\end{theorem}

\begin{proof}
By induction on $  \emptyset   \vdash  \ottnt{a}  :^{ \ottnt{q} }  \ottnt{A} $.
\begin{itemize}

\item \Rref{ST-Var}. Does not apply since the context here is empty.

\item \Rref{ST-App}. Have: $  \emptyset   \vdash   \ottnt{b}  \:  \ottnt{a} ^{ \ottnt{r} }   :^{ \ottnt{q} }  \ottnt{B} $ where $  \emptyset   \vdash  \ottnt{b}  :^{ \ottnt{q} }   {}^{ \ottnt{r} }\!  \ottnt{A}  \to  \ottnt{B}  $ and $  \emptyset   \vdash  \ottnt{a}  :^{  \ottnt{q}  \cdot  \ottnt{r}  }  \ottnt{A} $. \\
Need to show: $\exists c,  \vdash   \ottnt{b}  \:  \ottnt{a} ^{ \ottnt{r} }   \leadsto  \ottnt{c} $.\\
By IH, $\ottnt{b}$ is either a value or $ \vdash  \ottnt{b}  \leadsto  \ottnt{b'} $.\\
If $\ottnt{b}$ is a value, then $\ottnt{b} =  \lambda^{ \ottnt{r} }  \ottmv{x}  :  \ottnt{A}  .  \ottnt{b'} $ for some $\ottnt{b'}$. Therefore, $ \vdash   \ottnt{b}  \:  \ottnt{a} ^{ \ottnt{r} }   \leadsto   \ottnt{b'}  \{  \ottnt{a}  /  \ottmv{x}  \}  $.\\
Otherwise, $ \vdash   \ottnt{b}  \:  \ottnt{a} ^{ \ottnt{r} }   \leadsto   \ottnt{b'}  \:  \ottnt{a} ^{ \ottnt{r} }  $.

\item \Rref{ST-LetPair}. Have: $  \emptyset   \vdash   \mathbf{let}_{ \ottnt{q_{{\mathrm{0}}}} } \: (  \ottmv{x} ^{ \ottnt{r} } ,  \ottmv{y}  ) \: \mathbf{be} \:  \ottnt{a}  \: \mathbf{in} \:  \ottnt{b}   :^{ \ottnt{q} }  \ottnt{B} $ where $  \emptyset   \vdash  \ottnt{a}  :^{  \ottnt{q}  \cdot  \ottnt{q_{{\mathrm{0}}}}  }   {}^{ \ottnt{r} }\!  \ottnt{A_{{\mathrm{1}}}}  \: \times \:  \ottnt{A_{{\mathrm{2}}}}  $ and $   \ottmv{x}  :^{    \ottnt{q}  \cdot  \ottnt{q_{{\mathrm{0}}}}    \cdot  \ottnt{r}  }  \ottnt{A_{{\mathrm{1}}}}   ,   \ottmv{y}  :^{  \ottnt{q}  \cdot  \ottnt{q_{{\mathrm{0}}}}  }  \ottnt{A_{{\mathrm{2}}}}    \vdash  \ottnt{b}  :^{ \ottnt{q} }  \ottnt{B} $.\\
Need to show: $\exists c,  \vdash   \mathbf{let}_{ \ottnt{q_{{\mathrm{0}}}} } \: (  \ottmv{x} ^{ \ottnt{r} } ,  \ottmv{y}  ) \: \mathbf{be} \:  \ottnt{a}  \: \mathbf{in} \:  \ottnt{b}   \leadsto  \ottnt{c} $.\\
By IH, $\ottnt{a}$ is either a value or $ \vdash  \ottnt{a}  \leadsto  \ottnt{a'} $.\\
If $\ottnt{a}$ is a value, then $\ottnt{a} =  (  \ottnt{a_{{\mathrm{1}}}} ^{ \ottnt{r} } ,  \ottnt{a_{{\mathrm{2}}}}  ) $. Therefore,
$ \vdash   \mathbf{let}_{ \ottnt{q_{{\mathrm{0}}}} } \: (  \ottmv{x} ^{ \ottnt{r} } ,  \ottmv{y}  ) \: \mathbf{be} \:  \ottnt{a}  \: \mathbf{in} \:  \ottnt{b}   \leadsto    \ottnt{b}  \{  \ottnt{a_{{\mathrm{1}}}}  /  \ottmv{x}  \}   \{  \ottnt{a_{{\mathrm{2}}}}  /  \ottmv{y}  \}  $.\\
Otherwise, $ \vdash   \mathbf{let}_{ \ottnt{q_{{\mathrm{0}}}} } \: (  \ottmv{x} ^{ \ottnt{r} } ,  \ottmv{y}  ) \: \mathbf{be} \:  \ottnt{a}  \: \mathbf{in} \:  \ottnt{b}   \leadsto   \mathbf{let}_{ \ottnt{q_{{\mathrm{0}}}} } \: (  \ottmv{x} ^{ \ottnt{r} } ,  \ottmv{y}  ) \: \mathbf{be} \:  \ottnt{a'}  \: \mathbf{in} \:  \ottnt{b}  $.

\item \Rref{ST-LetUnit}. Have: $  \emptyset   \vdash   \mathbf{let}_{ \ottnt{q_{{\mathrm{0}}}} } \: \mathbf{unit} \: \mathbf{be} \:  \ottnt{a}  \: \mathbf{in} \:  \ottnt{b}   :^{ \ottnt{q} }  \ottnt{B} $ where $  \emptyset   \vdash  \ottnt{a}  :^{  \ottnt{q}  \cdot  \ottnt{q_{{\mathrm{0}}}}  }   \mathbf{Unit}  $ and $  \emptyset   \vdash  \ottnt{b}  :^{ \ottnt{q} }  \ottnt{B} $.\\
Need to show: $\exists c,  \vdash   \mathbf{let}_{ \ottnt{q_{{\mathrm{0}}}} } \: \mathbf{unit} \: \mathbf{be} \:  \ottnt{a}  \: \mathbf{in} \:  \ottnt{b}   \leadsto  \ottnt{c} $.\\
By IH, $\ottnt{a}$ is either a value or $ \vdash  \ottnt{a}  \leadsto  \ottnt{a'} $.\\
If $\ottnt{a}$ is a value, then $\ottnt{a} =  \mathbf{unit} $. Therefore,
$ \vdash   \mathbf{let}_{ \ottnt{q_{{\mathrm{0}}}} } \: \mathbf{unit} \: \mathbf{be} \:  \ottnt{a}  \: \mathbf{in} \:  \ottnt{b}   \leadsto  \ottnt{b} $.\\
Otherwise, $ \vdash   \mathbf{let}_{ \ottnt{q_{{\mathrm{0}}}} } \: \mathbf{unit} \: \mathbf{be} \:  \ottnt{a}  \: \mathbf{in} \:  \ottnt{b}   \leadsto   \mathbf{let}_{ \ottnt{q_{{\mathrm{0}}}} } \: \mathbf{unit} \: \mathbf{be} \:  \ottnt{a'}  \: \mathbf{in} \:  \ottnt{b}  $.

\item \Rref{ST-Case}. Have: $  \emptyset   \vdash   \mathbf{case}_{ \ottnt{q_{{\mathrm{0}}}} } \:  \ottnt{a}  \: \mathbf{of} \:  \ottmv{x_{{\mathrm{1}}}}  .  \ottnt{b_{{\mathrm{1}}}}  \: ; \:  \ottmv{x_{{\mathrm{2}}}}  .  \ottnt{b_{{\mathrm{2}}}}   :^{ \ottnt{q} }  \ottnt{B} $ where $  \emptyset   \vdash  \ottnt{a}  :^{  \ottnt{q}  \cdot  \ottnt{q_{{\mathrm{0}}}}  }   \ottnt{A_{{\mathrm{1}}}}  +  \ottnt{A_{{\mathrm{2}}}}  $ and $  \ottmv{x_{{\mathrm{1}}}}  :^{  \ottnt{q}  \cdot  \ottnt{q_{{\mathrm{0}}}}  }  \ottnt{A_{{\mathrm{1}}}}   \vdash  \ottnt{b_{{\mathrm{1}}}}  :^{ \ottnt{q} }  \ottnt{B} $ and $  \ottmv{x_{{\mathrm{2}}}}  :^{  \ottnt{q}  \cdot  \ottnt{q_{{\mathrm{0}}}}  }  \ottnt{A_{{\mathrm{2}}}}   \vdash  \ottnt{b_{{\mathrm{2}}}}  :^{ \ottnt{q} }  \ottnt{B} $.\\
Need to show: $\exists c,  \vdash   \mathbf{case}_{ \ottnt{q_{{\mathrm{0}}}} } \:  \ottnt{a}  \: \mathbf{of} \:  \ottmv{x_{{\mathrm{1}}}}  .  \ottnt{b_{{\mathrm{1}}}}  \: ; \:  \ottmv{x_{{\mathrm{2}}}}  .  \ottnt{b_{{\mathrm{2}}}}   \leadsto  \ottnt{c} $.\\
By IH, $\ottnt{a}$ is either a value or $ \vdash  \ottnt{a}  \leadsto  \ottnt{a'} $.\\
If $\ottnt{a}$ is a value, then $\ottnt{a} =  \mathbf{inj}_1 \:  \ottnt{a_{{\mathrm{1}}}} $ or $\ottnt{a} =  \mathbf{inj}_2 \:  \ottnt{a_{{\mathrm{2}}}} $. \\
Then, $ \vdash   \mathbf{case}_{ \ottnt{q_{{\mathrm{0}}}} } \:  \ottnt{a}  \: \mathbf{of} \:  \ottmv{x_{{\mathrm{1}}}}  .  \ottnt{b_{{\mathrm{1}}}}  \: ; \:  \ottmv{x_{{\mathrm{2}}}}  .  \ottnt{b_{{\mathrm{2}}}}   \leadsto   \ottnt{b_{{\mathrm{1}}}}  \{  \ottnt{a_{{\mathrm{1}}}}  /  \ottmv{x_{{\mathrm{1}}}}  \}  $ or $ \vdash   \mathbf{case}_{ \ottnt{q_{{\mathrm{0}}}} } \:  \ottnt{a}  \: \mathbf{of} \:  \ottmv{x_{{\mathrm{1}}}}  .  \ottnt{b_{{\mathrm{1}}}}  \: ; \:  \ottmv{x_{{\mathrm{2}}}}  .  \ottnt{b_{{\mathrm{2}}}}   \leadsto   \ottnt{b_{{\mathrm{2}}}}  \{  \ottnt{a_{{\mathrm{2}}}}  /  \ottmv{x_{{\mathrm{2}}}}  \}  $.\\
Otherwise, $ \vdash   \mathbf{case}_{ \ottnt{q_{{\mathrm{0}}}} } \:  \ottnt{a}  \: \mathbf{of} \:  \ottmv{x_{{\mathrm{1}}}}  .  \ottnt{b_{{\mathrm{1}}}}  \: ; \:  \ottmv{x_{{\mathrm{2}}}}  .  \ottnt{b_{{\mathrm{2}}}}   \leadsto   \mathbf{case}_{ \ottnt{q_{{\mathrm{0}}}} } \:  \ottnt{a'}  \: \mathbf{of} \:  \ottmv{x_{{\mathrm{1}}}}  .  \ottnt{b_{{\mathrm{1}}}}  \: ; \:  \ottmv{x_{{\mathrm{2}}}}  .  \ottnt{b_{{\mathrm{2}}}}  $.

\item \Rref{ST-SubL, ST-SubR}. Follows by IH.

\item \Rref{ST-Lam,ST-Pair,ST-Unit}. The terms typed by these rules are values.

\end{itemize}
\end{proof}

%------------------------------------------------------------------------------------------------

%---------------------------------------------------------------------------------
\section{Dependency Analysis in Simply-Typed LDC}
%---------------------------------------------------------------------------------

\begin{lemma}[Multiplication (Lemma \ref{DSMult})] \label{DSMultP}
If $ \Gamma  \vdash  \ottnt{a}  :^{ \ell }  \ottnt{A} $, then $  \ottnt{m_{{\mathrm{0}}}}  \sqcup  \Gamma   \vdash  \ottnt{a}  :^{  \ottnt{m_{{\mathrm{0}}}}  \: \sqcup \:  \ell  }  \ottnt{A} $.
\end{lemma}

\begin{proof}
By induction on $ \Gamma  \vdash  \ottnt{a}  :^{ \ell }  \ottnt{A} $.

\begin{itemize}
\item \Rref{ST-VarD}. Have: $      \top   \sqcup  \Gamma_{{\mathrm{1}}}   ,   \ottmv{x}  :^{ \ell }  \ottnt{A}     ,    \top   \sqcup  \Gamma_{{\mathrm{2}}}    \vdash   \ottmv{x}   :^{ \ell }  \ottnt{A} $.\\
Need to show: $      \top   \sqcup  \Gamma_{{\mathrm{1}}}   ,   \ottmv{x}  :^{  \ottnt{m_{{\mathrm{0}}}}  \: \sqcup \:  \ell  }  \ottnt{A}     ,    \top   \sqcup  \Gamma_{{\mathrm{2}}}    \vdash   \ottmv{x}   :^{  \ottnt{m_{{\mathrm{0}}}}  \: \sqcup \:  \ell  }  \ottnt{A} $.\\
This case follows by \rref{ST-VarD}.

\item \Rref{ST-LamD}. Have: $ \Gamma  \vdash   \lambda^{ \ottnt{m} }  \ottmv{x}  :  \ottnt{A}  .  \ottnt{b}   :^{ \ell }   {}^{ \ottnt{m} }\!  \ottnt{A}  \to  \ottnt{B}  $ where $  \Gamma  ,   \ottmv{x}  :^{  \ell  \: \sqcup \:  \ottnt{m}  }  \ottnt{A}    \vdash  \ottnt{b}  :^{ \ell }  \ottnt{B} $.\\
Need to show: $  \ottnt{m_{{\mathrm{0}}}}  \sqcup  \Gamma   \vdash   \lambda^{ \ottnt{m} }  \ottmv{x}  :  \ottnt{A}  .  \ottnt{b}   :^{  \ottnt{m_{{\mathrm{0}}}}  \: \sqcup \:  \ell  }   {}^{ \ottnt{m} }\!  \ottnt{A}  \to  \ottnt{B}  $.\\
By IH, $   \ottnt{m_{{\mathrm{0}}}}  \sqcup  \Gamma   ,   \ottmv{x}  :^{  \ottnt{m_{{\mathrm{0}}}}  \: \sqcup \:   (   \ell  \: \sqcup \:  \ottnt{m}   )   }  \ottnt{A}    \vdash  \ottnt{b}  :^{  \ottnt{m_{{\mathrm{0}}}}  \: \sqcup \:  \ell  }  \ottnt{B} $.\\
This case, then, follows by \rref{ST-LamD} using associativity of $\sqcup$.

\item \Rref{ST-AppD}. Have: $  \Gamma_{{\mathrm{1}}}  \sqcap  \Gamma_{{\mathrm{2}}}   \vdash   \ottnt{b}  \:  \ottnt{a} ^{ \ottnt{m} }   :^{ \ell }  \ottnt{B} $ where $ \Gamma_{{\mathrm{1}}}  \vdash  \ottnt{b}  :^{ \ell }   {}^{ \ottnt{m} }\!  \ottnt{A}  \to  \ottnt{B}  $ and $ \Gamma_{{\mathrm{2}}}  \vdash  \ottnt{a}  :^{  \ell  \: \sqcup \:  \ottnt{m}  }  \ottnt{A} $.\\
Need to show: $  \ottnt{m_{{\mathrm{0}}}}  \sqcup   (   \Gamma_{{\mathrm{1}}}  \sqcap  \Gamma_{{\mathrm{2}}}   )    \vdash   \ottnt{b}  \:  \ottnt{a} ^{ \ottnt{m} }   :^{  \ottnt{m_{{\mathrm{0}}}}  \: \sqcup \:  \ell  }  \ottnt{B} $.\\
By IH, $  \ottnt{m_{{\mathrm{0}}}}  \sqcup  \Gamma_{{\mathrm{1}}}   \vdash  \ottnt{b}  :^{  \ottnt{m_{{\mathrm{0}}}}  \: \sqcup \:  \ell  }   {}^{ \ottnt{m} }\!  \ottnt{A}  \to  \ottnt{B}  $ and $  \ottnt{m_{{\mathrm{0}}}}  \sqcup  \Gamma_{{\mathrm{2}}}   \vdash  \ottnt{a}  :^{  \ottnt{m_{{\mathrm{0}}}}  \: \sqcup \:   (   \ell  \: \sqcup \:  \ottnt{m}   )   }  \ottnt{A} $.\\
By \rref{ST-AppD}, using associativity of $\sqcup$, $   (   \ottnt{m_{{\mathrm{0}}}}  \sqcup  \Gamma_{{\mathrm{1}}}   )   \sqcap   (   \ottnt{m_{{\mathrm{0}}}}  \sqcup  \Gamma_{{\mathrm{2}}}   )    \vdash   \ottnt{b}  \:  \ottnt{a} ^{ \ottnt{m} }   :^{  \ottnt{m_{{\mathrm{0}}}}  \: \sqcup \:  \ell  }  \ottnt{B} $.\\
Now, for elements $\ell_{{\mathrm{1}}}, \ell_{{\mathrm{2}}}$ and $\ell_{{\mathrm{3}}}$ of any lattice, $  \ell_{{\mathrm{1}}}  \: \sqcup \:   (   \ell_{{\mathrm{2}}}  \: \sqcap \:  \ell_{{\mathrm{3}}}   )    \sqsubseteq    (   \ell_{{\mathrm{1}}}  \: \sqcup \:  \ell_{{\mathrm{2}}}   )   \: \sqcap \:   (   \ell_{{\mathrm{1}}}  \: \sqcup \:  \ell_{{\mathrm{3}}}   )   $.\\
This case, then, follows by \rref{ST-SubLD}, using the above relation.

\item \Rref{ST-PairD}. Have: $  \Gamma_{{\mathrm{1}}}  \sqcap  \Gamma_{{\mathrm{2}}}   \vdash   (  \ottnt{a_{{\mathrm{1}}}} ^{ \ottnt{m} } ,  \ottnt{a_{{\mathrm{2}}}}  )   :^{ \ell }   {}^{ \ottnt{m} }\!  \ottnt{A_{{\mathrm{1}}}}  \: \times \:  \ottnt{A_{{\mathrm{2}}}}  $ where $ \Gamma_{{\mathrm{1}}}  \vdash  \ottnt{a_{{\mathrm{1}}}}  :^{  \ell  \: \sqcup \:  \ottnt{m}  }  \ottnt{A_{{\mathrm{1}}}} $ and $ \Gamma_{{\mathrm{2}}}  \vdash  \ottnt{a_{{\mathrm{2}}}}  :^{ \ell }  \ottnt{A_{{\mathrm{2}}}} $.\\
Need to show: $  \ottnt{m_{{\mathrm{0}}}}  \sqcup   (   \Gamma_{{\mathrm{1}}}  \sqcap  \Gamma_{{\mathrm{2}}}   )    \vdash   (  \ottnt{a_{{\mathrm{1}}}} ^{ \ottnt{m} } ,  \ottnt{a_{{\mathrm{2}}}}  )   :^{  \ottnt{m_{{\mathrm{0}}}}  \: \sqcup \:  \ell  }   {}^{ \ottnt{m} }\!  \ottnt{A_{{\mathrm{1}}}}  \: \times \:  \ottnt{A_{{\mathrm{2}}}}  $.\\
By IH, $  \ottnt{m_{{\mathrm{0}}}}  \sqcup  \Gamma_{{\mathrm{1}}}   \vdash  \ottnt{a_{{\mathrm{1}}}}  :^{  \ottnt{m_{{\mathrm{0}}}}  \: \sqcup \:   (   \ell  \: \sqcup \:  \ottnt{m}   )   }  \ottnt{A_{{\mathrm{1}}}} $ and $  \ottnt{m_{{\mathrm{0}}}}  \sqcup  \Gamma_{{\mathrm{2}}}   \vdash  \ottnt{a_{{\mathrm{2}}}}  :^{  \ottnt{m_{{\mathrm{0}}}}  \: \sqcup \:  \ell  }  \ottnt{A_{{\mathrm{2}}}} $.\\
By \rref{ST-PairD}, using associativity of $\sqcup$, $   (   \ottnt{m_{{\mathrm{0}}}}  \sqcup  \Gamma_{{\mathrm{1}}}   )   \sqcap   (   \ottnt{m_{{\mathrm{0}}}}  \sqcup  \Gamma_{{\mathrm{2}}}   )    \vdash   (  \ottnt{a_{{\mathrm{1}}}} ^{ \ottnt{m} } ,  \ottnt{a_{{\mathrm{2}}}}  )   :^{  \ottnt{m_{{\mathrm{0}}}}  \: \sqcup \:  \ell  }   {}^{ \ottnt{m} }\!  \ottnt{A_{{\mathrm{1}}}}  \: \times \:  \ottnt{A_{{\mathrm{2}}}}  $.\\
For $\ell_{{\mathrm{1}}}, \ell_{{\mathrm{2}}}, \ell_{{\mathrm{3}}}$, we have, $  \ell_{{\mathrm{1}}}  \: \sqcup \:   (   \ell_{{\mathrm{2}}}  \: \sqcap \:  \ell_{{\mathrm{3}}}   )    \sqsubseteq    (   \ell_{{\mathrm{1}}}  \: \sqcup \:  \ell_{{\mathrm{2}}}   )   \: \sqcap \:   (   \ell_{{\mathrm{1}}}  \: \sqcup \:  \ell_{{\mathrm{3}}}   )   $.\\
This case follows by \rref{ST-SubLD}, using the above relation.

\item \Rref{ST-LetPairD}. Have: $  \Gamma_{{\mathrm{1}}}  \sqcap  \Gamma_{{\mathrm{2}}}   \vdash   \mathbf{let} \: (  \ottmv{x} ^{ \ottnt{m} } ,  \ottmv{y}  ) \: \mathbf{be} \:  \ottnt{a}  \: \mathbf{in} \:  \ottnt{b}   :^{ \ell }  \ottnt{B} $ where $ \Gamma_{{\mathrm{1}}}  \vdash  \ottnt{a}  :^{ \ell }   {}^{ \ottnt{m} }\!  \ottnt{A_{{\mathrm{1}}}}  \: \times \:  \ottnt{A_{{\mathrm{2}}}}  $ and $    \Gamma_{{\mathrm{2}}}  ,   \ottmv{x}  :^{  \ell  \: \sqcup \:  \ottnt{m}  }  \ottnt{A_{{\mathrm{1}}}}     ,   \ottmv{y}  :^{ \ell }  \ottnt{A_{{\mathrm{2}}}}    \vdash  \ottnt{b}  :^{ \ell }  \ottnt{B} $.\\
Need to show: $  \ottnt{m_{{\mathrm{0}}}}  \sqcup   (   \Gamma_{{\mathrm{1}}}  \sqcap  \Gamma_{{\mathrm{2}}}   )    \vdash   \mathbf{let} \: (  \ottmv{x} ^{ \ottnt{m} } ,  \ottmv{y}  ) \: \mathbf{be} \:  \ottnt{a}  \: \mathbf{in} \:  \ottnt{b}   :^{  \ottnt{m_{{\mathrm{0}}}}  \: \sqcup \:  \ell  }  \ottnt{B} $.\\
By IH, $  \ottnt{m_{{\mathrm{0}}}}  \sqcup  \Gamma_{{\mathrm{1}}}   \vdash  \ottnt{a}  :^{  \ottnt{m_{{\mathrm{0}}}}  \: \sqcup \:  \ell  }   {}^{ \ottnt{m} }\!  \ottnt{A_{{\mathrm{1}}}}  \: \times \:  \ottnt{A_{{\mathrm{2}}}}  $ and $     \ottnt{m_{{\mathrm{0}}}}  \sqcup  \Gamma_{{\mathrm{2}}}   ,   \ottmv{x}  :^{  \ottnt{m_{{\mathrm{0}}}}  \: \sqcup \:   (   \ell  \: \sqcup \:  \ottnt{m}   )   }  \ottnt{A_{{\mathrm{1}}}}     ,   \ottmv{y}  :^{  \ottnt{m_{{\mathrm{0}}}}  \: \sqcup \:  \ell  }  \ottnt{A_{{\mathrm{2}}}}    \vdash  \ottnt{b}  :^{  \ottnt{m_{{\mathrm{0}}}}  \: \sqcup \:  \ell  }  \ottnt{B} $. \\
This case follows by \rref{ST-LetPairD,ST-SubLD}, using associativity of $\sqcup$ and the distributive inequality.

\item \Rref{ST-UnitD}. Have: $   \top   \sqcup  \Gamma   \vdash   \mathbf{unit}   :^{ \ell }   \mathbf{Unit}  $.\\
Need to show: $   \top   \sqcup  \Gamma   \vdash   \mathbf{unit}   :^{  \ottnt{m_{{\mathrm{0}}}}  \: \sqcup \:  \ell  }   \mathbf{Unit}  $.\\
Follows by \rref{ST-UnitD}.

\item \Rref{ST-LetUnitD}. Have: $  \Gamma_{{\mathrm{1}}}  \sqcap  \Gamma_{{\mathrm{2}}}   \vdash   \mathbf{let} \: \mathbf{unit} \: \mathbf{be} \:  \ottnt{a}  \: \mathbf{in} \:  \ottnt{b}   :^{ \ell }  \ottnt{B} $ where $ \Gamma_{{\mathrm{1}}}  \vdash  \ottnt{a}  :^{ \ell }   \mathbf{Unit}  $ and $ \Gamma_{{\mathrm{2}}}  \vdash  \ottnt{b}  :^{ \ell }  \ottnt{B} $.\\
Need to show: $  \ottnt{m_{{\mathrm{0}}}}  \sqcup   (   \Gamma_{{\mathrm{1}}}  \sqcap  \Gamma_{{\mathrm{2}}}   )    \vdash   \mathbf{let} \: \mathbf{unit} \: \mathbf{be} \:  \ottnt{a}  \: \mathbf{in} \:  \ottnt{b}   :^{  \ottnt{m_{{\mathrm{0}}}}  \: \sqcup \:  \ell  }  \ottnt{B} $.\\
By IH, $  \ottnt{m_{{\mathrm{0}}}}  \sqcup  \Gamma_{{\mathrm{1}}}   \vdash  \ottnt{a}  :^{  \ottnt{m_{{\mathrm{0}}}}  \: \sqcup \:  \ell  }   \mathbf{Unit}  $ and $  \ottnt{m_{{\mathrm{0}}}}  \sqcup  \Gamma_{{\mathrm{2}}}   \vdash  \ottnt{b}  :^{  \ottnt{m_{{\mathrm{0}}}}  \: \sqcup \:  \ell  }  \ottnt{B} $.\\
This case follows by \rref{ST-LetUnitD,ST-SubLD}, using the distributive inequality.

\item \Rref{ST-Inj1D}. Have: $ \Gamma  \vdash   \mathbf{inj}_1 \:  \ottnt{a_{{\mathrm{1}}}}   :^{ \ell }   \ottnt{A_{{\mathrm{1}}}}  +  \ottnt{A_{{\mathrm{2}}}}  $ where $ \Gamma  \vdash  \ottnt{a_{{\mathrm{1}}}}  :^{ \ell }  \ottnt{A_{{\mathrm{1}}}} $.\\
Need to show: $  \ottnt{m_{{\mathrm{0}}}}  \sqcup  \Gamma   \vdash   \mathbf{inj}_1 \:  \ottnt{a_{{\mathrm{1}}}}   :^{  \ottnt{m_{{\mathrm{0}}}}  \: \sqcup \:  \ell  }   \ottnt{A_{{\mathrm{1}}}}  +  \ottnt{A_{{\mathrm{2}}}}  $.\\
By IH, $  \ottnt{m_{{\mathrm{0}}}}  \sqcup  \Gamma   \vdash  \ottnt{a_{{\mathrm{1}}}}  :^{  \ottnt{m_{{\mathrm{0}}}}  \: \sqcup \:  \ell  }  \ottnt{A_{{\mathrm{1}}}} $.\\
This case, then, follows by \rref{ST-Inj1D}.

\item \Rref{ST-Inj2D}. Similar to \rref{ST-Inj1D}.

\item \Rref{ST-CaseD}. Have: $  \Gamma_{{\mathrm{1}}}  \sqcap  \Gamma_{{\mathrm{2}}}   \vdash   \mathbf{case} \:  \ottnt{a}  \: \mathbf{of} \:  \ottmv{x_{{\mathrm{1}}}}  .  \ottnt{b_{{\mathrm{1}}}}  \: ; \:  \ottmv{x_{{\mathrm{2}}}}  .  \ottnt{b_{{\mathrm{2}}}}   :^{ \ell }  \ottnt{B} $ where $ \Gamma_{{\mathrm{1}}}  \vdash  \ottnt{a}  :^{ \ell }   \ottnt{A_{{\mathrm{1}}}}  +  \ottnt{A_{{\mathrm{2}}}}  $ and $  \Gamma_{{\mathrm{2}}}  ,   \ottmv{x_{{\mathrm{1}}}}  :^{ \ell }  \ottnt{A_{{\mathrm{1}}}}    \vdash  \ottnt{b_{{\mathrm{1}}}}  :^{ \ell }  \ottnt{B} $ and $  \Gamma_{{\mathrm{2}}}  ,   \ottmv{x_{{\mathrm{2}}}}  :^{ \ell }  \ottnt{A_{{\mathrm{2}}}}    \vdash  \ottnt{b_{{\mathrm{2}}}}  :^{ \ell }  \ottnt{B} $.\\
Need to show: $  \ottnt{m_{{\mathrm{0}}}}  \sqcup   (   \Gamma_{{\mathrm{1}}}  \sqcap  \Gamma_{{\mathrm{2}}}   )    \vdash   \mathbf{case} \:  \ottnt{a}  \: \mathbf{of} \:  \ottmv{x_{{\mathrm{1}}}}  .  \ottnt{b_{{\mathrm{1}}}}  \: ; \:  \ottmv{x_{{\mathrm{2}}}}  .  \ottnt{b_{{\mathrm{2}}}}   :^{  \ottnt{m_{{\mathrm{0}}}}  \: \sqcup \:  \ell  }  \ottnt{B} $.\\
By IH, $  \ottnt{m_{{\mathrm{0}}}}  \sqcup  \Gamma_{{\mathrm{1}}}   \vdash  \ottnt{a_{{\mathrm{1}}}}  :^{  \ottnt{m_{{\mathrm{0}}}}  \: \sqcup \:  \ell  }   \ottnt{A_{{\mathrm{1}}}}  +  \ottnt{A_{{\mathrm{2}}}}  $ and $   \ottnt{m_{{\mathrm{0}}}}  \sqcup  \Gamma_{{\mathrm{2}}}   ,   \ottmv{x_{{\mathrm{1}}}}  :^{  \ottnt{m_{{\mathrm{0}}}}  \: \sqcup \:  \ell  }  \ottnt{A_{{\mathrm{1}}}}    \vdash  \ottnt{b_{{\mathrm{1}}}}  :^{  \ottnt{m_{{\mathrm{0}}}}  \: \sqcup \:  \ell  }  \ottnt{B} $ and $   \ottnt{m_{{\mathrm{0}}}}  \sqcup  \Gamma_{{\mathrm{2}}}   ,   \ottmv{x_{{\mathrm{2}}}}  :^{  \ottnt{m_{{\mathrm{0}}}}  \: \sqcup \:  \ell  }  \ottnt{A_{{\mathrm{2}}}}    \vdash  \ottnt{b_{{\mathrm{2}}}}  :^{  \ottnt{m_{{\mathrm{0}}}}  \: \sqcup \:  \ell  }  \ottnt{B} $.\\
This case follows by \rref{ST-CaseD,ST-SubLD}, using the distributive inequality.

\item \Rref{ST-SubLD}. Have: $ \Gamma  \vdash  \ottnt{a}  :^{ \ell }  \ottnt{A} $ where $ \Gamma'  \vdash  \ottnt{a}  :^{ \ell }  \ottnt{A} $ and $ \Gamma   \sqsubseteq   \Gamma' $.\\
Need to show: $  \ottnt{m_{{\mathrm{0}}}}  \sqcup  \Gamma   \vdash  \ottnt{a}  :^{  \ottnt{m_{{\mathrm{0}}}}  \: \sqcup \:  \ell  }  \ottnt{A} $.\\
By IH, $  \ottnt{m_{{\mathrm{0}}}}  \sqcup  \Gamma'   \vdash  \ottnt{a}  :^{  \ottnt{m_{{\mathrm{0}}}}  \: \sqcup \:  \ell  }  \ottnt{A} $.\\
Since $ \Gamma   \sqsubseteq   \Gamma' $, so $  \ottnt{m_{{\mathrm{0}}}}  \sqcup  \Gamma    \sqsubseteq    \ottnt{m_{{\mathrm{0}}}}  \sqcup  \Gamma'  $.\\
This case follows by \rref{ST-SubLD}.

\item \Rref{ST-SubRD}. Have: $ \Gamma  \vdash  \ottnt{a}  :^{ \ell' }  \ottnt{A} $ where $ \Gamma  \vdash  \ottnt{a}  :^{ \ell }  \ottnt{A} $ and $ \ell  \sqsubseteq  \ell' $.\\
Need to show: $  \ottnt{m_{{\mathrm{0}}}}  \sqcup  \Gamma   \vdash  \ottnt{a}  :^{  \ottnt{m_{{\mathrm{0}}}}  \: \sqcup \:  \ell'  }  \ottnt{A} $.\\
By IH, $  \ottnt{m_{{\mathrm{0}}}}  \sqcup  \Gamma   \vdash  \ottnt{a}  :^{  \ottnt{m_{{\mathrm{0}}}}  \: \sqcup \:  \ell  }  \ottnt{A} $.\\\
Since $ \ell  \sqsubseteq  \ell' $, so $  \ottnt{m_{{\mathrm{0}}}}  \: \sqcup \:  \ell   \sqsubseteq   \ottnt{m_{{\mathrm{0}}}}  \: \sqcup \:  \ell'  $.\\
This case follows by \rref{ST-SubRD}.

\end{itemize}
\end{proof}

%-----------------------------------------------------------------------------------------------

\begin{lemma}[Splitting (Lemma \ref{DSSplit})] \label{DSSplitP}
If $ \Gamma  \vdash  \ottnt{a}  :^{  \ell_{{\mathrm{1}}}  \: \sqcap \:  \ell_{{\mathrm{2}}}  }  \ottnt{A} $, then there exists $\Gamma_{{\mathrm{1}}}$ and $\Gamma_{{\mathrm{2}}}$ such that $ \Gamma_{{\mathrm{1}}}  \vdash  \ottnt{a}  :^{ \ell_{{\mathrm{1}}} }  \ottnt{A} $ and $ \Gamma_{{\mathrm{2}}}  \vdash  \ottnt{a}  :^{ \ell_{{\mathrm{2}}} }  \ottnt{A} $ and $ \Gamma  =   \Gamma_{{\mathrm{1}}}  \sqcap  \Gamma_{{\mathrm{2}}}  $.
\end{lemma}

\begin{proof}
Have: $ \Gamma  \vdash  \ottnt{a}  :^{  \ell_{{\mathrm{1}}}  \: \sqcap \:  \ell_{{\mathrm{2}}}  }  \ottnt{A} $. By \rref{ST-SubRD}, $ \Gamma  \vdash  \ottnt{a}  :^{ \ell_{{\mathrm{1}}} }  \ottnt{A} $ and $ \Gamma  \vdash  \ottnt{a}  :^{ \ell_{{\mathrm{2}}} }  \ottnt{A} $. The lemma follows by setting $\Gamma_{{\mathrm{1}}} := \Gamma$ and $\Gamma_{{\mathrm{2}}} := \Gamma$.
\end{proof}

%-----------------------------------------------------------------------------------------------

\begin{lemma}[Weakening (Lemma \ref{DSWeak})] \label{DSWeakP}
If $  \Gamma_{{\mathrm{1}}}  ,  \Gamma_{{\mathrm{2}}}   \vdash  \ottnt{a}  :^{ \ell }  \ottnt{A} $, then $    \Gamma_{{\mathrm{1}}}  ,   \ottmv{z}  :^{  \top  }  \ottnt{C}     ,  \Gamma_{{\mathrm{2}}}   \vdash  \ottnt{a}  :^{ \ell }  \ottnt{A} $.
\end{lemma}

\begin{proof}
By induction on $  \Gamma_{{\mathrm{1}}}  ,  \Gamma_{{\mathrm{2}}}   \vdash  \ottnt{a}  :^{ \ell }  \ottnt{A} $.
\end{proof}

%-------------------------------------------------------------------------------------------------

\begin{lemma}[Substitution (Lemma \ref{DSSubst})] \label{DSSubstP}
If $    \Gamma_{{\mathrm{1}}}  ,   \ottmv{z}  :^{ \ottnt{m_{{\mathrm{0}}}} }  \ottnt{C}     ,  \Gamma_{{\mathrm{2}}}   \vdash  \ottnt{a}  :^{ \ell }  \ottnt{A} $ and $ \Gamma  \vdash  \ottnt{c}  :^{ \ottnt{m_{{\mathrm{0}}}} }  \ottnt{C} $ and $  \lfloor  \Gamma_{{\mathrm{1}}}  \rfloor   =   \lfloor  \Gamma  \rfloor  $, then $    \Gamma_{{\mathrm{1}}}  \sqcap  \Gamma    ,  \Gamma_{{\mathrm{2}}}   \vdash   \ottnt{a}  \{  \ottnt{c}  /  \ottmv{z}  \}   :^{ \ell }  \ottnt{A} $.
\end{lemma}

\begin{proof}
By induction on $    \Gamma_{{\mathrm{1}}}  ,   \ottmv{z}  :^{ \ottnt{m_{{\mathrm{0}}}} }  \ottnt{C}     ,  \Gamma_{{\mathrm{2}}}   \vdash  \ottnt{a}  :^{ \ell }  \ottnt{A} $.

\begin{itemize}

\item \Rref{ST-VarD}. There are three cases to consider.

\begin{itemize}
\item $          \top   \sqcup  \Gamma_{{\mathrm{11}}}   ,   \ottmv{z}  :^{  \top  }  \ottnt{C}     ,    \top   \sqcup  \Gamma_{{\mathrm{12}}}     ,   \ottmv{x}  :^{ \ell }  \ottnt{A}     ,    \top   \sqcup  \Gamma_{{\mathrm{2}}}    \vdash   \ottmv{x}   :^{ \ell }  \ottnt{A} $. Also, $ \Gamma  \vdash  \ottnt{c}  :^{  \top  }  \ottnt{C} $ where $  \lfloor  \Gamma_{{\mathrm{11}}}  \rfloor   =   \lfloor  \Gamma  \rfloor  $.\\
Need to show: $      \Gamma  ,    \top   \sqcup  \Gamma_{{\mathrm{12}}}     ,   \ottmv{x}  :^{ \ell }  \ottnt{A}     ,    \top   \sqcup  \Gamma_{{\mathrm{2}}}    \vdash   \ottmv{x}   :^{ \ell }  \ottnt{A} $.\\
Follows by \rref{ST-VarD} and \rref{ST-SubLD}.

\item $      \top   \sqcup  \Gamma_{{\mathrm{1}}}   ,   \ottmv{x}  :^{ \ell }  \ottnt{A}     ,    \top   \sqcup  \Gamma_{{\mathrm{2}}}    \vdash   \ottmv{x}   :^{ \ell }  \ottnt{A} $. Also, $ \Gamma  \vdash  \ottnt{a}  :^{ \ell }  \ottnt{A} $ where $  \lfloor  \Gamma_{{\mathrm{1}}}  \rfloor   =   \lfloor  \Gamma  \rfloor  $.\\
Need to show: $  \Gamma  ,    \top   \sqcup  \Gamma_{{\mathrm{2}}}    \vdash  \ottnt{a}  :^{ \ell }  \ottnt{A} $.\\
Follows by lemma \ref{DSWeakP}.

\item $      \top   \sqcup  \Gamma_{{\mathrm{1}}}   ,   \ottmv{x}  :^{ \ell }  \ottnt{A}     ,        \top   \sqcup  \Gamma_{{\mathrm{21}}}   ,   \ottmv{z}  :^{  \top  }  \ottnt{C}     ,    \top   \sqcup  \Gamma_{{\mathrm{22}}}      \vdash   \ottmv{x}   :^{ \ell }  \ottnt{A} $. Also, $    \Gamma_{{\mathrm{31}}}  ,   \ottmv{x}  :^{ \ottnt{m} }  \ottnt{A}     ,  \Gamma_{{\mathrm{32}}}   \vdash  \ottnt{c}  :^{  \top  }  \ottnt{C} $ where $  \lfloor  \Gamma_{{\mathrm{31}}}  \rfloor   =   \lfloor  \Gamma_{{\mathrm{1}}}  \rfloor  $ and $  \lfloor  \Gamma_{{\mathrm{32}}}  \rfloor   =   \lfloor  \Gamma_{{\mathrm{21}}}  \rfloor  $.\\
Need to show: $    \Gamma_{{\mathrm{31}}}  ,   \ottmv{x}  :^{  (   \ell  \: \sqcap \:  \ottnt{m}   )  }  \ottnt{A}     ,    \Gamma_{{\mathrm{32}}}  ,    \top   \sqcup  \Gamma_{{\mathrm{22}}}      \vdash   \ottmv{x}   :^{ \ell }  \ottnt{A} $.\\
Follows by \rref{ST-VarD} and \rref{ST-SubLD}.
\end{itemize}

\item \Rref{ST-LamD}. Have: $    \Gamma_{{\mathrm{1}}}  ,   \ottmv{z}  :^{ \ottnt{m_{{\mathrm{0}}}} }  \ottnt{C}     ,  \Gamma_{{\mathrm{2}}}   \vdash   \lambda^{ \ottnt{m} }  \ottmv{x}  :  \ottnt{A}  .  \ottnt{b}   :^{ \ell }   {}^{ \ottnt{m} }\!  \ottnt{A}  \to  \ottnt{B}  $ where $    \Gamma_{{\mathrm{1}}}  ,   \ottmv{z}  :^{ \ottnt{m_{{\mathrm{0}}}} }  \ottnt{C}     ,    \Gamma_{{\mathrm{2}}}  ,   \ottmv{x}  :^{  \ell  \: \sqcup \:  \ottnt{m}  }  \ottnt{A}      \vdash  \ottnt{b}  :^{ \ell }  \ottnt{B} $. Also, $ \Gamma  \vdash  \ottnt{c}  :^{ \ottnt{m_{{\mathrm{0}}}} }  \ottnt{C} $ where $  \lfloor  \Gamma  \rfloor   =   \lfloor  \Gamma_{{\mathrm{1}}}  \rfloor  $.\\
Need to show: $    \Gamma_{{\mathrm{1}}}  \sqcap  \Gamma    ,  \Gamma_{{\mathrm{2}}}   \vdash    \lambda^{ \ottnt{m} }  \ottmv{x}  :  \ottnt{A}  .  \ottnt{b}   \{  \ottnt{c}  /  \ottmv{z}  \}   :^{ \ell }   {}^{ \ottnt{m} }\!  \ottnt{A}  \to  \ottnt{B}  $.\\
Follows by IH and \rref{ST-LamD}. 

\item \Rref{ST-AppD}. Have: $      \Gamma_{{\mathrm{11}}}  \sqcap  \Gamma_{{\mathrm{12}}}    ,   \ottmv{z}  :^{  \ottnt{m_{{\mathrm{01}}}}  \: \sqcap \:  \ottnt{m_{{\mathrm{02}}}}  }  \ottnt{C}     ,    \Gamma_{{\mathrm{21}}}  \sqcap  \Gamma_{{\mathrm{22}}}     \vdash   \ottnt{b}  \:  \ottnt{a} ^{ \ottnt{m} }   :^{ \ell }  \ottnt{B} $ where $    \Gamma_{{\mathrm{11}}}  ,   \ottmv{z}  :^{ \ottnt{m_{{\mathrm{01}}}} }  \ottnt{C}     ,  \Gamma_{{\mathrm{21}}}   \vdash  \ottnt{b}  :^{ \ell }   {}^{ \ottnt{m} }\!  \ottnt{A}  \to  \ottnt{B}  $ and $    \Gamma_{{\mathrm{12}}}  ,   \ottmv{z}  :^{ \ottnt{m_{{\mathrm{02}}}} }  \ottnt{C}     ,  \Gamma_{{\mathrm{22}}}   \vdash  \ottnt{a}  :^{  \ell  \: \sqcup \:  \ottnt{m}  }  \ottnt{A} $. Also, $ \Gamma  \vdash  \ottnt{c}  :^{  \ottnt{m_{{\mathrm{01}}}}  \: \sqcap \:  \ottnt{m_{{\mathrm{02}}}}  }  \ottnt{C} $ where $  \lfloor  \Gamma  \rfloor   =   \lfloor  \Gamma_{{\mathrm{11}}}  \rfloor  $.\\
Need to show: $      \Gamma_{{\mathrm{11}}}  \sqcap  \Gamma_{{\mathrm{12}}}    \sqcap  \Gamma    ,    \Gamma_{{\mathrm{21}}}  \sqcap  \Gamma_{{\mathrm{22}}}     \vdash    \ottnt{b}  \{  \ottnt{c}  /  \ottmv{z}  \}   \:   \ottnt{a}  \{  \ottnt{c}  /  \ottmv{z}  \}  ^{ \ottnt{m} }   :^{ \ell }  \ottnt{B} $.\\
By lemma \ref{DSSplitP}, $\exists \Gamma_{{\mathrm{31}}}, \Gamma_{{\mathrm{32}}}$ such that $ \Gamma_{{\mathrm{31}}}  \vdash  \ottnt{c}  :^{ \ottnt{m_{{\mathrm{01}}}} }  \ottnt{C} $ and $ \Gamma_{{\mathrm{32}}}  \vdash  \ottnt{c}  :^{ \ottnt{m_{{\mathrm{02}}}} }  \ottnt{C} $ and $\Gamma =  \Gamma_{{\mathrm{31}}}  \sqcap  \Gamma_{{\mathrm{32}}} $.\\
By IH, $    \Gamma_{{\mathrm{11}}}  \sqcap  \Gamma_{{\mathrm{31}}}    ,  \Gamma_{{\mathrm{21}}}   \vdash   \ottnt{b}  \{  \ottnt{c}  /  \ottmv{z}  \}   :^{ \ell }   {}^{ \ottnt{m} }\!  \ottnt{A}  \to  \ottnt{B}  $ and $    \Gamma_{{\mathrm{12}}}  \sqcap  \Gamma_{{\mathrm{32}}}    ,  \Gamma_{{\mathrm{22}}}   \vdash   \ottnt{a}  \{  \ottnt{c}  /  \ottmv{z}  \}   :^{  \ell  \: \sqcup \:  \ottnt{m}  }  \ottnt{A} $.\\
This case, then, follows by \rref{ST-AppD}.

\item \Rref{ST-PairD}. Have: $      \Gamma_{{\mathrm{11}}}  \sqcap  \Gamma_{{\mathrm{12}}}    ,   \ottmv{z}  :^{  \ottnt{m_{{\mathrm{01}}}}  \: \sqcap \:  \ottnt{m_{{\mathrm{02}}}}  }  \ottnt{C}     ,    \Gamma_{{\mathrm{21}}}  \sqcap  \Gamma_{{\mathrm{22}}}     \vdash   (  \ottnt{a_{{\mathrm{1}}}} ^{ \ottnt{m} } ,  \ottnt{a_{{\mathrm{2}}}}  )   :^{ \ell }   {}^{ \ottnt{m} }\!  \ottnt{A_{{\mathrm{1}}}}  \: \times \:  \ottnt{A_{{\mathrm{2}}}}  $ where $    \Gamma_{{\mathrm{11}}}  ,   \ottmv{z}  :^{ \ottnt{m_{{\mathrm{01}}}} }  \ottnt{C}     ,  \Gamma_{{\mathrm{21}}}   \vdash  \ottnt{a_{{\mathrm{1}}}}  :^{  \ell  \: \sqcup \:  \ottnt{m}  }  \ottnt{A_{{\mathrm{1}}}} $ and $    \Gamma_{{\mathrm{12}}}  ,   \ottmv{z}  :^{ \ottnt{m_{{\mathrm{02}}}} }  \ottnt{C}     ,  \Gamma_{{\mathrm{22}}}   \vdash  \ottnt{a_{{\mathrm{2}}}}  :^{ \ell }  \ottnt{A_{{\mathrm{2}}}} $. Also, $ \Gamma  \vdash  \ottnt{c}  :^{  \ottnt{m_{{\mathrm{01}}}}  \: \sqcap \:  \ottnt{m_{{\mathrm{02}}}}  }  \ottnt{C} $ where $  \lfloor  \Gamma  \rfloor   =   \lfloor  \Gamma_{{\mathrm{11}}}  \rfloor  $.\\
Need to show: $      \Gamma_{{\mathrm{11}}}  \sqcap  \Gamma_{{\mathrm{12}}}    \sqcap  \Gamma    ,    \Gamma_{{\mathrm{21}}}  \sqcap  \Gamma_{{\mathrm{22}}}     \vdash   (   \ottnt{a_{{\mathrm{1}}}}  \{  \ottnt{c}  /  \ottmv{z}  \}  ^{ \ottnt{r} } ,   \ottnt{a_{{\mathrm{2}}}}  \{  \ottnt{c}  /  \ottmv{z}  \}   )   :^{ \ell }   {}^{ \ottnt{m} }\!  \ottnt{A_{{\mathrm{1}}}}  \: \times \:  \ottnt{A_{{\mathrm{2}}}}  $.\\
By lemma \ref{DSSplitP}, $\exists \Gamma_{{\mathrm{31}}}, \Gamma_{{\mathrm{32}}}$ such that $ \Gamma_{{\mathrm{31}}}  \vdash  \ottnt{c}  :^{ \ottnt{m_{{\mathrm{01}}}} }  \ottnt{C} $ and $ \Gamma_{{\mathrm{32}}}  \vdash  \ottnt{c}  :^{ \ottnt{m_{{\mathrm{02}}}} }  \ottnt{C} $ and $\Gamma =  \Gamma_{{\mathrm{31}}}  \sqcap  \Gamma_{{\mathrm{32}}} $.\\
By IH, $    \Gamma_{{\mathrm{11}}}  \sqcap  \Gamma_{{\mathrm{31}}}    ,  \Gamma_{{\mathrm{21}}}   \vdash   \ottnt{a_{{\mathrm{1}}}}  \{  \ottnt{c}  /  \ottmv{z}  \}   :^{  \ell  \: \sqcup \:  \ottnt{m}  }  \ottnt{A_{{\mathrm{1}}}} $ and $    \Gamma_{{\mathrm{12}}}  \sqcap  \Gamma_{{\mathrm{32}}}    ,  \Gamma_{{\mathrm{22}}}   \vdash   \ottnt{a_{{\mathrm{2}}}}  \{  \ottnt{c}  /  \ottmv{z}  \}   :^{ \ell }  \ottnt{A_{{\mathrm{2}}}} $.\\
This case, then, follows by \rref{ST-PairD}.

\item \Rref{ST-LetPairD}. Have: $      \Gamma_{{\mathrm{11}}}  \sqcap  \Gamma_{{\mathrm{12}}}    ,   \ottmv{z}  :^{  \ottnt{m_{{\mathrm{01}}}}  \: \sqcap \:  \ottnt{m_{{\mathrm{02}}}}  }  \ottnt{C}     ,    \Gamma_{{\mathrm{21}}}  \sqcap  \Gamma_{{\mathrm{22}}}     \vdash   \mathbf{let} \: (  \ottmv{x} ^{ \ottnt{m} } ,  \ottmv{y}  ) \: \mathbf{be} \:  \ottnt{a}  \: \mathbf{in} \:  \ottnt{b}   :^{ \ell }  \ottnt{B} $ where $    \Gamma_{{\mathrm{11}}}  ,   \ottmv{z}  :^{ \ottnt{m_{{\mathrm{01}}}} }  \ottnt{C}     ,  \Gamma_{{\mathrm{21}}}   \vdash  \ottnt{a}  :^{ \ell }   {}^{ \ottnt{m} }\!  \ottnt{A_{{\mathrm{1}}}}  \: \times \:  \ottnt{A_{{\mathrm{2}}}}  $ and $    \Gamma_{{\mathrm{12}}}  ,   \ottmv{z}  :^{ \ottnt{m_{{\mathrm{02}}}} }  \ottnt{C}     ,    \Gamma_{{\mathrm{22}}}  ,     \ottmv{x}  :^{  \ell  \: \sqcup \:  \ottnt{m}  }  \ottnt{A_{{\mathrm{1}}}}   ,   \ottmv{y}  :^{ \ell }  \ottnt{A_{{\mathrm{2}}}}        \vdash  \ottnt{b}  :^{ \ell }  \ottnt{B} $. Also, $ \Gamma  \vdash  \ottnt{c}  :^{  \ottnt{m_{{\mathrm{01}}}}  \: \sqcap \:  \ottnt{m_{{\mathrm{02}}}}  }  \ottnt{C} $ where $  \lfloor  \Gamma  \rfloor   =   \lfloor  \Gamma_{{\mathrm{11}}}  \rfloor  $.\\
Need to show: $      \Gamma_{{\mathrm{11}}}  \sqcap  \Gamma_{{\mathrm{12}}}    \sqcap  \Gamma    ,    \Gamma_{{\mathrm{21}}}  \sqcap  \Gamma_{{\mathrm{22}}}     \vdash    \mathbf{let} \: (  \ottmv{x} ^{ \ottnt{m} } ,  \ottmv{y}  ) \: \mathbf{be} \:   \ottnt{a}  \{  \ottnt{c}  /  \ottmv{z}  \}   \: \mathbf{in} \:  \ottnt{b}   \{  \ottnt{c}  /  \ottmv{z}  \}   :^{ \ell }  \ottnt{B} $.\\
By lemma \ref{DSSplitP}, $\exists \Gamma_{{\mathrm{31}}}, \Gamma_{{\mathrm{32}}}$ such that $ \Gamma_{{\mathrm{31}}}  \vdash  \ottnt{c}  :^{ \ottnt{m_{{\mathrm{01}}}} }  \ottnt{C} $ and $ \Gamma_{{\mathrm{32}}}  \vdash  \ottnt{c}  :^{ \ottnt{m_{{\mathrm{02}}}} }  \ottnt{C} $ and $\Gamma =  \Gamma_{{\mathrm{31}}}  \sqcap  \Gamma_{{\mathrm{32}}} $.\\
By IH, $    \Gamma_{{\mathrm{11}}}  \sqcap  \Gamma_{{\mathrm{31}}}    ,  \Gamma_{{\mathrm{21}}}   \vdash   \ottnt{a}  \{  \ottnt{c}  /  \ottmv{z}  \}   :^{ \ell }   {}^{ \ottnt{m} }\!  \ottnt{A_{{\mathrm{1}}}}  \: \times \:  \ottnt{A_{{\mathrm{2}}}}  $ and $    \Gamma_{{\mathrm{12}}}  \sqcap  \Gamma_{{\mathrm{32}}}    ,    \Gamma_{{\mathrm{22}}}  ,     \ottmv{x}  :^{  \ell  \: \sqcup \:  \ottnt{m}  }  \ottnt{A_{{\mathrm{1}}}}   ,   \ottmv{y}  :^{ \ell }  \ottnt{A_{{\mathrm{2}}}}        \vdash   \ottnt{b}  \{  \ottnt{c}  /  \ottmv{z}  \}   :^{ \ell }  \ottnt{B} $.\\
This case, then, follows by \rref{ST-LetPairD}.

\item \Rref{ST-UnitD}. Have: $      \top   \sqcup  \Gamma_{{\mathrm{1}}}   ,   \ottmv{z}  :^{  \top  }  \ottnt{C}     ,    \top   \sqcup  \Gamma_{{\mathrm{2}}}    \vdash   \mathbf{unit}   :^{ \ell }   \mathbf{Unit}  $. Also $ \Gamma  \vdash  \ottnt{c}  :^{  \top  }  \ottnt{C} $ where $ \lfloor  \Gamma  \rfloor  =  \lfloor  \Gamma_{{\mathrm{1}}}  \rfloor $.\\
Need to show: $  \Gamma  ,    \top   \sqcup  \Gamma_{{\mathrm{2}}}    \vdash   \mathbf{unit}   :^{ \ell }   \mathbf{Unit}  $.\\
Follows by \rref{ST-UnitD} and \rref{ST-SubLD}.

\item \Rref{ST-LetUnitD}. Have: $      \Gamma_{{\mathrm{11}}}  \sqcap  \Gamma_{{\mathrm{12}}}    ,   \ottmv{z}  :^{  \ottnt{m_{{\mathrm{01}}}}  \: \sqcap \:  \ottnt{m_{{\mathrm{02}}}}  }  \ottnt{C}     ,    \Gamma_{{\mathrm{21}}}  \sqcap  \Gamma_{{\mathrm{22}}}     \vdash   \mathbf{let} \: \mathbf{unit} \: \mathbf{be} \:  \ottnt{a}  \: \mathbf{in} \:  \ottnt{b}   :^{ \ell }  \ottnt{B} $ where $    \Gamma_{{\mathrm{11}}}  ,   \ottmv{z}  :^{ \ottnt{m_{{\mathrm{01}}}} }  \ottnt{C}     ,  \Gamma_{{\mathrm{21}}}   \vdash  \ottnt{a}  :^{ \ell }   \mathbf{Unit}  $ and $    \Gamma_{{\mathrm{12}}}  ,   \ottmv{z}  :^{ \ottnt{m_{{\mathrm{02}}}} }  \ottnt{C}     ,  \Gamma_{{\mathrm{22}}}   \vdash  \ottnt{b}  :^{ \ell }  \ottnt{B} $. Also, $ \Gamma  \vdash  \ottnt{c}  :^{  \ottnt{m_{{\mathrm{01}}}}  \: \sqcap \:  \ottnt{m_{{\mathrm{02}}}}  }  \ottnt{C} $ where $  \lfloor  \Gamma  \rfloor   =   \lfloor  \Gamma_{{\mathrm{11}}}  \rfloor  $.\\
Need to show: $      \Gamma_{{\mathrm{11}}}  \sqcap  \Gamma_{{\mathrm{12}}}    \sqcap  \Gamma    ,    \Gamma_{{\mathrm{21}}}  \sqcap  \Gamma_{{\mathrm{22}}}     \vdash    \mathbf{let} \: \mathbf{unit} \: \mathbf{be} \:   \ottnt{a}  \{  \ottnt{c}  /  \ottmv{z}  \}   \: \mathbf{in} \:  \ottnt{b}   \{  \ottnt{c}  /  \ottmv{z}  \}   :^{ \ell }  \ottnt{B} $.\\
By lemma \ref{DSSplitP}, $\exists \Gamma_{{\mathrm{31}}}, \Gamma_{{\mathrm{32}}}$ such that $ \Gamma_{{\mathrm{31}}}  \vdash  \ottnt{c}  :^{ \ottnt{m_{{\mathrm{01}}}} }  \ottnt{C} $ and $ \Gamma_{{\mathrm{32}}}  \vdash  \ottnt{c}  :^{ \ottnt{m_{{\mathrm{02}}}} }  \ottnt{C} $ and $\Gamma =  \Gamma_{{\mathrm{31}}}  \sqcap  \Gamma_{{\mathrm{32}}} $.\\
By IH, $    \Gamma_{{\mathrm{11}}}  \sqcap  \Gamma_{{\mathrm{31}}}    ,  \Gamma_{{\mathrm{21}}}   \vdash   \ottnt{a}  \{  \ottnt{c}  /  \ottmv{z}  \}   :^{ \ell }   \mathbf{Unit}  $ and $    \Gamma_{{\mathrm{12}}}  \sqcap  \Gamma_{{\mathrm{32}}}    ,  \Gamma_{{\mathrm{22}}}   \vdash   \ottnt{b}  \{  \ottnt{c}  /  \ottmv{z}  \}   :^{ \ell }  \ottnt{B} $.\\
This case, then, follows by \rref{ST-LetUnitD}. 

\item \Rref{ST-Inj1D}. Have: $    \Gamma_{{\mathrm{1}}}  ,   \ottmv{z}  :^{ \ottnt{m_{{\mathrm{0}}}} }  \ottnt{C}     ,  \Gamma_{{\mathrm{2}}}   \vdash   \mathbf{inj}_1 \:  \ottnt{a_{{\mathrm{1}}}}   :^{ \ell }   \ottnt{A_{{\mathrm{1}}}}  +  \ottnt{A_{{\mathrm{2}}}}  $ where $    \Gamma_{{\mathrm{1}}}  ,   \ottmv{z}  :^{ \ottnt{m_{{\mathrm{0}}}} }  \ottnt{C}     ,  \Gamma_{{\mathrm{2}}}   \vdash  \ottnt{a_{{\mathrm{1}}}}  :^{ \ell }  \ottnt{A_{{\mathrm{1}}}} $. Also, $ \Gamma  \vdash  \ottnt{c}  :^{ \ottnt{m_{{\mathrm{0}}}} }  \ottnt{C} $ where $  \lfloor  \Gamma  \rfloor   =   \lfloor  \Gamma_{{\mathrm{1}}}  \rfloor  $.\\
Need to show: $    \Gamma_{{\mathrm{1}}}  \sqcap  \Gamma    ,  \Gamma_{{\mathrm{2}}}   \vdash    \mathbf{inj}_1 \:  \ottnt{a_{{\mathrm{1}}}}   \{  \ottnt{c}  /  \ottmv{z}  \}   :^{ \ell }   \ottnt{A_{{\mathrm{1}}}}  +  \ottnt{A_{{\mathrm{2}}}}  $.\\
By IH, $    \Gamma_{{\mathrm{1}}}  \sqcap  \Gamma    ,  \Gamma_{{\mathrm{2}}}   \vdash   \ottnt{a_{{\mathrm{1}}}}  \{  \ottnt{c}  /  \ottmv{z}  \}   :^{ \ell }  \ottnt{A_{{\mathrm{1}}}} $.\\
This case, then, follows by \rref{ST-Inj1D}.

\item \Rref{ST-Inj2D}. Similar to \rref{ST-Inj1D}.

\item \Rref{ST-CaseD}. Have: $      \Gamma_{{\mathrm{11}}}  \sqcap  \Gamma_{{\mathrm{12}}}    ,   \ottmv{z}  :^{  \ottnt{m_{{\mathrm{01}}}}  \: \sqcap \:  \ottnt{m_{{\mathrm{02}}}}  }  \ottnt{C}     ,    \Gamma_{{\mathrm{21}}}  \sqcap  \Gamma_{{\mathrm{22}}}     \vdash   \mathbf{case} \:  \ottnt{a}  \: \mathbf{of} \:  \ottmv{x_{{\mathrm{1}}}}  .  \ottnt{b_{{\mathrm{1}}}}  \: ; \:  \ottmv{x_{{\mathrm{2}}}}  .  \ottnt{b_{{\mathrm{2}}}}   :^{ \ell }  \ottnt{B} $ where $    \Gamma_{{\mathrm{11}}}  ,   \ottmv{z}  :^{ \ottnt{m_{{\mathrm{01}}}} }  \ottnt{C}     ,  \Gamma_{{\mathrm{21}}}   \vdash  \ottnt{a}  :^{ \ell }   \ottnt{A_{{\mathrm{1}}}}  +  \ottnt{A_{{\mathrm{2}}}}  $ and $    \Gamma_{{\mathrm{12}}}  ,   \ottmv{z}  :^{ \ottnt{m_{{\mathrm{02}}}} }  \ottnt{C}     ,    \Gamma_{{\mathrm{22}}}  ,   \ottmv{x_{{\mathrm{1}}}}  :^{ \ell }  \ottnt{A_{{\mathrm{1}}}}      \vdash  \ottnt{b_{{\mathrm{1}}}}  :^{ \ell }  \ottnt{B} $ and $    \Gamma_{{\mathrm{12}}}  ,   \ottmv{z}  :^{ \ottnt{m_{{\mathrm{02}}}} }  \ottnt{C}     ,    \Gamma_{{\mathrm{22}}}  ,   \ottmv{x_{{\mathrm{2}}}}  :^{ \ell }  \ottnt{A_{{\mathrm{2}}}}      \vdash  \ottnt{b_{{\mathrm{2}}}}  :^{ \ell }  \ottnt{B} $. Also, $ \Gamma  \vdash  \ottnt{c}  :^{  \ottnt{m_{{\mathrm{01}}}}  \: \sqcap \:  \ottnt{m_{{\mathrm{02}}}}  }  \ottnt{C} $ where $  \lfloor  \Gamma  \rfloor   =   \lfloor  \Gamma_{{\mathrm{11}}}  \rfloor  $.\\
Need to show: $      \Gamma_{{\mathrm{11}}}  \sqcap  \Gamma_{{\mathrm{12}}}    \sqcap  \Gamma    ,    \Gamma_{{\mathrm{21}}}  \sqcap  \Gamma_{{\mathrm{22}}}     \vdash    \mathbf{case} \:   \ottnt{a}  \{  \ottnt{c}  /  \ottmv{z}  \}   \: \mathbf{of} \:  \ottmv{x_{{\mathrm{1}}}}  .   \ottnt{b_{{\mathrm{1}}}}  \{  \ottnt{c}  /  \ottmv{z}  \}   \: ; \:  \ottmv{x_{{\mathrm{2}}}}  .  \ottnt{b_{{\mathrm{2}}}}   \{  \ottnt{c}  /  \ottmv{z}  \}   :^{ \ell }  \ottnt{B} $.\\
By lemma \ref{DSSplitP}, $\exists \Gamma_{{\mathrm{31}}}, \Gamma_{{\mathrm{32}}}$ such that $ \Gamma_{{\mathrm{31}}}  \vdash  \ottnt{c}  :^{ \ottnt{m_{{\mathrm{01}}}} }  \ottnt{C} $ and $ \Gamma_{{\mathrm{32}}}  \vdash  \ottnt{c}  :^{ \ottnt{m_{{\mathrm{02}}}} }  \ottnt{C} $ and $\Gamma =  \Gamma_{{\mathrm{31}}}  \sqcap  \Gamma_{{\mathrm{32}}} $.\\
By IH, $    \Gamma_{{\mathrm{11}}}  \sqcap  \Gamma_{{\mathrm{31}}}    ,  \Gamma_{{\mathrm{21}}}   \vdash   \ottnt{a}  \{  \ottnt{c}  /  \ottmv{z}  \}   :^{ \ell }   \ottnt{A_{{\mathrm{1}}}}  +  \ottnt{A_{{\mathrm{2}}}}  $ and $    \Gamma_{{\mathrm{12}}}  \sqcap  \Gamma_{{\mathrm{32}}}    ,    \Gamma_{{\mathrm{22}}}  ,   \ottmv{x_{{\mathrm{1}}}}  :^{ \ell }  \ottnt{A_{{\mathrm{1}}}}      \vdash   \ottnt{b_{{\mathrm{1}}}}  \{  \ottnt{c}  /  \ottmv{z}  \}   :^{ \ell }  \ottnt{B} $ and $    \Gamma_{{\mathrm{12}}}  \sqcap  \Gamma_{{\mathrm{32}}}    ,    \Gamma_{{\mathrm{22}}}  ,   \ottmv{x_{{\mathrm{2}}}}  :^{ \ell }  \ottnt{A_{{\mathrm{2}}}}      \vdash   \ottnt{b_{{\mathrm{2}}}}  \{  \ottnt{c}  /  \ottmv{z}  \}   :^{ \ell }  \ottnt{B} $.\\
This case, then, follows by \rref{ST-CaseD}.

\item \Rref{ST-SubLD}. Have: $    \Gamma_{{\mathrm{1}}}  ,   \ottmv{z}  :^{ \ottnt{m_{{\mathrm{0}}}} }  \ottnt{C}     ,  \Gamma_{{\mathrm{2}}}   \vdash  \ottnt{a}  :^{ \ell }  \ottnt{A} $ where $    \Gamma'_{{\mathrm{1}}}  ,   \ottmv{z}  :^{ \ottnt{m'_{{\mathrm{0}}}} }  \ottnt{C}     ,  \Gamma'_{{\mathrm{2}}}   \vdash  \ottnt{a}  :^{ \ell }  \ottnt{A} $ where $ \Gamma_{{\mathrm{1}}}   \sqsubseteq   \Gamma'_{{\mathrm{1}}} $ and $ \ottnt{m_{{\mathrm{0}}}}  \sqsubseteq  \ottnt{m'_{{\mathrm{0}}}} $ and $ \Gamma_{{\mathrm{2}}}   \sqsubseteq   \Gamma'_{{\mathrm{2}}} $. Also, $ \Gamma  \vdash  \ottnt{c}  :^{ \ottnt{m_{{\mathrm{0}}}} }  \ottnt{C} $ where $  \lfloor  \Gamma  \rfloor   =   \lfloor  \Gamma_{{\mathrm{1}}}  \rfloor  $. \\
Need to show: $    \Gamma_{{\mathrm{1}}}  \sqcap  \Gamma    ,  \Gamma_{{\mathrm{2}}}   \vdash   \ottnt{a}  \{  \ottnt{c}  /  \ottmv{z}  \}   :^{ \ell }  \ottnt{A} $.\\
Since $ \ottnt{m_{{\mathrm{0}}}}  \sqsubseteq  \ottnt{m'_{{\mathrm{0}}}} $, by \rref{ST-SubRD}, $ \Gamma  \vdash  \ottnt{c}  :^{ \ottnt{m'_{{\mathrm{0}}}} }  \ottnt{C} $.\\
By IH, $    \Gamma'_{{\mathrm{1}}}  \sqcap  \Gamma    ,  \Gamma'_{{\mathrm{2}}}   \vdash   \ottnt{a}  \{  \ottnt{c}  /  \ottmv{z}  \}   :^{ \ell }  \ottnt{A} $.\\
This case, then, follows by \rref{ST-SubLD}. 

\item \Rref{ST-SubRD}. Have: $    \Gamma_{{\mathrm{1}}}  ,   \ottmv{z}  :^{ \ottnt{m_{{\mathrm{0}}}} }  \ottnt{C}     ,  \Gamma_{{\mathrm{2}}}   \vdash  \ottnt{a}  :^{ \ell' }  \ottnt{A} $ where $    \Gamma_{{\mathrm{1}}}  ,   \ottmv{z}  :^{ \ottnt{m_{{\mathrm{0}}}} }  \ottnt{C}     ,  \Gamma_{{\mathrm{2}}}   \vdash  \ottnt{a}  :^{ \ell }  \ottnt{A} $ and $ \ell  \sqsubseteq  \ell' $. Also, $ \Gamma  \vdash  \ottnt{c}  :^{ \ottnt{m_{{\mathrm{0}}}} }  \ottnt{C} $ where $  \lfloor  \Gamma  \rfloor   =   \lfloor  \Gamma_{{\mathrm{1}}}  \rfloor  $.\\
Need to show: $    \Gamma_{{\mathrm{1}}}  \sqcap  \Gamma    ,  \Gamma_{{\mathrm{2}}}   \vdash   \ottnt{a}  \{  \ottnt{c}  /  \ottmv{z}  \}   :^{ \ell' }  \ottnt{A} $.\\
By IH, $    \Gamma_{{\mathrm{1}}}  \sqcap  \Gamma    ,  \Gamma_{{\mathrm{2}}}   \vdash   \ottnt{a}  \{  \ottnt{c}  /  \ottmv{z}  \}   :^{ \ell }  \ottnt{A} $.\\
This case, then, follows by \rref{ST-SubRD}.

\end{itemize}

\end{proof}

%--------------------------------------------------------------------------------------------------

\begin{lemma}[Restricted Upgrading] \label{ResUpP}
If $    \Gamma_{{\mathrm{1}}}  ,   \ottmv{x}  :^{ \ottnt{m} }  \ottnt{A}     ,  \Gamma_{{\mathrm{2}}}   \vdash  \ottnt{b}  :^{ \ell }  \ottnt{B} $ and $ \ell_{{\mathrm{0}}}  \sqsubseteq  \ell $, then $    \Gamma_{{\mathrm{1}}}  ,   \ottmv{x}  :^{  \ell_{{\mathrm{0}}}  \: \sqcup \:  \ottnt{m}  }  \ottnt{A}     ,  \Gamma_{{\mathrm{2}}}   \vdash  \ottnt{b}  :^{ \ell }  \ottnt{B} $. 
\end{lemma}

\begin{proof}
By Lemma \ref{DSMultP}, $     \ell_{{\mathrm{0}}}  \sqcup  \Gamma_{{\mathrm{1}}}   ,   \ottmv{x}  :^{  \ell_{{\mathrm{0}}}  \: \sqcup \:  \ottnt{m}  }  \ottnt{A}     ,   \ell_{{\mathrm{0}}}  \sqcup  \Gamma_{{\mathrm{2}}}    \vdash  \ottnt{b}  :^{  \ell_{{\mathrm{0}}}  \: \sqcup \:  \ell  }  \ottnt{B} $. \\
Since $ \ell_{{\mathrm{0}}}  \sqsubseteq  \ell $, $     \ell_{{\mathrm{0}}}  \sqcup  \Gamma_{{\mathrm{1}}}   ,   \ottmv{x}  :^{  \ell_{{\mathrm{0}}}  \: \sqcup \:  \ottnt{m}  }  \ottnt{A}     ,   \ell_{{\mathrm{0}}}  \sqcup  \Gamma_{{\mathrm{2}}}    \vdash  \ottnt{b}  :^{ \ell }  \ottnt{B} $. \\
By \rref{ST-SubLD}, $    \Gamma_{{\mathrm{1}}}  ,   \ottmv{x}  :^{  \ell_{{\mathrm{0}}}  \: \sqcup \:  \ottnt{m}  }  \ottnt{A}     ,  \Gamma_{{\mathrm{2}}}   \vdash  \ottnt{b}  :^{ \ell }  \ottnt{B} $ because $ \Gamma_{{\mathrm{1}}}   \sqsubseteq    \ell_{{\mathrm{0}}}  \sqcup  \Gamma_{{\mathrm{1}}}  $ and $ \Gamma_{{\mathrm{2}}}   \sqsubseteq    \ell_{{\mathrm{0}}}  \sqcup  \Gamma_{{\mathrm{2}}}  $.
\end{proof}

%---------------------------------------------------------------------------------------------------

\begin{theorem}[Preservation (Theorem \ref{DSPreserve})]
If $ \Gamma  \vdash  \ottnt{a}  :^{ \ell }  \ottnt{A} $ and $ \vdash  \ottnt{a}  \leadsto  \ottnt{a'} $, then $ \Gamma  \vdash  \ottnt{a'}  :^{ \ell }  \ottnt{A} $.
\end{theorem}

\begin{proof}
By induction on $ \Gamma  \vdash  \ottnt{a}  :^{ \ell }  \ottnt{A} $ and inversion on $ \vdash  \ottnt{a}  \leadsto  \ottnt{a'} $.

\begin{itemize}
\item \Rref{ST-AppD}. Have: $  \Gamma_{{\mathrm{1}}}  \sqcap  \Gamma_{{\mathrm{2}}}   \vdash   \ottnt{b}  \:  \ottnt{a} ^{ \ottnt{m} }   :^{ \ell }  \ottnt{B} $ where $ \Gamma_{{\mathrm{1}}}  \vdash  \ottnt{b}  :^{ \ell }   {}^{ \ottnt{m} }\!  \ottnt{A}  \to  \ottnt{B}  $ and $ \Gamma_{{\mathrm{2}}}  \vdash  \ottnt{a}  :^{  \ell  \: \sqcup \:  \ottnt{m}  }  \ottnt{A} $. \\ Let $ \vdash   \ottnt{b}  \:  \ottnt{a} ^{ \ottnt{m} }   \leadsto  \ottnt{c} $. By inversion:

\begin{itemize}
\item $ \vdash   \ottnt{b}  \:  \ottnt{a} ^{ \ottnt{m} }   \leadsto   \ottnt{b'}  \:  \ottnt{a} ^{ \ottnt{m} }  $, when $ \vdash  \ottnt{b}  \leadsto  \ottnt{b'} $. \\
Need to show: $  \Gamma_{{\mathrm{1}}}  \sqcap  \Gamma_{{\mathrm{2}}}   \vdash   \ottnt{b'}  \:  \ottnt{a} ^{ \ottnt{m} }   :^{ \ell }  \ottnt{B} $.\\
Follows by IH and \rref{ST-AppD}.

\item $\ottnt{b} =  \lambda^{ \ottnt{m} }  \ottmv{x}  :  \ottnt{A'}  .  \ottnt{b'} $ and $ \vdash   \ottnt{b}  \:  \ottnt{a} ^{ \ottnt{m} }   \leadsto   \ottnt{b'}  \{  \ottnt{a}  /  \ottmv{x}  \}  $.\\
Need to show: $  \Gamma_{{\mathrm{1}}}  \sqcap  \Gamma_{{\mathrm{2}}}   \vdash   \ottnt{b'}  \{  \ottnt{a}  /  \ottmv{x}  \}   :^{ \ell }  \ottnt{B} $.\\
By inversion on $ \Gamma_{{\mathrm{1}}}  \vdash   \lambda^{ \ottnt{m} }  \ottmv{x}  :  \ottnt{A'}  .  \ottnt{b'}   :^{ \ell }   {}^{ \ottnt{m} }\!  \ottnt{A}  \to  \ottnt{B}  $, we get $\ottnt{A'} = \ottnt{A}$ and $  \Gamma_{{\mathrm{1}}}  ,   \ottmv{x}  :^{  \ell_{{\mathrm{0}}}  \: \sqcup \:  \ottnt{m}  }  \ottnt{A}    \vdash  \ottnt{b'}  :^{ \ell_{{\mathrm{0}}} }  \ottnt{B} $ for some $ \ell_{{\mathrm{0}}}  \sqsubseteq  \ell $.\\
By \rref{ST-SubRD}, $  \Gamma_{{\mathrm{1}}}  ,   \ottmv{x}  :^{  \ell_{{\mathrm{0}}}  \: \sqcup \:  \ottnt{m}  }  \ottnt{A}    \vdash  \ottnt{b'}  :^{ \ell }  \ottnt{B} $.\\
By lemma \ref{ResUpP}, $  \Gamma_{{\mathrm{1}}}  ,   \ottmv{x}  :^{  \ell  \: \sqcup \:  \ottnt{m}  }  \ottnt{A}    \vdash  \ottnt{b'}  :^{ \ell }  \ottnt{B} $. \\
This case, then, follows by the substitution lemma.

\end{itemize} 

\item \Rref{ST-LetPairD}. Have: $  \Gamma_{{\mathrm{1}}}  \sqcap  \Gamma_{{\mathrm{2}}}   \vdash   \mathbf{let} \: (  \ottmv{x} ^{ \ottnt{m} } ,  \ottmv{y}  ) \: \mathbf{be} \:  \ottnt{a}  \: \mathbf{in} \:  \ottnt{b}   :^{ \ell }  \ottnt{B} $ where $ \Gamma_{{\mathrm{1}}}  \vdash  \ottnt{a}  :^{ \ell }   {}^{ \ottnt{m} }\!  \ottnt{A_{{\mathrm{1}}}}  \: \times \:  \ottnt{A_{{\mathrm{2}}}}  $ and $    \Gamma_{{\mathrm{2}}}  ,   \ottmv{x}  :^{  \ell  \: \sqcup \:  \ottnt{m}  }  \ottnt{A_{{\mathrm{1}}}}     ,   \ottmv{y}  :^{ \ell }  \ottnt{A_{{\mathrm{2}}}}    \vdash  \ottnt{b}  :^{ \ell }  \ottnt{B} $. \\ Let $ \vdash   \mathbf{let} \: (  \ottmv{x} ^{ \ottnt{m} } ,  \ottmv{y}  ) \: \mathbf{be} \:  \ottnt{a}  \: \mathbf{in} \:  \ottnt{b}   \leadsto  \ottnt{c} $. By inversion:

\begin{itemize}
\item $ \vdash   \mathbf{let} \: (  \ottmv{x} ^{ \ottnt{m} } ,  \ottmv{y}  ) \: \mathbf{be} \:  \ottnt{a}  \: \mathbf{in} \:  \ottnt{b}   \leadsto   \mathbf{let} \: (  \ottmv{x} ^{ \ottnt{m} } ,  \ottmv{y}  ) \: \mathbf{be} \:  \ottnt{a'}  \: \mathbf{in} \:  \ottnt{b}  $, when $ \vdash  \ottnt{a}  \leadsto  \ottnt{a'} $.\\
Need to show: $  \Gamma_{{\mathrm{1}}}  \sqcap  \Gamma_{{\mathrm{2}}}   \vdash   \mathbf{let} \: (  \ottmv{x} ^{ \ottnt{m} } ,  \ottmv{y}  ) \: \mathbf{be} \:  \ottnt{a'}  \: \mathbf{in} \:  \ottnt{b}   :^{ \ell }  \ottnt{B} $.\\
Follows by IH and \rref{ST-LetPairD}.

\item $ \vdash   \mathbf{let} \: (  \ottmv{x} ^{ \ottnt{m} } ,  \ottmv{y}  ) \: \mathbf{be} \:   (  \ottnt{a_{{\mathrm{1}}}} ^{ \ottnt{m} } ,  \ottnt{a_{{\mathrm{2}}}}  )   \: \mathbf{in} \:  \ottnt{b}   \leadsto    \ottnt{b}  \{  \ottnt{a_{{\mathrm{1}}}}  /  \ottmv{x}  \}   \{  \ottnt{a_{{\mathrm{2}}}}  /  \ottmv{y}  \}  $.\\
Need to show: $  \Gamma_{{\mathrm{1}}}  \sqcap  \Gamma_{{\mathrm{2}}}   \vdash    \ottnt{b}  \{  \ottnt{a_{{\mathrm{1}}}}  /  \ottmv{x}  \}   \{  \ottnt{a_{{\mathrm{2}}}}  /  \ottmv{y}  \}   :^{ \ell }  \ottnt{B} $.\\
By inversion on $ \Gamma_{{\mathrm{1}}}  \vdash   (  \ottnt{a_{{\mathrm{1}}}} ^{ \ottnt{m} } ,  \ottnt{a_{{\mathrm{2}}}}  )   :^{ \ell }   {}^{ \ottnt{m} }\!  \ottnt{A_{{\mathrm{1}}}}  \: \times \:  \ottnt{A_{{\mathrm{2}}}}  $, we have:\\
$\exists \Gamma_{{\mathrm{11}}}, \Gamma_{{\mathrm{12}}}$ such that $ \Gamma_{{\mathrm{11}}}  \vdash  \ottnt{a_{{\mathrm{1}}}}  :^{  \ell  \: \sqcup \:  \ottnt{m}  }  \ottnt{A_{{\mathrm{1}}}} $ and $ \Gamma_{{\mathrm{12}}}  \vdash  \ottnt{a_{{\mathrm{2}}}}  :^{ \ell }  \ottnt{A_{{\mathrm{2}}}} $ and $\Gamma_{{\mathrm{1}}} =  \Gamma_{{\mathrm{11}}}  \sqcap  \Gamma_{{\mathrm{12}}} $.\\
This case, then, follows by applying the substitution lemma twice.
\end{itemize}

\item \Rref{ST-LetUnitD}. Have: $  \Gamma_{{\mathrm{1}}}  \sqcap  \Gamma_{{\mathrm{2}}}   \vdash   \mathbf{let} \: \mathbf{unit} \: \mathbf{be} \:  \ottnt{a}  \: \mathbf{in} \:  \ottnt{b}   :^{ \ell }  \ottnt{B} $ where $ \Gamma_{{\mathrm{1}}}  \vdash  \ottnt{a}  :^{ \ell }   \mathbf{Unit}  $ and $ \Gamma_{{\mathrm{2}}}  \vdash  \ottnt{b}  :^{ \ell }  \ottnt{B} $.\\
Let $ \vdash   \mathbf{let} \: \mathbf{unit} \: \mathbf{be} \:  \ottnt{a}  \: \mathbf{in} \:  \ottnt{b}   \leadsto  \ottnt{c} $. By inversion:

\begin{itemize}
\item $ \vdash   \mathbf{let} \: \mathbf{unit} \: \mathbf{be} \:  \ottnt{a}  \: \mathbf{in} \:  \ottnt{b}   \leadsto   \mathbf{let} \: \mathbf{unit} \: \mathbf{be} \:  \ottnt{a'}  \: \mathbf{in} \:  \ottnt{b}  $, when $ \vdash  \ottnt{a}  \leadsto  \ottnt{a'} $.\\
Need to show: $  \Gamma_{{\mathrm{1}}}  \sqcap  \Gamma_{{\mathrm{2}}}   \vdash   \mathbf{let} \: \mathbf{unit} \: \mathbf{be} \:  \ottnt{a'}  \: \mathbf{in} \:  \ottnt{b}   :^{ \ell }  \ottnt{B} $.\\
Follows by IH and \rref{ST-LetUnitD}.

\item $ \vdash   \mathbf{let} \: \mathbf{unit} \: \mathbf{be} \:   \mathbf{unit}   \: \mathbf{in} \:  \ottnt{b}   \leadsto  \ottnt{b} $.\\
Need to show: $  \Gamma_{{\mathrm{1}}}  \sqcap  \Gamma_{{\mathrm{2}}}   \vdash  \ottnt{b}  :^{ \ell }  \ottnt{B} $.\\
Follows by \rref{ST-SubLD}.
\end{itemize}

\item \Rref{ST-CaseD}. Have: $  \Gamma_{{\mathrm{1}}}  \sqcap  \Gamma_{{\mathrm{2}}}   \vdash   \mathbf{case} \:  \ottnt{a}  \: \mathbf{of} \:  \ottmv{x_{{\mathrm{1}}}}  .  \ottnt{b_{{\mathrm{1}}}}  \: ; \:  \ottmv{x_{{\mathrm{2}}}}  .  \ottnt{b_{{\mathrm{2}}}}   :^{ \ell }  \ottnt{B} $ where $ \Gamma_{{\mathrm{1}}}  \vdash  \ottnt{a}  :^{ \ell }   \ottnt{A_{{\mathrm{1}}}}  +  \ottnt{A_{{\mathrm{2}}}}  $ and $  \Gamma_{{\mathrm{2}}}  ,   \ottmv{x_{{\mathrm{1}}}}  :^{ \ell }  \ottnt{A_{{\mathrm{1}}}}    \vdash  \ottnt{b_{{\mathrm{1}}}}  :^{ \ell }  \ottnt{B} $ and $  \Gamma_{{\mathrm{2}}}  ,   \ottmv{x_{{\mathrm{2}}}}  :^{ \ell }  \ottnt{A_{{\mathrm{2}}}}    \vdash  \ottnt{b_{{\mathrm{2}}}}  :^{ \ell }  \ottnt{B} $.\\ Let $ \vdash   \mathbf{case} \:  \ottnt{a}  \: \mathbf{of} \:  \ottmv{x_{{\mathrm{1}}}}  .  \ottnt{b_{{\mathrm{1}}}}  \: ; \:  \ottmv{x_{{\mathrm{2}}}}  .  \ottnt{b_{{\mathrm{2}}}}   \leadsto  \ottnt{c} $. By inversion:

\begin{itemize}
\item $ \vdash   \mathbf{case} \:  \ottnt{a}  \: \mathbf{of} \:  \ottmv{x_{{\mathrm{1}}}}  .  \ottnt{b_{{\mathrm{1}}}}  \: ; \:  \ottmv{x_{{\mathrm{2}}}}  .  \ottnt{b_{{\mathrm{2}}}}   \leadsto   \mathbf{case} \:  \ottnt{a'}  \: \mathbf{of} \:  \ottmv{x_{{\mathrm{1}}}}  .  \ottnt{b_{{\mathrm{1}}}}  \: ; \:  \ottmv{x_{{\mathrm{2}}}}  .  \ottnt{b_{{\mathrm{2}}}}  $, when $ \vdash  \ottnt{a}  \leadsto  \ottnt{a'} $.\\
Need to show: $  \Gamma_{{\mathrm{1}}}  \sqcap  \Gamma_{{\mathrm{2}}}   \vdash   \mathbf{case} \:  \ottnt{a'}  \: \mathbf{of} \:  \ottmv{x_{{\mathrm{1}}}}  .  \ottnt{b_{{\mathrm{1}}}}  \: ; \:  \ottmv{x_{{\mathrm{2}}}}  .  \ottnt{b_{{\mathrm{2}}}}   :^{ \ell }  \ottnt{B} $.\\
Follows by IH and \rref{ST-CaseD}.

\item $ \vdash   \mathbf{case} \:   (   \mathbf{inj}_1 \:  \ottnt{a_{{\mathrm{1}}}}   )   \: \mathbf{of} \:  \ottmv{x_{{\mathrm{1}}}}  .  \ottnt{b_{{\mathrm{1}}}}  \: ; \:  \ottmv{x_{{\mathrm{2}}}}  .  \ottnt{b_{{\mathrm{2}}}}   \leadsto   \ottnt{b_{{\mathrm{1}}}}  \{  \ottnt{a_{{\mathrm{1}}}}  /  \ottmv{x_{{\mathrm{1}}}}  \}  $.\\
Need to show $  \Gamma_{{\mathrm{1}}}  \sqcap  \Gamma_{{\mathrm{2}}}   \vdash   \ottnt{b_{{\mathrm{1}}}}  \{  \ottnt{a_{{\mathrm{1}}}}  /  \ottmv{x_{{\mathrm{1}}}}  \}   :^{ \ell }  \ottnt{B} $.\\
By inversion on $ \Gamma_{{\mathrm{1}}}  \vdash   \mathbf{inj}_1 \:  \ottnt{a_{{\mathrm{1}}}}   :^{ \ell }   \ottnt{A_{{\mathrm{1}}}}  +  \ottnt{A_{{\mathrm{2}}}}  $, we have: $ \Gamma_{{\mathrm{1}}}  \vdash  \ottnt{a_{{\mathrm{1}}}}  :^{ \ell }  \ottnt{A_{{\mathrm{1}}}} $.\\
This case, then, follows by applying the substitution lemma.

\item $ \vdash   \mathbf{case} \:   (   \mathbf{inj}_2 \:  \ottnt{a_{{\mathrm{2}}}}   )   \: \mathbf{of} \:  \ottmv{x_{{\mathrm{1}}}}  .  \ottnt{b_{{\mathrm{1}}}}  \: ; \:  \ottmv{x_{{\mathrm{2}}}}  .  \ottnt{b_{{\mathrm{2}}}}   \leadsto   \ottnt{b_{{\mathrm{2}}}}  \{  \ottnt{a_{{\mathrm{2}}}}  /  \ottmv{x_{{\mathrm{2}}}}  \}  $.\\
Similar to the previous case.

\end{itemize}
\item \Rref{ST-SubLD,ST-SubRD}. Follows by IH.
\end{itemize}
\end{proof}

%--------------------------------------------------------------------------------------------------

\begin{theorem}[Progress (Theorem \ref{DSProgress})]
If $  \emptyset   \vdash  \ottnt{a}  :^{ \ell }  \ottnt{A} $, then either $\ottnt{a}$ is a value or there exists $\ottnt{a'}$ such that $ \vdash  \ottnt{a}  \leadsto  \ottnt{a'} $.
\end{theorem}

\begin{proof}
By induction on $  \emptyset   \vdash  \ottnt{a}  :^{ \ell }  \ottnt{A} $.
\begin{itemize}

\item \Rref{ST-VarD}. Does not apply since the context here is empty.

\item \Rref{ST-AppD}. Have: $  \emptyset   \vdash   \ottnt{b}  \:  \ottnt{a} ^{ \ottnt{m} }   :^{ \ell }  \ottnt{B} $ where $  \emptyset   \vdash  \ottnt{b}  :^{ \ell }   {}^{ \ottnt{m} }\!  \ottnt{A}  \to  \ottnt{B}  $ and $  \emptyset   \vdash  \ottnt{a}  :^{  \ell  \: \sqcup \:  \ottnt{m}  }  \ottnt{A} $. \\
Need to show: $\exists c,  \vdash   \ottnt{b}  \:  \ottnt{a} ^{ \ottnt{m} }   \leadsto  \ottnt{c} $.\\
By IH, $\ottnt{b}$ is either a value or $ \vdash  \ottnt{b}  \leadsto  \ottnt{b'} $.\\
If $\ottnt{b}$ is a value, then $\ottnt{b} =  \lambda^{ \ottnt{m} }  \ottmv{x}  :  \ottnt{A}  .  \ottnt{b'} $ for some $\ottnt{b'}$. Therefore, $ \vdash   \ottnt{b}  \:  \ottnt{a} ^{ \ottnt{m} }   \leadsto   \ottnt{b'}  \{  \ottnt{a}  /  \ottmv{x}  \}  $.\\
Otherwise, $ \vdash   \ottnt{b}  \:  \ottnt{a} ^{ \ottnt{m} }   \leadsto   \ottnt{b'}  \:  \ottnt{a} ^{ \ottnt{m} }  $.

\item \Rref{ST-LetPairD}. Have: $  \emptyset   \vdash   \mathbf{let} \: (  \ottmv{x} ^{ \ottnt{m} } ,  \ottmv{y}  ) \: \mathbf{be} \:  \ottnt{a}  \: \mathbf{in} \:  \ottnt{b}   :^{ \ell }  \ottnt{B} $ where $  \emptyset   \vdash  \ottnt{a}  :^{ \ell }   {}^{ \ottnt{m} }\!  \ottnt{A_{{\mathrm{1}}}}  \: \times \:  \ottnt{A_{{\mathrm{2}}}}  $ and $   \ottmv{x}  :^{  \ell  \: \sqcup \:  \ottnt{m}  }  \ottnt{A_{{\mathrm{1}}}}   ,   \ottmv{y}  :^{ \ell }  \ottnt{A_{{\mathrm{2}}}}    \vdash  \ottnt{b}  :^{ \ell }  \ottnt{B} $.\\
Need to show: $\exists c,  \vdash   \mathbf{let} \: (  \ottmv{x} ^{ \ottnt{m} } ,  \ottmv{y}  ) \: \mathbf{be} \:  \ottnt{a}  \: \mathbf{in} \:  \ottnt{b}   \leadsto  \ottnt{c} $.\\
By IH, $\ottnt{a}$ is either a value or $ \vdash  \ottnt{a}  \leadsto  \ottnt{a'} $.\\
If $\ottnt{a}$ is a value, then $\ottnt{a} =  (  \ottnt{a_{{\mathrm{1}}}} ^{ \ottnt{m} } ,  \ottnt{a_{{\mathrm{2}}}}  ) $. Therefore,
$ \vdash   \mathbf{let} \: (  \ottmv{x} ^{ \ottnt{m} } ,  \ottmv{y}  ) \: \mathbf{be} \:  \ottnt{a}  \: \mathbf{in} \:  \ottnt{b}   \leadsto    \ottnt{b}  \{  \ottnt{a_{{\mathrm{1}}}}  /  \ottmv{x}  \}   \{  \ottnt{a_{{\mathrm{2}}}}  /  \ottmv{y}  \}  $.\\
Otherwise, $ \vdash   \mathbf{let} \: (  \ottmv{x} ^{ \ottnt{m} } ,  \ottmv{y}  ) \: \mathbf{be} \:  \ottnt{a}  \: \mathbf{in} \:  \ottnt{b}   \leadsto   \mathbf{let} \: (  \ottmv{x} ^{ \ottnt{m} } ,  \ottmv{y}  ) \: \mathbf{be} \:  \ottnt{a'}  \: \mathbf{in} \:  \ottnt{b}  $.

\item \Rref{ST-LetUnitD}. Have: $  \emptyset   \vdash   \mathbf{let} \: \mathbf{unit} \: \mathbf{be} \:  \ottnt{a}  \: \mathbf{in} \:  \ottnt{b}   :^{ \ell }  \ottnt{B} $ where $  \emptyset   \vdash  \ottnt{a}  :^{ \ell }   \mathbf{Unit}  $ and $  \emptyset   \vdash  \ottnt{b}  :^{ \ell }  \ottnt{B} $.\\
Need to show: $\exists c,  \vdash   \mathbf{let} \: \mathbf{unit} \: \mathbf{be} \:  \ottnt{a}  \: \mathbf{in} \:  \ottnt{b}   \leadsto  \ottnt{c} $.\\
By IH, $\ottnt{a}$ is either a value or $ \vdash  \ottnt{a}  \leadsto  \ottnt{a'} $.\\
If $\ottnt{a}$ is a value, then $\ottnt{a} =  \mathbf{unit} $. Therefore,
$ \vdash   \mathbf{let} \: \mathbf{unit} \: \mathbf{be} \:  \ottnt{a}  \: \mathbf{in} \:  \ottnt{b}   \leadsto  \ottnt{b} $.\\
Otherwise, $ \vdash   \mathbf{let} \: \mathbf{unit} \: \mathbf{be} \:  \ottnt{a}  \: \mathbf{in} \:  \ottnt{b}   \leadsto   \mathbf{let} \: \mathbf{unit} \: \mathbf{be} \:  \ottnt{a'}  \: \mathbf{in} \:  \ottnt{b}  $.

\item \Rref{ST-CaseD}. Have: $  \emptyset   \vdash   \mathbf{case} \:  \ottnt{a}  \: \mathbf{of} \:  \ottmv{x_{{\mathrm{1}}}}  .  \ottnt{b_{{\mathrm{1}}}}  \: ; \:  \ottmv{x_{{\mathrm{2}}}}  .  \ottnt{b_{{\mathrm{2}}}}   :^{ \ell }  \ottnt{B} $ where $  \emptyset   \vdash  \ottnt{a}  :^{ \ell }   \ottnt{A_{{\mathrm{1}}}}  +  \ottnt{A_{{\mathrm{2}}}}  $ and $  \ottmv{x_{{\mathrm{1}}}}  :^{ \ell }  \ottnt{A_{{\mathrm{1}}}}   \vdash  \ottnt{b_{{\mathrm{1}}}}  :^{ \ell }  \ottnt{B} $ and $  \ottmv{x_{{\mathrm{2}}}}  :^{ \ell }  \ottnt{A_{{\mathrm{2}}}}   \vdash  \ottnt{b_{{\mathrm{2}}}}  :^{ \ell }  \ottnt{B} $.\\
Need to show: $\exists c,  \vdash   \mathbf{case} \:  \ottnt{a}  \: \mathbf{of} \:  \ottmv{x_{{\mathrm{1}}}}  .  \ottnt{b_{{\mathrm{1}}}}  \: ; \:  \ottmv{x_{{\mathrm{2}}}}  .  \ottnt{b_{{\mathrm{2}}}}   \leadsto  \ottnt{c} $.\\
By IH, $\ottnt{a}$ is either a value or $ \vdash  \ottnt{a}  \leadsto  \ottnt{a'} $.\\
If $\ottnt{a}$ is a value, then $\ottnt{a} =  \mathbf{inj}_1 \:  \ottnt{a_{{\mathrm{1}}}} $ or $\ottnt{a} =  \mathbf{inj}_2 \:  \ottnt{a_{{\mathrm{2}}}} $. \\
Then, $ \vdash   \mathbf{case} \:  \ottnt{a}  \: \mathbf{of} \:  \ottmv{x_{{\mathrm{1}}}}  .  \ottnt{b_{{\mathrm{1}}}}  \: ; \:  \ottmv{x_{{\mathrm{2}}}}  .  \ottnt{b_{{\mathrm{2}}}}   \leadsto   \ottnt{b_{{\mathrm{1}}}}  \{  \ottnt{a_{{\mathrm{1}}}}  /  \ottmv{x_{{\mathrm{1}}}}  \}  $ or $ \vdash   \mathbf{case} \:  \ottnt{a}  \: \mathbf{of} \:  \ottmv{x_{{\mathrm{1}}}}  .  \ottnt{b_{{\mathrm{1}}}}  \: ; \:  \ottmv{x_{{\mathrm{2}}}}  .  \ottnt{b_{{\mathrm{2}}}}   \leadsto   \ottnt{b_{{\mathrm{2}}}}  \{  \ottnt{a_{{\mathrm{2}}}}  /  \ottmv{x_{{\mathrm{2}}}}  \}  $.\\
Otherwise, $ \vdash   \mathbf{case} \:  \ottnt{a}  \: \mathbf{of} \:  \ottmv{x_{{\mathrm{1}}}}  .  \ottnt{b_{{\mathrm{1}}}}  \: ; \:  \ottmv{x_{{\mathrm{2}}}}  .  \ottnt{b_{{\mathrm{2}}}}   \leadsto   \mathbf{case} \:  \ottnt{a'}  \: \mathbf{of} \:  \ottmv{x_{{\mathrm{1}}}}  .  \ottnt{b_{{\mathrm{1}}}}  \: ; \:  \ottmv{x_{{\mathrm{2}}}}  .  \ottnt{b_{{\mathrm{2}}}}  $.

\item \Rref{ST-SubLD, ST-SubRD}. Follows by IH.

\item \Rref{ST-LamD,ST-PairD,ST-UnitD}. The terms typed by these rules are values.

\end{itemize}
\end{proof}

%---------------------------------------------------------------------------------------------------
\section{Heap Semantics For Simply-Typed LDC}
%---------------------------------------------------------------------------------------------------

\begin{lemma}[Similarity] \label{HeapSim}
If $ [  \ottnt{H}  ]  \ottnt{a}  \Longrightarrow^{ \ottnt{q} }_{ \ottnt{S} } [  \ottnt{H'}  ]  \ottnt{a'} $, then $  \ottnt{a}  \{  \ottnt{H}  \}   =   \ottnt{a'}  \{  \ottnt{H'}  \}  $ or $ \vdash   \ottnt{a}  \{  \ottnt{H}  \}   \leadsto   \ottnt{a'}  \{  \ottnt{H'}  \}  $. Here, $ \ottnt{a}  \{  \ottnt{H}  \} $ denotes the term obtained by substituting in $\ottnt{a}$ the definitions in $\ottnt{H}$, in reverse order.
\end{lemma} 

\begin{proof}
By induction on $ [  \ottnt{H}  ]  \ottnt{a}  \Longrightarrow^{ \ottnt{q} }_{ \ottnt{S} } [  \ottnt{H'}  ]  \ottnt{a'} $.
\end{proof}

%---------------------------------------------------------------------------------------------------

\begin{lemma}[Unchanged (Lemma \ref{HeapUnchanged})]
If  $ [     \ottnt{H_{{\mathrm{1}}}}  ,   \ottmv{x}  \overset{ \ottnt{r} }{\mapsto}  \ottnt{a}     ,  \ottnt{H_{{\mathrm{2}}}}   ]  \ottnt{c}  \Longrightarrow^{ \ottnt{q} }_{ \ottnt{S} } [     \ottnt{H'_{{\mathrm{1}}}}  ,   \ottmv{x}  \overset{ \ottnt{r'} }{\mapsto}  \ottnt{a}     ,  \ottnt{H'_{{\mathrm{2}}}}   ]  \ottnt{c'} $ (where $\lvert H_1 \rvert = \lvert H'_1 \rvert$) and $\neg(\exists q_0, r = q + q_0)$, then $\ottnt{r'} = r$. 
\end{lemma}

\begin{proof}
By induction on $ [     \ottnt{H_{{\mathrm{1}}}}  ,   \ottmv{x}  \overset{ \ottnt{r} }{\mapsto}  \ottnt{a}     ,  \ottnt{H_{{\mathrm{2}}}}   ]  \ottnt{c}  \Longrightarrow^{ \ottnt{q} }_{ \ottnt{S} } [     \ottnt{H'_{{\mathrm{1}}}}  ,   \ottmv{x}  \overset{ \ottnt{r'} }{\mapsto}  \ottnt{a}     ,  \ottnt{H'_{{\mathrm{2}}}}   ]  \ottnt{c'} $. There are two interesting cases:
\begin{itemize}
\item \Rref{HeapStep-Var}. Have: $ [     \ottnt{H_{{\mathrm{1}}}}  ,   \ottmv{x}  \overset{ \ottnt{r} }{\mapsto}  \ottnt{a}     ,  \ottnt{H_{{\mathrm{2}}}}   ]  \ottnt{c}  \Longrightarrow^{ \ottnt{q} }_{ \ottnt{S} } [     \ottnt{H'_{{\mathrm{1}}}}  ,   \ottmv{x}  \overset{ \ottnt{r'} }{\mapsto}  \ottnt{a}     ,  \ottnt{H'_{{\mathrm{2}}}}   ]  \ottnt{c'} $, where $c$ is a variable that is defined either in $\ottnt{H_{{\mathrm{1}}}}$ or in $\ottnt{H_{{\mathrm{2}}}}$. In either case, $\ottnt{r'} = r$. Observe that $\ottnt{c}$ cannot be $\ottmv{x}$ because then this rule would not apply, owing to the condition, $\neg(\exists q_0, r = q + q_0)$.  
\item \Rref{HeapStep-Discard}. We case split based on the parametrizing structure: 
\begin{itemize}
\item When the parametrizing structure is $ \mathbb{N}_{=} $, this rule is superfluous.
\item When the parametrizing structure is $ \mathbb{N}_{\geq} $, we have:\\
$ [     \ottnt{H_{{\mathrm{1}}}}  ,   \ottmv{x}  \overset{ \ottnt{r} }{\mapsto}  \ottnt{a}     ,  \ottnt{H_{{\mathrm{2}}}}   ]  \ottnt{c}  \Longrightarrow^{ \ottnt{q'} }_{ \ottnt{S} } [     \ottnt{H'_{{\mathrm{1}}}}  ,   \ottmv{x}  \overset{ \ottnt{r'} }{\mapsto}  \ottnt{a}     ,  \ottnt{H'_{{\mathrm{2}}}}   ]  \ottnt{c'} $ where $ [     \ottnt{H_{{\mathrm{1}}}}  ,   \ottmv{x}  \overset{ \ottnt{r} }{\mapsto}  \ottnt{a}     ,  \ottnt{H_{{\mathrm{2}}}}   ]  \ottnt{c}  \Longrightarrow^{ \ottnt{q} }_{ \ottnt{S} } [     \ottnt{H'_{{\mathrm{1}}}}  ,   \ottmv{x}  \overset{ \ottnt{r'} }{\mapsto}  \ottnt{a}     ,  \ottnt{H'_{{\mathrm{2}}}}   ]  \ottnt{c'} $ and $ \ottnt{q}  <:  \ottnt{q'} $. 
Further, $\neg(\exists q_0, r = q' + q_0)$. \\
Now, since $ \ottnt{q}  <:  \ottnt{q'} $, we know, $\exists \ottnt{q_{{\mathrm{1}}}}, \ottnt{q} =  \ottnt{q'}  +  \ottnt{q_{{\mathrm{1}}}} $. Then, if for some $\ottnt{q_{{\mathrm{2}}}}$, the equation $\ottnt{r} =  \ottnt{q}  +  \ottnt{q_{{\mathrm{2}}}} $ holds, we would reach a contradiction because in that case, the equation, $\ottnt{r} =  \ottnt{q'}  +   (   \ottnt{q_{{\mathrm{1}}}}  +  \ottnt{q_{{\mathrm{2}}}}   )  $, would also hold.\\
Therefore, $\neg(\exists q_0, r = q + q_0)$.\\
This case, then, follows by IH.
\item When the parametrizing structure is a lattice $ \mathcal{L} $, we have:\\
$ [     \ottnt{H_{{\mathrm{1}}}}  ,   \ottmv{x}  \overset{ \ottnt{m} }{\mapsto}  \ottnt{a}     ,  \ottnt{H_{{\mathrm{2}}}}   ]  \ottnt{c}  \Longrightarrow^{ \ell' }_{ \ottnt{S} } [     \ottnt{H'_{{\mathrm{1}}}}  ,   \ottmv{x}  \overset{ \ottnt{m'} }{\mapsto}  \ottnt{a}     ,  \ottnt{H'_{{\mathrm{2}}}}   ]  \ottnt{c'} $ where $ [     \ottnt{H_{{\mathrm{1}}}}  ,   \ottmv{x}  \overset{ \ottnt{m} }{\mapsto}  \ottnt{a}     ,  \ottnt{H_{{\mathrm{2}}}}   ]  \ottnt{c}  \Longrightarrow^{ \ell }_{ \ottnt{S} } [     \ottnt{H'_{{\mathrm{1}}}}  ,   \ottmv{x}  \overset{ \ottnt{m'} }{\mapsto}  \ottnt{a}     ,  \ottnt{H'_{{\mathrm{2}}}}   ]  \ottnt{c'} $ and $ \ell  \sqsubseteq  \ell' $. 
Further, $\neg( \ottnt{m}  \sqsubseteq  \ell' )$. \\
Then, $\neg( \ottnt{m}  \sqsubseteq  \ell )$ because otherwise, $ \ottnt{m}  \sqsubseteq  \ell   \sqsubseteq  \ell'$.\\
This case, then, follows by IH.
\end{itemize}
\end{itemize}
\end{proof}

%--------------------------------------------------------------------------------------------------

\begin{lemma}[Irrelevant (Lemma \ref{irrel})] \label{irrelP}
If $ [     \ottnt{H_{{\mathrm{1}}}}  ,   \ottmv{x}  \overset{ \ottnt{r} }{\mapsto}  \ottnt{a}     ,  \ottnt{H_{{\mathrm{2}}}}   ]  \ottnt{c}  \Longrightarrow^{ \ottnt{q} }_{  \ottnt{S}  \, \cup \,   \textit{fv} \:  \ottnt{b}   } [     \ottnt{H'_{{\mathrm{1}}}}  ,   \ottmv{x}  \overset{ \ottnt{r'} }{\mapsto}  \ottnt{a}     ,  \ottnt{H'_{{\mathrm{2}}}}   ]  \ottnt{c'} $ (where $\lvert \ottnt{H_{{\mathrm{1}}}} \rvert = \lvert \ottnt{H'_{{\mathrm{1}}}} \rvert$) and $\neg(\exists q_0, r = q + q_0)$, then $ [     \ottnt{H_{{\mathrm{1}}}}  ,   \ottmv{x}  \overset{ \ottnt{r} }{\mapsto}  \ottnt{b}     ,  \ottnt{H_{{\mathrm{2}}}}   ]  \ottnt{c}  \Longrightarrow^{ \ottnt{q} }_{  \ottnt{S}  \, \cup \,   \textit{fv} \:  \ottnt{a}   } [     \ottnt{H'_{{\mathrm{1}}}}  ,   \ottmv{x}  \overset{ \ottnt{r'} }{\mapsto}  \ottnt{b}     ,  \ottnt{H'_{{\mathrm{2}}}}   ]  \ottnt{c'} $.
\end{lemma}

\begin{proof}
By induction on  $ [     \ottnt{H_{{\mathrm{1}}}}  ,   \ottmv{x}  \overset{ \ottnt{r} }{\mapsto}  \ottnt{a}     ,  \ottnt{H_{{\mathrm{2}}}}   ]  \ottnt{c}  \Longrightarrow^{ \ottnt{q} }_{  \ottnt{S}  \, \cup \,   \textit{fv} \:  \ottnt{b}   } [     \ottnt{H'_{{\mathrm{1}}}}  ,   \ottmv{x}  \overset{ \ottnt{r'} }{\mapsto}  \ottnt{a}     ,  \ottnt{H'_{{\mathrm{2}}}}   ]  \ottnt{c'} $.
\end{proof}

%---------------------------------------------------------------------------------------------------

\begin{lemma} \label{heaphelper}
If $ \ottnt{H}  \models   \Gamma_{{\mathrm{1}}}  +  \Gamma_{{\mathrm{2}}}  $ and $ \Gamma_{{\mathrm{2}}}  \vdash  \ottnt{a}  :^{ \ottnt{q} }  \ottnt{A} $ and $q \neq 0$, then either $\ottnt{a}$ is a value or there exists $\ottnt{H'}, \Gamma'_{{\mathrm{2}}}, \ottnt{a'}$ such that:
\begin{itemize}
\item $ [  \ottnt{H}  ]  \ottnt{a}  \Longrightarrow^{ \ottnt{q} }_{ \ottnt{S} } [  \ottnt{H'}  ]  \ottnt{a'} $
\item $ \ottnt{H'}  \models   \Gamma_{{\mathrm{1}}}  +  \Gamma'_{{\mathrm{2}}}  $
\item $ \Gamma'_{{\mathrm{2}}}  \vdash  \ottnt{a'}  :^{ \ottnt{q} }  \ottnt{A} $
\end{itemize} 
(Note that $+$ is overloaded here: $ \Gamma_{{\mathrm{1}}}  +  \Gamma_{{\mathrm{2}}} $ denotes addition of contexts $\Gamma_{{\mathrm{1}}}$ and $\Gamma_{{\mathrm{2}}}$ after padding them as necessary.)
\end{lemma}

\begin{proof}
By induction on $ \Gamma_{{\mathrm{2}}}  \vdash  \ottnt{a}  :^{ \ottnt{q} }  \ottnt{A} $.

\begin{itemize}

\item \Rref{ST-Var}. Have $      0   \cdot  \Gamma_{{\mathrm{21}}}   ,   \ottmv{x}  :^{ \ottnt{q_{{\mathrm{2}}}} }  \ottnt{A}     ,    0   \cdot  \Gamma_{{\mathrm{22}}}    \vdash   \ottmv{x}   :^{ \ottnt{q_{{\mathrm{2}}}} }  \ottnt{A} $.\\
Further, $ \ottnt{H}  \models    (     \Gamma_{{\mathrm{11}}}  ,   \ottmv{x}  :^{ \ottnt{q_{{\mathrm{1}}}} }  \ottnt{A}     ,  \Gamma_{{\mathrm{12}}}   )   +   (       0   \cdot  \Gamma_{{\mathrm{21}}}   ,   \ottmv{x}  :^{ \ottnt{q_{{\mathrm{2}}}} }  \ottnt{A}     ,    0   \cdot  \Gamma_{{\mathrm{22}}}    )   $.\\
By inversion, $\exists \ottnt{H_{{\mathrm{1}}}}, \ottnt{q}, \ottnt{a}, \ottnt{H_{{\mathrm{2}}}}, \Gamma_{{\mathrm{0}}}$ and $\ottnt{q_{{\mathrm{0}}}}$ such that $H =    \ottnt{H_{{\mathrm{1}}}}  ,   \ottmv{x}  \overset{ \ottnt{q} }{\mapsto}  \ottnt{a}     ,  \ottnt{H_{{\mathrm{2}}}} $ and $ \Gamma_{{\mathrm{0}}}  \vdash  \ottnt{a}  :^{ \ottnt{q} }  \ottnt{A} $ and $q =  \ottnt{q_{{\mathrm{1}}}}  +  \ottnt{q_{{\mathrm{2}}}}  + \ottnt{q_{{\mathrm{0}}}}$.\\
By lemma \ref{BLSSplitP}, $\exists \Gamma_{{\mathrm{01}}}, \Gamma_{{\mathrm{02}}}$ such that $ \Gamma_{{\mathrm{01}}}  \vdash  \ottnt{a}  :^{  \ottnt{q_{{\mathrm{0}}}}  +  \ottnt{q_{{\mathrm{1}}}}  }  \ottnt{A} $ and $ \Gamma_{{\mathrm{02}}}  \vdash  \ottnt{a}  :^{ \ottnt{q_{{\mathrm{2}}}} }  \ottnt{A} $ and $ \Gamma_{{\mathrm{0}}}  =   \Gamma_{{\mathrm{01}}}  +  \Gamma_{{\mathrm{02}}}  $.\\
Then, we have,
\begin{itemize}
\item $ [     \ottnt{H_{{\mathrm{1}}}}  ,   \ottmv{x}  \overset{     \ottnt{q_{{\mathrm{0}}}}  +  \ottnt{q_{{\mathrm{1}}}}    +  \ottnt{q_{{\mathrm{2}}}}   }{\mapsto}  \ottnt{a}     ,  \ottnt{H_{{\mathrm{2}}}}   ]   \ottmv{x}   \Longrightarrow^{ \ottnt{q_{{\mathrm{2}}}} }_{ \ottnt{S} } [     \ottnt{H_{{\mathrm{1}}}}  ,   \ottmv{x}  \overset{   \ottnt{q_{{\mathrm{0}}}}  +  \ottnt{q_{{\mathrm{1}}}}   }{\mapsto}  \ottnt{a}     ,  \ottnt{H_{{\mathrm{2}}}}   ]  \ottnt{a} $.
\item $    \ottnt{H_{{\mathrm{1}}}}  ,   \ottmv{x}  \overset{   \ottnt{q_{{\mathrm{0}}}}  +  \ottnt{q_{{\mathrm{1}}}}   }{\mapsto}  \ottnt{a}     ,  \ottnt{H_{{\mathrm{2}}}}   \models    (     \Gamma_{{\mathrm{11}}}  ,   \ottmv{x}  :^{ \ottnt{q_{{\mathrm{1}}}} }  \ottnt{A}     ,  \Gamma_{{\mathrm{12}}}   )   +   (     \Gamma_{{\mathrm{02}}}  ,   \ottmv{x}  :^{  0  }  \ottnt{A}     ,    0   \cdot  \Gamma_{{\mathrm{22}}}    )   $.
\item $    \Gamma_{{\mathrm{02}}}  ,   \ottmv{x}  :^{  0  }  \ottnt{A}     ,    0   \cdot  \Gamma_{{\mathrm{22}}}    \vdash  \ottnt{a}  :^{ \ottnt{q_{{\mathrm{2}}}} }  \ottnt{A} $.
\end{itemize}

\item \Rref{ST-App}. Have: $  \Gamma_{{\mathrm{21}}}  +  \Gamma_{{\mathrm{22}}}   \vdash   \ottnt{b}  \:  \ottnt{a} ^{ \ottnt{r} }   :^{ \ottnt{q} }  \ottnt{B} $ where $ \Gamma_{{\mathrm{21}}}  \vdash  \ottnt{b}  :^{ \ottnt{q} }   {}^{ \ottnt{r} }\!  \ottnt{A}  \to  \ottnt{B}  $ and $ \Gamma_{{\mathrm{22}}}  \vdash  \ottnt{a}  :^{  \ottnt{q}  \cdot  \ottnt{r}  }  \ottnt{A} $.\\
Further, $ \ottnt{H}  \models   \Gamma_{{\mathrm{1}}}  +   (   \Gamma_{{\mathrm{21}}}  +  \Gamma_{{\mathrm{22}}}   )   $.\\
By IH, if $\ottnt{b}$ steps, then $ [  \ottnt{H}  ]  \ottnt{b}  \Longrightarrow^{ \ottnt{q} }_{  \ottnt{S}  \, \cup \,   \textit{fv} \:  \ottnt{a}   } [  \ottnt{H'}  ]  \ottnt{b'} $ and $ \Gamma'_{{\mathrm{21}}}  \vdash  \ottnt{b'}  :^{ \ottnt{q} }   {}^{ \ottnt{r} }\!  \ottnt{A}  \to  \ottnt{B}  $ and $ \ottnt{H'}  \models   \Gamma_{{\mathrm{1}}}  +   (   \Gamma'_{{\mathrm{21}}}  +  \Gamma_{{\mathrm{22}}}   )   $.\\
Then, $ [  \ottnt{H}  ]   \ottnt{b}  \:  \ottnt{a} ^{ \ottnt{r} }   \Longrightarrow^{ \ottnt{q} }_{ \ottnt{S} } [  \ottnt{H'}  ]   \ottnt{b'}  \:  \ottnt{a} ^{ \ottnt{r} }  $ and $  \Gamma'_{{\mathrm{21}}}  +  \Gamma_{{\mathrm{22}}}   \vdash   \ottnt{b'}  \:  \ottnt{a} ^{ \ottnt{r} }   :^{ \ottnt{q} }  \ottnt{B} $ and $ \ottnt{H'}  \models   \Gamma_{{\mathrm{1}}}  +   (   \Gamma'_{{\mathrm{21}}}  +  \Gamma_{{\mathrm{22}}}   )   $.\\

Otherwise, $\ottnt{b}$ is a value. By inversion $\ottnt{b} =  \lambda^{ \ottnt{r} }  \ottmv{x}  :  \ottnt{A}  .  \ottnt{b'} $.\\
Also, by inversion, $  \Gamma_{{\mathrm{21}}}  ,   \ottmv{x}  :^{  \ottnt{q}  \cdot  \ottnt{r}  }  \ottnt{A}    \vdash  \ottnt{b'}  :^{ \ottnt{q} }  \ottnt{B} $.\\
Then, $ [  \ottnt{H}  ]    (   \lambda^{ \ottnt{r} }  \ottmv{x}  :  \ottnt{A}  .  \ottnt{b'}   )   \:  \ottnt{a}   \Longrightarrow^{ \ottnt{q} }_{ \ottnt{S} } [   \ottnt{H}  ,   \ottmv{x}  \overset{  \ottnt{q}  \cdot  \ottnt{r}  }{\mapsto}  \ottnt{a}    ]  \ottnt{b'} $ (assuming $ \ottmv{x}  \: \textit{fresh} $).\\
Further, $  \ottnt{H}  ,   \ottmv{x}  \overset{  \ottnt{q}  \cdot  \ottnt{r}  }{\mapsto}  \ottnt{a}    \models    (   \Gamma_{{\mathrm{1}}}  +  \Gamma_{{\mathrm{21}}}   )   ,   \ottmv{x}  :^{  \ottnt{q}  \cdot  \ottnt{r}  }  \ottnt{A}   $. 

\item \Rref{ST-LetUnit}. Have: $  \Gamma_{{\mathrm{21}}}  +  \Gamma_{{\mathrm{22}}}   \vdash   \mathbf{let}_{ \ottnt{q_{{\mathrm{0}}}} } \: \mathbf{unit} \: \mathbf{be} \:  \ottnt{a}  \: \mathbf{in} \:  \ottnt{b}   :^{ \ottnt{q} }  \ottnt{B} $ where $ \Gamma_{{\mathrm{21}}}  \vdash  \ottnt{a}  :^{  \ottnt{q}  \cdot  \ottnt{q_{{\mathrm{0}}}}  }   \mathbf{Unit}  $ and $ \Gamma_{{\mathrm{22}}}  \vdash  \ottnt{b}  :^{ \ottnt{q} }  \ottnt{B} $. \\
Further, $ \ottnt{H}  \models   \Gamma_{{\mathrm{1}}}  +   (   \Gamma_{{\mathrm{21}}}  +  \Gamma_{{\mathrm{22}}}   )   $.\\
By IH, if $\ottnt{a}$ steps, then $ [  \ottnt{H}  ]  \ottnt{a}  \Longrightarrow^{  \ottnt{q}  \cdot  \ottnt{q_{{\mathrm{0}}}}  }_{  \ottnt{S}  \, \cup \,   \textit{fv} \:  \ottnt{b}   } [  \ottnt{H'}  ]  \ottnt{a'} $ and $ \ottnt{H'}  \models   \Gamma_{{\mathrm{1}}}  +   (   \Gamma'_{{\mathrm{21}}}  +  \Gamma_{{\mathrm{22}}}   )   $ and $ \Gamma'_{{\mathrm{21}}}  \vdash  \ottnt{a'}  :^{  \ottnt{q}  \cdot  \ottnt{q_{{\mathrm{0}}}}  }   \mathbf{Unit}  $.\\
Then, $ [  \ottnt{H}  ]   \mathbf{let}_{ \ottnt{q_{{\mathrm{0}}}} } \: \mathbf{unit} \: \mathbf{be} \:  \ottnt{a}  \: \mathbf{in} \:  \ottnt{b}   \Longrightarrow^{  \ottnt{q}  \cdot  \ottnt{q_{{\mathrm{0}}}}  }_{ \ottnt{S} } [  \ottnt{H'}  ]   \mathbf{let}_{ \ottnt{q_{{\mathrm{0}}}} } \: \mathbf{unit} \: \mathbf{be} \:  \ottnt{a'}  \: \mathbf{in} \:  \ottnt{b}  $.\\
By \rref{HeapStep-Discard}, $ [  \ottnt{H}  ]   \mathbf{let}_{ \ottnt{q_{{\mathrm{0}}}} } \: \mathbf{unit} \: \mathbf{be} \:  \ottnt{a}  \: \mathbf{in} \:  \ottnt{b}   \Longrightarrow^{ \ottnt{q} }_{ \ottnt{S} } [  \ottnt{H'}  ]   \mathbf{let}_{ \ottnt{q_{{\mathrm{0}}}} } \: \mathbf{unit} \: \mathbf{be} \:  \ottnt{a'}  \: \mathbf{in} \:  \ottnt{b}  $. ($\because  \ottnt{q_{{\mathrm{0}}}}  <:   1  $)\\
And by \rref{ST-LetUnit}, $  \Gamma'_{{\mathrm{21}}}  +  \Gamma_{{\mathrm{22}}}   \vdash   \mathbf{let}_{ \ottnt{q_{{\mathrm{0}}}} } \: \mathbf{unit} \: \mathbf{be} \:  \ottnt{a}  \: \mathbf{in} \:  \ottnt{b}   :^{ \ottnt{q} }  \ottnt{B} $.\\

Otherwise, $\ottnt{a}$ is a value. By inversion, $\ottnt{a} =  \mathbf{unit} $.\\
Then, $ [  \ottnt{H}  ]   \mathbf{let}_{ \ottnt{q_{{\mathrm{0}}}} } \: \mathbf{unit} \: \mathbf{be} \:   \mathbf{unit}   \: \mathbf{in} \:  \ottnt{b}   \Longrightarrow^{ \ottnt{q} }_{ \ottnt{S} } [  \ottnt{H}  ]  \ottnt{b} $. \\
Further, $  \Gamma_{{\mathrm{21}}}  +  \Gamma_{{\mathrm{22}}}   \vdash  \ottnt{b}  :^{ \ottnt{q} }  \ottnt{B} $. ($\because$ In case of $ \mathbb{N}_{=} $, $ \overline{ \Gamma_{{\mathrm{21}}} }  = \overline{0}$ and in case of $ \mathbb{N}_{\geq} $, $  \Gamma_{{\mathrm{21}}}  +  \Gamma_{{\mathrm{22}}}   <:  \Gamma_{{\mathrm{22}}} $.)

\item \Rref{ST-LetPair}. Have: $  \Gamma_{{\mathrm{21}}}  +  \Gamma_{{\mathrm{22}}}   \vdash   \mathbf{let}_{ \ottnt{q_{{\mathrm{0}}}} } \: (  \ottmv{x} ^{ \ottnt{r} } ,  \ottmv{y}  ) \: \mathbf{be} \:  \ottnt{a}  \: \mathbf{in} \:  \ottnt{b}   :^{ \ottnt{q} }  \ottnt{B} $ where $ \Gamma_{{\mathrm{21}}}  \vdash  \ottnt{a}  :^{  \ottnt{q}  \cdot  \ottnt{q_{{\mathrm{0}}}}  }   {}^{ \ottnt{r} }\!  \ottnt{A_{{\mathrm{1}}}}  \: \times \:  \ottnt{A_{{\mathrm{2}}}}  $ and $  \Gamma_{{\mathrm{22}}}  ,     \ottmv{x}  :^{  \ottnt{q}  \cdot    \ottnt{q_{{\mathrm{0}}}}  \cdot  \ottnt{r}    }  \ottnt{A_{{\mathrm{1}}}}   ,   \ottmv{y}  :^{  \ottnt{q}  \cdot  \ottnt{q_{{\mathrm{0}}}}  }  \ottnt{A_{{\mathrm{2}}}}      \vdash  \ottnt{b}  :^{ \ottnt{q} }  \ottnt{B} $.\\
Further, $ \ottnt{H}  \models   \Gamma_{{\mathrm{1}}}  +   (   \Gamma_{{\mathrm{21}}}  +  \Gamma_{{\mathrm{22}}}   )   $.\\
By IH, if $\ottnt{a}$ steps, then $ [  \ottnt{H}  ]  \ottnt{a}  \Longrightarrow^{  \ottnt{q}  \cdot  \ottnt{q_{{\mathrm{0}}}}  }_{  \ottnt{S}  \, \cup \,   \textit{fv} \:  \ottnt{b}   } [  \ottnt{H'}  ]  \ottnt{a'} $ and $ \ottnt{H'}  \models   \Gamma_{{\mathrm{1}}}  +   (   \Gamma'_{{\mathrm{21}}}  +  \Gamma_{{\mathrm{22}}}   )   $ and $ \Gamma'_{{\mathrm{21}}}  \vdash  \ottnt{a'}  :^{  \ottnt{q}  \cdot  \ottnt{q_{{\mathrm{0}}}}  }   {}^{ \ottnt{r} }\!  \ottnt{A_{{\mathrm{1}}}}  \: \times \:  \ottnt{A_{{\mathrm{2}}}}  $.\\
Then, $ [  \ottnt{H}  ]   \mathbf{let}_{ \ottnt{q_{{\mathrm{0}}}} } \: (  \ottmv{x} ^{ \ottnt{r} } ,  \ottmv{y}  ) \: \mathbf{be} \:  \ottnt{a}  \: \mathbf{in} \:  \ottnt{b}   \Longrightarrow^{  \ottnt{q}  \cdot  \ottnt{q_{{\mathrm{0}}}}  }_{ \ottnt{S} } [  \ottnt{H'}  ]   \mathbf{let}_{ \ottnt{q_{{\mathrm{0}}}} } \: (  \ottmv{x} ^{ \ottnt{r} } ,  \ottmv{y}  ) \: \mathbf{be} \:  \ottnt{a'}  \: \mathbf{in} \:  \ottnt{b}  $.\\
By \rref{HeapStep-Discard}, $ [  \ottnt{H}  ]   \mathbf{let}_{ \ottnt{q_{{\mathrm{0}}}} } \: (  \ottmv{x} ^{ \ottnt{r} } ,  \ottmv{y}  ) \: \mathbf{be} \:  \ottnt{a}  \: \mathbf{in} \:  \ottnt{b}   \Longrightarrow^{ \ottnt{q} }_{ \ottnt{S} } [  \ottnt{H'}  ]   \mathbf{let}_{ \ottnt{q_{{\mathrm{0}}}} } \: (  \ottmv{x} ^{ \ottnt{r} } ,  \ottmv{y}  ) \: \mathbf{be} \:  \ottnt{a'}  \: \mathbf{in} \:  \ottnt{b}  $. ($\because  \ottnt{q_{{\mathrm{0}}}}  <:   1  $)\\ 
And by \rref{ST-LetPair}, $  \Gamma'_{{\mathrm{21}}}  +  \Gamma_{{\mathrm{22}}}   \vdash   \mathbf{let}_{ \ottnt{q_{{\mathrm{0}}}} } \: (  \ottmv{x} ^{ \ottnt{r} } ,  \ottmv{y}  ) \: \mathbf{be} \:  \ottnt{a'}  \: \mathbf{in} \:  \ottnt{b}   : \:  \ottnt{B} $.\\

Otherwise, $\ottnt{a}$ is a value. By inversion, $\ottnt{a} =  (  \ottnt{a_{{\mathrm{1}}}} ^{ \ottnt{r} } ,  \ottnt{a_{{\mathrm{2}}}}  ) $.\\
Further, $\exists \Gamma_{{\mathrm{211}}}, \Gamma_{{\mathrm{212}}}$ such that $ \Gamma_{{\mathrm{211}}}  \vdash  \ottnt{a_{{\mathrm{1}}}}  :^{  \ottnt{q}  \cdot    \ottnt{q_{{\mathrm{0}}}}  \cdot  \ottnt{r}    }  \ottnt{A_{{\mathrm{1}}}} $ and $ \Gamma_{{\mathrm{212}}}  \vdash  \ottnt{a_{{\mathrm{2}}}}  :^{  \ottnt{q}  \cdot  \ottnt{q_{{\mathrm{0}}}}  }  \ottnt{A_{{\mathrm{2}}}} $ and $\Gamma_{{\mathrm{21}}} =  \Gamma_{{\mathrm{211}}}  +  \Gamma_{{\mathrm{212}}} $.\\
Now, $ [  \ottnt{H}  ]   \mathbf{let}_{ \ottnt{q_{{\mathrm{0}}}} } \: (  \ottmv{x} ^{ \ottnt{r} } ,  \ottmv{y}  ) \: \mathbf{be} \:   (  \ottnt{a_{{\mathrm{1}}}} ^{ \ottnt{r} } ,  \ottnt{a_{{\mathrm{2}}}}  )   \: \mathbf{in} \:  \ottnt{b}   \Longrightarrow^{ \ottnt{q} }_{ \ottnt{S} } [     \ottnt{H}  ,   \ottmv{x}  \overset{  \ottnt{q}  \cdot    \ottnt{q_{{\mathrm{0}}}}  \cdot  \ottnt{r}    }{\mapsto}  \ottnt{a_{{\mathrm{1}}}}     ,   \ottmv{y}  \overset{  \ottnt{q}  \cdot  \ottnt{q_{{\mathrm{0}}}}  }{\mapsto}  \ottnt{a_{{\mathrm{2}}}}    ]  \ottnt{b} $ (assuming $x, \ottmv{y}  \: \textit{fresh} $).\\
And, $    \ottnt{H}  ,   \ottmv{x}  \overset{  \ottnt{q}  \cdot    \ottnt{q_{{\mathrm{0}}}}  \cdot  \ottnt{r}    }{\mapsto}  \ottnt{a_{{\mathrm{1}}}}     ,   \ottmv{y}  \overset{  \ottnt{q}  \cdot  \ottnt{q_{{\mathrm{0}}}}  }{\mapsto}  \ottnt{a_{{\mathrm{2}}}}    \models    (   \Gamma_{{\mathrm{1}}}  +  \Gamma_{{\mathrm{22}}}   )   ,     \ottmv{x}  :^{  \ottnt{q}  \cdot    \ottnt{q_{{\mathrm{0}}}}  \cdot  \ottnt{r}    }  \ottnt{A_{{\mathrm{1}}}}   ,   \ottmv{y}  :^{  \ottnt{q}  \cdot  \ottnt{q_{{\mathrm{0}}}}  }  \ottnt{A_{{\mathrm{2}}}}     $.

\item \Rref{ST-Case}. Have: $  \Gamma_{{\mathrm{21}}}  +  \Gamma_{{\mathrm{22}}}   \vdash   \mathbf{case}_{ \ottnt{q_{{\mathrm{0}}}} } \:  \ottnt{a}  \: \mathbf{of} \:  \ottmv{x_{{\mathrm{1}}}}  .  \ottnt{b_{{\mathrm{1}}}}  \: ; \:  \ottmv{x_{{\mathrm{2}}}}  .  \ottnt{b_{{\mathrm{2}}}}   :^{ \ottnt{q} }  \ottnt{B} $ where $ \Gamma_{{\mathrm{21}}}  \vdash  \ottnt{a}  :^{  \ottnt{q}  \cdot  \ottnt{q_{{\mathrm{0}}}}  }   \ottnt{A_{{\mathrm{1}}}}  +  \ottnt{A_{{\mathrm{2}}}}  $ and $  \Gamma_{{\mathrm{22}}}  ,   \ottmv{x_{{\mathrm{1}}}}  :^{  \ottnt{q}  \cdot  \ottnt{q_{{\mathrm{0}}}}  }  \ottnt{A_{{\mathrm{1}}}}    \vdash  \ottnt{b_{{\mathrm{1}}}}  :^{ \ottnt{q} }  \ottnt{B} $ and $  \Gamma_{{\mathrm{22}}}  ,   \ottmv{x_{{\mathrm{2}}}}  :^{  \ottnt{q}  \cdot  \ottnt{q_{{\mathrm{0}}}}  }  \ottnt{A_{{\mathrm{2}}}}    \vdash  \ottnt{b_{{\mathrm{2}}}}  :^{ \ottnt{q} }  \ottnt{B} $.\\
Further, $ \ottnt{H}  \models   \Gamma_{{\mathrm{1}}}  +   (   \Gamma_{{\mathrm{21}}}  +  \Gamma_{{\mathrm{22}}}   )   $.\\
By IH, if $\ottnt{a}$ steps, then $ [  \ottnt{H}  ]  \ottnt{a}  \Longrightarrow^{  \ottnt{q}  \cdot  \ottnt{q_{{\mathrm{0}}}}  }_{  \ottnt{S}  \, \cup \,     \textit{fv} \:  \ottnt{b_{{\mathrm{1}}}}   \, \cup \,   \textit{fv} \:  \ottnt{b_{{\mathrm{2}}}}     } [  \ottnt{H'}  ]  \ottnt{a'} $ and $ \ottnt{H'}  \models   \Gamma_{{\mathrm{1}}}  +   (   \Gamma'_{{\mathrm{21}}}  +  \Gamma_{{\mathrm{22}}}   )   $ and $ \Gamma'_{{\mathrm{21}}}  \vdash  \ottnt{a'}  :^{  \ottnt{q}  \cdot  \ottnt{q_{{\mathrm{0}}}}  }   {}^{ \ottnt{r} }\!  \ottnt{A_{{\mathrm{1}}}}  \: \times \:  \ottnt{A_{{\mathrm{2}}}}  $.\\
Then, $ [  \ottnt{H}  ]   \mathbf{case}_{ \ottnt{q_{{\mathrm{0}}}} } \:  \ottnt{a}  \: \mathbf{of} \:  \ottmv{x_{{\mathrm{1}}}}  .  \ottnt{b_{{\mathrm{1}}}}  \: ; \:  \ottmv{x_{{\mathrm{2}}}}  .  \ottnt{b_{{\mathrm{2}}}}   \Longrightarrow^{  \ottnt{q}  \cdot  \ottnt{q_{{\mathrm{0}}}}  }_{ \ottnt{S} } [  \ottnt{H'}  ]   \mathbf{case}_{ \ottnt{q_{{\mathrm{0}}}} } \:  \ottnt{a'}  \: \mathbf{of} \:  \ottmv{x_{{\mathrm{1}}}}  .  \ottnt{b_{{\mathrm{1}}}}  \: ; \:  \ottmv{x_{{\mathrm{2}}}}  .  \ottnt{b_{{\mathrm{2}}}}  $.\\
By \rref{HeapStep-Discard}, $ [  \ottnt{H}  ]   \mathbf{case}_{ \ottnt{q_{{\mathrm{0}}}} } \:  \ottnt{a}  \: \mathbf{of} \:  \ottmv{x_{{\mathrm{1}}}}  .  \ottnt{b_{{\mathrm{1}}}}  \: ; \:  \ottmv{x_{{\mathrm{2}}}}  .  \ottnt{b_{{\mathrm{2}}}}   \Longrightarrow^{ \ottnt{q} }_{ \ottnt{S} } [  \ottnt{H'}  ]   \mathbf{case}_{ \ottnt{q_{{\mathrm{0}}}} } \:  \ottnt{a'}  \: \mathbf{of} \:  \ottmv{x_{{\mathrm{1}}}}  .  \ottnt{b_{{\mathrm{1}}}}  \: ; \:  \ottmv{x_{{\mathrm{2}}}}  .  \ottnt{b_{{\mathrm{2}}}}  $. ($\because  \ottnt{q_{{\mathrm{0}}}}  <:   1  $)\\
And by \rref{ST-Case}, $  \Gamma'_{{\mathrm{21}}}  +  \Gamma_{{\mathrm{22}}}   \vdash   \mathbf{case}_{ \ottnt{q_{{\mathrm{0}}}} } \:  \ottnt{a'}  \: \mathbf{of} \:  \ottmv{x_{{\mathrm{1}}}}  .  \ottnt{b_{{\mathrm{1}}}}  \: ; \:  \ottmv{x_{{\mathrm{2}}}}  .  \ottnt{b_{{\mathrm{2}}}}   :^{ \ottnt{q} }  \ottnt{B} $.\\

Otherwise, $\ottnt{a}$ is a value. By inversion, $\ottnt{a} =  \mathbf{inj}_1 \:  \ottnt{a_{{\mathrm{1}}}} $ or $\ottnt{a} =  \mathbf{inj}_2 \:  \ottnt{a_{{\mathrm{2}}}} $.\\
Say $\ottnt{a} =  \mathbf{inj}_1 \:  \ottnt{a_{{\mathrm{1}}}} $. Then, $ \Gamma_{{\mathrm{21}}}  \vdash  \ottnt{a_{{\mathrm{1}}}}  :^{  \ottnt{q}  \cdot  \ottnt{q_{{\mathrm{0}}}}  }  \ottnt{A_{{\mathrm{1}}}} $.\\
Now, $ [  \ottnt{H}  ]   \mathbf{case}_{ \ottnt{q_{{\mathrm{0}}}} } \:   (   \mathbf{inj}_1 \:  \ottnt{a_{{\mathrm{1}}}}   )   \: \mathbf{of} \:  \ottmv{x_{{\mathrm{1}}}}  .  \ottnt{b_{{\mathrm{1}}}}  \: ; \:  \ottmv{x_{{\mathrm{2}}}}  .  \ottnt{b_{{\mathrm{2}}}}   \Longrightarrow^{ \ottnt{q} }_{ \ottnt{S} } [   \ottnt{H}  ,   \ottmv{x_{{\mathrm{1}}}}  \overset{  \ottnt{q}  \cdot  \ottnt{q_{{\mathrm{0}}}}  }{\mapsto}  \ottnt{a_{{\mathrm{1}}}}    ]  \ottnt{b_{{\mathrm{1}}}} $ (assuming $ \ottmv{x_{{\mathrm{1}}}}  \: \textit{fresh} $).\\
And, $  \ottnt{H}  ,   \ottmv{x_{{\mathrm{1}}}}  \overset{  \ottnt{q}  \cdot  \ottnt{q_{{\mathrm{0}}}}  }{\mapsto}  \ottnt{a_{{\mathrm{1}}}}    \models    (   \Gamma_{{\mathrm{1}}}  +  \Gamma_{{\mathrm{22}}}   )   ,   \ottmv{x_{{\mathrm{1}}}}  :^{  \ottnt{q}  \cdot  \ottnt{q_{{\mathrm{0}}}}  }  \ottnt{A_{{\mathrm{1}}}}   $.\\
The case when $\ottnt{a} =  \mathbf{inj}_2 \:  \ottnt{a_{{\mathrm{2}}}} $ is similar.

\item \Rref{ST-SubL}. Have: $ \Gamma_{{\mathrm{2}}}  \vdash  \ottnt{a}  :^{ \ottnt{q} }  \ottnt{A} $ where $ \Gamma'_{{\mathrm{2}}}  \vdash  \ottnt{a}  :^{ \ottnt{q} }  \ottnt{A} $ and $ \Gamma_{{\mathrm{2}}}  <:  \Gamma'_{{\mathrm{2}}} $.\\
Further, $ \ottnt{H}  \models   \Gamma_{{\mathrm{1}}}  +  \Gamma_{{\mathrm{2}}}  $.\\
Since $ \Gamma_{{\mathrm{2}}}  <:  \Gamma'_{{\mathrm{2}}} $, there exists $\Gamma_{{\mathrm{0}}}$ such that $\Gamma_{{\mathrm{2}}} =  \Gamma'_{{\mathrm{2}}}  +  \Gamma_{{\mathrm{0}}} $. (In case of $ \mathbb{N}_{=} ,  \overline{ \Gamma_{{\mathrm{0}}} }  = \overline{0}$.)\\
By IH, if $\ottnt{a}$ is not value, then $ [  \ottnt{H}  ]  \ottnt{a}  \Longrightarrow^{ \ottnt{q} }_{ \ottnt{S} } [  \ottnt{H'}  ]  \ottnt{a'} $ and $ \ottnt{H'}  \models    (   \Gamma_{{\mathrm{1}}}  +  \Gamma_{{\mathrm{0}}}   )   +  \Gamma''_{{\mathrm{2}}}  $ and $ \Gamma''_{{\mathrm{2}}}  \vdash  \ottnt{a'}  :^{ \ottnt{q} }  \ottnt{A} $.\\
Then, by \rref{ST-SubL}, $  \Gamma_{{\mathrm{0}}}  +  \Gamma''_{{\mathrm{2}}}   \vdash  \ottnt{a'}  :^{ \ottnt{q} }  \ottnt{A} $, since $  \Gamma_{{\mathrm{0}}}  +  \Gamma''_{{\mathrm{2}}}   <:  \Gamma''_{{\mathrm{2}}} $.  

\item \Rref{ST-SubR}. Have $ \Gamma_{{\mathrm{2}}}  \vdash  \ottnt{a}  :^{ \ottnt{q'} }  \ottnt{A} $ where $ \Gamma_{{\mathrm{2}}}  \vdash  \ottnt{a}  :^{ \ottnt{q} }  \ottnt{A} $ and $ \ottnt{q}  <:  \ottnt{q'} $.\\
Further, $ \ottnt{H}  \models   \Gamma_{{\mathrm{1}}}  +  \Gamma_{{\mathrm{2}}}  $.\\
By IH, if $\ottnt{a}$ is not a value, then $ [  \ottnt{H}  ]  \ottnt{a}  \Longrightarrow^{ \ottnt{q} }_{ \ottnt{S} } [  \ottnt{H'}  ]  \ottnt{a'} $ and $ \ottnt{H'}  \models   \Gamma_{{\mathrm{1}}}  +  \Gamma'_{{\mathrm{2}}}  $ and $ \Gamma'_{{\mathrm{2}}}  \vdash  \ottnt{a'}  :^{ \ottnt{q} }  \ottnt{A} $.\\
Then, since $ \ottnt{q}  <:  \ottnt{q'} $, by \rref{HeapStep-Discard}, $ [  \ottnt{H}  ]  \ottnt{a}  \Longrightarrow^{ \ottnt{q'} }_{ \ottnt{S} } [  \ottnt{H'}  ]  \ottnt{a'} $. And by \rref{ST-SubR}, $ \Gamma'_{{\mathrm{2}}}  \vdash  \ottnt{a'}  :^{ \ottnt{q'} }  \ottnt{A} $. 

\end{itemize}
\end{proof}

%----------------------------------------------------------------------------------------------------

\begin{theorem}[Soundness (Theorem \ref{heapsound})]
If $ \ottnt{H}  \models  \Gamma $ and $ \Gamma  \vdash  \ottnt{a}  :^{ \ottnt{q} }  \ottnt{A} $ and $q \neq 0$, then either $\ottnt{a}$ is a value or there exists $\ottnt{H'}, \Gamma', \ottnt{a'}$ such that $ [  \ottnt{H}  ]  \ottnt{a}  \Longrightarrow^{ \ottnt{q} }_{ \ottnt{S} } [  \ottnt{H'}  ]  \ottnt{a'} $ and $ \ottnt{H'}  \models  \Gamma' $ and $ \Gamma'  \vdash  \ottnt{a'}  :^{ \ottnt{q} }  \ottnt{A} $.
\end{theorem}

\begin{proof}
Use lemma \ref{heaphelper} with $\Gamma_{{\mathrm{1}}} :=   0   \cdot  \Gamma $ and $\Gamma_{{\mathrm{2}}} := \Gamma$.
\end{proof}

%---------------------------------------------------------------------------------------------------

\begin{lemma} \label{heaphelperD}
If $ \ottnt{H}  \models   \Gamma_{{\mathrm{1}}}  \sqcap  \Gamma_{{\mathrm{2}}}  $ and $ \Gamma_{{\mathrm{2}}}  \vdash  \ottnt{a}  :^{ \ell }  \ottnt{A} $ and $l \neq \top$, then either $\ottnt{a}$ is a value or there exists $\ottnt{H'}, \Gamma'_{{\mathrm{2}}}, \ottnt{a'}$ such that:
\begin{itemize}
\item $ [  \ottnt{H}  ]  \ottnt{a}  \Longrightarrow^{ \ell }_{ \ottnt{S} } [  \ottnt{H'}  ]  \ottnt{a'} $
\item $ \ottnt{H'}  \models   \Gamma_{{\mathrm{1}}}  \sqcap  \Gamma'_{{\mathrm{2}}}  $
\item $ \Gamma'_{{\mathrm{2}}}  \vdash  \ottnt{a'}  :^{ \ell }  \ottnt{A} $
\end{itemize} 
\end{lemma}

\begin{proof}
By induction on $ \Gamma_{{\mathrm{2}}}  \vdash  \ottnt{a}  :^{ \ell }  \ottnt{A} $. Follow lemma \ref{heaphelper}.
\end{proof}

%---------------------------------------------------------------------------------------------------

\begin{theorem}[Soundness (Theorem \ref{heapsound})]
If $ \ottnt{H}  \models  \Gamma $ and $ \Gamma  \vdash  \ottnt{a}  :^{ \ell }  \ottnt{A} $ and $l \neq \top$, then either $\ottnt{a}$ is a value or there exists $\ottnt{H'}, \Gamma', \ottnt{a'}$ such that $ [  \ottnt{H}  ]  \ottnt{a}  \Longrightarrow^{ \ell }_{ \ottnt{S} } [  \ottnt{H'}  ]  \ottnt{a'} $ and $ \ottnt{H'}  \models  \Gamma' $ and $ \Gamma'  \vdash  \ottnt{a'}  :^{ \ell }  \ottnt{A} $.
\end{theorem}

\begin{proof}
Use lemma \ref{heaphelperD} with $\Gamma_{{\mathrm{1}}} :=   \top   \sqcup  \Gamma $ and $\Gamma_{{\mathrm{2}}} := \Gamma$.
\end{proof}

%---------------------------------------------------------------------------------------------------

\begin{corollary}[No Usage (Corollary \ref{nonint})] \label{nonintprf}
In LDC($ \mathcal{Q}_{\mathbb{N} } $): Let $  \emptyset   \vdash  \ottnt{f}  :^{  1  }   {}^{  0  }\!  \ottnt{A}  \to  \ottnt{A}  $. Then, for any $  \emptyset   \vdash  \ottnt{a_{{\mathrm{1}}}}  :^{  0  }  \ottnt{A} $ and $  \emptyset   \vdash  \ottnt{a_{{\mathrm{2}}}}  :^{  0  }  \ottnt{A} $, the terms $ \ottnt{f}  \:  \ottnt{a_{{\mathrm{1}}}} ^{  0  } $ and $ \ottnt{f}  \:  \ottnt{a_{{\mathrm{2}}}} ^{  0  } $ have the same operational behavior. 
\end{corollary}

\begin{proof}
To see why, we consider the reduction of $ \ottnt{f}  \:  \ottnt{a_{{\mathrm{1}}}} ^{  0  } $ and $ \ottnt{f}  \:  \ottnt{a_{{\mathrm{2}}}} ^{  0  } $. 

Let $ [  \ottnt{H}  ]  \ottnt{a}  \Longrightarrow^{ \ottnt{q} }_{ \ottnt{j} } [  \ottnt{H'}  ]  \ottnt{a'} $ denote a reduction of $\ottnt{j}$ steps where $\ottnt{H}$ and $\ottnt{a}$ are the initial heap and term, $\ottnt{H'}$ and $\ottnt{a'}$ are the final heap and term, and $\ottnt{q}$ is the label at which the reduction takes place. 

For some $\ottnt{j}$, $\ottnt{H}$ and $\ottnt{b}$, we have, $ [   \emptyset   ]   \ottnt{f}  \:  \ottnt{a_{{\mathrm{1}}}} ^{  0  }   \Longrightarrow^{  1  }_{ \ottnt{j} } [  \ottnt{H}  ]    (   \lambda^{  0  }  \ottmv{x}  .  \ottnt{b}   )   \:  \ottnt{a_{{\mathrm{1}}}} ^{  0  }  $ and  $ [   \emptyset   ]   \ottnt{f}  \:  \ottnt{a_{{\mathrm{2}}}} ^{  0  }   \Longrightarrow^{  1  }_{ \ottnt{j} } [  \ottnt{H}  ]    (   \lambda^{  0  }  \ottmv{x}  .  \ottnt{b}   )   \:  \ottnt{a_{{\mathrm{2}}}} ^{  0  }  $. 

Then, $ [  \ottnt{H}  ]    (   \lambda^{  0  }  \ottmv{x}  .  \ottnt{b}   )   \:  \ottnt{a_{{\mathrm{1}}}} ^{  0  }   \Longrightarrow^{  1  }_{ \ottnt{S} } [   \ottnt{H}  ,   \ottmv{x}  \overset{  0  }{\mapsto}  \ottnt{a_{{\mathrm{1}}}}    ]  \ottnt{b} $ and $ [  \ottnt{H}  ]    (   \lambda^{  0  }  \ottmv{x}  .  \ottnt{b}   )   \:  \ottnt{a_{{\mathrm{2}}}} ^{  0  }   \Longrightarrow^{  1  }_{ \ottnt{S} } [   \ottnt{H}  ,   \ottmv{x}  \overset{  0  }{\mapsto}  \ottnt{a_{{\mathrm{2}}}}    ]  \ottnt{b} $ (for $ \ottmv{x}  \: \textit{fresh} $).

But, if $ [   \ottnt{H}  ,   \ottmv{x}  \overset{  0  }{\mapsto}  \ottnt{a_{{\mathrm{1}}}}    ]  \ottnt{b}  \Longrightarrow^{  1  }_{ \ottnt{k} } [     \ottnt{H'}  ,   \ottmv{x}  \overset{  0  }{\mapsto}  \ottnt{a_{{\mathrm{1}}}}     ,  \ottnt{H''}   ]  \ottnt{b'} $, then $ [   \ottnt{H}  ,   \ottmv{x}  \overset{  0  }{\mapsto}  \ottnt{a_{{\mathrm{2}}}}    ]  \ottnt{b}  \Longrightarrow^{  1  }_{ \ottnt{k} } [     \ottnt{H'}  ,   \ottmv{x}  \overset{  0  }{\mapsto}  \ottnt{a_{{\mathrm{2}}}}     ,  \ottnt{H''}   ]  \ottnt{b'} $, for any $\ottnt{k}$ (By Lemma \ref{irrel}).

Therefore, $ \ottnt{f}  \:  \ottnt{a_{{\mathrm{1}}}} ^{  0  } $ and $ \ottnt{f}  \:  \ottnt{a_{{\mathrm{2}}}} ^{  0  } $ have the same operational behavior.

\end{proof}

%---------------------------------------------------------------------------------------------------

\begin{corollary}[Noninterference (Corollary \ref{noninterference})]
In LDC($ \mathcal{L} $): Let $  \emptyset   \vdash  \ottnt{f}  :^{  \mathbf{L}  }   {}^{  \mathbf{H}  }\!  \ottnt{A}  \to  \ottnt{A}  $. Then, for any $  \emptyset   \vdash  \ottnt{a_{{\mathrm{1}}}}  :^{  \mathbf{H}  }  \ottnt{A} $ and $  \emptyset   \vdash  \ottnt{a_{{\mathrm{2}}}}  :^{  \mathbf{H}  }  \ottnt{A} $, the terms $ \ottnt{f}  \:  \ottnt{a_{{\mathrm{1}}}} ^{  \mathbf{H}  } $ and $ \ottnt{f}  \:  \ottnt{a_{{\mathrm{2}}}} ^{  \mathbf{H}  } $ have the same operational behavior. 
\end{corollary}

\begin{proof}
Same as the proof of Corollary \ref{nonintprf} with $ 0 $ and $ 1 $ replaced by $ \mathbf{H} $ and $ \mathbf{L} $ respectively.
\end{proof}

%----------------------------------------------------------------------------------------------------

\begin{corollary}[Affine Usage (Corollary \ref{single})]
In LDC($ \mathbb{N}_{\geq} $): Let $  \emptyset   \vdash  \ottnt{f}  :^{  1  }   {}^{  1  }\!  \ottnt{A}  \to  \ottnt{A}  $. Then, for any $  \emptyset   \vdash  \ottnt{a}  :^{  1  }  \ottnt{A} $, the term $ \ottnt{f}  \:  \ottnt{a} ^{  1  } $ uses $\ottnt{a}$ at most once during reduction.
\end{corollary}

\begin{proof}
To see why, consider the reduction of $ \ottnt{f}  \:  \ottnt{a} ^{  1  } $.

For some $\ottnt{j}$, $\ottnt{H}$ and $\ottnt{b}$, we have, $ [   \emptyset   ]   \ottnt{f}  \:  \ottnt{a} ^{  1  }   \Longrightarrow^{  1  }_{ \ottnt{j} } [  \ottnt{H}  ]    (   \lambda^{  1  }  \ottmv{x}  .  \ottnt{b}   )   \:  \ottnt{a} ^{  1  }  $.

Then, $ [  \ottnt{H}  ]    (   \lambda^{  1  }  \ottmv{x}  .  \ottnt{b}   )   \:  \ottnt{a} ^{  1  }   \Longrightarrow^{  1  }_{ \ottnt{S} } [   \ottnt{H}  ,   \ottmv{x}  \overset{  1  }{\mapsto}  \ottnt{a}    ]  \ottnt{b} $ (for $ \ottmv{x}  \: \textit{fresh} $).

Now, $\ottnt{b}$ may reduce to a value without ever looking-up for $\ottmv{x}$ or $\ottnt{b}$ may look-up the value of $\ottmv{x}$ exactly once. A single use of $\ottmv{x}$ will change the allowed usage from $ 1 $ to $ 0 $, making it essentially unusable thereafter. Hence, $\ottnt{a}$ cannot be used more than once.
\end{proof}

%---------------------------------------------------------------------------------------------------

%---------------------------------------------------------------------------------------------------
\section{Linearity Analysis in PTS Version of LDC}
%---------------------------------------------------------------------------------------------------

\begin{lemma}[Multiplication] \label{BLDMultP}
If $ \Gamma  \vdash  \ottnt{a}  :^{ \ottnt{q} }  \ottnt{A} $, then $  \ottnt{r_{{\mathrm{0}}}}  \cdot  \Gamma   \vdash  \ottnt{a}  :^{  \ottnt{r_{{\mathrm{0}}}}  \cdot  \ottnt{q}  }  \ottnt{A} $.
\end{lemma}

\begin{proof}
By induction on $ \Gamma  \vdash  \ottnt{a}  :^{ \ottnt{q} }  \ottnt{A} $. Follow the proof of lemma \ref{BLSMultP}.
\end{proof}

%--------------------------------------------------------------------------------------------------

\begin{lemma}[Factorization] \label{BLDFactP}
If $ \Gamma  \vdash  \ottnt{a}  :^{ \ottnt{q} }  \ottnt{A} $ and $q \neq 0$, then there exists $\Gamma'$ such that $ \Gamma'  \vdash  \ottnt{a}  :^{  1  }  \ottnt{A} $ and $ \Gamma  <:   \ottnt{q}  \cdot  \Gamma'  $.
\end{lemma}

\begin{proof}
By induction on $ \Gamma  \vdash  \ottnt{a}  :^{ \ottnt{q} }  \ottnt{A} $. Follow the proof of lemma \ref{BLSFactP}.
\end{proof}

%---------------------------------------------------------------------------------------------------

\begin{lemma}[Splitting] \label{BLDSplitP}
If $ \Gamma  \vdash  \ottnt{a}  :^{  \ottnt{q_{{\mathrm{1}}}}  +  \ottnt{q_{{\mathrm{2}}}}  }  \ottnt{A} $, then there exists $\Gamma_{{\mathrm{1}}}$ and $\Gamma_{{\mathrm{2}}}$ such that $ \Gamma_{{\mathrm{1}}}  \vdash  \ottnt{a}  :^{ \ottnt{q_{{\mathrm{1}}}} }  \ottnt{A} $ and $ \Gamma_{{\mathrm{2}}}  \vdash  \ottnt{a}  :^{ \ottnt{q_{{\mathrm{2}}}} }  \ottnt{A} $ and $ \Gamma  =   \Gamma_{{\mathrm{1}}}  +  \Gamma_{{\mathrm{2}}}  $. 
\end{lemma}

\begin{proof}
If $ \ottnt{q_{{\mathrm{1}}}}  +  \ottnt{q_{{\mathrm{2}}}}  = 0$, then $\Gamma_{{\mathrm{1}}} :=   0   \cdot  \Gamma $ and $\Gamma_{{\mathrm{2}}} := \Gamma$.\\
Otherwise, by Lemma \ref{BLDFactP}, $\exists \Gamma'$ such that $ \Gamma'  \vdash  \ottnt{a}  :^{  1  }  \ottnt{A} $ and $ \Gamma  <:    (   \ottnt{q_{{\mathrm{1}}}}  +  \ottnt{q_{{\mathrm{2}}}}   )   \cdot  \Gamma'  $.\\
Then, $\Gamma =    (   \ottnt{q_{{\mathrm{1}}}}  +  \ottnt{q_{{\mathrm{2}}}}   )   \cdot  \Gamma'   +  \Gamma_{{\mathrm{0}}} $ for some $\Gamma_{{\mathrm{0}}}$ (in case of $ \mathbb{N}_{=} $, $ \overline{ \Gamma_{{\mathrm{0}}} }  = \overline{0}$).\\
Now, by Lemma \ref{BLDMultP}, $  \ottnt{q_{{\mathrm{1}}}}  \cdot  \Gamma'   \vdash  \ottnt{a}  :^{ \ottnt{q_{{\mathrm{1}}}} }  \ottnt{A} $ and $  \ottnt{q_{{\mathrm{2}}}}  \cdot  \Gamma'   \vdash  \ottnt{a}  :^{ \ottnt{q_{{\mathrm{2}}}} }  \ottnt{A} $.\\
By \rref{PTS-SubL}, $   \ottnt{q_{{\mathrm{2}}}}  \cdot  \Gamma'   +  \Gamma_{{\mathrm{0}}}   \vdash  \ottnt{a}  :^{ \ottnt{q_{{\mathrm{2}}}} }  \ottnt{A} $.\\
The lemma follows by setting $\Gamma_{{\mathrm{1}}} :=  \ottnt{q_{{\mathrm{1}}}}  \cdot  \Gamma' $ and $\Gamma_{{\mathrm{2}}} :=   \ottnt{q_{{\mathrm{2}}}}  \cdot  \Gamma'   +  \Gamma_{{\mathrm{0}}} $.
\end{proof}

%--------------------------------------------------------------------------------------------------

\begin{lemma}[Weakening] \label{BLDWeakP}
If $  \Gamma_{{\mathrm{1}}}  ,  \Gamma_{{\mathrm{2}}}   \vdash  \ottnt{a}  :^{ \ottnt{q} }  \ottnt{A} $ and $ \Delta_{{\mathrm{1}}}  \vdash_{0}  \ottnt{C}  :   \ottmv{s}  $ and $  \lfloor  \Gamma_{{\mathrm{1}}}  \rfloor   =  \Delta_{{\mathrm{1}}} $, then $    \Gamma_{{\mathrm{1}}}  ,   \ottmv{z}  :^{  0  }  \ottnt{C}     ,  \Gamma_{{\mathrm{2}}}   \vdash  \ottnt{a}  :^{ \ottnt{q} }  \ottnt{A} $.
\end{lemma}

\begin{proof}
By induction on $  \Gamma_{{\mathrm{1}}}  ,  \Gamma_{{\mathrm{2}}}   \vdash  \ottnt{a}  :^{ \ottnt{q} }  \ottnt{A} $.
\end{proof}

%---------------------------------------------------------------------------------------------------

\begin{lemma}[Substitution (Lemma \ref{DSubst})] \label{BLDSubstP}
If $    \Gamma_{{\mathrm{1}}}  ,   \ottmv{z}  :^{ \ottnt{r_{{\mathrm{0}}}} }  \ottnt{C}     ,  \Gamma_{{\mathrm{2}}}   \vdash  \ottnt{a}  :^{ \ottnt{q} }  \ottnt{A} $ and $ \Gamma  \vdash  \ottnt{c}  :^{ \ottnt{r_{{\mathrm{0}}}} }  \ottnt{C} $ and $  \lfloor  \Gamma_{{\mathrm{1}}}  \rfloor   =   \lfloor  \Gamma  \rfloor  $, then $     \Gamma_{{\mathrm{1}}}  +  \Gamma    ,  \Gamma_{{\mathrm{2}}}   \{  \ottnt{c}  /  \ottmv{z}  \}   \vdash   \ottnt{a}  \{  \ottnt{c}  /  \ottmv{z}  \}   :^{ \ottnt{q} }   \ottnt{A}  \{  \ottnt{c}  /  \ottmv{z}  \}  $. 
\end{lemma}

\begin{proof}
By induction on $    \Gamma_{{\mathrm{1}}}  ,   \ottmv{z}  :^{ \ottnt{r_{{\mathrm{0}}}} }  \ottnt{C}     ,  \Gamma_{{\mathrm{2}}}   \vdash  \ottnt{a}  :^{ \ottnt{q} }  \ottnt{A} $.

\begin{itemize}

\item \Rref{PTS-Var}. Have: $     0   \cdot  \Gamma_{{\mathrm{1}}}    ,   \ottmv{x}  :^{ \ottnt{q} }  \ottnt{A}    \vdash   \ottmv{x}   :^{ \ottnt{q} }  \ottnt{A} $ where $ \Delta  \vdash_{0}  \ottnt{A}  :   \ottmv{s}  $ and $  \lfloor  \Gamma_{{\mathrm{1}}}  \rfloor   =  \Delta $. There are two cases to consider.

\begin{itemize}
\item Have: $     0   \cdot  \Gamma_{{\mathrm{1}}}    ,   \ottmv{x}  :^{ \ottnt{q} }  \ottnt{A}    \vdash   \ottmv{x}   :^{ \ottnt{q} }  \ottnt{A} $ and $ \Gamma  \vdash  \ottnt{a}  :^{ \ottnt{q} }  \ottnt{A} $ where $  \lfloor  \Gamma_{{\mathrm{1}}}  \rfloor   =   \lfloor  \Gamma  \rfloor  $.\\
Need to show: $ \Gamma  \vdash  \ottnt{a}  :^{ \ottnt{q} }  \ottnt{A} $. Follows from what's given.
\item Have: $         0   \cdot  \Gamma_{{\mathrm{11}}}    ,   \ottmv{z}  :^{  0  }  \ottnt{C}     ,    0   \cdot  \Gamma_{{\mathrm{12}}}     ,   \ottmv{x}  :^{ \ottnt{q} }  \ottnt{A}    \vdash   \ottmv{x}   :^{ \ottnt{q} }  \ottnt{A} $ and $ \Gamma  \vdash  \ottnt{c}  :^{  0  }  \ottnt{C} $ where $  \lfloor  \Gamma_{{\mathrm{11}}}  \rfloor   =   \lfloor  \Gamma  \rfloor  $.\\
Need to show: $      \Gamma  ,    0   \cdot  \Gamma_{{\mathrm{12}}}    \{  \ottnt{c}  /  \ottmv{z}  \}    ,   \ottmv{x}  :^{ \ottnt{q} }  \ottnt{A}    \{  \ottnt{c}  /  \ottmv{z}  \}   \vdash   \ottmv{x}   :^{ \ottnt{q} }   \ottnt{A}  \{  \ottnt{c}  /  \ottmv{z}  \}  $.\\
Follows by \rref{PTS-Var,PTS-SubL} (note that in case of $ \mathbb{N}_{=} $, $ \overline{ \Gamma }  = \overline{0}$).
\end{itemize}  

\item \Rref{PTS-Weak}. Have: $  \Gamma_{{\mathrm{1}}}  ,   \ottmv{y}  :^{  0  }  \ottnt{B}    \vdash  \ottnt{a}  :^{ \ottnt{q} }  \ottnt{A} $ where $ \Gamma_{{\mathrm{1}}}  \vdash  \ottnt{a}  :^{ \ottnt{q} }  \ottnt{A} $ and $ \Delta  \vdash_{0}  \ottnt{B}  :   \ottmv{s}  $ and $  \lfloor  \Gamma_{{\mathrm{1}}}  \rfloor   =  \Delta $. There are two cases to consider.
\begin{itemize}
\item Have: $  \Gamma_{{\mathrm{1}}}  ,   \ottmv{y}  :^{  0  }  \ottnt{B}    \vdash  \ottnt{a}  :^{ \ottnt{q} }  \ottnt{A} $ and $ \Gamma  \vdash  \ottnt{b}  :^{  0  }  \ottnt{B} $ where $  \lfloor  \Gamma_{{\mathrm{1}}}  \rfloor   =   \lfloor  \Gamma  \rfloor  $.\\
Need to show: $  \Gamma_{{\mathrm{1}}}  +  \Gamma   \vdash   \ottnt{a}  \{  \ottnt{b}  /  \ottmv{y}  \}   :^{ \ottnt{q} }   \ottnt{A}  \{  \ottnt{b}  /  \ottmv{y}  \}  $.\\
Since $y \notin \: \text{fv } a$ and $y \notin \: \text{fv } A$, need to show: $  \Gamma_{{\mathrm{1}}}  +  \Gamma   \vdash  \ottnt{a}  :^{ \ottnt{q} }  \ottnt{A} $.\\
This case follows by \rref{PTS-SubL} (note that in case of $ \mathbb{N}_{=} $, $ \overline{ \Gamma }  = \overline{0}$).
\item Have: $    \Gamma_{{\mathrm{11}}}  ,     \ottmv{z}  :^{ \ottnt{r_{{\mathrm{0}}}} }  \ottnt{C}   ,  \Gamma_{{\mathrm{12}}}      ,   \ottmv{y}  :^{  0  }  \ottnt{B}    \vdash  \ottnt{a}  :^{ \ottnt{q} }  \ottnt{A} $ and $ \Gamma  \vdash  \ottnt{c}  :^{ \ottnt{r_{{\mathrm{0}}}} }  \ottnt{C} $ where $  \lfloor  \Gamma_{{\mathrm{11}}}  \rfloor   =   \lfloor  \Gamma  \rfloor  $.\\
Need to show: $        \Gamma_{{\mathrm{11}}}  +  \Gamma    ,  \Gamma_{{\mathrm{12}}}   \{  \ottnt{c}  /  \ottmv{z}  \}    ,   \ottmv{y}  :^{  0  }  \ottnt{B}    \{  \ottnt{c}  /  \ottmv{z}  \}   \vdash   \ottnt{a}  \{  \ottnt{c}  /  \ottmv{z}  \}   :^{ \ottnt{q} }   \ottnt{A}  \{  \ottnt{c}  /  \ottmv{z}  \}  $.\\
Follows by IH and \rref{PTS-Weak}.
\end{itemize}

\item \Rref{PTS-Pi}. Have: $      \Gamma_{{\mathrm{11}}}  +  \Gamma_{{\mathrm{21}}}    ,   \ottmv{z}  :^{  \ottnt{r_{{\mathrm{01}}}}  +  \ottnt{r_{{\mathrm{02}}}}  }  \ottnt{C}     ,    \Gamma_{{\mathrm{12}}}  +  \Gamma_{{\mathrm{22}}}     \vdash   \Pi  \ottmv{x}  :^{ \ottnt{r} } \!  \ottnt{A}  .  \ottnt{B}   :^{ \ottnt{q} }   \ottmv{s_{{\mathrm{3}}}}  $ where $    \Gamma_{{\mathrm{11}}}  ,   \ottmv{z}  :^{ \ottnt{r_{{\mathrm{01}}}} }  \ottnt{C}     ,  \Gamma_{{\mathrm{12}}}   \vdash  \ottnt{A}  :^{ \ottnt{q} }   \ottmv{s_{{\mathrm{1}}}}  $ and $    \Gamma_{{\mathrm{21}}}  ,   \ottmv{z}  :^{ \ottnt{r_{{\mathrm{02}}}} }  \ottnt{C}     ,    \Gamma_{{\mathrm{22}}}  ,   \ottmv{x}  :^{ \ottnt{q_{{\mathrm{0}}}} }  \ottnt{A}      \vdash  \ottnt{B}  :^{ \ottnt{q} }   \ottmv{s_{{\mathrm{2}}}}  $ and $ \mathcal{R} ( \ottmv{s_{{\mathrm{1}}}}  ,  \ottmv{s_{{\mathrm{2}}}}  ,  \ottmv{s_{{\mathrm{3}}}} ) $.\\
Further, $ \Gamma  \vdash  \ottnt{c}  :^{  \ottnt{r_{{\mathrm{01}}}}  +  \ottnt{r_{{\mathrm{02}}}}  }  \ottnt{C} $ where $  \lfloor  \Gamma  \rfloor   =   \lfloor  \Gamma_{{\mathrm{11}}}  \rfloor  $.\\
Need to show: $      \Gamma_{{\mathrm{11}}}  +  \Gamma_{{\mathrm{21}}}    +  \Gamma    ,      \Gamma_{{\mathrm{12}}}  \{  \ottnt{c}  /  \ottmv{z}  \}   +  \Gamma_{{\mathrm{22}}}   \{  \ottnt{c}  /  \ottmv{z}  \}     \vdash    \Pi  \ottmv{x}  :^{ \ottnt{r} } \!   \ottnt{A}  \{  \ottnt{c}  /  \ottmv{z}  \}   .  \ottnt{B}   \{  \ottnt{c}  /  \ottmv{z}  \}   :^{ \ottnt{q} }   \ottmv{s_{{\mathrm{3}}}}  $.\\
By lemma \ref{BLDSplitP}, $\exists \Gamma_{{\mathrm{31}}}, \Gamma_{{\mathrm{32}}}$ such that $ \Gamma_{{\mathrm{31}}}  \vdash  \ottnt{c}  :^{ \ottnt{r_{{\mathrm{01}}}} }  \ottnt{C} $ and $ \Gamma_{{\mathrm{32}}}  \vdash  \ottnt{c}  :^{ \ottnt{r_{{\mathrm{02}}}} }  \ottnt{C} $ and $  \Gamma_{{\mathrm{31}}}  +  \Gamma_{{\mathrm{32}}}   =  \Gamma $.\\
By IH, $     \Gamma_{{\mathrm{11}}}  +  \Gamma_{{\mathrm{31}}}    ,  \Gamma_{{\mathrm{12}}}   \{  \ottnt{c}  /  \ottmv{z}  \}   \vdash   \ottnt{A}  \{  \ottnt{c}  /  \ottmv{z}  \}   :^{ \ottnt{q} }   \ottmv{s_{{\mathrm{1}}}}  $ and $    \Gamma_{{\mathrm{21}}}  +  \Gamma_{{\mathrm{32}}}    ,      \Gamma_{{\mathrm{22}}}  \{  \ottnt{c}  /  \ottmv{z}  \}   ,   \ottmv{x}  :^{ \ottnt{q_{{\mathrm{0}}}} }  \ottnt{A}    \{  \ottnt{c}  /  \ottmv{z}  \}     \vdash   \ottnt{B}  \{  \ottnt{c}  /  \ottmv{z}  \}   :^{ \ottnt{q} }   \ottmv{s_{{\mathrm{2}}}}  $.\\
This case, then, follows by \rref{PTS-Pi}.

\item \Rref{PTS-Lam}. Have: $    \Gamma_{{\mathrm{1}}}  ,   \ottmv{z}  :^{ \ottnt{r_{{\mathrm{0}}}} }  \ottnt{C}     ,  \Gamma_{{\mathrm{2}}}   \vdash   \lambda^{ \ottnt{r} }  \ottmv{x}  :  \ottnt{A}  .  \ottnt{b}   :^{ \ottnt{q} }   \Pi  \ottmv{x}  :^{ \ottnt{r} } \!  \ottnt{A}  .  \ottnt{B}  $ where\\ $    \Gamma_{{\mathrm{1}}}  ,   \ottmv{z}  :^{ \ottnt{r_{{\mathrm{0}}}} }  \ottnt{C}     ,    \Gamma_{{\mathrm{2}}}  ,   \ottmv{x}  :^{  \ottnt{q}  \cdot  \ottnt{r}  }  \ottnt{A}      \vdash  \ottnt{b}  :^{ \ottnt{q} }  \ottnt{B} $.\\
Further, $ \Gamma  \vdash  \ottnt{c}  :^{ \ottnt{r_{{\mathrm{0}}}} }  \ottnt{C} $ where $  \lfloor  \Gamma  \rfloor   =   \lfloor  \Gamma_{{\mathrm{1}}}  \rfloor  $.\\
Need to show: $     \Gamma_{{\mathrm{1}}}  +  \Gamma    ,  \Gamma_{{\mathrm{2}}}   \{  \ottnt{c}  /  \ottmv{z}  \}   \vdash    \lambda^{ \ottnt{r} }  \ottmv{x}  :   \ottnt{A}  \{  \ottnt{c}  /  \ottmv{z}  \}   .  \ottnt{b}   \{  \ottnt{c}  /  \ottmv{z}  \}   :^{ \ottnt{q} }    \Pi  \ottmv{x}  :^{ \ottnt{r} } \!   \ottnt{A}  \{  \ottnt{c}  /  \ottmv{z}  \}   .  \ottnt{B}   \{  \ottnt{c}  /  \ottmv{z}  \}  $.\\
By IH, $    \Gamma_{{\mathrm{1}}}  +  \Gamma    ,      \Gamma_{{\mathrm{2}}}  \{  \ottnt{c}  /  \ottmv{z}  \}   ,   \ottmv{x}  :^{  \ottnt{q}  \cdot  \ottnt{r}  }  \ottnt{A}    \{  \ottnt{c}  /  \ottmv{z}  \}     \vdash   \ottnt{b}  \{  \ottnt{c}  /  \ottmv{z}  \}   :^{ \ottnt{q} }   \ottnt{B}  \{  \ottnt{c}  /  \ottmv{z}  \}  $.\\
This case, then, follows by \rref{PTS-Lam}.

\item \Rref{PTS-App}. Have: $      \Gamma_{{\mathrm{11}}}  +  \Gamma_{{\mathrm{21}}}    ,   \ottmv{z}  :^{  \ottnt{r_{{\mathrm{01}}}}  +  \ottnt{r_{{\mathrm{02}}}}  }  \ottnt{C}     ,    \Gamma_{{\mathrm{12}}}  +  \Gamma_{{\mathrm{22}}}     \vdash   \ottnt{b}  \:  \ottnt{a} ^{ \ottnt{r} }   :^{ \ottnt{q} }   \ottnt{B}  \{  \ottnt{a}  /  \ottmv{x}  \}  $ where\\ $    \Gamma_{{\mathrm{11}}}  ,   \ottmv{z}  :^{ \ottnt{r_{{\mathrm{01}}}} }  \ottnt{C}     ,  \Gamma_{{\mathrm{12}}}   \vdash  \ottnt{b}  :^{ \ottnt{q} }   \Pi  \ottmv{x}  :^{ \ottnt{r} } \!  \ottnt{A}  .  \ottnt{B}  $ and $    \Gamma_{{\mathrm{21}}}  ,   \ottmv{z}  :^{ \ottnt{r_{{\mathrm{02}}}} }  \ottnt{C}     ,  \Gamma_{{\mathrm{22}}}   \vdash  \ottnt{a}  :^{  \ottnt{q}  \cdot  \ottnt{r}  }  \ottnt{A} $.\\
Further, $ \Gamma  \vdash  \ottnt{c}  :^{  \ottnt{r_{{\mathrm{01}}}}  +  \ottnt{r_{{\mathrm{02}}}}  }  \ottnt{C} $ where $  \lfloor  \Gamma  \rfloor   =   \lfloor  \Gamma_{{\mathrm{11}}}  \rfloor  $.\\
Need to show: $      \Gamma_{{\mathrm{11}}}  +  \Gamma_{{\mathrm{21}}}    +  \Gamma    ,      \Gamma_{{\mathrm{12}}}  \{  \ottnt{c}  /  \ottmv{z}  \}   +  \Gamma_{{\mathrm{22}}}   \{  \ottnt{c}  /  \ottmv{z}  \}     \vdash    \ottnt{b}  \{  \ottnt{c}  /  \ottmv{z}  \}   \:   \ottnt{a}  \{  \ottnt{c}  /  \ottmv{z}  \}  ^{ \ottnt{r} }   :^{ \ottnt{q} }    \ottnt{B}  \{  \ottnt{c}  /  \ottmv{z}  \}   \{   \ottnt{a}  \{  \ottnt{c}  /  \ottmv{z}  \}   /  \ottmv{x}  \}  $.\\
By lemma \ref{BLDSplitP}, $\exists \Gamma_{{\mathrm{31}}}, \Gamma_{{\mathrm{32}}}$ such that $ \Gamma_{{\mathrm{31}}}  \vdash  \ottnt{c}  :^{ \ottnt{r_{{\mathrm{01}}}} }  \ottnt{C} $ and $ \Gamma_{{\mathrm{32}}}  \vdash  \ottnt{c}  :^{ \ottnt{r_{{\mathrm{02}}}} }  \ottnt{C} $ and $  \Gamma_{{\mathrm{31}}}  +  \Gamma_{{\mathrm{32}}}   =  \Gamma $.\\
By IH, $     \Gamma_{{\mathrm{11}}}  +  \Gamma_{{\mathrm{31}}}    ,  \Gamma_{{\mathrm{12}}}   \{  \ottnt{c}  /  \ottmv{z}  \}   \vdash   \ottnt{b}  \{  \ottnt{c}  /  \ottmv{z}  \}   :^{ \ottnt{q} }    \Pi  \ottmv{x}  :^{ \ottnt{r} } \!   \ottnt{A}  \{  \ottnt{c}  /  \ottmv{z}  \}   .  \ottnt{B}   \{  \ottnt{c}  /  \ottmv{z}  \}  $\\ and $     \Gamma_{{\mathrm{21}}}  +  \Gamma_{{\mathrm{32}}}    ,  \Gamma_{{\mathrm{22}}}   \{  \ottnt{c}  /  \ottmv{z}  \}   \vdash   \ottnt{a}  \{  \ottnt{c}  /  \ottmv{z}  \}   :^{  \ottnt{q}  \cdot  \ottnt{r}  }   \ottnt{A}  \{  \ottnt{c}  /  \ottmv{z}  \}  $.\\
This case, then, follows by \rref{PTS-App}.

\item \Rref{PTS-Conv}. Have: $    \Gamma_{{\mathrm{1}}}  ,   \ottmv{z}  :^{ \ottnt{r_{{\mathrm{0}}}} }  \ottnt{C}     ,  \Gamma_{{\mathrm{2}}}   \vdash  \ottnt{a}  :^{ \ottnt{q} }  \ottnt{B} $ where $    \Gamma_{{\mathrm{1}}}  ,   \ottmv{z}  :^{ \ottnt{r_{{\mathrm{0}}}} }  \ottnt{C}     ,  \Gamma_{{\mathrm{2}}}   \vdash  \ottnt{a}  :^{ \ottnt{q} }  \ottnt{A} $ and  $ \ottnt{A}  =_{\beta}  \ottnt{B} $.\\
Further, $ \Gamma  \vdash  \ottnt{c}  :^{ \ottnt{r_{{\mathrm{0}}}} }  \ottnt{C} $ where $  \lfloor  \Gamma  \rfloor   =   \lfloor  \Gamma_{{\mathrm{1}}}  \rfloor  $.\\
Need to show: $     \Gamma_{{\mathrm{11}}}  +  \Gamma    ,  \Gamma_{{\mathrm{2}}}   \{  \ottnt{c}  /  \ottmv{z}  \}   \vdash   \ottnt{a}  \{  \ottnt{c}  /  \ottmv{z}  \}   :^{ \ottnt{q} }   \ottnt{B}  \{  \ottnt{c}  /  \ottmv{z}  \}  $.\\
By IH, $     \Gamma_{{\mathrm{11}}}  +  \Gamma    ,  \Gamma_{{\mathrm{2}}}   \{  \ottnt{c}  /  \ottmv{z}  \}   \vdash   \ottnt{a}  \{  \ottnt{c}  /  \ottmv{z}  \}   :^{ \ottnt{q} }   \ottnt{A}  \{  \ottnt{c}  /  \ottmv{z}  \}  $.\\
Also, since $ \ottnt{A}  =_{\beta}  \ottnt{B} $, so $  \ottnt{A}  \{  \ottnt{c}  /  \ottmv{z}  \}   =_{\beta}   \ottnt{B}  \{  \ottnt{c}  /  \ottmv{z}  \}  $.\\
This case, then, follows by \rref{PTS-Conv}. 

\item \Rref{PTS-Pair}. Have: $      \Gamma_{{\mathrm{11}}}  +  \Gamma_{{\mathrm{21}}}    ,   \ottmv{z}  :^{  \ottnt{r_{{\mathrm{01}}}}  +  \ottnt{r_{{\mathrm{02}}}}  }  \ottnt{C}     ,    \Gamma_{{\mathrm{12}}}  +  \Gamma_{{\mathrm{22}}}     \vdash   (  \ottnt{a_{{\mathrm{1}}}} ^{ \ottnt{r} } ,  \ottnt{a_{{\mathrm{2}}}}  )   :^{ \ottnt{q} }   \Sigma  \ottmv{x}  :^{ \ottnt{r} } \!  \ottnt{A_{{\mathrm{1}}}}  .  \ottnt{A_{{\mathrm{2}}}}  $ where $    \Gamma_{{\mathrm{11}}}  ,   \ottmv{z}  :^{ \ottnt{r_{{\mathrm{01}}}} }  \ottnt{C}     ,  \Gamma_{{\mathrm{12}}}   \vdash  \ottnt{a_{{\mathrm{1}}}}  :^{  \ottnt{q}  \cdot  \ottnt{r}  }  \ottnt{A_{{\mathrm{1}}}} $ and $    \Gamma_{{\mathrm{21}}}  ,   \ottmv{z}  :^{ \ottnt{r_{{\mathrm{02}}}} }  \ottnt{C}     ,  \Gamma_{{\mathrm{22}}}   \vdash  \ottnt{a_{{\mathrm{2}}}}  :^{ \ottnt{q} }   \ottnt{A_{{\mathrm{2}}}}  \{  \ottnt{a_{{\mathrm{1}}}}  /  \ottmv{x}  \}  $.\\
Further, $ \Gamma  \vdash  \ottnt{c}  :^{  \ottnt{r_{{\mathrm{01}}}}  +  \ottnt{r_{{\mathrm{02}}}}  }  \ottnt{C} $ where $  \lfloor  \Gamma  \rfloor   =   \lfloor  \Gamma_{{\mathrm{11}}}  \rfloor  $.\\
Need to show: $      \Gamma_{{\mathrm{11}}}  +  \Gamma_{{\mathrm{21}}}    +  \Gamma    ,      \Gamma_{{\mathrm{12}}}  \{  \ottnt{c}  /  \ottmv{z}  \}   +  \Gamma_{{\mathrm{22}}}   \{  \ottnt{c}  /  \ottmv{z}  \}     \vdash   (   \ottnt{a_{{\mathrm{1}}}}  \{  \ottnt{c}  /  \ottmv{z}  \}  ^{ \ottnt{r} } ,   \ottnt{a_{{\mathrm{2}}}}  \{  \ottnt{c}  /  \ottmv{z}  \}   )   :^{ \ottnt{q} }    \Sigma  \ottmv{x}  :^{ \ottnt{r} } \!   \ottnt{A_{{\mathrm{1}}}}  \{  \ottnt{c}  /  \ottmv{z}  \}   .  \ottnt{A_{{\mathrm{2}}}}   \{  \ottnt{c}  /  \ottmv{z}  \}  $.\\
By lemma \ref{BLDSplitP}, $\exists \Gamma_{{\mathrm{31}}}, \Gamma_{{\mathrm{32}}}$ such that $ \Gamma_{{\mathrm{31}}}  \vdash  \ottnt{c}  :^{ \ottnt{r_{{\mathrm{01}}}} }  \ottnt{C} $ and $ \Gamma_{{\mathrm{32}}}  \vdash  \ottnt{c}  :^{ \ottnt{r_{{\mathrm{02}}}} }  \ottnt{C} $ and $ \Gamma  =   \Gamma_{{\mathrm{31}}}  +  \Gamma_{{\mathrm{32}}}  $.\\
By IH, $     \Gamma_{{\mathrm{11}}}  +  \Gamma_{{\mathrm{31}}}    ,  \Gamma_{{\mathrm{12}}}   \{  \ottnt{c}  /  \ottmv{z}  \}   \vdash   \ottnt{a_{{\mathrm{1}}}}  \{  \ottnt{c}  /  \ottmv{z}  \}   :^{  \ottnt{q}  \cdot  \ottnt{r}  }   \ottnt{A_{{\mathrm{1}}}}  \{  \ottnt{c}  /  \ottmv{z}  \}  $\\ and $     \Gamma_{{\mathrm{21}}}  +  \Gamma_{{\mathrm{32}}}    ,  \Gamma_{{\mathrm{22}}}   \{  \ottnt{c}  /  \ottmv{z}  \}   \vdash   \ottnt{a_{{\mathrm{2}}}}  \{  \ottnt{c}  /  \ottmv{z}  \}   :^{ \ottnt{q} }    \ottnt{A_{{\mathrm{2}}}}  \{  \ottnt{a_{{\mathrm{1}}}}  /  \ottmv{x}  \}   \{  \ottnt{c}  /  \ottmv{z}  \}  $.\\
Now, $  \ottnt{A_{{\mathrm{2}}}}  \{  \ottnt{a_{{\mathrm{1}}}}  /  \ottmv{x}  \}   \{  \ottnt{c}  /  \ottmv{z}  \}  =   \ottnt{A_{{\mathrm{2}}}}  \{  \ottnt{c}  /  \ottmv{z}  \}   \{   \ottnt{a_{{\mathrm{1}}}}  \{  \ottnt{c}  /  \ottmv{z}  \}   /  \ottmv{x}  \} $.\\
This case, then, follows by \rref{PTS-Pair}.

\item \Rref{PTS-LetPair}. Have: $      \Gamma_{{\mathrm{11}}}  +  \Gamma_{{\mathrm{21}}}    ,   \ottmv{z}  :^{  \ottnt{r_{{\mathrm{01}}}}  +  \ottnt{r_{{\mathrm{02}}}}  }  \ottnt{C}     ,    \Gamma_{{\mathrm{12}}}  +  \Gamma_{{\mathrm{22}}}     \vdash   \mathbf{let}_{ \ottnt{q_{{\mathrm{0}}}} } \: (  \ottmv{x} ^{ \ottnt{r} } ,  \ottmv{y}  ) \: \mathbf{be} \:  \ottnt{a}  \: \mathbf{in} \:  \ottnt{b}   :^{ \ottnt{q} }   \ottnt{B}  \{  \ottnt{a}  /  w  \}  $ where \\
$  \Delta  ,   w  :   \Sigma  \ottmv{x}  :^{ \ottnt{r} } \!  \ottnt{A_{{\mathrm{1}}}}  .  \ottnt{A_{{\mathrm{2}}}}     \vdash_{0}  \ottnt{B}  :   \ottmv{s}  $ and $    \Gamma_{{\mathrm{11}}}  ,   \ottmv{z}  :^{ \ottnt{r_{{\mathrm{01}}}} }  \ottnt{C}     ,  \Gamma_{{\mathrm{12}}}   \vdash  \ottnt{a}  :^{  \ottnt{q}  \cdot  \ottnt{q_{{\mathrm{0}}}}  }   \Sigma  \ottmv{x}  :^{ \ottnt{r} } \!  \ottnt{A_{{\mathrm{1}}}}  .  \ottnt{A_{{\mathrm{2}}}}  $ and $ \Delta  =      \lfloor  \Gamma_{{\mathrm{11}}}  \rfloor   ,   \ottmv{z}  :  \ottnt{C}     ,   \lfloor  \Gamma_{{\mathrm{12}}}  \rfloor   $ and $      \Gamma_{{\mathrm{21}}}  ,   \ottmv{z}  :^{ \ottnt{r_{{\mathrm{02}}}} }  \ottnt{C}     ,  \Gamma_{{\mathrm{22}}}    ,     \ottmv{x}  :^{  \ottnt{q}  \cdot    \ottnt{q_{{\mathrm{0}}}}  \cdot  \ottnt{r}    }  \ottnt{A_{{\mathrm{1}}}}   ,   \ottmv{y}  :^{  \ottnt{q}  \cdot  \ottnt{q_{{\mathrm{0}}}}  }  \ottnt{A_{{\mathrm{2}}}}      \vdash  \ottnt{b}  :^{ \ottnt{q} }   \ottnt{B}  \{   (   \ottmv{x}  ^{ \ottnt{r} } ,   \ottmv{y}   )   /  w  \}  $.\\
Further, $ \Gamma  \vdash  \ottnt{c}  :^{  \ottnt{r_{{\mathrm{01}}}}  +  \ottnt{r_{{\mathrm{02}}}}  }  \ottnt{C} $ where $  \lfloor  \Gamma  \rfloor   =   \lfloor  \Gamma_{{\mathrm{11}}}  \rfloor  $.\\
Need to show: $      \Gamma_{{\mathrm{11}}}  +  \Gamma_{{\mathrm{21}}}    +  \Gamma    ,      \Gamma_{{\mathrm{12}}}  \{  \ottnt{c}  /  \ottmv{z}  \}   +  \Gamma_{{\mathrm{22}}}   \{  \ottnt{c}  /  \ottmv{z}  \}     \vdash    \mathbf{let}_{ \ottnt{q_{{\mathrm{0}}}} } \: (  \ottmv{x} ^{ \ottnt{r} } ,  \ottmv{y}  ) \: \mathbf{be} \:   \ottnt{a}  \{  \ottnt{c}  /  \ottmv{z}  \}   \: \mathbf{in} \:  \ottnt{b}   \{  \ottnt{c}  /  \ottmv{z}  \}   :^{ \ottnt{q} }    \ottnt{B}  \{  \ottnt{c}  /  \ottmv{z}  \}   \{   \ottnt{a}  \{  \ottnt{c}  /  \ottmv{z}  \}   /  w  \}  $.\\
By lemma \ref{BLDSplitP}, $\exists \Gamma_{{\mathrm{31}}}, \Gamma_{{\mathrm{32}}}$ such that $ \Gamma_{{\mathrm{31}}}  \vdash  \ottnt{c}  :^{ \ottnt{r_{{\mathrm{01}}}} }  \ottnt{C} $ and $ \Gamma_{{\mathrm{32}}}  \vdash  \ottnt{c}  :^{ \ottnt{r_{{\mathrm{02}}}} }  \ottnt{C} $ and $ \Gamma  =   \Gamma_{{\mathrm{31}}}  +  \Gamma_{{\mathrm{32}}}  $.\\
By IH, $     \Gamma_{{\mathrm{11}}}  +  \Gamma_{{\mathrm{31}}}    ,  \Gamma_{{\mathrm{12}}}   \{  \ottnt{c}  /  \ottmv{z}  \}   \vdash   \ottnt{a}  \{  \ottnt{c}  /  \ottmv{z}  \}   :^{  \ottnt{q}  \cdot  \ottnt{q_{{\mathrm{0}}}}  }    \Sigma  \ottmv{x}  :^{ \ottnt{r} } \!   \ottnt{A_{{\mathrm{1}}}}  \{  \ottnt{c}  /  \ottmv{z}  \}   .  \ottnt{A_{{\mathrm{2}}}}   \{  \ottnt{c}  /  \ottmv{z}  \}  $ and\\
$       \Gamma_{{\mathrm{21}}}  +  \Gamma_{{\mathrm{32}}}    ,  \Gamma_{{\mathrm{22}}}   \{  \ottnt{c}  /  \ottmv{z}  \}    ,      \ottmv{x}  :^{  \ottnt{q}  \cdot    \ottnt{q_{{\mathrm{0}}}}  \cdot  \ottnt{r}    }   \ottnt{A_{{\mathrm{1}}}}  \{  \ottnt{c}  /  \ottmv{z}  \}    ,   \ottmv{y}  :^{  \ottnt{q}  \cdot  \ottnt{q_{{\mathrm{0}}}}  }  \ottnt{A_{{\mathrm{2}}}}    \{  \ottnt{c}  /  \ottmv{z}  \}     \vdash   \ottnt{b}  \{  \ottnt{c}  /  \ottmv{z}  \}   :^{ \ottnt{q} }    \ottnt{B}  \{   (   \ottmv{x}  ^{ \ottnt{r} } ,   \ottmv{y}   )   /  w  \}   \{  \ottnt{c}  /  \ottmv{z}  \}  $.\\  
This case, then, follows by \rref{PTS-LetPair}.

\item \Rref{PTS-Sum, PTS-Inj1, PTS-Inj2}. By IH.

\item \Rref{PTS-Case}. Have: $      \Gamma_{{\mathrm{11}}}  +  \Gamma_{{\mathrm{21}}}    ,   \ottmv{z}  :^{  \ottnt{r_{{\mathrm{01}}}}  +  \ottnt{r_{{\mathrm{02}}}}  }  \ottnt{C}     ,    \Gamma_{{\mathrm{12}}}  +  \Gamma_{{\mathrm{22}}}     \vdash   \mathbf{case}_{ \ottnt{q_{{\mathrm{0}}}} } \:  \ottnt{a}  \: \mathbf{of} \:  \ottmv{x_{{\mathrm{1}}}}  .  \ottnt{b_{{\mathrm{1}}}}  \: ; \:  \ottmv{x_{{\mathrm{2}}}}  .  \ottnt{b_{{\mathrm{2}}}}   :^{ \ottnt{q} }   \ottnt{B}  \{  \ottnt{a}  /  w  \}  $ where \\
$  \Delta  ,   w  :   \ottnt{A_{{\mathrm{1}}}}  +  \ottnt{A_{{\mathrm{2}}}}     \vdash_{0}  \ottnt{B}  :   \ottmv{s}  $ and $    \Gamma_{{\mathrm{11}}}  ,   \ottmv{z}  :^{ \ottnt{r_{{\mathrm{01}}}} }  \ottnt{C}     ,  \Gamma_{{\mathrm{12}}}   \vdash  \ottnt{a}  :^{  \ottnt{q}  \cdot  \ottnt{q_{{\mathrm{0}}}}  }   \ottnt{A_{{\mathrm{1}}}}  +  \ottnt{A_{{\mathrm{2}}}}  $ and $ \Delta  =      \lfloor  \Gamma_{{\mathrm{11}}}  \rfloor   ,   \ottmv{z}  :  \ottnt{C}     ,   \lfloor  \Gamma_{{\mathrm{12}}}  \rfloor   $ and\\ $      \Gamma_{{\mathrm{21}}}  ,   \ottmv{z}  :^{ \ottnt{r_{{\mathrm{02}}}} }  \ottnt{C}     ,  \Gamma_{{\mathrm{22}}}    ,   \ottmv{x_{{\mathrm{1}}}}  :^{  \ottnt{q}  \cdot  \ottnt{q_{{\mathrm{0}}}}  }  \ottnt{A_{{\mathrm{1}}}}    \vdash  \ottnt{b_{{\mathrm{1}}}}  :^{ \ottnt{q} }   \ottnt{B}  \{   \mathbf{inj}_1 \:   \ottmv{x_{{\mathrm{1}}}}    /  w  \}  $ and\\ $      \Gamma_{{\mathrm{21}}}  ,   \ottmv{z}  :^{ \ottnt{r_{{\mathrm{02}}}} }  \ottnt{C}     ,  \Gamma_{{\mathrm{22}}}    ,   \ottmv{x_{{\mathrm{2}}}}  :^{  \ottnt{q}  \cdot  \ottnt{q_{{\mathrm{0}}}}  }  \ottnt{A_{{\mathrm{2}}}}    \vdash  \ottnt{b_{{\mathrm{2}}}}  :^{ \ottnt{q} }   \ottnt{B}  \{   \mathbf{inj}_2 \:   \ottmv{x_{{\mathrm{2}}}}    /  w  \}  $.\\
Further, $ \Gamma  \vdash  \ottnt{c}  :^{  \ottnt{r_{{\mathrm{01}}}}  +  \ottnt{r_{{\mathrm{02}}}}  }  \ottnt{C} $ where $  \lfloor  \Gamma  \rfloor   =   \lfloor  \Gamma_{{\mathrm{11}}}  \rfloor  $.\\
Need to show: $      \Gamma_{{\mathrm{11}}}  +  \Gamma_{{\mathrm{21}}}    +  \Gamma    ,      \Gamma_{{\mathrm{12}}}  \{  \ottnt{c}  /  \ottmv{z}  \}   +  \Gamma_{{\mathrm{22}}}   \{  \ottnt{c}  /  \ottmv{z}  \}     \vdash    \mathbf{case}_{ \ottnt{q_{{\mathrm{0}}}} } \:   \ottnt{a}  \{  \ottnt{c}  /  \ottmv{z}  \}   \: \mathbf{of} \:  \ottmv{x_{{\mathrm{1}}}}  .   \ottnt{b_{{\mathrm{1}}}}  \{  \ottnt{c}  /  \ottmv{z}  \}   \: ; \:  \ottmv{x_{{\mathrm{2}}}}  .  \ottnt{b_{{\mathrm{2}}}}   \{  \ottnt{c}  /  \ottmv{z}  \}   :^{ \ottnt{q} }    \ottnt{B}  \{  \ottnt{c}  /  \ottmv{z}  \}   \{   \ottnt{a}  \{  \ottnt{c}  /  \ottmv{z}  \}   /  w  \}  $.\\
By lemma \ref{BLDSplitP}, $\exists \Gamma_{{\mathrm{31}}}, \Gamma_{{\mathrm{32}}}$ such that $ \Gamma_{{\mathrm{31}}}  \vdash  \ottnt{c}  :^{ \ottnt{r_{{\mathrm{01}}}} }  \ottnt{C} $ and $ \Gamma_{{\mathrm{32}}}  \vdash  \ottnt{c}  :^{ \ottnt{r_{{\mathrm{02}}}} }  \ottnt{C} $ and $ \Gamma  =   \Gamma_{{\mathrm{31}}}  +  \Gamma_{{\mathrm{32}}}  $.\\
By IH, $     \Gamma_{{\mathrm{11}}}  +  \Gamma_{{\mathrm{31}}}    ,  \Gamma_{{\mathrm{12}}}   \{  \ottnt{c}  /  \ottmv{z}  \}   \vdash   \ottnt{a}  \{  \ottnt{c}  /  \ottmv{z}  \}   :^{  \ottnt{q}  \cdot  \ottnt{q_{{\mathrm{0}}}}  }     \ottnt{A_{{\mathrm{1}}}}  \{  \ottnt{c}  /  \ottmv{z}  \}   +  \ottnt{A_{{\mathrm{2}}}}   \{  \ottnt{c}  /  \ottmv{z}  \}  $ and\\
$        \Gamma_{{\mathrm{21}}}  +  \Gamma_{{\mathrm{32}}}    ,  \Gamma_{{\mathrm{22}}}   \{  \ottnt{c}  /  \ottmv{z}  \}    ,   \ottmv{x_{{\mathrm{1}}}}  :^{  \ottnt{q}  \cdot  \ottnt{q_{{\mathrm{0}}}}  }  \ottnt{A_{{\mathrm{1}}}}    \{  \ottnt{c}  /  \ottmv{z}  \}   \vdash   \ottnt{b_{{\mathrm{1}}}}  \{  \ottnt{c}  /  \ottmv{z}  \}   :^{ \ottnt{q} }    \ottnt{B}  \{   \mathbf{inj}_1 \:   \ottmv{x_{{\mathrm{1}}}}    /  w  \}   \{  \ottnt{c}  /  \ottmv{z}  \}  $ and\\ $        \Gamma_{{\mathrm{21}}}  +  \Gamma_{{\mathrm{32}}}    ,  \Gamma_{{\mathrm{22}}}   \{  \ottnt{c}  /  \ottmv{z}  \}    ,   \ottmv{x_{{\mathrm{2}}}}  :^{  \ottnt{q}  \cdot  \ottnt{q_{{\mathrm{0}}}}  }  \ottnt{A_{{\mathrm{2}}}}    \{  \ottnt{c}  /  \ottmv{z}  \}   \vdash   \ottnt{b_{{\mathrm{2}}}}  \{  \ottnt{c}  /  \ottmv{z}  \}   :^{ \ottnt{q} }    \ottnt{B}  \{   \mathbf{inj}_2 \:   \ottmv{x_{{\mathrm{2}}}}    /  w  \}   \{  \ottnt{c}  /  \ottmv{z}  \}  $.\\  
This case, then, follows by \rref{PTS-Case}.

\item \Rref{PTS-SubL, PTS-SubR}. By IH.

\end{itemize}

\end{proof}

%---------------------------------------------------------------------------------------------------

\begin{theorem}[Preservation (Theorem \ref{preservationDep})] \label{Dpreserve}
If $ \Gamma  \vdash  \ottnt{a}  :^{ \ottnt{q} }  \ottnt{A} $ and $ \vdash  \ottnt{a}  \leadsto  \ottnt{a'} $, then $ \Gamma  \vdash  \ottnt{a'}  :^{ \ottnt{q} }  \ottnt{A} $.
\end{theorem}

\begin{proof}

By induction on $ \Gamma  \vdash  \ottnt{a}  :^{ \ottnt{q} }  \ottnt{A} $ and inversion on $ \vdash  \ottnt{a}  \leadsto  \ottnt{a'} $.

\begin{itemize}

\item \Rref{PTS-App}. Have: $  \Gamma_{{\mathrm{1}}}  +  \Gamma_{{\mathrm{2}}}   \vdash   \ottnt{b}  \:  \ottnt{a} ^{ \ottnt{r} }   :^{ \ottnt{q} }   \ottnt{B}  \{  \ottnt{a}  /  \ottmv{x}  \}  $ where $ \Gamma_{{\mathrm{1}}}  \vdash  \ottnt{b}  :^{ \ottnt{q} }   \Pi  \ottmv{x}  :^{ \ottnt{r} } \!  \ottnt{A}  .  \ottnt{B}  $ and $ \Gamma_{{\mathrm{2}}}  \vdash  \ottnt{a}  :^{  \ottnt{q}  \cdot  \ottnt{r}  }  \ottnt{A} $. \\ Let $ \vdash   \ottnt{b}  \:  \ottnt{a} ^{ \ottnt{r} }   \leadsto  \ottnt{c} $. By inversion:

\begin{itemize}
\item $ \vdash   \ottnt{b}  \:  \ottnt{a} ^{ \ottnt{r} }   \leadsto   \ottnt{b'}  \:  \ottnt{a} ^{ \ottnt{r} }  $, when $ \vdash  \ottnt{b}  \leadsto  \ottnt{b'} $. \\
Need to show: $  \Gamma_{{\mathrm{1}}}  +  \Gamma_{{\mathrm{2}}}   \vdash   \ottnt{b'}  \:  \ottnt{a} ^{ \ottnt{r} }   :^{ \ottnt{q} }   \ottnt{B}  \{  \ottnt{a}  /  \ottmv{x}  \}  $.\\
Follows by IH and \rref{PTS-App}.

\item $\ottnt{b} =  \lambda^{ \ottnt{r} }  \ottmv{x}  :  \ottnt{A'}  .  \ottnt{b'} $ and $ \vdash   \ottnt{b}  \:  \ottnt{a} ^{ \ottnt{r} }   \leadsto   \ottnt{b'}  \{  \ottnt{a}  /  \ottmv{x}  \}  $.\\
Need to show: $  \Gamma_{{\mathrm{1}}}  +  \Gamma_{{\mathrm{2}}}   \vdash   \ottnt{b'}  \{  \ottnt{a}  /  \ottmv{x}  \}   :^{ \ottnt{q} }   \ottnt{B}  \{  \ottnt{a}  /  \ottmv{x}  \}  $.\\
By inversion on $ \Gamma_{{\mathrm{1}}}  \vdash   \lambda^{ \ottnt{r} }  \ottmv{x}  :  \ottnt{A'}  .  \ottnt{b'}   :^{ \ottnt{q} }   \Pi  \ottmv{x}  :^{ \ottnt{r} } \!  \ottnt{A}  .  \ottnt{B}  $, we get $  \Gamma_{{\mathrm{1}}}  ,   \ottmv{x}  :^{  \ottnt{q_{{\mathrm{0}}}}  \cdot  \ottnt{r}  }  \ottnt{A''}    \vdash  \ottnt{b'}  :^{ \ottnt{q_{{\mathrm{0}}}} }  \ottnt{B} $ for some $ \ottnt{q_{{\mathrm{0}}}}  <:  \ottnt{q} $ and $ \ottnt{A''}  =_{\beta}  \ottnt{A} $.\\
Now, there are two cases to consider.
\begin{itemize}
\item $\ottnt{q_{{\mathrm{0}}}} = 0$. Since $ \ottnt{q_{{\mathrm{0}}}}  <:  \ottnt{q} $, so $\ottnt{q} = 0$.\\
Then, by \rref{PTS-Conv} and the substitution lemma, $  \Gamma_{{\mathrm{1}}}  +  \Gamma_{{\mathrm{2}}}   \vdash   \ottnt{b'}  \{  \ottnt{a}  /  \ottmv{x}  \}   :^{  0  }   \ottnt{B}  \{  \ottnt{a}  /  \ottmv{x}  \}  $. 
\item $\ottnt{q_{{\mathrm{0}}}} \neq 0$. By lemma \ref{BLDFactP}, $\exists \Gamma'_{{\mathrm{1}}}$ and $\ottnt{r'}$ such that $  \Gamma'_{{\mathrm{1}}}  ,   \ottmv{x}  :^{ \ottnt{r'} }  \ottnt{A''}    \vdash  \ottnt{b'}  :^{  1  }  \ottnt{B} $ and $ \Gamma_{{\mathrm{1}}}  <:   \ottnt{q_{{\mathrm{0}}}}  \cdot  \Gamma'_{{\mathrm{1}}}  $ and $  \ottnt{q_{{\mathrm{0}}}}  \cdot  \ottnt{r}   <:   \ottnt{q_{{\mathrm{0}}}}  \cdot  \ottnt{r'}  $.\\
Since $\ottnt{q_{{\mathrm{0}}}} \neq 0$, therefore $ \ottnt{r}  <:  \ottnt{r'} $. Hence, by \rref{PTS-SubL}, $  \Gamma'_{{\mathrm{1}}}  ,   \ottmv{x}  :^{ \ottnt{r} }  \ottnt{A''}    \vdash  \ottnt{b'}  :^{  1  }  \ottnt{B} $.\\
Now, by lemma \ref{BLDMultP}, $   \ottnt{q}  \cdot  \Gamma'_{{\mathrm{1}}}   ,   \ottmv{x}  :^{  \ottnt{q}  \cdot  \ottnt{r}  }  \ottnt{A''}    \vdash  \ottnt{b'}  :^{ \ottnt{q} }  \ottnt{B} $. By \rref{PTS-SubL}, $  \Gamma_{{\mathrm{1}}}  ,   \ottmv{x}  :^{  \ottnt{q}  \cdot  \ottnt{r}  }  \ottnt{A''}    \vdash  \ottnt{b'}  :^{ \ottnt{q} }  \ottnt{B} $.\\
This case, then, follows by \rref{PTS-Conv} and the substitution lemma.
\end{itemize}

\end{itemize} 

\item \Rref{PTS-LetUnit}. Have: $  \Gamma_{{\mathrm{1}}}  +  \Gamma_{{\mathrm{2}}}   \vdash   \mathbf{let}_{ \ottnt{q_{{\mathrm{0}}}} } \: \mathbf{unit} \: \mathbf{be} \:  \ottnt{a}  \: \mathbf{in} \:  \ottnt{b}   :^{ \ottnt{q} }   \ottnt{B}  \{  \ottnt{a}  /  \ottmv{z}  \}  $ where $  \Delta  ,   \ottmv{z}  :   \mathbf{Unit}     \vdash_{0}  \ottnt{B}  :   \ottmv{s}  $ and $\Delta =  \lfloor  \Gamma_{{\mathrm{1}}}  \rfloor $ and $ \Gamma_{{\mathrm{1}}}  \vdash  \ottnt{a}  :^{  \ottnt{q}  \cdot  \ottnt{q_{{\mathrm{0}}}}  }   \mathbf{Unit}  $ and $ \Gamma_{{\mathrm{2}}}  \vdash  \ottnt{b}  :^{ \ottnt{q} }   \ottnt{B}  \{   \mathbf{unit}   /  \ottmv{z}  \}  $.\\ Let $ \vdash   \mathbf{let}_{ \ottnt{q_{{\mathrm{0}}}} } \: \mathbf{unit} \: \mathbf{be} \:  \ottnt{a}  \: \mathbf{in} \:  \ottnt{b}   \leadsto  \ottnt{c} $. By inversion:

\begin{itemize}

\item $ \vdash   \mathbf{let}_{ \ottnt{q_{{\mathrm{0}}}} } \: \mathbf{unit} \: \mathbf{be} \:  \ottnt{a}  \: \mathbf{in} \:  \ottnt{b}   \leadsto   \mathbf{let}_{ \ottnt{q_{{\mathrm{0}}}} } \: \mathbf{unit} \: \mathbf{be} \:  \ottnt{a'}  \: \mathbf{in} \:  \ottnt{b}  $, when $ \vdash  \ottnt{a}  \leadsto  \ottnt{a'} $.\\
Need to show: $  \Gamma_{{\mathrm{1}}}  +  \Gamma_{{\mathrm{2}}}   \vdash   \mathbf{let}_{ \ottnt{q_{{\mathrm{0}}}} } \: \mathbf{unit} \: \mathbf{be} \:  \ottnt{a'}  \: \mathbf{in} \:  \ottnt{b}   :^{ \ottnt{q} }   \ottnt{B}  \{  \ottnt{a}  /  \ottmv{z}  \}  $.\\
Follows by IH and \rref{PTS-LetUnit,PTS-Conv}.

\item $ \vdash   \mathbf{let}_{ \ottnt{q_{{\mathrm{0}}}} } \: \mathbf{unit} \: \mathbf{be} \:   \mathbf{unit}   \: \mathbf{in} \:  \ottnt{b}   \leadsto  \ottnt{b} $.\\
Need to show: $  \Gamma_{{\mathrm{1}}}  +  \Gamma_{{\mathrm{2}}}   \vdash  \ottnt{b}  :^{ \ottnt{q} }   \ottnt{B}  \{   \mathbf{unit}   /  \ottmv{z}  \}  $.\\
Follows by \rref{PTS-SubL} (note that in case of $ \mathbb{N}_{=} $, $ \overline{ \Gamma_{{\mathrm{1}}} }  = \overline{0}$).

\end{itemize}

\item \Rref{PTS-LetPair}. Have: $  \Gamma_{{\mathrm{1}}}  +  \Gamma_{{\mathrm{2}}}   \vdash   \mathbf{let}_{ \ottnt{q_{{\mathrm{0}}}} } \: (  \ottmv{x} ^{ \ottnt{r} } ,  \ottmv{y}  ) \: \mathbf{be} \:  \ottnt{a}  \: \mathbf{in} \:  \ottnt{b}   :^{ \ottnt{q} }   \ottnt{B}  \{  \ottnt{a}  /  \ottmv{z}  \}  $ where $  \Delta  ,   \ottmv{z}  :   \Sigma  \ottmv{x}  :^{ \ottnt{r} } \!  \ottnt{A_{{\mathrm{1}}}}  .  \ottnt{A_{{\mathrm{2}}}}     \vdash_{0}  \ottnt{B}  :   \ottmv{s}  $ and $\Delta =  \lfloor  \Gamma_{{\mathrm{1}}}  \rfloor $ and $ \Gamma_{{\mathrm{1}}}  \vdash  \ottnt{a}  :^{  \ottnt{q}  \cdot  \ottnt{q_{{\mathrm{0}}}}  }   \Sigma  \ottmv{x}  :^{ \ottnt{r} } \!  \ottnt{A_{{\mathrm{1}}}}  .  \ottnt{A_{{\mathrm{2}}}}  $ and $    \Gamma_{{\mathrm{2}}}  ,   \ottmv{x}  :^{    \ottnt{q}  \cdot  \ottnt{q_{{\mathrm{0}}}}    \cdot  \ottnt{r}  }  \ottnt{A_{{\mathrm{1}}}}     ,   \ottmv{y}  :^{  \ottnt{q}  \cdot  \ottnt{q_{{\mathrm{0}}}}  }  \ottnt{A_{{\mathrm{2}}}}    \vdash  \ottnt{b}  :^{ \ottnt{q} }   \ottnt{B}  \{   (   \ottmv{x}  ^{ \ottnt{r} } ,   \ottmv{y}   )   /  \ottmv{z}  \}  $. \\ Let $ \vdash   \mathbf{let}_{ \ottnt{q_{{\mathrm{0}}}} } \: (  \ottmv{x} ^{ \ottnt{r} } ,  \ottmv{y}  ) \: \mathbf{be} \:  \ottnt{a}  \: \mathbf{in} \:  \ottnt{b}   \leadsto  \ottnt{c} $. By inversion:

\begin{itemize}
\item $ \vdash   \mathbf{let}_{ \ottnt{q_{{\mathrm{0}}}} } \: (  \ottmv{x} ^{ \ottnt{r} } ,  \ottmv{y}  ) \: \mathbf{be} \:  \ottnt{a}  \: \mathbf{in} \:  \ottnt{b}   \leadsto   \mathbf{let}_{ \ottnt{q_{{\mathrm{0}}}} } \: (  \ottmv{x} ^{ \ottnt{r} } ,  \ottmv{y}  ) \: \mathbf{be} \:  \ottnt{a'}  \: \mathbf{in} \:  \ottnt{b}  $, when $ \vdash  \ottnt{a}  \leadsto  \ottnt{a'} $.\\
Need to show: $  \Gamma_{{\mathrm{1}}}  +  \Gamma_{{\mathrm{2}}}   \vdash   \mathbf{let}_{ \ottnt{q_{{\mathrm{0}}}} } \: (  \ottmv{x} ^{ \ottnt{r} } ,  \ottmv{y}  ) \: \mathbf{be} \:  \ottnt{a'}  \: \mathbf{in} \:  \ottnt{b}   :^{ \ottnt{q} }   \ottnt{B}  \{  \ottnt{a}  /  \ottmv{z}  \}  $.\\
Follows by IH and \rref{PTS-LetPair,PTS-Conv}.

\item $ \vdash   \mathbf{let}_{ \ottnt{q_{{\mathrm{0}}}} } \: (  \ottmv{x} ^{ \ottnt{r} } ,  \ottmv{y}  ) \: \mathbf{be} \:   (  \ottnt{a_{{\mathrm{1}}}} ^{ \ottnt{r} } ,  \ottnt{a_{{\mathrm{2}}}}  )   \: \mathbf{in} \:  \ottnt{b}   \leadsto    \ottnt{b}  \{  \ottnt{a_{{\mathrm{1}}}}  /  \ottmv{x}  \}   \{  \ottnt{a_{{\mathrm{2}}}}  /  \ottmv{y}  \}  $.\\
Need to show: $  \Gamma_{{\mathrm{1}}}  +  \Gamma_{{\mathrm{2}}}   \vdash    \ottnt{b}  \{  \ottnt{a_{{\mathrm{1}}}}  /  \ottmv{x}  \}   \{  \ottnt{a_{{\mathrm{2}}}}  /  \ottmv{y}  \}   :^{ \ottnt{q} }   \ottnt{B}  \{   (  \ottnt{a_{{\mathrm{1}}}} ^{ \ottnt{r} } ,  \ottnt{a_{{\mathrm{2}}}}  )   /  \ottmv{z}  \}  $.\\
By inversion on $ \Gamma_{{\mathrm{1}}}  \vdash   (  \ottnt{a_{{\mathrm{1}}}} ^{ \ottnt{r} } ,  \ottnt{a_{{\mathrm{2}}}}  )   :^{  \ottnt{q}  \cdot  \ottnt{q_{{\mathrm{0}}}}  }   \Sigma  \ottmv{x}  :^{ \ottnt{r} } \!  \ottnt{A_{{\mathrm{1}}}}  .  \ottnt{A_{{\mathrm{2}}}}  $, we have:\\
$\exists \Gamma_{{\mathrm{11}}}, \Gamma_{{\mathrm{12}}}$ such that $ \Gamma_{{\mathrm{11}}}  \vdash  \ottnt{a_{{\mathrm{1}}}}  :^{    \ottnt{q}  \cdot  \ottnt{q_{{\mathrm{0}}}}    \cdot  \ottnt{r}  }  \ottnt{A_{{\mathrm{1}}}} $ and $ \Gamma_{{\mathrm{12}}}  \vdash  \ottnt{a_{{\mathrm{2}}}}  :^{  \ottnt{q}  \cdot  \ottnt{q_{{\mathrm{0}}}}  }   \ottnt{A_{{\mathrm{2}}}}  \{  \ottnt{a_{{\mathrm{1}}}}  /  \ottmv{x}  \}  $ and $\Gamma_{{\mathrm{1}}} =  \Gamma_{{\mathrm{11}}}  +  \Gamma_{{\mathrm{12}}} $.\\
Applying the substitution lemma, we get, $     \Gamma_{{\mathrm{2}}}  +  \Gamma_{{\mathrm{11}}}    ,   \ottmv{y}  :^{  \ottnt{q}  \cdot  \ottnt{q_{{\mathrm{0}}}}  }  \ottnt{A_{{\mathrm{2}}}}    \{  \ottnt{a_{{\mathrm{1}}}}  /  \ottmv{x}  \}   \vdash   \ottnt{b}  \{  \ottnt{a_{{\mathrm{1}}}}  /  \ottmv{x}  \}   :^{ \ottnt{q} }   \ottnt{B}  \{   (  \ottnt{a_{{\mathrm{1}}}} ^{ \ottnt{r} } ,   \ottmv{y}   )   /  \ottmv{z}  \}  $.\\
Applying the lemma again, we get, $    \Gamma_{{\mathrm{2}}}  +  \Gamma_{{\mathrm{11}}}    +  \Gamma_{{\mathrm{12}}}   \vdash    \ottnt{b}  \{  \ottnt{a_{{\mathrm{1}}}}  /  \ottmv{x}  \}   \{  \ottnt{a_{{\mathrm{2}}}}  /  \ottmv{y}  \}   :^{ \ottnt{q} }   \ottnt{B}  \{   (  \ottnt{a_{{\mathrm{1}}}} ^{ \ottnt{r} } ,  \ottnt{a_{{\mathrm{2}}}}  )   /  \ottmv{z}  \}  $, as required.
\end{itemize}

\item \Rref{PTS-Case}. Have: $  \Gamma_{{\mathrm{1}}}  +  \Gamma_{{\mathrm{2}}}   \vdash   \mathbf{case}_{ \ottnt{q_{{\mathrm{0}}}} } \:  \ottnt{a}  \: \mathbf{of} \:  \ottmv{x_{{\mathrm{1}}}}  .  \ottnt{b_{{\mathrm{1}}}}  \: ; \:  \ottmv{x_{{\mathrm{2}}}}  .  \ottnt{b_{{\mathrm{2}}}}   :^{ \ottnt{q} }   \ottnt{B}  \{  \ottnt{a}  /  \ottmv{z}  \}  $ where $  \Delta  ,   \ottmv{z}  :   \ottnt{A_{{\mathrm{1}}}}  +  \ottnt{A_{{\mathrm{2}}}}     \vdash_{0}  \ottnt{B}  :   \ottmv{s}  $ and $\Delta =  \lfloor  \Gamma_{{\mathrm{1}}}  \rfloor $ and $ \Gamma_{{\mathrm{1}}}  \vdash  \ottnt{a}  :^{  \ottnt{q}  \cdot  \ottnt{q_{{\mathrm{0}}}}  }   \ottnt{A_{{\mathrm{1}}}}  +  \ottnt{A_{{\mathrm{2}}}}  $ and $  \Gamma_{{\mathrm{2}}}  ,   \ottmv{x_{{\mathrm{1}}}}  :^{  \ottnt{q}  \cdot  \ottnt{q_{{\mathrm{0}}}}  }  \ottnt{A_{{\mathrm{1}}}}    \vdash  \ottnt{b_{{\mathrm{1}}}}  :^{ \ottnt{q} }   \ottnt{B}  \{   \mathbf{inj}_1 \:   \ottmv{x_{{\mathrm{1}}}}    /  \ottmv{z}  \}  $ and $  \Gamma_{{\mathrm{2}}}  ,   \ottmv{x_{{\mathrm{2}}}}  :^{  \ottnt{q}  \cdot  \ottnt{q_{{\mathrm{0}}}}  }  \ottnt{A_{{\mathrm{2}}}}    \vdash  \ottnt{b_{{\mathrm{2}}}}  :^{ \ottnt{q} }   \ottnt{B}  \{   \mathbf{inj}_2 \:   \ottmv{x_{{\mathrm{2}}}}    /  \ottmv{z}  \}  $.\\ Let $ \vdash   \mathbf{case}_{ \ottnt{q_{{\mathrm{0}}}} } \:  \ottnt{a}  \: \mathbf{of} \:  \ottmv{x_{{\mathrm{1}}}}  .  \ottnt{b_{{\mathrm{1}}}}  \: ; \:  \ottmv{x_{{\mathrm{2}}}}  .  \ottnt{b_{{\mathrm{2}}}}   \leadsto  \ottnt{c} $. By inversion:

\begin{itemize}
\item $ \vdash   \mathbf{case}_{ \ottnt{q_{{\mathrm{0}}}} } \:  \ottnt{a}  \: \mathbf{of} \:  \ottmv{x_{{\mathrm{1}}}}  .  \ottnt{b_{{\mathrm{1}}}}  \: ; \:  \ottmv{x_{{\mathrm{2}}}}  .  \ottnt{b_{{\mathrm{2}}}}   \leadsto   \mathbf{case}_{ \ottnt{q_{{\mathrm{0}}}} } \:  \ottnt{a'}  \: \mathbf{of} \:  \ottmv{x_{{\mathrm{1}}}}  .  \ottnt{b_{{\mathrm{1}}}}  \: ; \:  \ottmv{x_{{\mathrm{2}}}}  .  \ottnt{b_{{\mathrm{2}}}}  $, when $ \vdash  \ottnt{a}  \leadsto  \ottnt{a'} $.\\
Need to show: $  \Gamma_{{\mathrm{1}}}  +  \Gamma_{{\mathrm{2}}}   \vdash   \mathbf{case}_{ \ottnt{q_{{\mathrm{0}}}} } \:  \ottnt{a'}  \: \mathbf{of} \:  \ottmv{x_{{\mathrm{1}}}}  .  \ottnt{b_{{\mathrm{1}}}}  \: ; \:  \ottmv{x_{{\mathrm{2}}}}  .  \ottnt{b_{{\mathrm{2}}}}   :^{ \ottnt{q} }   \ottnt{B}  \{  \ottnt{a}  /  \ottmv{z}  \}  $.\\
Follows by IH and \rref{PTS-Case,PTS-Conv}.

\item $ \vdash   \mathbf{case}_{ \ottnt{q_{{\mathrm{0}}}} } \:   (   \mathbf{inj}_1 \:  \ottnt{a_{{\mathrm{1}}}}   )   \: \mathbf{of} \:  \ottmv{x_{{\mathrm{1}}}}  .  \ottnt{b_{{\mathrm{1}}}}  \: ; \:  \ottmv{x_{{\mathrm{2}}}}  .  \ottnt{b_{{\mathrm{2}}}}   \leadsto   \ottnt{b_{{\mathrm{1}}}}  \{  \ottnt{a_{{\mathrm{1}}}}  /  \ottmv{x_{{\mathrm{1}}}}  \}  $.\\
Need to show $  \Gamma_{{\mathrm{1}}}  +  \Gamma_{{\mathrm{2}}}   \vdash   \ottnt{b_{{\mathrm{1}}}}  \{  \ottnt{a_{{\mathrm{1}}}}  /  \ottmv{x_{{\mathrm{1}}}}  \}   :^{ \ottnt{q} }   \ottnt{B}  \{   \mathbf{inj}_1 \:  \ottnt{a_{{\mathrm{1}}}}   /  \ottmv{z}  \}  $.\\
By inversion on $ \Gamma_{{\mathrm{1}}}  \vdash   \mathbf{inj}_1 \:  \ottnt{a_{{\mathrm{1}}}}   :^{  \ottnt{q}  \cdot  \ottnt{q_{{\mathrm{0}}}}  }   \ottnt{A_{{\mathrm{1}}}}  +  \ottnt{A_{{\mathrm{2}}}}  $, we have $ \Gamma_{{\mathrm{1}}}  \vdash  \ottnt{a_{{\mathrm{1}}}}  :^{  \ottnt{q}  \cdot  \ottnt{q_{{\mathrm{0}}}}  }  \ottnt{A_{{\mathrm{1}}}} $.\\
This case, then, follows by applying the substitution lemma.

\item $ \vdash   \mathbf{case}_{ \ottnt{q_{{\mathrm{0}}}} } \:   (   \mathbf{inj}_2 \:  \ottnt{a_{{\mathrm{2}}}}   )   \: \mathbf{of} \:  \ottmv{x_{{\mathrm{1}}}}  .  \ottnt{b_{{\mathrm{1}}}}  \: ; \:  \ottmv{x_{{\mathrm{2}}}}  .  \ottnt{b_{{\mathrm{2}}}}   \leadsto   \ottnt{b_{{\mathrm{2}}}}  \{  \ottnt{a_{{\mathrm{2}}}}  /  \ottmv{x_{{\mathrm{2}}}}  \}  $.\\
Similar to the previous case.

\end{itemize}

\item \Rref{PTS-Weak,PTS-Conv,PTS-SubL,PTS-SubR}. Follows by IH.

\end{itemize}

\end{proof}

%----------------------------------------------------------------------------------------------------

\begin{theorem}[Progress (Theorem \ref{progressDep})] \label{Dprogress}
If $  \emptyset   \vdash  \ottnt{a}  :^{ \ottnt{q} }  \ottnt{A} $, then either $\ottnt{a}$ is a value or there exists $\ottnt{a'}$ such that $ \vdash  \ottnt{a}  \leadsto  \ottnt{a'} $.
\end{theorem}

\begin{proof}

By induction on $  \emptyset   \vdash  \ottnt{a}  :^{ \ottnt{q} }  \ottnt{A} $.
\begin{itemize}

\item \Rref{PTS-App}. Have: $  \emptyset   \vdash   \ottnt{b}  \:  \ottnt{a} ^{ \ottnt{r} }   :^{ \ottnt{q} }   \ottnt{B}  \{  \ottnt{a}  /  \ottmv{x}  \}  $ where $  \emptyset   \vdash  \ottnt{b}  :^{ \ottnt{q} }   \Pi  \ottmv{x}  :^{ \ottnt{r} } \!  \ottnt{A}  .  \ottnt{B}  $ and $  \emptyset   \vdash  \ottnt{a}  :^{  \ottnt{q}  \cdot  \ottnt{r}  }  \ottnt{A} $. \\
Need to show: $\exists c,  \vdash   \ottnt{b}  \:  \ottnt{a} ^{ \ottnt{r} }   \leadsto  \ottnt{c} $.\\
By IH, $\ottnt{b}$ is either a value or $ \vdash  \ottnt{b}  \leadsto  \ottnt{b'} $.\\
If $\ottnt{b}$ is a value, then, by inversion, $\ottnt{b} =  \lambda^{ \ottnt{r} }  \ottmv{x}  :  \ottnt{A'}  .  \ottnt{b'} $ for some $\ottnt{b'}$. Therefore, $ \vdash   \ottnt{b}  \:  \ottnt{a} ^{ \ottnt{r} }   \leadsto   \ottnt{b'}  \{  \ottnt{a}  /  \ottmv{x}  \}  $.\\
Otherwise, $ \vdash   \ottnt{b}  \:  \ottnt{a} ^{ \ottnt{r} }   \leadsto   \ottnt{b'}  \:  \ottnt{a} ^{ \ottnt{r} }  $.

\item \Rref{PTS-LetUnit}. Have: $  \emptyset   \vdash   \mathbf{let}_{ \ottnt{q_{{\mathrm{0}}}} } \: \mathbf{unit} \: \mathbf{be} \:  \ottnt{a}  \: \mathbf{in} \:  \ottnt{b}   :^{ \ottnt{q} }   \ottnt{B}  \{  \ottnt{a}  /  \ottmv{z}  \}  $ where $  \ottmv{z}  :   \mathbf{Unit}    \vdash_{0}  \ottnt{B}  :   \ottmv{s}  $ and $  \emptyset   \vdash  \ottnt{a}  :^{  \ottnt{q}  \cdot  \ottnt{q_{{\mathrm{0}}}}  }   \mathbf{Unit}  $ and $  \emptyset   \vdash  \ottnt{b}  :^{ \ottnt{q} }   \ottnt{B}  \{   \mathbf{unit}   /  \ottmv{z}  \}  $.\\
Ned to show: $\exists c,  \vdash   \mathbf{let}_{ \ottnt{q_{{\mathrm{0}}}} } \: \mathbf{unit} \: \mathbf{be} \:  \ottnt{a}  \: \mathbf{in} \:  \ottnt{b}   \leadsto  \ottnt{c} $.\\
By IH, $\ottnt{a}$ is either a value or $ \vdash  \ottnt{a}  \leadsto  \ottnt{a'} $.\\
If $\ottnt{a}$ is a value, then, by inversion, $\ottnt{a} =  \mathbf{unit} $. Therefore, $ \vdash   \mathbf{let}_{ \ottnt{q_{{\mathrm{0}}}} } \: \mathbf{unit} \: \mathbf{be} \:  \ottnt{a}  \: \mathbf{in} \:  \ottnt{b}   \leadsto  \ottnt{b} $.\\
Otherwise, $ \vdash   \mathbf{let}_{ \ottnt{q_{{\mathrm{0}}}} } \: \mathbf{unit} \: \mathbf{be} \:  \ottnt{a}  \: \mathbf{in} \:  \ottnt{b}   \leadsto   \mathbf{let}_{ \ottnt{q_{{\mathrm{0}}}} } \: \mathbf{unit} \: \mathbf{be} \:  \ottnt{a'}  \: \mathbf{in} \:  \ottnt{b}  $.

\item \Rref{PTS-LetPair}. Have: $  \emptyset   \vdash   \mathbf{let}_{ \ottnt{q_{{\mathrm{0}}}} } \: (  \ottmv{x} ^{ \ottnt{r} } ,  \ottmv{y}  ) \: \mathbf{be} \:  \ottnt{a}  \: \mathbf{in} \:  \ottnt{b}   :^{ \ottnt{q} }   \ottnt{B}  \{  \ottnt{a}  /  \ottmv{z}  \}  $ where $  \ottmv{z}  :   \Sigma  \ottmv{x}  :^{ \ottnt{r} } \!  \ottnt{A_{{\mathrm{1}}}}  .  \ottnt{A_{{\mathrm{2}}}}    \vdash_{0}  \ottnt{B}  :   \ottmv{s}  $ and $  \emptyset   \vdash  \ottnt{a}  :^{  \ottnt{q}  \cdot  \ottnt{q_{{\mathrm{0}}}}  }   \Sigma  \ottmv{x}  :^{ \ottnt{r} } \!  \ottnt{A_{{\mathrm{1}}}}  .  \ottnt{A_{{\mathrm{2}}}}  $ and $   \ottmv{x}  :^{    \ottnt{q}  \cdot  \ottnt{q_{{\mathrm{0}}}}    \cdot  \ottnt{r}  }  \ottnt{A_{{\mathrm{1}}}}   ,   \ottmv{y}  :^{  \ottnt{q}  \cdot  \ottnt{q_{{\mathrm{0}}}}  }  \ottnt{A_{{\mathrm{2}}}}    \vdash  \ottnt{b}  :^{ \ottnt{q} }   \ottnt{B}  \{   (   \ottmv{x}  ^{ \ottnt{r} } ,   \ottmv{y}   )   /  \ottmv{z}  \}  $.\\
Need to show: $\exists c,  \vdash   \mathbf{let}_{ \ottnt{q_{{\mathrm{0}}}} } \: (  \ottmv{x} ^{ \ottnt{r} } ,  \ottmv{y}  ) \: \mathbf{be} \:  \ottnt{a}  \: \mathbf{in} \:  \ottnt{b}   \leadsto  \ottnt{c} $.\\
By IH, $\ottnt{a}$ is either a value or $ \vdash  \ottnt{a}  \leadsto  \ottnt{a'} $.\\
If $\ottnt{a}$ is a value, then, by inversion, $\ottnt{a} =  (  \ottnt{a_{{\mathrm{1}}}} ^{ \ottnt{r} } ,  \ottnt{a_{{\mathrm{2}}}}  ) $. Therefore,
$ \vdash   \mathbf{let}_{ \ottnt{q_{{\mathrm{0}}}} } \: (  \ottmv{x} ^{ \ottnt{r} } ,  \ottmv{y}  ) \: \mathbf{be} \:  \ottnt{a}  \: \mathbf{in} \:  \ottnt{b}   \leadsto    \ottnt{b}  \{  \ottnt{a_{{\mathrm{1}}}}  /  \ottmv{x}  \}   \{  \ottnt{a_{{\mathrm{2}}}}  /  \ottmv{y}  \}  $.\\
Otherwise, $ \vdash   \mathbf{let}_{ \ottnt{q_{{\mathrm{0}}}} } \: (  \ottmv{x} ^{ \ottnt{r} } ,  \ottmv{y}  ) \: \mathbf{be} \:  \ottnt{a}  \: \mathbf{in} \:  \ottnt{b}   \leadsto   \mathbf{let}_{ \ottnt{q_{{\mathrm{0}}}} } \: (  \ottmv{x} ^{ \ottnt{r} } ,  \ottmv{y}  ) \: \mathbf{be} \:  \ottnt{a'}  \: \mathbf{in} \:  \ottnt{b}  $.

\item \Rref{ST-Case}. Have: $  \emptyset   \vdash   \mathbf{case}_{ \ottnt{q_{{\mathrm{0}}}} } \:  \ottnt{a}  \: \mathbf{of} \:  \ottmv{x_{{\mathrm{1}}}}  .  \ottnt{b_{{\mathrm{1}}}}  \: ; \:  \ottmv{x_{{\mathrm{2}}}}  .  \ottnt{b_{{\mathrm{2}}}}   :^{ \ottnt{q} }   \ottnt{B}  \{  \ottnt{a}  /  \ottmv{z}  \}  $ where $  \ottmv{z}  :   \ottnt{A_{{\mathrm{1}}}}  +  \ottnt{A_{{\mathrm{2}}}}    \vdash_{0}  \ottnt{B}  :   \ottmv{s}  $ and $  \emptyset   \vdash  \ottnt{a}  :^{  \ottnt{q}  \cdot  \ottnt{q_{{\mathrm{0}}}}  }   \ottnt{A_{{\mathrm{1}}}}  +  \ottnt{A_{{\mathrm{2}}}}  $ and $  \ottmv{x_{{\mathrm{1}}}}  :^{  \ottnt{q}  \cdot  \ottnt{q_{{\mathrm{0}}}}  }  \ottnt{A_{{\mathrm{1}}}}   \vdash  \ottnt{b_{{\mathrm{1}}}}  :^{ \ottnt{q} }   \ottnt{B}  \{   \mathbf{inj}_1 \:   \ottmv{x_{{\mathrm{1}}}}    /  \ottmv{z}  \}  $ and $  \ottmv{x_{{\mathrm{2}}}}  :^{  \ottnt{q}  \cdot  \ottnt{q_{{\mathrm{0}}}}  }  \ottnt{A_{{\mathrm{2}}}}   \vdash  \ottnt{b_{{\mathrm{2}}}}  :^{ \ottnt{q} }   \ottnt{B}  \{   \mathbf{inj}_2 \:   \ottmv{x_{{\mathrm{2}}}}    /  \ottmv{z}  \}  $.\\
Need to show: $\exists c,  \vdash   \mathbf{case}_{ \ottnt{q_{{\mathrm{0}}}} } \:  \ottnt{a}  \: \mathbf{of} \:  \ottmv{x_{{\mathrm{1}}}}  .  \ottnt{b_{{\mathrm{1}}}}  \: ; \:  \ottmv{x_{{\mathrm{2}}}}  .  \ottnt{b_{{\mathrm{2}}}}   \leadsto  \ottnt{c} $.\\
By IH, $\ottnt{a}$ is either a value or $ \vdash  \ottnt{a}  \leadsto  \ottnt{a'} $.\\
If $\ottnt{a}$ is a value, then $\ottnt{a} =  \mathbf{inj}_1 \:  \ottnt{a_{{\mathrm{1}}}} $ or $\ottnt{a} =  \mathbf{inj}_2 \:  \ottnt{a_{{\mathrm{2}}}} $. \\
Then, $ \vdash   \mathbf{case}_{ \ottnt{q_{{\mathrm{0}}}} } \:  \ottnt{a}  \: \mathbf{of} \:  \ottmv{x_{{\mathrm{1}}}}  .  \ottnt{b_{{\mathrm{1}}}}  \: ; \:  \ottmv{x_{{\mathrm{2}}}}  .  \ottnt{b_{{\mathrm{2}}}}   \leadsto   \ottnt{b_{{\mathrm{1}}}}  \{  \ottnt{a_{{\mathrm{1}}}}  /  \ottmv{x_{{\mathrm{1}}}}  \}  $ or $ \vdash   \mathbf{case}_{ \ottnt{q_{{\mathrm{0}}}} } \:  \ottnt{a}  \: \mathbf{of} \:  \ottmv{x_{{\mathrm{1}}}}  .  \ottnt{b_{{\mathrm{1}}}}  \: ; \:  \ottmv{x_{{\mathrm{2}}}}  .  \ottnt{b_{{\mathrm{2}}}}   \leadsto   \ottnt{b_{{\mathrm{2}}}}  \{  \ottnt{a_{{\mathrm{2}}}}  /  \ottmv{x_{{\mathrm{2}}}}  \}  $.\\
Otherwise, $ \vdash   \mathbf{case}_{ \ottnt{q_{{\mathrm{0}}}} } \:  \ottnt{a}  \: \mathbf{of} \:  \ottmv{x_{{\mathrm{1}}}}  .  \ottnt{b_{{\mathrm{1}}}}  \: ; \:  \ottmv{x_{{\mathrm{2}}}}  .  \ottnt{b_{{\mathrm{2}}}}   \leadsto   \mathbf{case}_{ \ottnt{q_{{\mathrm{0}}}} } \:  \ottnt{a'}  \: \mathbf{of} \:  \ottmv{x_{{\mathrm{1}}}}  .  \ottnt{b_{{\mathrm{1}}}}  \: ; \:  \ottmv{x_{{\mathrm{2}}}}  .  \ottnt{b_{{\mathrm{2}}}}  $.

\item \Rref{PTS-Weak,PTS-Conv,PTS-SubL,PTS-SubR}. Follows by IH.

\item \Rref{PTS-Var}. Does not apply since the context here is empty.

\item The terms typed by the remaining rules are values.

\end{itemize}

\end{proof}

%-----------------------------------------------------------------------------------------------------

%----------------------------------------------------------------------------------
\section{Heap Semantics for PTS Version of LDC} \label{heapPTS}
%----------------------------------------------------------------------------------

Recall that to show heap soundness for a dependent type system, we need to allow delayed substitution in types. Towards this end, we extend contexts with definitions, which mimic substitutions. The definitions are only used to derive type equalities and are orthogonal to linearity and dependency analyses. The typing rules for definitions and the conversion rule where definitions are used are shown in Figure \ref{def}. We add these extra rules to the type system of the calculus. 

\begin{figure}
\text{context}, $\Gamma ::=  \emptyset  \: | \:  \Gamma  ,   \ottmv{x}  :^{ \ottnt{q} }  \ottnt{A}   \: | \:  \Gamma  ,   \ottmv{x}  =  \ottnt{a}  :^{ \ottnt{q} }  \ottnt{A}  $
\drules[PTS]{$ \Gamma  \vdash  \ottnt{a}  :^{ \ottnt{q} }  \ottnt{A} $}{Extra Rules}{DefVar,DefWeak,DefConv}
\caption{Typing Rules With Definitions}
\label{def}
\end{figure}  

There are a few things to note:
\begin{itemize}
\item Definitions do not interact with the analyses, as we see in \rref{PTS-DefVar,PTS-DefWeak}.
\item The conversion \rref{PTS-DefConv} checks for equality of types after substituting the definitions in the types. Here, $ \ottnt{A}  \{  \Delta  \} $ denotes the type $\ottnt{A}$ with the definitions in $\Delta$ substituted in reverse order. With this rule, we can show: $  \ottmv{x}  =   \mathbf{Unit}   :^{  0  }   \ottmv{s}    \vdash   \lambda^{  1  }  \ottmv{y}  :   \ottmv{x}   .   \ottmv{y}    :^{  1  }   \Pi  \ottmv{y}  :^{  1  } \!   \mathbf{Unit}   .   \mathbf{Unit}   $.
\end{itemize}

The definitions allow us to communicate to the type system the substitutions that have been delayed by the heap. To enable this communication, we also need to update the compatibility relation, as shown in Figure \ref{defcompat}.
\begin{figure}
\drules[HeapCompat]{$ \ottnt{H}  \models  \Gamma $}{Updated Compatibility}{DefEmpty,DefCons}
\caption{Compatibility Relation With Definitions}
\label{defcompat}
\end{figure} 

These extensions do not alter the essential character of the underlying system. The following lemmas establish the correspondence between the underlying system and the extended system. To distinguish, let us denote the typing and the compatibility judgments of the underlying system as $ \Gamma  \vdash^{\mathbf{u} }  \ottnt{a}  :^{ \ottnt{q} }  \ottnt{A} $ and $ \ottnt{H}  \models^{\mathbf{u} }  \Gamma $ and those of the extended system as $ \Gamma  \vdash^{\mathbf{e} }  \ottnt{a}  :^{ \ottnt{q} }  \ottnt{A} $ and $ \ottnt{H}  \models^{\mathbf{e} }  \Gamma $ respectively. Also, for $ \ottnt{H}  \models^{\mathbf{u} }  \Gamma $, let $ \Gamma _{ \ottnt{H} } $ be the context $\Gamma$ with the variables defined according to $\ottnt{H}$. Then, the multi-substitution lemma below says that given a derivation in the extended system, substituting the definitions gives us a derivation in the underlying system. The elaboration lemma says that given a derivation in the underlying system, adding appropriate definitions to the context gives us a derivation in the extended system. % the context with definitions gives we can derive in the extended system with variables in the context defined according to a compatible heap.  
\begin{lemma}[Multi-Substitution]\label{MultiSubst}
If $ \ottnt{H}  \models^{\mathbf{e} }  \Gamma $ and $ \Gamma  \vdash^{\mathbf{e} }  \ottnt{a}  :^{ \ottnt{q} }  \ottnt{A} $, then $  \emptyset   \vdash^{\mathbf{u} }   \ottnt{a}  \{  \Delta  \}   :^{ \ottnt{q} }   \ottnt{A}  \{  \Delta  \}  $ where $ \Delta  =   \lfloor  \Gamma  \rfloor  $.
\end{lemma}

\begin{proof}
By induction on $ \ottnt{H}  \models^{\mathbf{e} }  \Gamma $.
\end{proof}

\begin{lemma}[Elaboration]\label{Elaboration}
If $ \ottnt{H}  \models^{\mathbf{u} }  \Gamma $ and $ \Gamma  \vdash^{\mathbf{u} }  \ottnt{a}  :^{ \ottnt{q} }  \ottnt{A} $, then $ \ottnt{H}  \models^{\mathbf{e} }   \Gamma _{ \ottnt{H} }  $ and $  \Gamma _{ \ottnt{H} }   \vdash^{\mathbf{e} }  \ottnt{a}  :^{ \ottnt{q} }  \ottnt{A} $.
\end{lemma}

\begin{proof}
By induction on $ \ottnt{H}  \models^{\mathbf{u} }  \Gamma $.
\end{proof}

Next, we prove the soundness theorem for the extended system.
%---------------------------------------------------------------------------------------------------

\begin{lemma} \label{heaphelperDep}
If $ \ottnt{H}  \models   \Gamma_{{\mathrm{1}}}  +  \Gamma_{{\mathrm{2}}}  $ and $ \Gamma_{{\mathrm{2}}}  \vdash  \ottnt{a}  :^{ \ottnt{q} }  \ottnt{A} $ and $q \neq 0$, then either $\ottnt{a}$ is a value or there exists $\ottnt{H'}, \Gamma'_{{\mathrm{2}}}, \ottnt{a'}$ such that:
\begin{itemize}
\item $ [  \ottnt{H}  ]  \ottnt{a}  \Longrightarrow^{ \ottnt{q} }_{ \ottnt{S} } [  \ottnt{H'}  ]  \ottnt{a'} $
\item $ \ottnt{H'}  \models   \Gamma_{{\mathrm{1}}}  +  \Gamma'_{{\mathrm{2}}}  $
\item $ \Gamma'_{{\mathrm{2}}}  \vdash  \ottnt{a'}  :^{ \ottnt{q} }  \ottnt{A} $
\end{itemize} 
\end{lemma}
(Note that $+$ is overloaded here: $ \Gamma_{{\mathrm{1}}}  +  \Gamma_{{\mathrm{2}}} $ denotes addition of contexts $\Gamma_{{\mathrm{1}}}$ and $\Gamma_{{\mathrm{2}}}$ after padding them as necessary.)

\begin{proof}
By induction on $ \Gamma_{{\mathrm{2}}}  \vdash  \ottnt{a}  :^{ \ottnt{q} }  \ottnt{A} $.

\begin{itemize}

\item \Rref{PTS-DefVar}. Have: $    0   \cdot  \Gamma_{{\mathrm{21}}}   ,   \ottmv{x}  =  \ottnt{a}  :^{ \ottnt{q_{{\mathrm{2}}}} }  \ottnt{A}    \vdash   \ottmv{x}   :^{ \ottnt{q_{{\mathrm{2}}}} }  \ottnt{A} $ where $ \Delta  \vdash_{0}  \ottnt{a}  :  \ottnt{A} $ and $ \Delta  =   \lfloor  \Gamma_{{\mathrm{21}}}  \rfloor  $.\\
Further, $ \ottnt{H}  \models    (   \Gamma_{{\mathrm{11}}}  ,   \ottmv{x}  =  \ottnt{a}  :^{ \ottnt{q_{{\mathrm{1}}}} }  \ottnt{A}    )   +   (     0   \cdot  \Gamma_{{\mathrm{21}}}   ,   \ottmv{x}  =  \ottnt{a}  :^{ \ottnt{q_{{\mathrm{2}}}} }  \ottnt{A}    )   $.\\
By inversion, $\exists \ottnt{H_{{\mathrm{1}}}}, \Gamma_{{\mathrm{0}}}$ such that $H =  \ottnt{H_{{\mathrm{1}}}}  ,   \ottmv{x}  \overset{  \ottnt{q_{{\mathrm{1}}}}  +  \ottnt{q_{{\mathrm{2}}}}  }{\mapsto}  \ottnt{a}  $ and $ \Gamma_{{\mathrm{0}}}  \vdash  \ottnt{a}  :^{  \ottnt{q_{{\mathrm{1}}}}  +  \ottnt{q_{{\mathrm{2}}}}  }  \ottnt{A} $.\\
By lemma \ref{BLDSplitP}, $\exists \Gamma_{{\mathrm{01}}}, \Gamma_{{\mathrm{02}}}$ such that $ \Gamma_{{\mathrm{01}}}  \vdash  \ottnt{a}  :^{ \ottnt{q_{{\mathrm{1}}}} }  \ottnt{A} $ and $ \Gamma_{{\mathrm{02}}}  \vdash  \ottnt{a}  :^{ \ottnt{q_{{\mathrm{2}}}} }  \ottnt{A} $ and $ \Gamma_{{\mathrm{0}}}  =   \Gamma_{{\mathrm{01}}}  +  \Gamma_{{\mathrm{02}}}  $.\\
Then, we have,
\begin{itemize}
\item $ [   \ottnt{H_{{\mathrm{1}}}}  ,   \ottmv{x}  \overset{  \ottnt{q_{{\mathrm{1}}}}  +  \ottnt{q_{{\mathrm{2}}}}  }{\mapsto}  \ottnt{a}    ]   \ottmv{x}   \Longrightarrow^{ \ottnt{q_{{\mathrm{2}}}} }_{ \ottnt{S} } [   \ottnt{H_{{\mathrm{1}}}}  ,   \ottmv{x}  \overset{ \ottnt{q_{{\mathrm{1}}}} }{\mapsto}  \ottnt{a}    ]  \ottnt{a} $.
\item $  \ottnt{H_{{\mathrm{1}}}}  ,   \ottmv{x}  \overset{ \ottnt{q_{{\mathrm{1}}}} }{\mapsto}  \ottnt{a}    \models    (   \Gamma_{{\mathrm{11}}}  ,   \ottmv{x}  =  \ottnt{a}  :^{ \ottnt{q_{{\mathrm{1}}}} }  \ottnt{A}    )   +   (   \Gamma_{{\mathrm{02}}}  ,   \ottmv{x}  =  \ottnt{a}  :^{  0  }  \ottnt{A}    )   $.
\item $  \Gamma_{{\mathrm{02}}}  ,   \ottmv{x}  =  \ottnt{a}  :^{  0  }  \ottnt{A}    \vdash  \ottnt{a}  :^{ \ottnt{q_{{\mathrm{2}}}} }  \ottnt{A} $.
\end{itemize}

\item \Rref{PTS-DefWeak}. Have: $  \Gamma_{{\mathrm{21}}}  ,   \ottmv{y}  =  \ottnt{b}  :^{  0  }  \ottnt{B}    \vdash  \ottnt{a}  :^{ \ottnt{q} }  \ottnt{A} $ where $ \Gamma_{{\mathrm{21}}}  \vdash  \ottnt{a}  :^{ \ottnt{q} }  \ottnt{A} $ and $ \Delta  \vdash_{0}  \ottnt{b}  :  \ottnt{B} $ and $ \Delta  =   \lfloor  \Gamma_{{\mathrm{21}}}  \rfloor  $.\\
Further, $ \ottnt{H}  \models    (   \Gamma_{{\mathrm{11}}}  ,   \ottmv{y}  =  \ottnt{b}  :^{ \ottnt{q_{{\mathrm{1}}}} }  \ottnt{B}    )   +   (   \Gamma_{{\mathrm{21}}}  ,   \ottmv{y}  =  \ottnt{b}  :^{  0  }  \ottnt{B}    )   $.\\
By inversion, $\exists \ottnt{H_{{\mathrm{1}}}}, \Gamma_{{\mathrm{0}}}$ such that $H =  \ottnt{H_{{\mathrm{1}}}}  ,   \ottmv{y}  \overset{ \ottnt{q_{{\mathrm{1}}}} }{\mapsto}  \ottnt{b}  $ and $ \Gamma_{{\mathrm{0}}}  \vdash  \ottnt{b}  :^{ \ottnt{q_{{\mathrm{1}}}} }  \ottnt{B} $ and $ \ottnt{H_{{\mathrm{1}}}}  \models     \Gamma_{{\mathrm{11}}}  +  \Gamma_{{\mathrm{21}}}    +  \Gamma_{{\mathrm{0}}}  $.\\
By IH, $\exists \ottnt{H'_{{\mathrm{1}}}} , \ottnt{H'_{{\mathrm{2}}}} , \Gamma'_{{\mathrm{21}}} , \Gamma'_{{\mathrm{22}}}, \ottnt{a'}$ such that
\begin{itemize}
\item $ [  \ottnt{H_{{\mathrm{1}}}}  ]  \ottnt{a}  \Longrightarrow^{ \ottnt{q} }_{  \ottnt{S}  \, \cup \,   \{  \ottmv{y}  \}   } [   \ottnt{H'_{{\mathrm{1}}}}  ,  \ottnt{H'_{{\mathrm{2}}}}   ]  \ottnt{a'} $ where $\lvert \ottnt{H'_{{\mathrm{1}}}} \rvert = \lvert \ottnt{H_{{\mathrm{1}}}} \rvert$
\item $  \ottnt{H'_{{\mathrm{1}}}}  ,  \ottnt{H'_{{\mathrm{2}}}}   \models      \Gamma_{{\mathrm{11}}}  +  \Gamma_{{\mathrm{0}}}   +  \Gamma'_{{\mathrm{21}}}    ,  \Gamma'_{{\mathrm{22}}}  $ 
\item $  \Gamma'_{{\mathrm{21}}}  ,  \Gamma'_{{\mathrm{22}}}   \vdash  \ottnt{a'}  :^{ \ottnt{q} }  \ottnt{A} $
\end{itemize}
Then, we have,
\begin{itemize}
\item $ [   \ottnt{H_{{\mathrm{1}}}}  ,   \ottmv{y}  \overset{ \ottnt{q_{{\mathrm{1}}}} }{\mapsto}  \ottnt{b}    ]  \ottnt{a}  \Longrightarrow^{ \ottnt{q} }_{ \ottnt{S} } [   \ottnt{H'_{{\mathrm{1}}}}  ,     \ottmv{y}  \overset{ \ottnt{q_{{\mathrm{1}}}} }{\mapsto}  \ottnt{b}   ,  \ottnt{H'_{{\mathrm{2}}}}     ]  \ottnt{a'} $
\item $  \ottnt{H'_{{\mathrm{1}}}}  ,     \ottmv{y}  \overset{ \ottnt{q_{{\mathrm{1}}}} }{\mapsto}  \ottnt{b}   ,  \ottnt{H'_{{\mathrm{2}}}}     \models    (    (   \Gamma_{{\mathrm{11}}}  ,   \ottmv{y}  =  \ottnt{b}  :^{ \ottnt{q_{{\mathrm{1}}}} }  \ottnt{B}    )   +   (   \Gamma'_{{\mathrm{21}}}  ,   \ottmv{y}  =  \ottnt{b}  :^{  0  }  \ottnt{B}    )    )   ,  \Gamma'_{{\mathrm{22}}}  $
\item $  \Gamma'_{{\mathrm{21}}}  ,     \ottmv{y}  =  \ottnt{b}  :^{  0  }  \ottnt{B}   ,  \Gamma'_{{\mathrm{22}}}     \vdash  \ottnt{a'}  :^{ \ottnt{q} }  \ottnt{A} $
\end{itemize}

\item \Rref{PTS-Var,PTS-Weak}. Does not apply since whenever $ \ottnt{H}  \models   \Gamma_{{\mathrm{1}}}  +  \Gamma_{{\mathrm{2}}}  $ holds, every assumption in $\Gamma_{{\mathrm{2}}}$ is a definition.

\item \Rref{PTS-App}. Have: $  \Gamma_{{\mathrm{21}}}  +  \Gamma_{{\mathrm{22}}}   \vdash   \ottnt{b}  \:  \ottnt{a} ^{ \ottnt{r} }   :^{ \ottnt{q} }   \ottnt{B}  \{  \ottnt{a}  /  \ottmv{x}  \}  $ where $ \Gamma_{{\mathrm{21}}}  \vdash  \ottnt{b}  :^{ \ottnt{q} }   \Pi  \ottmv{x}  :^{ \ottnt{r} } \!  \ottnt{A}  .  \ottnt{B}  $ and $ \Gamma_{{\mathrm{22}}}  \vdash  \ottnt{a}  :^{  \ottnt{q}  \cdot  \ottnt{r}  }  \ottnt{A} $.\\
Further, $ \ottnt{H}  \models   \Gamma_{{\mathrm{1}}}  +   (   \Gamma_{{\mathrm{21}}}  +  \Gamma_{{\mathrm{22}}}   )   $.\\
If $\ottnt{b}$ steps, then this case follows by IH.\\
Otherwise, $\ottnt{b}$ is a value. By inversion, $b =  \lambda^{ \ottnt{r} }  \ottmv{x}  :  \ottnt{A'}  .  \ottnt{b'} $.\\
Further, $\exists \ottnt{A_{{\mathrm{1}}}} , \ottnt{B_{{\mathrm{1}}}}$ such that $  \Gamma_{{\mathrm{21}}}  ,   \ottmv{x}  :^{  \ottnt{q}  \cdot  \ottnt{r}  }  \ottnt{A_{{\mathrm{1}}}}    \vdash  \ottnt{b'}  :^{ \ottnt{q} }  \ottnt{B_{{\mathrm{1}}}} $ and $   \Pi  \ottmv{x}  :^{ \ottnt{r} } \!  \ottnt{A_{{\mathrm{1}}}}  .  \ottnt{B_{{\mathrm{1}}}}   \{  \Delta_{{\mathrm{2}}}  \}   =_{\beta}    \Pi  \ottmv{x}  :^{ \ottnt{r} } \!  \ottnt{A}  .  \ottnt{B}   \{  \Delta_{{\mathrm{2}}}  \}  $ where $ \Delta_{{\mathrm{2}}}  =   \lfloor  \Gamma_{{\mathrm{21}}}  \rfloor  $.\\
Then, by \rref{PTS-DefConv}, $ \Gamma_{{\mathrm{22}}}  \vdash  \ottnt{a}  :^{  \ottnt{q}  \cdot  \ottnt{r}  }  \ottnt{A_{{\mathrm{1}}}} $. Further, $  \Gamma_{{\mathrm{21}}}  ,   \ottmv{x}  =  \ottnt{a}  :^{  \ottnt{q}  \cdot  \ottnt{r}  }  \ottnt{A_{{\mathrm{1}}}}    \vdash  \ottnt{b'}  :^{ \ottnt{q} }  \ottnt{B_{{\mathrm{1}}}} $.\\
Now,
\begin{itemize}
\item $ [  \ottnt{H}  ]    (   \lambda^{ \ottnt{r} }  \ottmv{x}  :  \ottnt{A'}  .  \ottnt{b'}   )   \:  \ottnt{a} ^{ \ottnt{r} }   \Longrightarrow^{ \ottnt{q} }_{ \ottnt{S} } [   \ottnt{H}  ,   \ottmv{x}  \overset{  \ottnt{q}  \cdot  \ottnt{r}  }{\mapsto}  \ottnt{a}    ]  \ottnt{b'} $
\item $  \ottnt{H}  ,   \ottmv{x}  \overset{  \ottnt{q}  \cdot  \ottnt{r}  }{\mapsto}  \ottnt{a}    \models    (   \Gamma_{{\mathrm{1}}}  +  \Gamma_{{\mathrm{21}}}   )   ,   \ottmv{x}  =  \ottnt{a}  :^{  \ottnt{q}  \cdot  \ottnt{r}  }  \ottnt{A_{{\mathrm{1}}}}   $ 
\item $  \Gamma_{{\mathrm{21}}}  ,   \ottmv{x}  =  \ottnt{a}  :^{  \ottnt{q}  \cdot  \ottnt{r}  }  \ottnt{A_{{\mathrm{1}}}}    \vdash  \ottnt{b'}  :^{ \ottnt{q} }   \ottnt{B}  \{  \ottnt{a}  /  \ottmv{x}  \}  $ by \rref{PTS-DefConv}. 
\end{itemize}

\item \Rref{PTS-DefConv}. Have: $ \Gamma_{{\mathrm{2}}}  \vdash  \ottnt{a}  :^{ \ottnt{q} }  \ottnt{B} $ where $ \Gamma_{{\mathrm{2}}}  \vdash  \ottnt{a}  :^{ \ottnt{q} }  \ottnt{A} $ and $  \ottnt{A}  \{  \Delta_{{\mathrm{2}}}  \}   =_{\beta}   \ottnt{B}  \{  \Delta_{{\mathrm{2}}}  \}  $ and $ \Delta_{{\mathrm{2}}}  =   \lfloor  \Gamma_{{\mathrm{2}}}  \rfloor  $.\\
Further, $ \ottnt{H}  \models   \Gamma_{{\mathrm{1}}}  +  \Gamma_{{\mathrm{2}}}  $.\\
By IH, $ [  \ottnt{H}  ]  \ottnt{a}  \Longrightarrow^{ \ottnt{q} }_{ \ottnt{S} } [  \ottnt{H'}  ]  \ottnt{a'} $ and $ \ottnt{H'}  \models   \Gamma_{{\mathrm{1}}}  +  \Gamma'_{{\mathrm{2}}}  $ and $ \Gamma'_{{\mathrm{2}}}  \vdash  \ottnt{a'}  :^{ \ottnt{q} }  \ottnt{A} $.\\
Let $ \Delta'_{{\mathrm{2}}}  =   \lfloor  \Gamma'_{{\mathrm{2}}}  \rfloor  $. Then, $ \ottnt{A}  \{  \Delta'_{{\mathrm{2}}}  \}  =   \ottnt{A}  \{  \Delta_{{\mathrm{2}}}  \}   =_{\beta}   \ottnt{B}  \{  \Delta_{{\mathrm{2}}}  \}   =  \ottnt{B}  \{  \Delta'_{{\mathrm{2}}}  \} $ ($\because  \textit{fv} \:  \ottnt{A}  ,  \textit{fv} \:  \ottnt{B}  \in \text{dom} \Delta_{{\mathrm{2}}} \subseteq \text{dom} \Delta'_{{\mathrm{2}}}$).\\
This case, then, follows by \rref{PTS-DefConv}.

\item \Rref{PTS-LetPair}. Have: $  \Gamma_{{\mathrm{21}}}  +  \Gamma_{{\mathrm{22}}}   \vdash   \mathbf{let}_{ \ottnt{q_{{\mathrm{0}}}} } \: (  \ottmv{x} ^{ \ottnt{r} } ,  \ottmv{y}  ) \: \mathbf{be} \:  \ottnt{a}  \: \mathbf{in} \:  \ottnt{b}   :^{ \ottnt{q} }   \ottnt{B}  \{  \ottnt{a}  /  \ottmv{z}  \}  $ where $  \Delta_{{\mathrm{2}}}  ,   \ottmv{z}  :   \Sigma  \ottmv{x}  :^{ \ottnt{r} } \!  \ottnt{A_{{\mathrm{1}}}}  .  \ottnt{A_{{\mathrm{2}}}}     \vdash_{0}  \ottnt{B}  :   \ottmv{s}  $ and $ \Gamma_{{\mathrm{21}}}  \vdash  \ottnt{a}  :^{  \ottnt{q}  \cdot  \ottnt{q_{{\mathrm{0}}}}  }   \Sigma  \ottmv{x}  :^{ \ottnt{r} } \!  \ottnt{A_{{\mathrm{1}}}}  .  \ottnt{A_{{\mathrm{2}}}}  $ and $  \Gamma_{{\mathrm{22}}}  ,     \ottmv{x}  :^{  \ottnt{q}  \cdot    \ottnt{q_{{\mathrm{0}}}}  \cdot  \ottnt{r}    }  \ottnt{A_{{\mathrm{1}}}}   ,   \ottmv{y}  :^{  \ottnt{q}  \cdot  \ottnt{q_{{\mathrm{0}}}}  }  \ottnt{A_{{\mathrm{2}}}}      \vdash  \ottnt{b}  :^{ \ottnt{q} }   \ottnt{B}  \{   (   \ottmv{x}  ^{ \ottnt{r} } ,   \ottmv{y}   )   /  \ottmv{z}  \}  $ and $ \Delta_{{\mathrm{2}}}  =   \lfloor  \Gamma_{{\mathrm{21}}}  \rfloor  $.\\
Further, $ \ottnt{H}  \models   \Gamma_{{\mathrm{1}}}  +   (   \Gamma_{{\mathrm{21}}}  +  \Gamma_{{\mathrm{22}}}   )   $.\\
If $\ottnt{a}$ steps, then $ [  \ottnt{H}  ]  \ottnt{a}  \Longrightarrow^{  \ottnt{q}  \cdot  \ottnt{q_{{\mathrm{0}}}}  }_{  \ottnt{S}  \, \cup \,   \textit{fv} \:  \ottnt{b}   } [  \ottnt{H'}  ]  \ottnt{a'} $ and $ \ottnt{H'}  \models   \Gamma_{{\mathrm{1}}}  +   (   \Gamma'_{{\mathrm{21}}}  +  \Gamma_{{\mathrm{22}}}   )   $ and $ \Gamma'_{{\mathrm{21}}}  \vdash  \ottnt{a'}  :^{  \ottnt{q}  \cdot  \ottnt{q_{{\mathrm{0}}}}  }   \Sigma  \ottmv{x}  :^{ \ottnt{r} } \!  \ottnt{A_{{\mathrm{1}}}}  .  \ottnt{A_{{\mathrm{2}}}}  $.\\
Then, 
\begin{itemize}
\item $ [  \ottnt{H}  ]   \mathbf{let}_{ \ottnt{q_{{\mathrm{0}}}} } \: (  \ottmv{x} ^{ \ottnt{r} } ,  \ottmv{y}  ) \: \mathbf{be} \:  \ottnt{a}  \: \mathbf{in} \:  \ottnt{b}   \Longrightarrow^{  \ottnt{q}  \cdot  \ottnt{q_{{\mathrm{0}}}}  }_{ \ottnt{S} } [  \ottnt{H'}  ]   \mathbf{let}_{ \ottnt{q_{{\mathrm{0}}}} } \: (  \ottmv{x} ^{ \ottnt{r} } ,  \ottmv{y}  ) \: \mathbf{be} \:  \ottnt{a'}  \: \mathbf{in} \:  \ottnt{b}  $.\\
By \rref{HeapStep-Discard}, $ [  \ottnt{H}  ]   \mathbf{let}_{ \ottnt{q_{{\mathrm{0}}}} } \: (  \ottmv{x} ^{ \ottnt{r} } ,  \ottmv{y}  ) \: \mathbf{be} \:  \ottnt{a}  \: \mathbf{in} \:  \ottnt{b}   \Longrightarrow^{ \ottnt{q} }_{ \ottnt{S} } [  \ottnt{H'}  ]   \mathbf{let}_{ \ottnt{q_{{\mathrm{0}}}} } \: (  \ottmv{x} ^{ \ottnt{r} } ,  \ottmv{y}  ) \: \mathbf{be} \:  \ottnt{a'}  \: \mathbf{in} \:  \ottnt{b}  $. ($\because  \ottnt{q_{{\mathrm{0}}}}  <:   1  $)
\item $ \ottnt{H'}  \models   \Gamma_{{\mathrm{1}}}  +   (   \Gamma'_{{\mathrm{21}}}  +  \Gamma_{{\mathrm{22}}}   )   $
\item By \rref{ST-LetPair}, $  \Gamma'_{{\mathrm{21}}}  +  \Gamma_{{\mathrm{22}}}   \vdash   \mathbf{let}_{ \ottnt{q_{{\mathrm{0}}}} } \: (  \ottmv{x} ^{ \ottnt{r} } ,  \ottmv{y}  ) \: \mathbf{be} \:  \ottnt{a'}  \: \mathbf{in} \:  \ottnt{b}   : \:   \ottnt{B}  \{  \ottnt{a'}  /  \ottmv{z}  \}  $.\\
Now, by lemma \ref{HeapSim}, $  \ottnt{a'}  \{  \Delta'_{{\mathrm{2}}}  \}   =_{\beta}   \ottnt{a}  \{  \Delta_{{\mathrm{2}}}  \}  $ where $ \Delta'_{{\mathrm{2}}}  =   \lfloor  \Gamma'_{{\mathrm{21}}}  \rfloor  $.\\
Then, $  \ottnt{B}  \{  \ottnt{a}  /  \ottmv{z}  \}   \{  \Delta_{{\mathrm{2}}}  \}  =   \ottnt{B}  \{  \Delta_{{\mathrm{2}}}  \}   \{   \ottnt{a}  \{  \Delta_{{\mathrm{2}}}  \}   /  \ottmv{z}  \}  =    \ottnt{B}  \{  \Delta'_{{\mathrm{2}}}  \}   \{   \ottnt{a}  \{  \Delta_{{\mathrm{2}}}  \}   /  \ottmv{z}  \}   =_{\beta}    \ottnt{B}  \{  \Delta'_{{\mathrm{2}}}  \}   \{   \ottnt{a'}  \{  \Delta'_{{\mathrm{2}}}  \}   /  \ottmv{z}  \}   =   \ottnt{B}  \{  \ottnt{a'}  /  \ottmv{z}  \}   \{  \Delta'_{{\mathrm{2}}}  \} $.\\
By \rref{PTS-DefConv}, $  \Gamma'_{{\mathrm{21}}}  +  \Gamma_{{\mathrm{22}}}   \vdash   \mathbf{let}_{ \ottnt{q_{{\mathrm{0}}}} } \: (  \ottmv{x} ^{ \ottnt{r} } ,  \ottmv{y}  ) \: \mathbf{be} \:  \ottnt{a'}  \: \mathbf{in} \:  \ottnt{b}   : \:   \ottnt{B}  \{  \ottnt{a}  /  \ottmv{z}  \}  $.
\end{itemize}
 
Otherwise, $\ottnt{a}$ is a value. By inversion, $a =  (  \ottnt{a_{{\mathrm{1}}}} ^{ \ottnt{r} } ,  \ottnt{a_{{\mathrm{2}}}}  ) $.\\
Further, $\exists \Gamma_{{\mathrm{211}}}, \Gamma_{{\mathrm{212}}}$ such that $ \Gamma_{{\mathrm{211}}}  \vdash  \ottnt{a_{{\mathrm{1}}}}  :^{  \ottnt{q}  \cdot    \ottnt{q_{{\mathrm{0}}}}  \cdot  \ottnt{r}    }  \ottnt{A_{{\mathrm{1}}}} $ and $ \Gamma_{{\mathrm{212}}}  \vdash  \ottnt{a_{{\mathrm{2}}}}  :^{  \ottnt{q}  \cdot  \ottnt{q_{{\mathrm{0}}}}  }   \ottnt{A_{{\mathrm{2}}}}  \{  \ottnt{a_{{\mathrm{1}}}}  /  \ottmv{x}  \}  $ and $\Gamma_{{\mathrm{21}}} =  \Gamma_{{\mathrm{211}}}  +  \Gamma_{{\mathrm{212}}} $.\\
Then,
\begin{itemize}
\item $ [  \ottnt{H}  ]   \mathbf{let}_{ \ottnt{q_{{\mathrm{0}}}} } \: (  \ottmv{x} ^{ \ottnt{r} } ,  \ottmv{y}  ) \: \mathbf{be} \:   (  \ottnt{a_{{\mathrm{1}}}} ^{ \ottnt{r} } ,  \ottnt{a_{{\mathrm{2}}}}  )   \: \mathbf{in} \:  \ottnt{b}   \Longrightarrow^{ \ottnt{q} }_{ \ottnt{S} } [     \ottnt{H}  ,   \ottmv{x}  \overset{  \ottnt{q}  \cdot    \ottnt{q_{{\mathrm{0}}}}  \cdot  \ottnt{r}    }{\mapsto}  \ottnt{a_{{\mathrm{1}}}}     ,   \ottmv{y}  \overset{  \ottnt{q}  \cdot  \ottnt{q_{{\mathrm{0}}}}  }{\mapsto}  \ottnt{a_{{\mathrm{2}}}}    ]  \ottnt{b} $ (assuming $x, \ottmv{y}  \: \textit{fresh} $)
\item $    \ottnt{H}  ,   \ottmv{x}  \overset{  \ottnt{q}  \cdot    \ottnt{q_{{\mathrm{0}}}}  \cdot  \ottnt{r}    }{\mapsto}  \ottnt{a_{{\mathrm{1}}}}     ,   \ottmv{y}  \overset{  \ottnt{q}  \cdot  \ottnt{q_{{\mathrm{0}}}}  }{\mapsto}  \ottnt{a_{{\mathrm{2}}}}    \models    (   \Gamma_{{\mathrm{1}}}  +  \Gamma_{{\mathrm{22}}}   )   ,     \ottmv{x}  =  \ottnt{a_{{\mathrm{1}}}}  :^{  \ottnt{q}  \cdot    \ottnt{q_{{\mathrm{0}}}}  \cdot  \ottnt{r}    }  \ottnt{A_{{\mathrm{1}}}}   ,   \ottmv{y}  =  \ottnt{a_{{\mathrm{2}}}}  :^{  \ottnt{q}  \cdot  \ottnt{q_{{\mathrm{0}}}}  }  \ottnt{A_{{\mathrm{2}}}}     $
\item $  \Gamma_{{\mathrm{22}}}  ,     \ottmv{x}  =  \ottnt{a_{{\mathrm{1}}}}  :^{  \ottnt{q}  \cdot    \ottnt{q_{{\mathrm{0}}}}  \cdot  \ottnt{r}    }  \ottnt{A_{{\mathrm{1}}}}   ,   \ottmv{y}  =  \ottnt{a_{{\mathrm{2}}}}  :^{  \ottnt{q}  \cdot  \ottnt{q_{{\mathrm{0}}}}  }  \ottnt{A_{{\mathrm{2}}}}      \vdash  \ottnt{b}  :^{ \ottnt{q} }   \ottnt{B}  \{   (  \ottnt{a_{{\mathrm{1}}}} ^{ \ottnt{r} } ,  \ottnt{a_{{\mathrm{2}}}}  )   /  \ottmv{z}  \}  $.
\end{itemize}

\item \Rref{PTS-LetUnit}. Similar to \rref{PTS-LetPair}.

\item \Rref{PTS-Case}. Have: $  \Gamma_{{\mathrm{21}}}  +  \Gamma_{{\mathrm{22}}}   \vdash   \mathbf{case}_{ \ottnt{q_{{\mathrm{0}}}} } \:  \ottnt{a}  \: \mathbf{of} \:  \ottmv{x_{{\mathrm{1}}}}  .  \ottnt{b_{{\mathrm{1}}}}  \: ; \:  \ottmv{x_{{\mathrm{2}}}}  .  \ottnt{b_{{\mathrm{2}}}}   :^{ \ottnt{q} }   \ottnt{B}  \{  \ottnt{a}  /  \ottmv{z}  \}  $ where $  \Delta_{{\mathrm{2}}}  ,   \ottmv{z}  :   \ottnt{A_{{\mathrm{1}}}}  +  \ottnt{A_{{\mathrm{2}}}}     \vdash_{0}  \ottnt{B}  :   \ottmv{s}  $ and $ \Gamma_{{\mathrm{21}}}  \vdash  \ottnt{a}  :^{  \ottnt{q}  \cdot  \ottnt{q_{{\mathrm{0}}}}  }   \ottnt{A_{{\mathrm{1}}}}  +  \ottnt{A_{{\mathrm{2}}}}  $ and $  \Gamma_{{\mathrm{22}}}  ,   \ottmv{x_{{\mathrm{1}}}}  :^{  \ottnt{q}  \cdot  \ottnt{q_{{\mathrm{0}}}}  }  \ottnt{A_{{\mathrm{1}}}}    \vdash  \ottnt{b_{{\mathrm{1}}}}  :^{ \ottnt{q} }   \ottnt{B}  \{   \mathbf{inj}_1 \:   \ottmv{x_{{\mathrm{1}}}}    /  \ottmv{z}  \}  $ and $  \Gamma_{{\mathrm{22}}}  ,   \ottmv{x_{{\mathrm{2}}}}  :^{  \ottnt{q}  \cdot  \ottnt{q_{{\mathrm{0}}}}  }  \ottnt{A_{{\mathrm{2}}}}    \vdash  \ottnt{b_{{\mathrm{2}}}}  :^{ \ottnt{q} }   \ottnt{B}  \{   \mathbf{inj}_2 \:   \ottmv{x_{{\mathrm{2}}}}    /  \ottmv{z}  \}  $ and $ \Delta_{{\mathrm{2}}}  =   \lfloor  \Gamma_{{\mathrm{21}}}  \rfloor  $.\\
Further, $ \ottnt{H}  \models   \Gamma_{{\mathrm{1}}}  +   (   \Gamma_{{\mathrm{21}}}  +  \Gamma_{{\mathrm{22}}}   )   $.\\
If $\ottnt{a}$ steps, then this case follows by IH.\\
Otherwise, $\ottnt{a}$ is a value. By inversion, $\ottnt{a} =  \mathbf{inj}_1 \:  \ottnt{a_{{\mathrm{1}}}} $ or $\ottnt{a} =  \mathbf{inj}_2 \:  \ottnt{a_{{\mathrm{2}}}} $.\\
Say $\ottnt{a} =  \mathbf{inj}_1 \:  \ottnt{a_{{\mathrm{1}}}} $. Then, $ \Gamma_{{\mathrm{21}}}  \vdash  \ottnt{a_{{\mathrm{1}}}}  :^{  \ottnt{q}  \cdot  \ottnt{q_{{\mathrm{0}}}}  }  \ottnt{A_{{\mathrm{1}}}} $.\\
Now, 
\begin{itemize}
\item  $ [  \ottnt{H}  ]   \mathbf{case}_{ \ottnt{q_{{\mathrm{0}}}} } \:   (   \mathbf{inj}_1 \:  \ottnt{a_{{\mathrm{1}}}}   )   \: \mathbf{of} \:  \ottmv{x_{{\mathrm{1}}}}  .  \ottnt{b_{{\mathrm{1}}}}  \: ; \:  \ottmv{x_{{\mathrm{2}}}}  .  \ottnt{b_{{\mathrm{2}}}}   \Longrightarrow^{ \ottnt{q} }_{ \ottnt{S} } [   \ottnt{H}  ,   \ottmv{x_{{\mathrm{1}}}}  \overset{  \ottnt{q}  \cdot  \ottnt{q_{{\mathrm{0}}}}  }{\mapsto}  \ottnt{a_{{\mathrm{1}}}}    ]  \ottnt{b_{{\mathrm{1}}}} $ (assuming $ \ottmv{x_{{\mathrm{1}}}}  \: \textit{fresh} $)
\item $  \ottnt{H}  ,   \ottmv{x_{{\mathrm{1}}}}  \overset{  \ottnt{q}  \cdot  \ottnt{q_{{\mathrm{0}}}}  }{\mapsto}  \ottnt{a_{{\mathrm{1}}}}    \models    (   \Gamma_{{\mathrm{1}}}  +  \Gamma_{{\mathrm{22}}}   )   ,   \ottmv{x_{{\mathrm{1}}}}  =  \ottnt{a_{{\mathrm{1}}}}  :^{  \ottnt{q}  \cdot  \ottnt{q_{{\mathrm{0}}}}  }  \ottnt{A_{{\mathrm{1}}}}   $
\item $  \Gamma_{{\mathrm{22}}}  ,   \ottmv{x_{{\mathrm{1}}}}  =  \ottnt{a_{{\mathrm{1}}}}  :^{  \ottnt{q}  \cdot  \ottnt{q_{{\mathrm{0}}}}  }  \ottnt{A_{{\mathrm{1}}}}    \vdash  \ottnt{b_{{\mathrm{1}}}}  :^{ \ottnt{q} }   \ottnt{B}  \{   \mathbf{inj}_1 \:  \ottnt{a_{{\mathrm{1}}}}   /  \ottmv{z}  \}  $
\end{itemize}
The case when $\ottnt{a} =  \mathbf{inj}_2 \:  \ottnt{a_{{\mathrm{2}}}} $ is similar.
 
\end{itemize}

\end{proof}

%----------------------------------------------------------------------------------------------------

\begin{theorem}[Soundness (Theorem \ref{heapDep})]
If $ \ottnt{H}  \models  \Gamma $ and $ \Gamma  \vdash  \ottnt{a}  :^{ \ottnt{q} }  \ottnt{A} $ and $q \neq 0$, then either $\ottnt{a}$ is a value or there exists $\ottnt{H'}, \Gamma', \ottnt{a'}$ such that $ [  \ottnt{H}  ]  \ottnt{a}  \Longrightarrow^{ \ottnt{q} }_{ \ottnt{S} } [  \ottnt{H'}  ]  \ottnt{a'} $ and $ \ottnt{H'}  \models  \Gamma' $ and $ \Gamma'  \vdash  \ottnt{a'}  :^{ \ottnt{q} }  \ottnt{A} $.
\end{theorem} 

\begin{proof}
Use lemma \ref{heaphelperDep} with $\Gamma_{{\mathrm{1}}} :=   0   \cdot  \Gamma $ and $\Gamma_{{\mathrm{2}}} := \Gamma$.
\end{proof}

%----------------------------------------------------------------------------------------------------

\begin{corollary}[Corollary \ref{Null}]
In LDC($ \mathcal{Q}_{\mathbb{N} } $): If $  \emptyset   \vdash  \ottnt{f}  :^{  1  }   \Pi  \ottmv{x}  :^{  0  } \!   \ottmv{s}   .   \ottmv{x}   $ and $  \emptyset   \vdash_{0}  \ottnt{A}  :   \ottmv{s}  $, then $ \ottnt{f}  \:  \ottnt{A} ^{  0  } $ must diverge.
\end{corollary}

\begin{proof}
To see why, let us assume, towards contradiction, that $ \ottnt{f}  \:  \ottnt{A} ^{  0  } $ terminates. 

Let $\ottnt{A_{{\mathrm{1}}}} :=  \mathbf{Unit} $ and $\ottnt{A_{{\mathrm{2}}}} :=  \mathbf{Bool} $.

Then, for some $\ottnt{j}$, $\ottnt{H}$ and $\ottnt{b}$, we have, $ [   \emptyset   ]   \ottnt{f}  \:  \ottnt{A_{{\mathrm{1}}}} ^{  0  }   \Longrightarrow^{  1  }_{ \ottnt{j} } [  \ottnt{H}  ]    (   \lambda^{  0  }  \ottmv{x}  .  \ottnt{b}   )   \:  \ottnt{A_{{\mathrm{1}}}} ^{  0  }  $ and  $ [   \emptyset   ]   \ottnt{f}  \:  \ottnt{A_{{\mathrm{2}}}} ^{  0  }   \Longrightarrow^{  1  }_{ \ottnt{j} } [  \ottnt{H}  ]    (   \lambda^{  0  }  \ottmv{x}  .  \ottnt{b}   )   \:  \ottnt{A_{{\mathrm{2}}}} ^{  0  }  $.

Again, $ [  \ottnt{H}  ]    (   \lambda^{  0  }  \ottmv{x}  .  \ottnt{b}   )   \:  \ottnt{A_{{\mathrm{1}}}} ^{  0  }   \Longrightarrow^{  1  }_{ \ottnt{S} } [   \ottnt{H}  ,   \ottmv{x}  \overset{  0  }{\mapsto}  \ottnt{A_{{\mathrm{1}}}}    ]  \ottnt{b} $ and $ [  \ottnt{H}  ]    (   \lambda^{  0  }  \ottmv{x}  .  \ottnt{b}   )   \:  \ottnt{A_{{\mathrm{2}}}} ^{  0  }   \Longrightarrow^{  1  }_{ \ottnt{S} } [   \ottnt{H}  ,   \ottmv{x}  \overset{  0  }{\mapsto}  \ottnt{A_{{\mathrm{2}}}}    ]  \ottnt{b} $ (say, $ \ottmv{x}  \: \textit{fresh} $).

Now, if $ [   \ottnt{H}  ,   \ottmv{x}  \overset{  0  }{\mapsto}  \ottnt{A_{{\mathrm{1}}}}    ]  \ottnt{b}  \Longrightarrow^{  1  }_{ \ottnt{k} } [     \ottnt{H'}  ,   \ottmv{x}  \overset{  0  }{\mapsto}  \ottnt{A_{{\mathrm{1}}}}     ,  \ottnt{H''}   ]  \ottnt{b'} $, then $ [   \ottnt{H}  ,   \ottmv{x}  \overset{  0  }{\mapsto}  \ottnt{A_{{\mathrm{2}}}}    ]  \ottnt{b}  \Longrightarrow^{  1  }_{ \ottnt{k} } [     \ottnt{H'}  ,   \ottmv{x}  \overset{  0  }{\mapsto}  \ottnt{A_{{\mathrm{2}}}}     ,  \ottnt{H''}   ]  \ottnt{b'} $, for any $\ottnt{k}$ (By Lemma \ref{irrelP}).

As and when the reduction stops, $\ottnt{b'}$ is a value. By soundness, $\ottnt{b'}$ must be a value of  types $\ottnt{A_{{\mathrm{1}}}}$ and $\ottnt{A_{{\mathrm{2}}}}$ simultaneously, a contradiction. ($\because$ there is no value that has types both $ \mathbf{Unit} $ and $ \mathbf{Bool} $.)\\
So $ \ottnt{f}  \:  \ottnt{A_{{\mathrm{1}}}} ^{  0  } $, $ \ottnt{f}  \:  \ottnt{A_{{\mathrm{2}}}} ^{  0  } $, and therefore $ \ottnt{f}  \:  \ottnt{A} ^{  0  } $, must diverge for any $\ottnt{A}$. 

As a corollary, we see that in a strongly normalizing PTS, no such $\ottnt{f}$ exists.

Using the same argument as above, we can also show that if $  \emptyset   \vdash  \ottnt{f}  :^{  \mathbf{L}  }   \Pi  \ottmv{x}  :^{  \mathbf{H}  } \!   \ottmv{s}   .   \ottmv{x}   $ and $  \emptyset   \vdash  \ottnt{A}  :^{  \mathbf{H}  }   \ottmv{s}  $, then $ \ottnt{f}  \:  \ottnt{A} ^{  \mathbf{H}  } $ must diverge.
\end{proof}

%----------------------------------------------------------------------------------------------------

\begin{corollary}[Corollary \ref{Identity}]
In LDC($ \mathcal{Q}_{\mathbb{N} } $): In a strongly normalizing PTS, if $  \emptyset   \vdash  \ottnt{f}  :^{  1  }   \Pi  \ottmv{x}  :^{  0  } \!   \ottmv{s}   .   \Pi  \ottmv{y}  :^{  1  } \!   \ottmv{x}   .   \ottmv{x}    $ and $  \emptyset   \vdash_{0}  \ottnt{A}  :   \ottmv{s}  $ and $  \emptyset   \vdash  \ottnt{a}  :^{  1  }  \ottnt{A} $, then $    \ottnt{f}  \:  \ottnt{A} ^{  0  }    \:  \ottnt{a} ^{  1  }   =_{\beta}  \ottnt{a} $.
\end{corollary}

\begin{proof}
To see why, let's look at how $   \ottnt{f}  \:  \ottnt{A} ^{  0  }    \:  \ottnt{a} ^{  1  } $ reduces.

Let, $\ottnt{A_{{\mathrm{1}}}} :=  \mathbf{Unit} $ and $\ottnt{A_{{\mathrm{2}}}} :=  \mathbf{Bool} $. Also, let $\ottnt{a_{{\mathrm{1}}}} :=  \mathbf{unit} $ and $\ottnt{a_{{\mathrm{2}}}} :=  \mathbf{true} $.

Then, for some $\ottnt{j}$, $\ottnt{H_{{\mathrm{1}}}}$ and $\ottnt{b}$, we have, $ [   \emptyset   ]     \ottnt{f}  \:  \ottnt{A_{{\mathrm{1}}}} ^{  0  }    \:  \ottnt{a_{{\mathrm{1}}}} ^{  1  }   \Longrightarrow^{  1  }_{ \ottnt{j} } [  \ottnt{H_{{\mathrm{1}}}}  ]      (   \lambda^{  0  }  \ottmv{x}  .  \ottnt{b}   )   \:  \ottnt{A_{{\mathrm{1}}}} ^{  0  }    \:  \ottnt{a_{{\mathrm{1}}}} ^{  1  }  $ and  $ [   \emptyset   ]     \ottnt{f}  \:  \ottnt{A_{{\mathrm{2}}}} ^{  0  }    \:  \ottnt{a_{{\mathrm{2}}}} ^{  1  }   \Longrightarrow^{  1  }_{ \ottnt{j} } [  \ottnt{H_{{\mathrm{1}}}}  ]      (   \lambda^{  0  }  \ottmv{x}  .  \ottnt{b}   )   \:  \ottnt{A_{{\mathrm{2}}}} ^{  0  }    \:  \ottnt{a_{{\mathrm{2}}}} ^{  1  }  $.  

Now, $ [  \ottnt{H_{{\mathrm{1}}}}  ]      (   \lambda^{  0  }  \ottmv{x}  .  \ottnt{b}   )   \:  \ottnt{A_{{\mathrm{1}}}} ^{  0  }    \:  \ottnt{a_{{\mathrm{1}}}}   \Longrightarrow^{  1  }_{ \ottnt{S} } [   \ottnt{H_{{\mathrm{1}}}}  ,   \ottmv{x}  \overset{  0  }{\mapsto}  \ottnt{A_{{\mathrm{1}}}}    ]   \ottnt{b}  \:  \ottnt{a_{{\mathrm{1}}}} ^{  1  }  $ and $ [  \ottnt{H_{{\mathrm{1}}}}  ]      (   \lambda^{  0  }  \ottmv{x}  .  \ottnt{b}   )   \:  \ottnt{A_{{\mathrm{2}}}} ^{  0  }    \:  \ottnt{a_{{\mathrm{2}}}} ^{  1  }   \Longrightarrow^{  1  }_{ \ottnt{S} } [   \ottnt{H_{{\mathrm{1}}}}  ,   \ottmv{x}  \overset{  0  }{\mapsto}  \ottnt{A_{{\mathrm{2}}}}    ]   \ottnt{b}  \:  \ottnt{a_{{\mathrm{2}}}} ^{  1  }  $ (say $ \ottmv{x}  \: \textit{fresh} $).

Then, $ [   \ottnt{H_{{\mathrm{1}}}}  ,   \ottmv{x}  \overset{  0  }{\mapsto}  \ottnt{A_{{\mathrm{1}}}}    ]   \ottnt{b}  \:  \ottnt{a_{{\mathrm{1}}}} ^{  1  }   \Longrightarrow^{  1  }_{ \ottnt{k} } [     \ottnt{H'_{{\mathrm{1}}}}  ,   \ottmv{x}  \overset{  0  }{\mapsto}  \ottnt{A_{{\mathrm{1}}}}     ,  \ottnt{H_{{\mathrm{2}}}}   ]    (   \lambda^{  1  }  \ottmv{y}  .  \ottnt{c}   )   \:  \ottnt{a_{{\mathrm{1}}}} ^{  1  }  $ and $ [   \ottnt{H_{{\mathrm{1}}}}  ,   \ottmv{x}  \overset{  0  }{\mapsto}  \ottnt{A_{{\mathrm{2}}}}    ]   \ottnt{b}  \:  \ottnt{a_{{\mathrm{2}}}} ^{  1  }   \Longrightarrow^{  1  }_{ \ottnt{k} } [     \ottnt{H'_{{\mathrm{1}}}}  ,   \ottmv{x}  \overset{  0  }{\mapsto}  \ottnt{A_{{\mathrm{2}}}}     ,  \ottnt{H_{{\mathrm{2}}}}   ]    (   \lambda^{  1  }  \ottmv{y}  .  \ottnt{c}   )   \:  \ottnt{a_{{\mathrm{2}}}} ^{  1  }  $, for some $k$ (By Lemma \ref{irrelP}). 

Next, $ [     \ottnt{H'_{{\mathrm{1}}}}  ,   \ottmv{x}  \overset{  0  }{\mapsto}  \ottnt{A_{{\mathrm{1}}}}     ,  \ottnt{H_{{\mathrm{2}}}}   ]    (   \lambda^{  1  }  \ottmv{y}  .  \ottnt{c}   )   \:  \ottnt{a_{{\mathrm{1}}}} ^{  1  }   \Longrightarrow^{  1  }_{ \ottnt{S} } [     \ottnt{H'_{{\mathrm{1}}}}  ,   \ottmv{x}  \overset{  0  }{\mapsto}  \ottnt{A_{{\mathrm{1}}}}     ,    \ottnt{H_{{\mathrm{2}}}}  ,   \ottmv{y}  \overset{  1  }{\mapsto}  \ottnt{a_{{\mathrm{1}}}}      ]  \ottnt{c} $ and \\ $ [     \ottnt{H'_{{\mathrm{1}}}}  ,   \ottmv{x}  \overset{  0  }{\mapsto}  \ottnt{A_{{\mathrm{2}}}}     ,  \ottnt{H_{{\mathrm{2}}}}   ]    (   \lambda^{  1  }  \ottmv{y}  .  \ottnt{c}   )   \:  \ottnt{a_{{\mathrm{2}}}} ^{  1  }   \Longrightarrow^{  1  }_{ \ottnt{S} } [     \ottnt{H'_{{\mathrm{1}}}}  ,   \ottmv{x}  \overset{  0  }{\mapsto}  \ottnt{A_{{\mathrm{2}}}}     ,    \ottnt{H_{{\mathrm{2}}}}  ,   \ottmv{y}  \overset{  1  }{\mapsto}  \ottnt{a_{{\mathrm{2}}}}      ]  \ottnt{c} $ (say $ \ottmv{y}  \: \textit{fresh} $).

Since the reductions terminate, the value of $\ottmv{y}$ gets looked up. Till that look-up, the two reductions are indistinguishable from one another. 

Let $ [     \ottnt{H'_{{\mathrm{1}}}}  ,   \ottmv{x}  \overset{  0  }{\mapsto}  \ottnt{A_{{\mathrm{1}}}}     ,    \ottnt{H_{{\mathrm{2}}}}  ,   \ottmv{y}  \overset{  1  }{\mapsto}  \ottnt{a_{{\mathrm{1}}}}      ]  \ottnt{c}  \Longrightarrow^{  1  }_{ \ottnt{k'} } [     \ottnt{H''_{{\mathrm{1}}}}  ,   \ottmv{x}  \overset{  0  }{\mapsto}  \ottnt{A_{{\mathrm{1}}}}     ,      \ottnt{H'_{{\mathrm{2}}}}  ,   \ottmv{y}  \overset{  1  }{\mapsto}  \ottnt{a_{{\mathrm{1}}}}     ,  \ottnt{H_{{\mathrm{3}}}}     ]  \ottnt{c'} $ be the point at which $y$ is looked-up. Analyzing the stepping rules, $\ottnt{c'} = y$ or $\ottnt{c'}$ is a proper path headed by $y$ (where a path is a series of nested elimination forms headed by a variable). By soundness, $ \ottnt{c'}$ is well-typed. But there is no well-typed proper path that can be headed by both $\ottnt{a_{{\mathrm{1}}}}$ and $\ottnt{a_{{\mathrm{2}}}}$ because $  \emptyset   \vdash  \ottnt{a_{{\mathrm{1}}}}  : \:   \mathbf{Unit}  $ and $  \emptyset   \vdash  \ottnt{a_{{\mathrm{2}}}}  : \:   \mathbf{Bool}  $. Therefore, $\ottnt{c'} = y $. 

This means that $   \ottnt{f}  \:  \ottnt{A_{{\mathrm{1}}}} ^{  0  }    \:  \ottnt{a_{{\mathrm{1}}}} ^{  1  } $ reduces to $\ottnt{a_{{\mathrm{1}}}}$ and $   \ottnt{f}  \:  \ottnt{A_{{\mathrm{2}}}} ^{  0  }    \:  \ottnt{a_{{\mathrm{2}}}} ^{  1  } $ to $\ottnt{a_{{\mathrm{2}}}}$, and hence $   \ottnt{f}  \:  \ottnt{A} ^{  0  }    \:  \ottnt{a} ^{  1  } $ to $\ottnt{a}$, for any $\ottnt{A}$ and $\ottnt{a}$. Therefore, $    \ottnt{f}  \:  \ottnt{A} ^{  0  }    \:  \ottnt{a} ^{  1  }   =_{\beta}  \ottnt{a} $.

\end{proof}

%----------------------------------------------------------------------------------------------------

%---------------------------------------------------------------------------------------------------
\section{Dependency Analysis in PTS Version of LDC}
%---------------------------------------------------------------------------------------------------

\begin{lemma}[Multiplication] \label{DDMultP}
If $ \Gamma  \vdash  \ottnt{a}  :^{ \ell }  \ottnt{A} $, then $  \ottnt{r_{{\mathrm{0}}}}  \sqcup  \Gamma   \vdash  \ottnt{a}  :^{  \ottnt{r_{{\mathrm{0}}}}  \: \sqcup \:  \ell  }  \ottnt{A} $.
\end{lemma}

\begin{proof}
By induction on $ \Gamma  \vdash  \ottnt{a}  :^{ \ell }  \ottnt{A} $. Follow the proof of lemma \ref{DSMultP}.
\end{proof}

%---------------------------------------------------------------------------------------------------

\begin{lemma}[Splitting] \label{DDSplitP}
If $ \Gamma  \vdash  \ottnt{a}  :^{  \ell_{{\mathrm{1}}}  \: \sqcap \:  \ell_{{\mathrm{2}}}  }  \ottnt{A} $, then there exists $\Gamma_{{\mathrm{1}}}$ and $\Gamma_{{\mathrm{2}}}$ such that $ \Gamma_{{\mathrm{1}}}  \vdash  \ottnt{a}  :^{ \ell_{{\mathrm{1}}} }  \ottnt{A} $ and $ \Gamma_{{\mathrm{2}}}  \vdash  \ottnt{a}  :^{ \ell_{{\mathrm{2}}} }  \ottnt{A} $ and $ \Gamma  =   \Gamma_{{\mathrm{1}}}  \sqcap  \Gamma_{{\mathrm{2}}}  $. 
\end{lemma}

\begin{proof}
Have: $ \Gamma  \vdash  \ottnt{a}  :^{  \ell_{{\mathrm{1}}}  \: \sqcap \:  \ell_{{\mathrm{2}}}  }  \ottnt{A} $. By \rref{PTS-SubRD}, $ \Gamma  \vdash  \ottnt{a}  :^{ \ell_{{\mathrm{1}}} }  \ottnt{A} $ and $ \Gamma  \vdash  \ottnt{a}  :^{ \ell_{{\mathrm{2}}} }  \ottnt{A} $. The lemma follows by setting $\Gamma_{{\mathrm{1}}} := \Gamma$ and $\Gamma_{{\mathrm{2}}} := \Gamma$.
\end{proof}

%--------------------------------------------------------------------------------------------------

\begin{lemma}[Weakening]
If $  \Gamma_{{\mathrm{1}}}  ,  \Gamma_{{\mathrm{2}}}   \vdash  \ottnt{a}  :^{ \ell }  \ottnt{A} $ and $ \Delta_{{\mathrm{1}}}  \vdash_{0}  \ottnt{C}  :   \ottmv{s}  $ and $  \lfloor  \Gamma_{{\mathrm{1}}}  \rfloor   =  \Delta_{{\mathrm{1}}} $, then $    \Gamma_{{\mathrm{1}}}  ,   \ottmv{z}  :^{  \top  }  \ottnt{C}     ,  \Gamma_{{\mathrm{2}}}   \vdash  \ottnt{a}  :^{ \ell }  \ottnt{A} $.
\end{lemma}

\begin{proof}
By induction on $  \Gamma_{{\mathrm{1}}}  ,  \Gamma_{{\mathrm{2}}}   \vdash  \ottnt{a}  :^{ \ell }  \ottnt{A} $.
\end{proof}

%---------------------------------------------------------------------------------------------------

\begin{lemma}[Substitution (Lemma \ref{DSubst})]
If $    \Gamma_{{\mathrm{1}}}  ,   \ottmv{z}  :^{ \ottnt{m_{{\mathrm{0}}}} }  \ottnt{C}     ,  \Gamma_{{\mathrm{2}}}   \vdash  \ottnt{a}  :^{ \ell }  \ottnt{A} $ and $ \Gamma  \vdash  \ottnt{c}  :^{ \ottnt{m_{{\mathrm{0}}}} }  \ottnt{C} $ and $  \lfloor  \Gamma_{{\mathrm{1}}}  \rfloor   =   \lfloor  \Gamma  \rfloor  $, then $     \Gamma_{{\mathrm{1}}}  \sqcap  \Gamma    ,  \Gamma_{{\mathrm{2}}}   \{  \ottnt{c}  /  \ottmv{z}  \}   \vdash   \ottnt{a}  \{  \ottnt{c}  /  \ottmv{z}  \}   :^{ \ell }   \ottnt{A}  \{  \ottnt{c}  /  \ottmv{z}  \}  $. 
\end{lemma}

\begin{proof}
By induction on $    \Gamma_{{\mathrm{1}}}  ,   \ottmv{z}  :^{ \ottnt{m_{{\mathrm{0}}}} }  \ottnt{C}     ,  \Gamma_{{\mathrm{2}}}   \vdash  \ottnt{a}  :^{ \ell }  \ottnt{A} $. Follow lemma \ref{DSSubstP}.
\end{proof}

%---------------------------------------------------------------------------------------------------

\begin{theorem}[Preservation (Theorem \ref{preservationDep})]
If $ \Gamma  \vdash  \ottnt{a}  :^{ \ell }  \ottnt{A} $ and $ \vdash  \ottnt{a}  \leadsto  \ottnt{a'} $, then $ \Gamma  \vdash  \ottnt{a'}  :^{ \ell }  \ottnt{A} $.
\end{theorem}

\begin{proof}

By induction on $ \Gamma  \vdash  \ottnt{a}  :^{ \ell }  \ottnt{A} $ and inversion on $ \vdash  \ottnt{a}  \leadsto  \ottnt{a'} $. Follow theorem \ref{Dpreserve}.

\end{proof}

%----------------------------------------------------------------------------------------------------

\begin{theorem}[Progress (Theorem \ref{progressDep})]
If $  \emptyset   \vdash  \ottnt{a}  :^{ \ell }  \ottnt{A} $, then either $\ottnt{a}$ is a value or there exists $\ottnt{a'}$ such that $ \vdash  \ottnt{a}  \leadsto  \ottnt{a'} $.
\end{theorem}

\begin{proof}

By induction on $  \emptyset   \vdash  \ottnt{a}  :^{ \ell }  \ottnt{A} $. Follow theorem \ref{Dprogress}.

\end{proof}

%-----------------------------------------------------------------------------------------------------

\begin{lemma} \label{heaphelperDDep}
If $ \ottnt{H}  \models   \Gamma_{{\mathrm{1}}}  \sqcap  \Gamma_{{\mathrm{2}}}  $ and $ \Gamma_{{\mathrm{2}}}  \vdash  \ottnt{a}  :^{ \ell }  \ottnt{A} $ and $l \neq \top$, then either $\ottnt{a}$ is a value or there exists $\ottnt{H'}, \Gamma'_{{\mathrm{2}}}, \ottnt{a'}$ such that:
\begin{itemize}
\item $ [  \ottnt{H}  ]  \ottnt{a}  \Longrightarrow^{ \ell }_{ \ottnt{S} } [  \ottnt{H'}  ]  \ottnt{a'} $
\item $ \ottnt{H'}  \models   \Gamma_{{\mathrm{1}}}  \sqcap  \Gamma'_{{\mathrm{2}}}  $
\item $ \Gamma'_{{\mathrm{2}}}  \vdash  \ottnt{a'}  :^{ \ell }  \ottnt{A} $
\end{itemize} 
\end{lemma}

\begin{proof}
By induction on $ \Gamma_{{\mathrm{2}}}  \vdash  \ottnt{a}  :^{ \ell }  \ottnt{A} $. Follow lemma \ref{heaphelperDep}.
\end{proof}

%---------------------------------------------------------------------------------------------------

\begin{theorem}[Soundness (Theorem \ref{heapDep})]
If $ \ottnt{H}  \models  \Gamma $ and $ \Gamma  \vdash  \ottnt{a}  :^{ \ell }  \ottnt{A} $ and $l \neq \top$, then either $\ottnt{a}$ is a value or there exists $\ottnt{H'}, \Gamma', \ottnt{a'}$ such that $ [  \ottnt{H}  ]  \ottnt{a}  \Longrightarrow^{ \ell }_{ \ottnt{S} } [  \ottnt{H'}  ]  \ottnt{a'} $ and $ \ottnt{H'}  \models  \Gamma' $ and $ \Gamma'  \vdash  \ottnt{a'}  :^{ \ell }  \ottnt{A} $.
\end{theorem}

\begin{proof}
Use lemma \ref{heaphelperDDep} with $\Gamma_{{\mathrm{1}}} :=   \top   \sqcup  \Gamma $ and $\Gamma_{{\mathrm{2}}} := \Gamma$.
\end{proof}

%----------------------------------------------------------------------------------------------------

%----------------------------------------------------------------------------------------------------
\section{Unrestricted Usage}
%----------------------------------------------------------------------------------------------------

\begin{lemma}[Multiplication] \label{BLwDMultP}
If $ \Gamma  \vdash  \ottnt{a}  :^{ \ottnt{q} }  \ottnt{A} $, then $  \ottnt{r_{{\mathrm{0}}}}  \cdot  \Gamma   \vdash  \ottnt{a}  :^{  \ottnt{r_{{\mathrm{0}}}}  \cdot  \ottnt{q}  }  \ottnt{A} $.
\end{lemma}

\begin{proof}
By induction on $ \Gamma  \vdash  \ottnt{a}  :^{ \ottnt{q} }  \ottnt{A} $. All the cases other than \rref{PTS-LamOmega} are similar to those of lemma \ref{BLDMultP}.

\begin{itemize}
\item \Rref{PTS-LamOmega}. Have: $  \ottnt{q_{{\mathrm{0}}}}  \cdot  \Gamma   \vdash   \lambda^{ \ottnt{r} }  \ottmv{x}  :  \ottnt{A}  .  \ottnt{b}   :^{  \ottnt{q_{{\mathrm{0}}}}  \cdot  \ottnt{q}  }   \Pi  \ottmv{x}  :^{ \ottnt{r} } \!  \ottnt{A}  .  \ottnt{B}  $ where $  \Gamma  ,   \ottmv{x}  :^{  \ottnt{q}  \cdot  \ottnt{r}  }  \ottnt{A}    \vdash  \ottnt{b}  :^{ \ottnt{q} }  \ottnt{B} $ and $  \ottnt{q}  =   \omega    \Rightarrow   \ottnt{r}  =   \omega   $ and $ \ottnt{q_{{\mathrm{0}}}}  \neq   0  $.\\
Need to show: $  \ottnt{r_{{\mathrm{0}}}}  \cdot    \ottnt{q_{{\mathrm{0}}}}  \cdot  \Gamma     \vdash   \lambda^{ \ottnt{r} }  \ottmv{x}  :  \ottnt{A}  .  \ottnt{b}   :^{  \ottnt{r_{{\mathrm{0}}}}  \cdot    \ottnt{q_{{\mathrm{0}}}}  \cdot  \ottnt{q}    }   \Pi  \ottmv{x}  :^{ \ottnt{r} } \!  \ottnt{A}  .  \ottnt{B}  $.\\
There are two cases to consider:
\begin{itemize}
\item $ \ottnt{r_{{\mathrm{0}}}}  =   0  $. By IH, $    0   \cdot  \Gamma   ,   \ottmv{x}  :^{  0  }  \ottnt{A}    \vdash  \ottnt{b}  :^{  0  }  \ottnt{B} $.\\
This case, then, follows by \rref{PTS-LamOmega}.
\item $ \ottnt{r_{{\mathrm{0}}}}  \neq   0  $. Therefore, $  \ottnt{r_{{\mathrm{0}}}}  \cdot  \ottnt{q_{{\mathrm{0}}}}   \neq   0  $.\\
This case, then, follows by \rref{PTS-LamOmega}.
\end{itemize}
\end{itemize}

\end{proof}

%--------------------------------------------------------------------------------------------------

\begin{lemma}[Independence] \label{ind}
If $    \Gamma_{{\mathrm{1}}}  ,   \ottmv{z}  :^{ \ottnt{r_{{\mathrm{0}}}} }  \ottnt{C}     ,  \Gamma_{{\mathrm{2}}}   \vdash  \ottnt{a}  :^{ \ottnt{q} }  \ottnt{A} $ and $\neg ( \ottnt{r_{{\mathrm{0}}}}  \ll:  \ottnt{q} )$, then $    \Gamma_{{\mathrm{1}}}  ,   \ottmv{z}  :^{  0  }  \ottnt{C}     ,  \Gamma_{{\mathrm{2}}}   \vdash  \ottnt{a}  :^{ \ottnt{q} }  \ottnt{A} $. Here, $ \ottnt{r_{{\mathrm{0}}}}  \ll:  \ottnt{q}  \triangleq \exists \ottnt{q_{{\mathrm{0}}}},  \ottnt{r_{{\mathrm{0}}}}  =   \ottnt{q}  +  \ottnt{q_{{\mathrm{0}}}}  $.
\end{lemma}

\begin{proof}

By induction on $    \Gamma_{{\mathrm{1}}}  ,   \ottmv{z}  :^{ \ottnt{r_{{\mathrm{0}}}} }  \ottnt{C}     ,  \Gamma_{{\mathrm{2}}}   \vdash  \ottnt{a}  :^{ \ottnt{q} }  \ottnt{A} $.

\begin{itemize}

\item \Rref{PTS-Var}. Have $    0   \cdot  \Gamma   ,   \ottmv{x}  :^{ \ottnt{q} }  \ottnt{A}    \vdash   \ottmv{x}   :^{ \ottnt{q} }  \ottnt{A} $ where $ \Delta  \vdash_{0}  \ottnt{A}  :   \ottmv{s}  $ and $\Delta =  \lfloor  \Gamma  \rfloor $.\\
This case follows immediately because $ \ottnt{q}  \ll:  \ottnt{q} $ and the other grades are all $0$.

\item \Rref{PTS-Pi}. Have: $     \Gamma_{{\mathrm{11}}}  +  \Gamma_{{\mathrm{21}}}    ,   \ottmv{z}  :^{  \ottnt{r_{{\mathrm{01}}}}  +  \ottnt{r_{{\mathrm{02}}}}  }  \ottnt{C}    ,    \Gamma_{{\mathrm{12}}}  +  \Gamma_{{\mathrm{22}}}     \vdash   \Pi  \ottmv{x}  :^{ \ottnt{r} } \!  \ottnt{A}  .  \ottnt{B}   :^{ \ottnt{q} }   \ottmv{s_{{\mathrm{3}}}}  $ where $   \Gamma_{{\mathrm{11}}}  ,   \ottmv{z}  :^{ \ottnt{r_{{\mathrm{01}}}} }  \ottnt{C}    ,  \Gamma_{{\mathrm{12}}}   \vdash  \ottnt{A}  :^{ \ottnt{q} }   \ottmv{s_{{\mathrm{1}}}}  $ and $    \Gamma_{{\mathrm{21}}}  ,   \ottmv{z}  :^{ \ottnt{r_{{\mathrm{02}}}} }  \ottnt{C}    ,  \Gamma_{{\mathrm{22}}}   ,   \ottmv{x}  :^{ \ottnt{r'} }  \ottnt{A}    \vdash  \ottnt{B}  :^{ \ottnt{q} }   \ottmv{s_{{\mathrm{2}}}}  $ and $ \mathcal{R} ( \ottmv{s_{{\mathrm{1}}}}  ,  \ottmv{s_{{\mathrm{2}}}}  ,  \ottmv{s_{{\mathrm{3}}}} ) $ and $ \neg(  \ottnt{r_{{\mathrm{01}}}}  +  \ottnt{r_{{\mathrm{02}}}}   \ll:  \ottnt{q} ) $.\\
Now, if $ \ottnt{r_{{\mathrm{01}}}}  \ll:  \ottnt{q} $ or $ \ottnt{r_{{\mathrm{02}}}}  \ll:  \ottnt{q} $, then $  \ottnt{r_{{\mathrm{01}}}}  +  \ottnt{r_{{\mathrm{02}}}}   \ll:  \ottnt{q} $, a contradiction.\\
Therefore, $ \neg( \ottnt{r_{{\mathrm{01}}}}  \ll:  \ottnt{q} ) $ and $ \neg( \ottnt{r_{{\mathrm{02}}}}  \ll:  \ottnt{q} ) $.\\
This case, then, follows by IH and \rref{PTS-Pi}.

\item \Rref{PTS-LamOmega}. Have: $    \ottnt{q_{{\mathrm{0}}}}  \cdot  \Gamma_{{\mathrm{1}}}   ,   \ottmv{z}  :^{  \ottnt{q_{{\mathrm{0}}}}  \cdot  \ottnt{r_{{\mathrm{0}}}}  }  \ottnt{C}    ,   \ottnt{q_{{\mathrm{0}}}}  \cdot  \Gamma_{{\mathrm{2}}}    \vdash   \lambda^{ \ottnt{r} }  \ottmv{x}  :  \ottnt{A}  .  \ottnt{b}   :^{  \ottnt{q_{{\mathrm{0}}}}  \cdot  \ottnt{q}  }   \Pi  \ottmv{x}  :^{ \ottnt{r} } \!  \ottnt{A}  .  \ottnt{B}  $ where $    \Gamma_{{\mathrm{1}}}  ,   \ottmv{z}  :^{ \ottnt{r_{{\mathrm{0}}}} }  \ottnt{C}    ,  \Gamma_{{\mathrm{2}}}   ,   \ottmv{x}  :^{  \ottnt{q}  \cdot  \ottnt{r}  }  \ottnt{A}    \vdash  \ottnt{b}  :^{ \ottnt{q} }  \ottnt{B} $ and $  \ottnt{q}  =   \omega    \Rightarrow   \ottnt{r}  =   \omega   $ and $ \ottnt{q_{{\mathrm{0}}}}  \neq   0  $ and $ \neg(  \ottnt{q_{{\mathrm{0}}}}  \cdot  \ottnt{r_{{\mathrm{0}}}}   \ll:   \ottnt{q_{{\mathrm{0}}}}  \cdot  \ottnt{q}  ) $.\\
If $ \ottnt{r_{{\mathrm{0}}}}  \ll:  \ottnt{q} $, then $  \ottnt{q_{{\mathrm{0}}}}  \cdot  \ottnt{r_{{\mathrm{0}}}}   \ll:   \ottnt{q_{{\mathrm{0}}}}  \cdot  \ottnt{q}  $, a contradiction. Therefore, $ \neg( \ottnt{r_{{\mathrm{0}}}}  \ll:  \ottnt{q} ) $.\\
This case, then, follows by IH and \rref{PTS-LamOmega}.

\item \Rref{PTS-App}. Have: $     \Gamma_{{\mathrm{11}}}  +  \Gamma_{{\mathrm{21}}}   ,   \ottmv{z}  :^{  \ottnt{r_{{\mathrm{01}}}}  +  \ottnt{r_{{\mathrm{02}}}}  }  \ottnt{C}    ,  \Gamma_{{\mathrm{12}}}   +  \Gamma_{{\mathrm{22}}}   \vdash   \ottnt{b}  \:  \ottnt{a} ^{ \ottnt{r} }   :^{ \ottnt{q} }   \ottnt{B}  \{  \ottnt{a}  /  \ottmv{x}  \}  $ where $   \Gamma_{{\mathrm{11}}}  ,   \ottmv{z}  :^{ \ottnt{r_{{\mathrm{01}}}} }  \ottnt{C}    ,  \Gamma_{{\mathrm{12}}}   \vdash  \ottnt{b}  :^{ \ottnt{q} }   \Pi  \ottmv{x}  :^{ \ottnt{r} } \!  \ottnt{A}  .  \ottnt{B}  $ and $   \Gamma_{{\mathrm{21}}}  ,   \ottmv{z}  :^{ \ottnt{r_{{\mathrm{02}}}} }  \ottnt{C}    ,  \Gamma_{{\mathrm{22}}}   \vdash  \ottnt{a}  :^{  \ottnt{q}  \cdot  \ottnt{r}  }  \ottnt{A} $ and $ \neg(  \ottnt{r_{{\mathrm{01}}}}  +  \ottnt{r_{{\mathrm{02}}}}   \ll:  \ottnt{q} ) $.\\
There are two cases to consider.
\begin{itemize}
\item $ \ottnt{r}  =   0  $. Then, $   \Gamma_{{\mathrm{21}}}  ,   \ottmv{z}  :^{  0  }  \ottnt{C}    ,  \Gamma_{{\mathrm{22}}}   \vdash  \ottnt{a}  :^{  0  }  \ottnt{A} $.\\
Now, if $ \ottnt{r_{{\mathrm{01}}}}  \ll:  \ottnt{q} $, then $  \ottnt{r_{{\mathrm{01}}}}  +  \ottnt{r_{{\mathrm{02}}}}   \ll:  \ottnt{q} $, a contradiction. Therefore, $ \neg( \ottnt{r_{{\mathrm{01}}}}  \ll:  \ottnt{q} ) $.\\
This case, then, follows by IH and \rref{PTS-App}.

\item $ \ottnt{r}  \neq   0  $. Then, $\ottnt{r} =   1   +  \ottnt{r'} $, for some $\ottnt{r'}$. Therefore, $ \neg( \ottnt{r_{{\mathrm{01}}}}  \ll:  \ottnt{q} ) $ and $ \neg( \ottnt{r_{{\mathrm{02}}}}  \ll:   \ottnt{q}  \cdot  \ottnt{r}  ) $.\\
This case, then, follows by IH and \rref{PTS-App}.
\end{itemize}

\item \Rref{PTS-Pair}. Have: $     \Gamma_{{\mathrm{11}}}  +  \Gamma_{{\mathrm{21}}}   ,   \ottmv{z}  :^{  \ottnt{r_{{\mathrm{01}}}}  +  \ottnt{r_{{\mathrm{02}}}}  }  \ottnt{C}    ,  \Gamma_{{\mathrm{12}}}   +  \Gamma_{{\mathrm{22}}}   \vdash   (  \ottnt{a_{{\mathrm{1}}}} ^{ \ottnt{r} } ,  \ottnt{a_{{\mathrm{2}}}}  )   :^{ \ottnt{q} }   \Sigma  \ottmv{x}  :^{ \ottnt{r} } \!  \ottnt{A_{{\mathrm{1}}}}  .  \ottnt{A_{{\mathrm{2}}}}  $ where $   \Gamma_{{\mathrm{11}}}  ,   \ottmv{z}  :^{ \ottnt{r_{{\mathrm{01}}}} }  \ottnt{C}    ,  \Gamma_{{\mathrm{12}}}   \vdash  \ottnt{a_{{\mathrm{1}}}}  :^{  \ottnt{q}  \cdot  \ottnt{r}  }  \ottnt{A_{{\mathrm{1}}}} $ and $   \Gamma_{{\mathrm{21}}}  ,   \ottmv{z}  :^{ \ottnt{r_{{\mathrm{02}}}} }  \ottnt{C}    ,  \Gamma_{{\mathrm{22}}}   \vdash  \ottnt{a_{{\mathrm{2}}}}  :^{ \ottnt{q} }   \ottnt{A_{{\mathrm{2}}}}  \{  \ottnt{a_{{\mathrm{1}}}}  /  \ottmv{x}  \}  $ and $ \neg(  \ottnt{r_{{\mathrm{01}}}}  +  \ottnt{r_{{\mathrm{02}}}}   \ll:  \ottnt{q} ) $.\\
Here too, there are two cases to consider.
\begin{itemize}
\item $ \ottnt{r}  =   0  $. Then, $   \Gamma_{{\mathrm{11}}}  ,   \ottmv{z}  :^{  0  }  \ottnt{C}    ,  \Gamma_{{\mathrm{12}}}   \vdash  \ottnt{a_{{\mathrm{1}}}}  :^{  0  }  \ottnt{A_{{\mathrm{1}}}} $.\\
Now, if $ \ottnt{r_{{\mathrm{02}}}}  \ll:  \ottnt{q} $, then $  \ottnt{r_{{\mathrm{01}}}}  +  \ottnt{r_{{\mathrm{02}}}}   \ll:  \ottnt{q} $, a contradiction. Therefore, $ \neg( \ottnt{r_{{\mathrm{02}}}}  \ll:  \ottnt{q} ) $.\\
This case, then, follows by IH and \rref{PTS-Pair}.

\item $ \ottnt{r}  \neq   0  $. Then, $\ottnt{r} =   1   +  \ottnt{r'} $, for some $\ottnt{r'}$. Therefore, $ \neg( \ottnt{r_{{\mathrm{01}}}}  \ll:   \ottnt{q}  \cdot  \ottnt{r}  ) $ and $ \neg( \ottnt{r_{{\mathrm{02}}}}  \ll:  \ottnt{q} ) $.\\
This case, then, follows by IH and \rref{PTS-Pair}.
\end{itemize}

\item \Rref{PTS-LetPair}. Have: $     \Gamma_{{\mathrm{11}}}  +  \Gamma_{{\mathrm{21}}}   ,   \ottmv{z}  :^{  \ottnt{r_{{\mathrm{01}}}}  +  \ottnt{r_{{\mathrm{02}}}}  }  \ottnt{C}    ,  \Gamma_{{\mathrm{12}}}   +  \Gamma_{{\mathrm{22}}}   \vdash   \mathbf{let}_{ \ottnt{q_{{\mathrm{0}}}} } \: (  \ottmv{x} ^{ \ottnt{r} } ,  \ottmv{y}  ) \: \mathbf{be} \:  \ottnt{a}  \: \mathbf{in} \:  \ottnt{b}   :^{ \ottnt{q} }   \ottnt{B}  \{  \ottnt{a}  /  \ottmv{z}  \}  $ where $  \Delta  ,   \ottmv{z}  :   \Sigma  \ottmv{x}  :^{ \ottnt{r} } \!  \ottnt{A_{{\mathrm{1}}}}  .  \ottnt{A_{{\mathrm{2}}}}     \vdash_{0}  \ottnt{B}  :   \ottmv{s}  $ and $   \Gamma_{{\mathrm{11}}}  ,   \ottmv{z}  :^{ \ottnt{r_{{\mathrm{01}}}} }  \ottnt{C}    ,  \Gamma_{{\mathrm{12}}}   \vdash  \ottnt{a}  :^{  \ottnt{q}  \cdot  \ottnt{q_{{\mathrm{0}}}}  }   \Sigma  \ottmv{x}  :^{ \ottnt{r} } \!  \ottnt{A_{{\mathrm{1}}}}  .  \ottnt{A_{{\mathrm{2}}}}  $ and $     \Gamma_{{\mathrm{21}}}  ,   \ottmv{z}  :^{ \ottnt{r_{{\mathrm{02}}}} }  \ottnt{C}    ,  \Gamma_{{\mathrm{22}}}   ,   \ottmv{x}  :^{  \ottnt{q}  \cdot   \ottnt{q_{{\mathrm{0}}}}  \cdot  \ottnt{r}   }  \ottnt{A_{{\mathrm{1}}}}    ,   \ottmv{y}  :^{  \ottnt{q}  \cdot  \ottnt{q_{{\mathrm{0}}}}  }  \ottnt{A_{{\mathrm{2}}}}    \vdash  \ottnt{b}  :^{ \ottnt{q} }   \ottnt{B}  \{   (   \ottmv{x}  ^{ \ottnt{r} } ,   \ottmv{y}   )   /  \ottmv{z}  \}  $ and $ \ottnt{q_{{\mathrm{0}}}}  <:   1  $ and $ \neg(  \ottnt{r_{{\mathrm{01}}}}  +  \ottnt{r_{{\mathrm{02}}}}   \ll:  \ottnt{q} ) $.\\
Since $ \ottnt{q_{{\mathrm{0}}}}  <:   1  $, so $\ottnt{q_{{\mathrm{0}}}} =   1   +  \ottnt{q'_{{\mathrm{0}}}} $, for some $\ottnt{q'_{{\mathrm{0}}}}$. Therefore, $ \neg( \ottnt{r_{{\mathrm{01}}}}  \ll:   \ottnt{q}  \cdot  \ottnt{q_{{\mathrm{0}}}}  ) $ and $ \neg( \ottnt{r_{{\mathrm{02}}}}  \ll:  \ottnt{q} ) $.\\
This case, then, follows by IH and \rref{PTS-LetPair}.

\item The other cases follow similarly. 

\end{itemize}

\end{proof}

%--------------------------------------------------------------------------------------------------

\begin{lemma}[Factorization] \label{BLwDFactP}
If $ \Gamma  \vdash  \ottnt{a}  :^{ \ottnt{q} }  \ottnt{A} $ and $q \neq 0$, then there exists $\Gamma'$ such that $ \Gamma'  \vdash  \ottnt{a}  :^{  1  }  \ottnt{A} $ and $ \Gamma  <:   \ottnt{q}  \cdot  \Gamma'  $.
\end{lemma}

\begin{proof}
By induction on $ \Gamma  \vdash  \ottnt{a}  :^{ \ottnt{q} }  \ottnt{A} $.

\begin{itemize}

\item \Rref{PTS-Pi}. Have: $  \Gamma_{{\mathrm{1}}}  +  \Gamma_{{\mathrm{2}}}   \vdash   \Pi  \ottmv{x}  :^{ \ottnt{r} } \!  \ottnt{A}  .  \ottnt{B}   :^{ \ottnt{q} }   \ottmv{s_{{\mathrm{3}}}}  $ where $ \Gamma_{{\mathrm{1}}}  \vdash  \ottnt{A}  :^{ \ottnt{q} }   \ottmv{s_{{\mathrm{1}}}}  $ and $  \Gamma_{{\mathrm{2}}}  ,   \ottmv{x}  :^{ \ottnt{r_{{\mathrm{0}}}} }  \ottnt{A}    \vdash  \ottnt{B}  :^{ \ottnt{q} }   \ottmv{s_{{\mathrm{2}}}}  $.\\
By IH, $\exists \Gamma'_{{\mathrm{1}}}, \Gamma'_{{\mathrm{2}}}$ such that $ \Gamma'_{{\mathrm{1}}}  \vdash  \ottnt{A}  :^{  1  }   \ottmv{s_{{\mathrm{1}}}}  $ and $  \Gamma'_{{\mathrm{2}}}  ,   \ottmv{x}  :^{ \ottnt{r'} }  \ottnt{A}    \vdash  \ottnt{B}  :^{  1  }   \ottmv{s_{{\mathrm{2}}}}  $ where $ \Gamma_{{\mathrm{1}}}  <:   \ottnt{q}  \cdot  \Gamma'_{{\mathrm{1}}}  $ and $ \Gamma_{{\mathrm{2}}}  <:   \ottnt{q}  \cdot  \Gamma'_{{\mathrm{2}}}  $.\\
This case, then, follows by \rref{PTS-Pi}.

\item \Rref{PTS-LamOmega}. Have: $  \ottnt{q_{{\mathrm{0}}}}  \cdot  \Gamma   \vdash   \lambda^{ \ottnt{r} }  \ottmv{x}  :  \ottnt{A}  .  \ottnt{b}   :^{  \ottnt{q_{{\mathrm{0}}}}  \cdot  \ottnt{q}  }   \Pi  \ottmv{x}  :^{ \ottnt{r} } \!  \ottnt{A}  .  \ottnt{B}  $ where $  \Gamma  ,   \ottmv{x}  :^{  \ottnt{q}  \cdot  \ottnt{r}  }  \ottnt{A}    \vdash  \ottnt{b}  :^{ \ottnt{q} }  \ottnt{B} $ and $  \ottnt{q}  =   \omega    \Rightarrow   \ottnt{r}  =   \omega   $ and $ \ottnt{q_{{\mathrm{0}}}}  \neq   0  $.\\
There are two cases to consider:
\begin{itemize}
\item $ \ottnt{q}  =   \omega  $. Then $  \ottnt{q_{{\mathrm{0}}}}  \cdot  \ottnt{q}   =   \omega  $. \\
Here, we need to define an operation on contexts. For a context $\Gamma$, define $ \Gamma ^{ \{ 0 , \omega \} } $ as:
\begin{align*}
  \emptyset  ^{ \{ 0 , \omega \} }  & =  \emptyset  \\
  (   \Gamma  ,   \ottmv{x}  :^{ \ottnt{q} }  \ottnt{A}    )  ^{ \{ 0 , \omega \} }  & = \begin{cases}
                                \Gamma ^{ \{ 0 , \omega \} }   ,   \ottmv{x}  :^{ \ottnt{q} }  \ottnt{A}   \text{ if } q \in \{ 0 ,  \omega  \} \\
                                \Gamma ^{ \{ 0 , \omega \} }   ,   \ottmv{x}  :^{  0  }  \ottnt{A}   \text{ otherwise}
                             \end{cases}
\end{align*}
Now, since $  \ottnt{q_{{\mathrm{0}}}}  \cdot  \Gamma   \vdash   \lambda^{ \ottnt{r} }  \ottmv{x}  :  \ottnt{A}  .  \ottnt{b}   :^{  \omega  }    \Pi  \ottmv{x}  :^{ \ottnt{r} } \!  \ottnt{A}  .  \ottnt{B}   $, by lemma \ref{ind}, $   (   \ottnt{q_{{\mathrm{0}}}}  \cdot  \Gamma   )  ^{ \{ 0 , \omega \} }   \vdash   \lambda^{ \ottnt{r} }  \ottmv{x}  :  \ottnt{A}  .  \ottnt{b}   :^{  \omega  }    \Pi  \ottmv{x}  :^{ \ottnt{r} } \!  \ottnt{A}  .  \ottnt{B}   $.\\
By \rref{PTS-SubR}, $   (   \ottnt{q_{{\mathrm{0}}}}  \cdot  \Gamma   )  ^{ \{ 0 , \omega \} }   \vdash   \lambda^{ \ottnt{r} }  \ottmv{x}  :  \ottnt{A}  .  \ottnt{b}   :^{  1  }   \Pi  \ottmv{x}  :^{ \ottnt{r} } \!  \ottnt{A}  .  \ottnt{B}  $.\\
Next, $  \ottnt{q_{{\mathrm{0}}}}  \cdot  \Gamma   <:    (   \ottnt{q_{{\mathrm{0}}}}  \cdot  \Gamma   )  ^{ \{ 0 , \omega \} }   =   \omega   \cdot    (   \ottnt{q_{{\mathrm{0}}}}  \cdot  \Gamma   )  ^{ \{ 0 , \omega \} }   =   \ottnt{q_{{\mathrm{0}}}}  \cdot  \ottnt{q}   \cdot    (   \ottnt{q_{{\mathrm{0}}}}  \cdot  \Gamma   )  ^{ \{ 0 , \omega \} }  $.\\
Note that in case of $ \mathbb{N}_{=}^{\omega} $ and $ \mathcal{Q}_{\text{Lin} } $, we have, $ \ottnt{q_{{\mathrm{0}}}}  \cdot  \Gamma  =   (   \ottnt{q_{{\mathrm{0}}}}  \cdot  \Gamma   )  ^{ \{ 0 , \omega \} } $. 

\item $ \ottnt{q}  \neq   \omega  $. By IH, $  \Gamma'  ,   \ottmv{x}  :^{ \ottnt{r'} }  \ottnt{A}    \vdash  \ottnt{b}  :^{  1  }  \ottnt{B} $ where $ \Gamma  <:   \ottnt{q}  \cdot  \Gamma'  $ and $  \ottnt{q}  \cdot  \ottnt{r}   <:   \ottnt{q}  \cdot  \ottnt{r'}  $.\\
Since $  \ottnt{q}  \cdot  \ottnt{r}   <:   \ottnt{q}  \cdot  \ottnt{r'}  $ and $q \notin \{ 0 ,  \omega  \}$, so $ \ottnt{r}  <:  \ottnt{r'} $.\\
By \rref{PTS-SubL}, $  \Gamma'  ,   \ottmv{x}  :^{ \ottnt{r} }  \ottnt{A}    \vdash  \ottnt{b}  :^{  1  }  \ottnt{B} $.\\
This case, then, follows by \rref{PTS-LamOmega}.
 
\end{itemize}

\item \Rref{PTS-App}. Have: $  \Gamma_{{\mathrm{1}}}  +  \Gamma_{{\mathrm{2}}}   \vdash   \ottnt{b}  \:  \ottnt{a} ^{ \ottnt{r} }   :^{ \ottnt{q} }   \ottnt{B}  \{  \ottnt{a}  /  \ottmv{x}  \}  $ where $ \Gamma_{{\mathrm{1}}}  \vdash  \ottnt{b}  :^{ \ottnt{q} }   \Pi  \ottmv{x}  :^{ \ottnt{r} } \!  \ottnt{A}  .  \ottnt{B}  $ and $ \Gamma_{{\mathrm{2}}}  \vdash  \ottnt{a}  :^{  \ottnt{q}  \cdot  \ottnt{r}  }  \ottnt{A} $.\\
There are two cases to consider:
\begin{itemize}
\item $ \ottnt{r}  =   0  $. By IH, $\exists \Gamma'_{{\mathrm{1}}}$ such that $ \Gamma'_{{\mathrm{1}}}  \vdash  \ottnt{b}  :^{  1  }   \Pi  \ottmv{x}  :^{ \ottnt{r} } \!  \ottnt{A}  .  \ottnt{B}  $ and $ \Gamma_{{\mathrm{1}}}  <:   \ottnt{q}  \cdot  \Gamma'_{{\mathrm{1}}}  $.\\
Next, by the multiplication lemma, $   0   \cdot  \Gamma_{{\mathrm{2}}}   \vdash  \ottnt{a}  :^{  0  }  \ottnt{A} $.\\
This case, then, follows by \rref{PTS-App}.
\item $ \ottnt{r}  \neq   0  $. By IH, $\exists \Gamma'_{{\mathrm{1}}}, \Gamma'_{{\mathrm{2}}}$ such that $ \Gamma'_{{\mathrm{1}}}  \vdash  \ottnt{b}  :^{  1  }   \Pi  \ottmv{x}  :^{ \ottnt{r} } \!  \ottnt{A}  .  \ottnt{B}  $ and $ \Gamma'_{{\mathrm{2}}}  \vdash  \ottnt{a}  :^{  1  }  \ottnt{A} $ where $ \Gamma_{{\mathrm{1}}}  <:   \ottnt{q}  \cdot  \Gamma'_{{\mathrm{1}}}  $ and $ \Gamma_{{\mathrm{2}}}  <:    \ottnt{q}  \cdot  \ottnt{r}   \cdot  \Gamma'_{{\mathrm{2}}}  $.\\
Then, by the multiplication lemma, $  \ottnt{r}  \cdot  \Gamma'_{{\mathrm{2}}}   \vdash  \ottnt{a}  :^{ \ottnt{r} }  \ottnt{A} $.\\
This case, then, follows by \rref{PTS-App}. 
\end{itemize}

\item \Rref{PTS-Pair}. Have: $  \Gamma_{{\mathrm{1}}}  +  \Gamma_{{\mathrm{2}}}   \vdash   (  \ottnt{a_{{\mathrm{1}}}} ^{ \ottnt{r} } ,  \ottnt{a_{{\mathrm{2}}}}  )   :^{ \ottnt{q} }   \Sigma  \ottmv{x}  :^{ \ottnt{r} } \!  \ottnt{A_{{\mathrm{1}}}}  .  \ottnt{A_{{\mathrm{2}}}}  $ where $ \Gamma_{{\mathrm{1}}}  \vdash  \ottnt{a_{{\mathrm{1}}}}  :^{  \ottnt{q}  \cdot  \ottnt{r}  }  \ottnt{A_{{\mathrm{1}}}} $ and $ \Gamma_{{\mathrm{2}}}  \vdash  \ottnt{a_{{\mathrm{2}}}}  :^{ \ottnt{q} }   \ottnt{A_{{\mathrm{2}}}}  \{  \ottnt{a_{{\mathrm{1}}}}  /  \ottmv{x}  \}  $.\\
There are two cases to consider.

\begin{itemize}

\item $r = 0$. Then, by the multiplication lemma, $   0   \cdot  \Gamma_{{\mathrm{1}}}   \vdash  \ottnt{a_{{\mathrm{1}}}}  :^{  0  }  \ottnt{A_{{\mathrm{1}}}} $.\\
By IH, $\exists \Gamma'_{{\mathrm{2}}}$ such that $ \Gamma'_{{\mathrm{2}}}  \vdash  \ottnt{a_{{\mathrm{2}}}}  :^{  1  }   \ottnt{A_{{\mathrm{2}}}}  \{  \ottnt{a_{{\mathrm{1}}}}  /  \ottmv{x}  \}  $ and $ \Gamma_{{\mathrm{2}}}  <:   \ottnt{q}  \cdot  \Gamma'_{{\mathrm{2}}}  $.\\
This case, then, follows by \rref{PTS-Pair}.

\item $r \neq 0$. By IH, $\exists \Gamma'_{{\mathrm{1}}}, \Gamma'_{{\mathrm{2}}}$ such that $ \Gamma'_{{\mathrm{1}}}  \vdash  \ottnt{a_{{\mathrm{1}}}}  :^{  1  }  \ottnt{A_{{\mathrm{1}}}} $ and $ \Gamma'_{{\mathrm{2}}}  \vdash  \ottnt{a_{{\mathrm{2}}}}  :^{  1  }   \ottnt{A_{{\mathrm{2}}}}  \{  \ottnt{a_{{\mathrm{1}}}}  /  \ottmv{x}  \}  $ and $ \Gamma_{{\mathrm{1}}}  <:    (   \ottnt{q}  \cdot  \ottnt{r}   )   \cdot  \Gamma'_{{\mathrm{1}}}  $ and $ \Gamma_{{\mathrm{2}}}  <:   \ottnt{q}  \cdot  \Gamma'_{{\mathrm{2}}}  $.\\
Then, by the multiplication lemma, $  \ottnt{r}  \cdot  \Gamma'_{{\mathrm{1}}}   \vdash  \ottnt{a_{{\mathrm{1}}}}  :^{ \ottnt{r} }  \ottnt{A_{{\mathrm{1}}}} $.\\
This case, then, follows by \rref{PTS-Pair}.

\end{itemize}

\item \Rref{PTS-LetPair}. Have: $  \Gamma_{{\mathrm{1}}}  +  \Gamma_{{\mathrm{2}}}   \vdash   \mathbf{let}_{ \ottnt{q_{{\mathrm{0}}}} } \: (  \ottmv{x} ^{ \ottnt{r} } ,  \ottmv{y}  ) \: \mathbf{be} \:  \ottnt{a}  \: \mathbf{in} \:  \ottnt{b}   :^{ \ottnt{q} }   \ottnt{B}  \{  \ottnt{a}  /  \ottmv{z}  \}  $ where $  \Delta  ,   \ottmv{z}  :   \Sigma  \ottmv{x}  :^{ \ottnt{r} } \!  \ottnt{A_{{\mathrm{1}}}}  .  \ottnt{A_{{\mathrm{2}}}}     \vdash_{0}  \ottnt{B}  :   \ottmv{s}  $ and $ \Gamma_{{\mathrm{1}}}  \vdash  \ottnt{a}  :^{  \ottnt{q}  \cdot  \ottnt{q_{{\mathrm{0}}}}  }   \Sigma  \ottmv{x}  :^{ \ottnt{r} } \!  \ottnt{A_{{\mathrm{1}}}}  .  \ottnt{A_{{\mathrm{2}}}}  $ and $   \Gamma_{{\mathrm{2}}}  ,   \ottmv{x}  :^{  \ottnt{q}  \cdot   \ottnt{q_{{\mathrm{0}}}}  \cdot  \ottnt{r}   }  \ottnt{A_{{\mathrm{1}}}}    ,   \ottmv{y}  :^{  \ottnt{q}  \cdot  \ottnt{q_{{\mathrm{0}}}}  }  \ottnt{A_{{\mathrm{2}}}}    \vdash  \ottnt{b}  :^{ \ottnt{q} }   \ottnt{B}  \{   (   \ottmv{x}  ^{ \ottnt{r} } ,   \ottmv{y}   )   /  \ottmv{z}  \}  $.\\
There are two cases to consider:
\begin{itemize}

\item $ \ottnt{q}  =   \omega  $. Now, by lemma \ref{ind}, $   (   \Gamma_{{\mathrm{1}}}  +  \Gamma_{{\mathrm{2}}}   )  ^{ \{ 0 , \omega \} }   \vdash   \mathbf{let}_{ \ottnt{q_{{\mathrm{0}}}} } \: (  \ottmv{x} ^{ \ottnt{r} } ,  \ottmv{y}  ) \: \mathbf{be} \:  \ottnt{a}  \: \mathbf{in} \:  \ottnt{b}   :^{  \omega  }   \ottnt{B}  \{  \ottnt{a}  /  \ottmv{z}  \}  $.\\
By \rref{PTS-SubR}, $   (   \Gamma_{{\mathrm{1}}}  +  \Gamma_{{\mathrm{2}}}   )  ^{ \{ 0 , \omega \} }   \vdash   \mathbf{let}_{ \ottnt{q_{{\mathrm{0}}}} } \: (  \ottmv{x} ^{ \ottnt{r} } ,  \ottmv{y}  ) \: \mathbf{be} \:  \ottnt{a}  \: \mathbf{in} \:  \ottnt{b}   :^{  1  }   \ottnt{B}  \{  \ottnt{a}  /  \ottmv{z}  \}  $.\\
Next, $  \Gamma_{{\mathrm{1}}}  +  \Gamma_{{\mathrm{2}}}   <:    (   \Gamma_{{\mathrm{1}}}  +  \Gamma_{{\mathrm{2}}}   )  ^{ \{ 0 , \omega \} }   =   \omega   \cdot    (   \Gamma_{{\mathrm{1}}}  +  \Gamma_{{\mathrm{2}}}   )  ^{ \{ 0 , \omega \} }  $.\\
Note that in case of $ \mathbb{N}_{=}^{\omega} $ and $ \mathcal{Q}_{\text{Lin} } $, we have, $ \Gamma_{{\mathrm{1}}}  +  \Gamma_{{\mathrm{2}}}  =   (   \Gamma_{{\mathrm{1}}}  +  \Gamma_{{\mathrm{2}}}   )  ^{ \{ 0 , \omega \} } $. 

\item $ \ottnt{q}  \neq   \omega  $. By IH, $ \Gamma'_{{\mathrm{1}}}  \vdash  \ottnt{a}  :^{  1  }   \Sigma  \ottmv{x}  :^{ \ottnt{r} } \!  \ottnt{A_{{\mathrm{1}}}}  .  \ottnt{A_{{\mathrm{2}}}}  $ and $   \Gamma'_{{\mathrm{2}}}  ,   \ottmv{x}  :^{ \ottnt{r'_{{\mathrm{1}}}} }  \ottnt{A_{{\mathrm{1}}}}    ,   \ottmv{y}  :^{ \ottnt{r'_{{\mathrm{2}}}} }  \ottnt{A_{{\mathrm{2}}}}    \vdash  \ottnt{b}  :^{  1  }   \ottnt{B}  \{   (   \ottmv{x}  ^{ \ottnt{r} } ,   \ottmv{y}   )   /  \ottmv{z}  \}  $ where $ \Gamma_{{\mathrm{1}}}  <:    \ottnt{q}  \cdot  \ottnt{q_{{\mathrm{0}}}}   \cdot  \Gamma'_{{\mathrm{1}}}  $ and $ \Gamma_{{\mathrm{2}}}  <:   \ottnt{q}  \cdot  \Gamma'_{{\mathrm{2}}}  $ and $   \ottnt{q}  \cdot  \ottnt{q_{{\mathrm{0}}}}   \cdot  \ottnt{r}   <:   \ottnt{q}  \cdot  \ottnt{r'_{{\mathrm{1}}}}  $ and $  \ottnt{q}  \cdot  \ottnt{q_{{\mathrm{0}}}}   <:   \ottnt{q}  \cdot  \ottnt{r'_{{\mathrm{2}}}}  $.\\
Now, since $q \notin \{ 0 ,  \omega  \}$, so $  \ottnt{q_{{\mathrm{0}}}}  \cdot  \ottnt{r}   <:  \ottnt{r'_{{\mathrm{1}}}} $ and $ \ottnt{q_{{\mathrm{0}}}}  <:  \ottnt{r'_{{\mathrm{2}}}} $.\\ 
Therefore, by \rref{PTS-SubL}, $   \Gamma'_{{\mathrm{2}}}  ,   \ottmv{x}  :^{  \ottnt{q_{{\mathrm{0}}}}  \cdot  \ottnt{r}  }  \ottnt{A_{{\mathrm{1}}}}    ,   \ottmv{y}  :^{ \ottnt{q_{{\mathrm{0}}}} }  \ottnt{A_{{\mathrm{2}}}}    \vdash  \ottnt{b}  :^{  1  }   \ottnt{B}  \{   (   \ottmv{x}  ^{ \ottnt{r} } ,   \ottmv{y}   )   /  \ottmv{z}  \}  $.\\
Next, by the multiplication lemma, $  \ottnt{q_{{\mathrm{0}}}}  \cdot  \Gamma'_{{\mathrm{1}}}   \vdash  \ottnt{a}  :^{ \ottnt{q_{{\mathrm{0}}}} }   \Sigma  \ottmv{x}  :^{ \ottnt{r} } \!  \ottnt{A_{{\mathrm{1}}}}  .  \ottnt{A_{{\mathrm{2}}}}  $.\\
This case, then, follows by \rref{PTS-LetPair}.

\end{itemize}

\item The other cases follow similarly.

\end{itemize}

\end{proof}

%---------------------------------------------------------------------------------------------------

\begin{lemma}[Splitting] \label{BLwDSplitP}
If $ \Gamma  \vdash  \ottnt{a}  :^{  \ottnt{q_{{\mathrm{1}}}}  +  \ottnt{q_{{\mathrm{2}}}}  }  \ottnt{A} $, then there exists $\Gamma_{{\mathrm{1}}}$ and $\Gamma_{{\mathrm{2}}}$ such that $ \Gamma_{{\mathrm{1}}}  \vdash  \ottnt{a}  :^{ \ottnt{q_{{\mathrm{1}}}} }  \ottnt{A} $ and $ \Gamma_{{\mathrm{2}}}  \vdash  \ottnt{a}  :^{ \ottnt{q_{{\mathrm{2}}}} }  \ottnt{A} $ and $ \Gamma  =   \Gamma_{{\mathrm{1}}}  +  \Gamma_{{\mathrm{2}}}  $. 
\end{lemma}

\begin{proof}
If $ \ottnt{q_{{\mathrm{1}}}}  +  \ottnt{q_{{\mathrm{2}}}}  = 0$, then $\Gamma_{{\mathrm{1}}} :=   0   \cdot  \Gamma $ and $\Gamma_{{\mathrm{2}}} := \Gamma$.\\
Otherwise, by lemma \ref{BLwDFactP}, $\exists \Gamma'$ such that $ \Gamma'  \vdash  \ottnt{a}  :^{  1  }  \ottnt{A} $ and $ \Gamma  <:    (   \ottnt{q_{{\mathrm{1}}}}  +  \ottnt{q_{{\mathrm{2}}}}   )   \cdot  \Gamma'  $.\\
Then, $\Gamma =    (   \ottnt{q_{{\mathrm{1}}}}  +  \ottnt{q_{{\mathrm{2}}}}   )   \cdot  \Gamma'   +  \Gamma_{{\mathrm{0}}} $ for some $\Gamma_{{\mathrm{0}}}$. In case of $ \mathbb{N}_{=}^{\omega} $ and $ \mathcal{Q}_{\text{Lin} } $, we choose $\Gamma_{{\mathrm{0}}}$ such that $\Gamma_{{\mathrm{0}}} =  \Gamma_{{\mathrm{0}}} ^{ \{ 0 , \omega \} } $.\\
Now, by lemma \ref{BLwDMultP}, $  \ottnt{q_{{\mathrm{1}}}}  \cdot  \Gamma'   \vdash  \ottnt{a}  :^{ \ottnt{q_{{\mathrm{1}}}} }  \ottnt{A} $ and $  \ottnt{q_{{\mathrm{2}}}}  \cdot  \Gamma'   \vdash  \ottnt{a}  :^{ \ottnt{q_{{\mathrm{2}}}} }  \ottnt{A} $.\\
Next, $   \ottnt{q_{{\mathrm{2}}}}  \cdot  \Gamma'   +  \Gamma_{{\mathrm{0}}}   <:   \ottnt{q_{{\mathrm{2}}}}  \cdot  \Gamma'  $. Then, by \rref{PTS-SubL}, $   \ottnt{q_{{\mathrm{2}}}}  \cdot  \Gamma'   +  \Gamma_{{\mathrm{0}}}   \vdash  \ottnt{a}  :^{ \ottnt{q_{{\mathrm{2}}}} }  \ottnt{A} $.\\
The lemma follows by setting $\Gamma_{{\mathrm{1}}} :=  \ottnt{q_{{\mathrm{1}}}}  \cdot  \Gamma' $ and $\Gamma_{{\mathrm{2}}} :=   \ottnt{q_{{\mathrm{2}}}}  \cdot  \Gamma'   +  \Gamma_{{\mathrm{0}}} $.
\end{proof}

%-------------------------------------------------------------------------------------------

\begin{lemma}[Weakening] \label{BLwDWeakP}
If $  \Gamma_{{\mathrm{1}}}  ,  \Gamma_{{\mathrm{2}}}   \vdash  \ottnt{a}  :^{ \ottnt{q} }  \ottnt{A} $ and $ \Delta_{{\mathrm{1}}}  \vdash_{0}  \ottnt{C}  :   \ottmv{s}  $ and $  \lfloor  \Gamma_{{\mathrm{1}}}  \rfloor   =  \Delta_{{\mathrm{1}}} $, then $    \Gamma_{{\mathrm{1}}}  ,   \ottmv{z}  :^{  0  }  \ottnt{C}     ,  \Gamma_{{\mathrm{2}}}   \vdash  \ottnt{a}  :^{ \ottnt{q} }  \ottnt{A} $.
\end{lemma}

\begin{proof}
By induction on $  \Gamma_{{\mathrm{1}}}  ,  \Gamma_{{\mathrm{2}}}   \vdash  \ottnt{a}  :^{ \ottnt{q} }  \ottnt{A} $.
\end{proof}

%---------------------------------------------------------------------------------------------------

\begin{lemma}[Substitution] \label{BLwDSubstP}
If $    \Gamma_{{\mathrm{1}}}  ,   \ottmv{z}  :^{ \ottnt{r_{{\mathrm{0}}}} }  \ottnt{C}     ,  \Gamma_{{\mathrm{2}}}   \vdash  \ottnt{a}  :^{ \ottnt{q} }  \ottnt{A} $ and $ \Gamma  \vdash  \ottnt{c}  :^{ \ottnt{r_{{\mathrm{0}}}} }  \ottnt{C} $ and $  \lfloor  \Gamma_{{\mathrm{1}}}  \rfloor   =   \lfloor  \Gamma  \rfloor  $, then $     \Gamma_{{\mathrm{1}}}  +  \Gamma    ,  \Gamma_{{\mathrm{2}}}   \{  \ottnt{c}  /  \ottmv{z}  \}   \vdash   \ottnt{a}  \{  \ottnt{c}  /  \ottmv{z}  \}   :^{ \ottnt{q} }   \ottnt{A}  \{  \ottnt{c}  /  \ottmv{z}  \}  $. 
\end{lemma}

\begin{proof}
By induction on $    \Gamma_{{\mathrm{1}}}  ,   \ottmv{z}  :^{ \ottnt{r_{{\mathrm{0}}}} }  \ottnt{C}     ,  \Gamma_{{\mathrm{2}}}   \vdash  \ottnt{a}  :^{ \ottnt{q} }  \ottnt{A} $. All the cases other than \rref{PTS-LamOmega} are similar to those of lemma \ref{BLDSubstP}.

\begin{itemize}
\item \Rref{PTS-LamOmega}. Have: $    \ottnt{q_{{\mathrm{0}}}}  \cdot  \Gamma_{{\mathrm{1}}}   ,   \ottmv{z}  :^{  \ottnt{q_{{\mathrm{0}}}}  \cdot  \ottnt{r_{{\mathrm{0}}}}  }  \ottnt{C}    ,   \ottnt{q_{{\mathrm{0}}}}  \cdot  \Gamma_{{\mathrm{2}}}    \vdash   \lambda^{ \ottnt{r} }  \ottmv{x}  :  \ottnt{A}  .  \ottnt{b}   :^{  \ottnt{q_{{\mathrm{0}}}}  \cdot  \ottnt{q}  }   \Pi  \ottmv{x}  :^{ \ottnt{r} } \!  \ottnt{A}  .  \ottnt{B}  $ where $    \Gamma_{{\mathrm{1}}}  ,   \ottmv{z}  :^{ \ottnt{r_{{\mathrm{0}}}} }  \ottnt{C}    ,  \Gamma_{{\mathrm{2}}}   ,   \ottmv{x}  :^{  \ottnt{q}  \cdot  \ottnt{r}  }  \ottnt{A}    \vdash  \ottnt{b}  :^{ \ottnt{q} }  \ottnt{B} $ and $  \ottnt{q}  =   \omega    \Rightarrow   \ottnt{r}  =   \omega   $ and $ \ottnt{q_{{\mathrm{0}}}}  \neq   0  $.\\
Further, $ \Gamma  \vdash  \ottnt{c}  :^{  \ottnt{q_{{\mathrm{0}}}}  \cdot  \ottnt{r_{{\mathrm{0}}}}  }  \ottnt{C} $ where $  \lfloor  \Gamma  \rfloor   =   \lfloor  \Gamma_{{\mathrm{1}}}  \rfloor  $.\\
There are two cases to consider:
\begin{itemize}

\item $ \ottnt{r_{{\mathrm{0}}}}  =   0  $. Then, $   0   \cdot  \Gamma   \vdash  \ottnt{c}  :^{  0  }  \ottnt{C} $.\\
By IH, $     \Gamma_{{\mathrm{1}}}  ,  \Gamma_{{\mathrm{2}}}   \{  \ottnt{c}  /  \ottmv{z}  \}   ,   \ottmv{x}  :^{  \ottnt{q}  \cdot  \ottnt{r}  }  \ottnt{A}    \{  \ottnt{c}  /  \ottmv{z}  \}   \vdash   \ottnt{b}  \{  \ottnt{c}  /  \ottmv{z}  \}   :^{ \ottnt{q} }   \ottnt{B}  \{  \ottnt{c}  /  \ottmv{z}  \}  $.\\
This case, then, follows by \rref{PTS-LamOmega,PTS-SubL}.

\item $ \ottnt{r_{{\mathrm{0}}}}  \neq   0  $. Then $  \ottnt{q_{{\mathrm{0}}}}  \cdot  \ottnt{r_{{\mathrm{0}}}}   \neq   0  $.\\
By the factorization lemma, $\exists \Gamma'$ such that $ \Gamma'  \vdash  \ottnt{c}  :^{  1  }  \ottnt{C} $ where $ \Gamma  <:    \ottnt{q_{{\mathrm{0}}}}  \cdot  \ottnt{r_{{\mathrm{0}}}}   \cdot  \Gamma'  $.\\
By the multiplication lemma, $  \ottnt{r_{{\mathrm{0}}}}  \cdot  \Gamma'   \vdash  \ottnt{c}  :^{ \ottnt{r_{{\mathrm{0}}}} }  \ottnt{C} $.\\
By IH, $      \Gamma_{{\mathrm{1}}}  +   \ottnt{r_{{\mathrm{0}}}}  \cdot  \Gamma'    ,  \Gamma_{{\mathrm{2}}}   \{  \ottnt{c}  /  \ottmv{z}  \}   ,   \ottmv{x}  :^{  \ottnt{q}  \cdot  \ottnt{r}  }  \ottnt{A}    \{  \ottnt{c}  /  \ottmv{z}  \}   \vdash   \ottnt{b}  \{  \ottnt{c}  /  \ottmv{z}  \}   :^{ \ottnt{q} }   \ottnt{B}  \{  \ottnt{c}  /  \ottmv{z}  \}  $.\\
By \rref{PTS-LamOmega}, $     \ottnt{q_{{\mathrm{0}}}}  \cdot  \Gamma_{{\mathrm{1}}}   +    \ottnt{q_{{\mathrm{0}}}}  \cdot  \ottnt{r_{{\mathrm{0}}}}   \cdot  \Gamma'    ,   \ottnt{q_{{\mathrm{0}}}}  \cdot  \Gamma_{{\mathrm{2}}}    \{  \ottnt{c}  /  \ottmv{z}  \}   \vdash    \lambda^{ \ottnt{r} }  \ottmv{x}  :   \ottnt{A}  \{  \ottnt{c}  /  \ottmv{z}  \}   .  \ottnt{b}   \{  \ottnt{c}  /  \ottmv{z}  \}   :^{  \ottnt{q_{{\mathrm{0}}}}  \cdot  \ottnt{q}  }    \Pi  \ottmv{x}  :^{ \ottnt{r} } \!   \ottnt{A}  \{  \ottnt{c}  /  \ottmv{z}  \}   .  \ottnt{B}   \{  \ottnt{c}  /  \ottmv{z}  \}  $.\\
This case, then, follows by \rref{PTS-SubL}.

\end{itemize}

\end{itemize}
\end{proof}

%---------------------------------------------------------------------------------------------------

\begin{lemma}[Lambda Inversion] \label{lambdainv}
If $ \Gamma  \vdash   \lambda^{ \ottnt{r} }  \ottmv{x}  :  \ottnt{A}  .  \ottnt{b}   :^{ \ottnt{q} }  \ottnt{C} $, then $\exists \ottnt{A'} , \ottnt{B'}, \Gamma_{{\mathrm{0}}}, \ottnt{q_{{\mathrm{0}}}}, \ottnt{q_{{\mathrm{1}}}}$ such that:
\begin{itemize}
\item $  \Gamma_{{\mathrm{0}}}  ,   \ottmv{x}  :^{  \ottnt{q_{{\mathrm{0}}}}  \cdot  \ottnt{r}  }  \ottnt{A'}    \vdash  \ottnt{b}  :^{ \ottnt{q_{{\mathrm{0}}}} }  \ottnt{B'} $ 
\item $ \Gamma  <:   \ottnt{q_{{\mathrm{1}}}}  \cdot  \Gamma_{{\mathrm{0}}}  $ and $  \ottnt{q_{{\mathrm{1}}}}  \cdot  \ottnt{q_{{\mathrm{0}}}}   <:  \ottnt{q} $
\item $  \ottnt{q_{{\mathrm{0}}}}  =   \omega    \Rightarrow   \ottnt{r}  =   \omega   $ and $ \ottnt{q_{{\mathrm{1}}}}  \neq   0  $
\item  $ \ottnt{A}  =_{\beta}  \ottnt{A'} $ and $ \ottnt{C}  =_{\beta}   \Pi  \ottmv{x}  :^{ \ottnt{r} } \!  \ottnt{A'}  .  \ottnt{B'}  $ and $ \Delta  \vdash_{0}  \ottnt{C}  :   \ottmv{s}  $, where $ \Delta  =   \lfloor  \Gamma  \rfloor  $.
\end{itemize}
\end{lemma}

\begin{proof}
By induction on $ \Gamma  \vdash   \lambda^{ \ottnt{r} }  \ottmv{x}  :  \ottnt{A}  .  \ottnt{b}   :^{ \ottnt{q} }  \ottnt{C} $.
\end{proof}

%----------------------------------------------------------------------------------------------------

\begin{theorem}[Preservation] \label{Dwpreserve}
If $ \Gamma  \vdash  \ottnt{a}  :^{ \ottnt{q} }  \ottnt{A} $ and $ \vdash  \ottnt{a}  \leadsto  \ottnt{a'} $, then $ \Gamma  \vdash  \ottnt{a'}  :^{ \ottnt{q} }  \ottnt{A} $.
\end{theorem}

\begin{proof}
By induction on $ \Gamma  \vdash  \ottnt{a}  :^{ \ottnt{q} }  \ottnt{A} $ and inversion on $ \vdash  \ottnt{a}  \leadsto  \ottnt{a'} $. All the cases other than \rref{PTS-App} are similar to those of lemma \ref{Dpreserve}.

\begin{itemize}
\item \Rref{PTS-App}. Have: $  \Gamma_{{\mathrm{1}}}  +  \Gamma_{{\mathrm{2}}}   \vdash   \ottnt{b}  \:  \ottnt{a} ^{ \ottnt{r} }   :^{ \ottnt{q} }   \ottnt{B}  \{  \ottnt{a}  /  \ottmv{x}  \}  $ where $ \Gamma_{{\mathrm{1}}}  \vdash  \ottnt{b}  :^{ \ottnt{q} }   \Pi  \ottmv{x}  :^{ \ottnt{r} } \!  \ottnt{A}  .  \ottnt{B}  $ and $ \Gamma_{{\mathrm{2}}}  \vdash  \ottnt{a}  :^{  \ottnt{q}  \cdot  \ottnt{r}  }  \ottnt{A} $. \\ Let $ \vdash   \ottnt{b}  \:  \ottnt{a} ^{ \ottnt{r} }   \leadsto  \ottnt{c} $. By inversion:

\begin{itemize}

\item $ \vdash   \ottnt{b}  \:  \ottnt{a} ^{ \ottnt{r} }   \leadsto   \ottnt{b'}  \:  \ottnt{a} ^{ \ottnt{r} }  $, when $ \vdash  \ottnt{b}  \leadsto  \ottnt{b'} $. \\
Need to show: $  \Gamma_{{\mathrm{1}}}  +  \Gamma_{{\mathrm{2}}}   \vdash   \ottnt{b'}  \:  \ottnt{a} ^{ \ottnt{r} }   :^{ \ottnt{q} }   \ottnt{B}  \{  \ottnt{a}  /  \ottmv{x}  \}  $.\\
Follows by IH and \rref{PTS-App}.

\item $\ottnt{b} =  \lambda^{ \ottnt{r} }  \ottmv{x}  :  \ottnt{A'}  .  \ottnt{b'} $ and $ \vdash   \ottnt{b}  \:  \ottnt{a} ^{ \ottnt{r} }   \leadsto   \ottnt{b'}  \{  \ottnt{a}  /  \ottmv{x}  \}  $.\\
Need to show: $  \Gamma_{{\mathrm{1}}}  +  \Gamma_{{\mathrm{2}}}   \vdash   \ottnt{b'}  \{  \ottnt{a}  /  \ottmv{x}  \}   :^{ \ottnt{q} }   \ottnt{B}  \{  \ottnt{a}  /  \ottmv{x}  \}  $.\\
Applying lemma \ref{lambdainv} on $ \Gamma_{{\mathrm{1}}}  \vdash   \lambda^{ \ottnt{r} }  \ottmv{x}  :  \ottnt{A'}  .  \ottnt{b'}   :^{ \ottnt{q} }   \Pi  \ottmv{x}  :^{ \ottnt{r} } \!  \ottnt{A}  .  \ottnt{B}  $, we get:
\begin{itemize}
\item $  \Gamma_{{\mathrm{0}}}  ,   \ottmv{x}  :^{  \ottnt{q_{{\mathrm{0}}}}  \cdot  \ottnt{r}  }  \ottnt{A''}    \vdash  \ottnt{b'}  :^{ \ottnt{q_{{\mathrm{0}}}} }  \ottnt{B'} $ 
\item $ \Gamma_{{\mathrm{1}}}  <:   \ottnt{q_{{\mathrm{1}}}}  \cdot  \Gamma_{{\mathrm{0}}}  $ and $  \ottnt{q_{{\mathrm{1}}}}  \cdot  \ottnt{q_{{\mathrm{0}}}}   <:  \ottnt{q} $
\item $  \ottnt{q_{{\mathrm{0}}}}  =   \omega    \Rightarrow   \ottnt{r}  =   \omega   $ and $ \ottnt{q_{{\mathrm{1}}}}  \neq   0  $
\item $ \ottnt{A''}  =_{\beta}  \ottnt{A'} $ and $  \Pi  \ottmv{x}  :^{ \ottnt{r} } \!  \ottnt{A}  .  \ottnt{B}   =_{\beta}   \Pi  \ottmv{x}  :^{ \ottnt{r} } \!  \ottnt{A''}  .  \ottnt{B'}  $.
\end{itemize}

Now, there are three cases to consider.
\begin{itemize}

\item $\ottnt{q_{{\mathrm{0}}}} = 0$. Since $  \ottnt{q_{{\mathrm{1}}}}  \cdot  \ottnt{q_{{\mathrm{0}}}}   <:  \ottnt{q} $, so $\ottnt{q} = 0$.\\
Then, by the substitution lemma and \rref{PTS-SubL}, $  \Gamma_{{\mathrm{1}}}  +  \Gamma_{{\mathrm{2}}}   \vdash   \ottnt{b'}  \{  \ottnt{a}  /  \ottmv{x}  \}   :^{  0  }   \ottnt{B}  \{  \ottnt{a}  /  \ottmv{x}  \}  $. 

\item $ \ottnt{q_{{\mathrm{0}}}}  =   \omega  $. So $ \ottnt{r}  =   \omega  $.\\
 Now, if $ \ottnt{q}  =   0  $, then by the substitution lemma and \rref{PTS-SubL}, $  \Gamma_{{\mathrm{1}}}  +  \Gamma_{{\mathrm{2}}}   \vdash   \ottnt{b'}  \{  \ottnt{a}  /  \ottmv{x}  \}   :^{  0  }   \ottnt{B}  \{  \ottnt{a}  /  \ottmv{x}  \}  $. \\
 If $ \ottnt{q}  \neq   0  $, then $  \ottnt{q}  \cdot  \ottnt{r}   =   \omega  $. So $ \Gamma_{{\mathrm{2}}}  \vdash  \ottnt{a}  :^{  \omega  }  \ottnt{A} $. \\ %By \rref{PTS-Conv}, $ \Gamma_{{\mathrm{2}}}  \vdash  \ottnt{a}  :^{  \omega  }  \ottnt{A''} $.\\
 Next, $  \Gamma_{{\mathrm{0}}}  ,   \ottmv{x}  :^{  \omega  }  \ottnt{A''}    \vdash  \ottnt{b'}  :^{  \omega  }  \ottnt{B'} $.\\
 By the multiplication lemma and \rref{PTS-SubL}, $  \Gamma_{{\mathrm{1}}}  ,   \ottmv{x}  :^{  \omega  }  \ottnt{A''}    \vdash  \ottnt{b'}  :^{  \omega  }  \ottnt{B'} $.\\  
 By the substitution lemma, $  \Gamma_{{\mathrm{1}}}  +  \Gamma_{{\mathrm{2}}}   \vdash   \ottnt{b'}  \{  \ottnt{a}  /  \ottmv{x}  \}   :^{  \omega  }   \ottnt{B'}  \{  \ottnt{a}  /  \ottmv{x}  \}  $.\\
 This case, then, follows by \rref{PTS-Conv,PTS-SubR}.

\item $\ottnt{q_{{\mathrm{0}}}} \notin \{ 0 ,  \omega  \}$. By lemma \ref{BLwDFactP}, $\exists \Gamma'_{{\mathrm{0}}}$ and $\ottnt{r'}$ such that $  \Gamma'_{{\mathrm{0}}}  ,   \ottmv{x}  :^{ \ottnt{r'} }  \ottnt{A''}    \vdash  \ottnt{b'}  :^{  1  }  \ottnt{B'} $ and $ \Gamma_{{\mathrm{0}}}  <:   \ottnt{q_{{\mathrm{0}}}}  \cdot  \Gamma'_{{\mathrm{0}}}  $ and $  \ottnt{q_{{\mathrm{0}}}}  \cdot  \ottnt{r}   <:   \ottnt{q_{{\mathrm{0}}}}  \cdot  \ottnt{r'}  $.\\
Since $\ottnt{q_{{\mathrm{0}}}} \notin \{ 0 ,  \omega  \} $, so $ \ottnt{r}  <:  \ottnt{r'} $. Hence, by \rref{PTS-SubL}, $  \Gamma'_{{\mathrm{0}}}  ,   \ottmv{x}  :^{ \ottnt{r} }  \ottnt{A''}    \vdash  \ottnt{b'}  :^{  1  }  \ottnt{B'} $.\\
Now, by lemma \ref{BLwDMultP}, $   \ottnt{q}  \cdot  \Gamma'_{{\mathrm{0}}}   ,   \ottmv{x}  :^{  \ottnt{q}  \cdot  \ottnt{r}  }  \ottnt{A''}    \vdash  \ottnt{b'}  :^{ \ottnt{q} }  \ottnt{B'} $. By \rref{PTS-SubL}, $  \Gamma_{{\mathrm{1}}}  ,   \ottmv{x}  :^{  \ottnt{q}  \cdot  \ottnt{r}  }  \ottnt{A''}    \vdash  \ottnt{b'}  :^{ \ottnt{q} }  \ottnt{B'} $.\\
This case, then, follows by \rref{PTS-Conv} and the substitution lemma.
\end{itemize}

\end{itemize} 
\end{itemize}
\end{proof}

%----------------------------------------------------------------------------------------------------

\begin{theorem}[Progress] \label{Dwprogress}
If $  \emptyset   \vdash  \ottnt{a}  :^{ \ottnt{q} }  \ottnt{A} $, then either $\ottnt{a}$ is a value or there exists $\ottnt{a'}$ such that $ \vdash  \ottnt{a}  \leadsto  \ottnt{a'} $.
\end{theorem}

\begin{proof}
By induction on $  \emptyset   \vdash  \ottnt{a}  :^{ \ottnt{q} }  \ottnt{A} $. Follow the proof of theorem \ref{Dprogress}.
\end{proof}

%----------------------------------------------------------------------------------------------------

\textbf{Note:} The multi-substitution and elaboration lemmas (Lemmas \ref{MultiSubst} and \ref{Elaboration}) are true of LDC($ \mathcal{Q}_{\mathbb{N} }^{\omega} $).

%----------------------------------------------------------------------------------------------------

\begin{lemma} \label{heaphelperDepw}
If $ \ottnt{H}  \models   \Gamma_{{\mathrm{1}}}  +  \Gamma_{{\mathrm{2}}}  $ and $ \Gamma_{{\mathrm{2}}}  \vdash  \ottnt{a}  :^{ \ottnt{q} }  \ottnt{A} $ and $q \neq 0$, then either $\ottnt{a}$ is a value or there exists $\ottnt{H'}, \Gamma'_{{\mathrm{2}}}, \ottnt{a'}$ such that:
\begin{itemize}
\item $ [  \ottnt{H}  ]  \ottnt{a}  \Longrightarrow^{ \ottnt{q} }_{ \ottnt{S} } [  \ottnt{H'}  ]  \ottnt{a'} $
\item $ \ottnt{H'}  \models   \Gamma_{{\mathrm{1}}}  +  \Gamma'_{{\mathrm{2}}}  $
\item $ \Gamma'_{{\mathrm{2}}}  \vdash  \ottnt{a'}  :^{ \ottnt{q} }  \ottnt{A} $
\end{itemize} 
\end{lemma}

\begin{proof}
By induction on $ \Gamma_{{\mathrm{2}}}  \vdash  \ottnt{a}  :^{ \ottnt{q} }  \ottnt{A} $. All the cases other than \rref{PTS-App} are similar to those of lemma \ref{heaphelperDep}.

\begin{itemize}

\item \Rref{PTS-App}. Have: $  \Gamma_{{\mathrm{21}}}  +  \Gamma_{{\mathrm{22}}}   \vdash   \ottnt{b}  \:  \ottnt{a} ^{ \ottnt{r} }   :^{ \ottnt{q} }   \ottnt{B}  \{  \ottnt{a}  /  \ottmv{x}  \}  $ where $ \Gamma_{{\mathrm{21}}}  \vdash  \ottnt{b}  :^{ \ottnt{q} }   \Pi  \ottmv{x}  :^{ \ottnt{r} } \!  \ottnt{A}  .  \ottnt{B}  $ and $ \Gamma_{{\mathrm{22}}}  \vdash  \ottnt{a}  :^{  \ottnt{q}  \cdot  \ottnt{r}  }  \ottnt{A} $.\\
Further, $ \ottnt{H}  \models   \Gamma_{{\mathrm{1}}}  +   (   \Gamma_{{\mathrm{21}}}  +  \Gamma_{{\mathrm{22}}}   )   $.\\
If $\ottnt{b}$ steps, then this case follows by IH.\\

Otherwise, $\ottnt{b}$ is a value. By inversion, $b =  \lambda^{ \ottnt{r} }  \ottmv{x}  :  \ottnt{A'}  .  \ottnt{b'} $.\\
Using lemma \ref{lambdainv} on $ \Gamma_{{\mathrm{21}}}  \vdash   \lambda^{ \ottnt{r} }  \ottmv{x}  :  \ottnt{A'}  .  \ottnt{b'}   :^{ \ottnt{q} }   \Pi  \ottmv{x}  :^{ \ottnt{r} } \!  \ottnt{A}  .  \ottnt{B}  $, we get:
\begin{itemize}
\item $  \Gamma_{{\mathrm{20}}}  ,   \ottmv{x}  :^{  \ottnt{q_{{\mathrm{0}}}}  \cdot  \ottnt{r}  }  \ottnt{A''}    \vdash  \ottnt{b'}  :^{ \ottnt{q_{{\mathrm{0}}}} }  \ottnt{B'} $ 
\item $ \Gamma_{{\mathrm{21}}}  <:   \ottnt{q_{{\mathrm{1}}}}  \cdot  \Gamma_{{\mathrm{20}}}  $ and $  \ottnt{q_{{\mathrm{1}}}}  \cdot  \ottnt{q_{{\mathrm{0}}}}   <:  \ottnt{q} $
\item $  \ottnt{q_{{\mathrm{0}}}}  =   \omega    \Rightarrow   \ottnt{r}  =   \omega   $ and $ \ottnt{q_{{\mathrm{1}}}}  \neq   0  $
\item $  \ottnt{A''}  \{  \Delta_{{\mathrm{2}}}  \}   =_{\beta}   \ottnt{A'}  \{  \Delta_{{\mathrm{2}}}  \}  $ and $   \Pi  \ottmv{x}  :^{ \ottnt{r} } \!  \ottnt{A}  .  \ottnt{B}   \{  \Delta_{{\mathrm{2}}}  \}   =_{\beta}    \Pi  \ottmv{x}  :^{ \ottnt{r} } \!  \ottnt{A''}  .  \ottnt{B'}   \{  \Delta_{{\mathrm{2}}}  \}  $ where $ \Delta_{{\mathrm{2}}}  =   \lfloor  \Gamma_{{\mathrm{21}}}  \rfloor  $.
\end{itemize}

Now, there are two cases to consider.
\begin{itemize} 

\item $ \ottnt{q_{{\mathrm{0}}}}  =   \omega  $. So $ \ottnt{r}  =   \omega  $. Further, since $ \ottnt{q}  \neq   0  $, so $  \ottnt{q}  \cdot  \ottnt{r}   =   \omega  $.\\
 So, $  \Gamma_{{\mathrm{20}}}  ,   \ottmv{x}  :^{  \omega  }  \ottnt{A''}    \vdash  \ottnt{b'}  :^{  \omega  }  \ottnt{B'} $.\\
 By the multiplication lemma and \rref{PTS-SubL}, $  \Gamma_{{\mathrm{21}}}  ,   \ottmv{x}  :^{  \omega  }  \ottnt{A''}    \vdash  \ottnt{b'}  :^{  \omega  }  \ottnt{B'} $.\\
 Then, by \rref{PTS-DefConv}, $ \Gamma_{{\mathrm{22}}}  \vdash  \ottnt{a}  :^{  \omega  }  \ottnt{A''} $.\\
 Further, $  \Gamma_{{\mathrm{21}}}  ,   \ottmv{x}  =  \ottnt{a}  :^{  \omega  }  \ottnt{A''}    \vdash  \ottnt{b'}  :^{  \omega  }  \ottnt{B'} $. And by \rref{PTS-SubR}, $  \Gamma_{{\mathrm{21}}}  ,   \ottmv{x}  =  \ottnt{a}  :^{  \omega  }  \ottnt{A''}    \vdash  \ottnt{b'}  :^{ \ottnt{q} }  \ottnt{B'} $.\\
 Now,
\begin{itemize}
\item $ [  \ottnt{H}  ]    (   \lambda^{ \ottnt{r} }  \ottmv{x}  :  \ottnt{A'}  .  \ottnt{b'}   )   \:  \ottnt{a} ^{ \ottnt{r} }   \Longrightarrow^{ \ottnt{q} }_{ \ottnt{S} } [   \ottnt{H}  ,   \ottmv{x}  \overset{  \omega  }{\mapsto}  \ottnt{a}    ]  \ottnt{b'} $
\item $  \ottnt{H}  ,   \ottmv{x}  \overset{  \omega  }{\mapsto}  \ottnt{a}    \models    (   \Gamma_{{\mathrm{1}}}  +  \Gamma_{{\mathrm{21}}}   )   ,   \ottmv{x}  =  \ottnt{a}  :^{  \omega  }  \ottnt{A''}   $ 
\item $  \Gamma_{{\mathrm{21}}}  ,   \ottmv{x}  =  \ottnt{a}  :^{  \omega  }  \ottnt{A''}    \vdash  \ottnt{b'}  :^{ \ottnt{q} }   \ottnt{B}  \{  \ottnt{a}  /  \ottmv{x}  \}  $ by \rref{PTS-DefConv}. 
\end{itemize} 

\item $ \ottnt{q_{{\mathrm{0}}}}  \neq   \omega  $. Note that if $ \ottnt{q_{{\mathrm{0}}}}  =   0  $, then $ \ottnt{q}  =   0  $, a contradiction. So $\ottnt{q_{{\mathrm{0}}}} \notin \{ 0 ,  \omega  \}$. \\
By lemma \ref{BLwDFactP}, $\exists \Gamma'_{{\mathrm{20}}}$ and $\ottnt{r'}$ such that $  \Gamma'_{{\mathrm{20}}}  ,   \ottmv{x}  :^{ \ottnt{r'} }  \ottnt{A''}    \vdash  \ottnt{b'}  :^{  1  }  \ottnt{B'} $ and $ \Gamma_{{\mathrm{20}}}  <:   \ottnt{q_{{\mathrm{0}}}}  \cdot  \Gamma'_{{\mathrm{20}}}  $ and $  \ottnt{q_{{\mathrm{0}}}}  \cdot  \ottnt{r}   <:   \ottnt{q_{{\mathrm{0}}}}  \cdot  \ottnt{r'}  $.\\
Since $\ottnt{q_{{\mathrm{0}}}} \notin \{ 0 ,  \omega  \} $, therefore $ \ottnt{r}  <:  \ottnt{r'} $. Hence, by \rref{PTS-SubL}, $  \Gamma'_{{\mathrm{20}}}  ,   \ottmv{x}  :^{ \ottnt{r} }  \ottnt{A''}    \vdash  \ottnt{b'}  :^{  1  }  \ottnt{B'} $.\\
Now, by lemma \ref{BLwDMultP}, $   \ottnt{q}  \cdot  \Gamma'_{{\mathrm{20}}}   ,   \ottmv{x}  :^{  \ottnt{q}  \cdot  \ottnt{r}  }  \ottnt{A''}    \vdash  \ottnt{b'}  :^{ \ottnt{q} }  \ottnt{B'} $. By \rref{PTS-SubL}, $  \Gamma_{{\mathrm{21}}}  ,   \ottmv{x}  :^{  \ottnt{q}  \cdot  \ottnt{r}  }  \ottnt{A''}    \vdash  \ottnt{b'}  :^{ \ottnt{q} }  \ottnt{B'} $.\\
Next, by \rref{PTS-DefConv}, $ \Gamma_{{\mathrm{22}}}  \vdash  \ottnt{a}  :^{  \ottnt{q}  \cdot  \ottnt{r}  }  \ottnt{A''} $. Then, $  \Gamma_{{\mathrm{21}}}  ,   \ottmv{x}  =  \ottnt{a}  :^{  \ottnt{q}  \cdot  \ottnt{r}  }  \ottnt{A''}    \vdash  \ottnt{b'}  :^{ \ottnt{q} }  \ottnt{B'} $.\\
 Therefore,
\begin{itemize}
\item $ [  \ottnt{H}  ]    (   \lambda^{ \ottnt{r} }  \ottmv{x}  :  \ottnt{A'}  .  \ottnt{b'}   )   \:  \ottnt{a} ^{ \ottnt{r} }   \Longrightarrow^{ \ottnt{q} }_{ \ottnt{S} } [   \ottnt{H}  ,   \ottmv{x}  \overset{  \ottnt{q}  \cdot  \ottnt{r}  }{\mapsto}  \ottnt{a}    ]  \ottnt{b'} $
\item $  \ottnt{H}  ,   \ottmv{x}  \overset{  \ottnt{q}  \cdot  \ottnt{r}  }{\mapsto}  \ottnt{a}    \models    (   \Gamma_{{\mathrm{1}}}  +  \Gamma_{{\mathrm{21}}}   )   ,   \ottmv{x}  =  \ottnt{a}  :^{  \ottnt{q}  \cdot  \ottnt{r}  }  \ottnt{A''}   $ 
\item $  \Gamma_{{\mathrm{21}}}  ,   \ottmv{x}  =  \ottnt{a}  :^{  \ottnt{q}  \cdot  \ottnt{r}  }  \ottnt{A''}    \vdash  \ottnt{b'}  :^{ \ottnt{q} }   \ottnt{B}  \{  \ottnt{a}  /  \ottmv{x}  \}  $ by \rref{PTS-DefConv}. 
\end{itemize}

\end{itemize}
\end{itemize}

\end{proof}

%-------------------------------------------------------------------------------------------------

\begin{theorem}[Soundness] \label{Dwheap}
If $ \ottnt{H}  \models  \Gamma $ and $ \Gamma  \vdash  \ottnt{a}  :^{ \ottnt{q} }  \ottnt{A} $ and $q \neq 0$, then either $\ottnt{a}$ is a value or there exists $\ottnt{H'}, \Gamma', \ottnt{a'}$ such that $ [  \ottnt{H}  ]  \ottnt{a}  \Longrightarrow^{ \ottnt{q} }_{ \ottnt{S} } [  \ottnt{H'}  ]  \ottnt{a'} $ and $ \ottnt{H'}  \models  \Gamma' $ and $ \Gamma'  \vdash  \ottnt{a'}  :^{ \ottnt{q} }  \ottnt{A} $.
\end{theorem} 

\begin{proof}
Use lemma \ref{heaphelperDepw} with $\Gamma_{{\mathrm{1}}} :=   0   \cdot  \Gamma $ and $\Gamma_{{\mathrm{2}}} := \Gamma$.
\end{proof}

%-------------------------------------------------------------------------------------------------

\begin{theorem}[Theorem \ref{Bwsound}]
LDC($ \mathcal{Q}_{\mathbb{N} }^{\omega} $) satisfies type soundness and heap soundness.
\end{theorem}

\begin{proof}
Follows by theorems \ref{Dwpreserve}, \ref{Dwprogress} and \ref{Dwheap}.
\end{proof}

%--------------------------------------------------------------------------------------------------

%--------------------------------------------------------------------------------------------------
\section{Comparison of LDC with Other Calculi}
%--------------------------------------------------------------------------------------------------

\begin{theorem}[Theorem \ref{LNLTyping}]
The translation from LNL $\lambda$-calculus to LDC($ \mathcal{Q}_{\text{Lin} } $), shown in Figure \ref{termT}, is sound:
\begin{itemize}
\item If $ \Theta  ;  \Gamma  \vdash_{\mathcal{L} }   \ottmv{e}   :  \ottnt{A} $, then $     \overline{ \Theta }  ^{  \omega  }   ,   \overline{ \Gamma }   ^{  1  }   \vdash   \overline{  \ottmv{e}  }   :^{  1  }   \overline{ \ottnt{A} }  $. 
\item If $ \Theta  \vdash_{\mathcal{C} }   \ottmv{t}   :  \ottnt{X} $, then $   \overline{ \Theta }  ^{  \omega  }   \vdash   \overline{  \ottmv{t}  }   :^{  \omega  }   \overline{ \ottnt{X} }  $.
\item If $  \ottmv{e}   =_{\beta}  \ottnt{f} $, then $  \overline{  \ottmv{e}  }   =_{\beta}   \overline{ \ottnt{f} }  $. If $  \ottmv{s}   =_{\beta}   \ottmv{t}  $ then $  \overline{  \ottmv{s}  }   =_{\beta}   \overline{  \ottmv{t}  }  $.
\end{itemize}
\end{theorem}

\begin{proof}
By mutual induction on $ \Theta  ;  \Gamma  \vdash_{\mathcal{L} }   \ottmv{e}   :  \ottnt{A} $ and $ \Theta  \vdash_{\mathcal{C} }   \ottmv{t}   :  \ottnt{X} $ for typing.\\
By case analysis on $  \ottmv{e}   =_{\beta}  \ottnt{f} $ and $  \ottmv{s}   =_{\beta}   \ottmv{t}  $ for $\beta$-equality. 
\end{proof}

%----------------------------------------------------------------------------------------------------

\begin{theorem}[Theorem \ref{GraD}]
With $ \mathcal{Q}_{\mathbb{N} }^{\omega} $ as the parametrizing structure, if $ \Gamma  \vdash  \ottnt{a}  : \:  \ottnt{A} $ in \textsc{GraD}, then $ \Gamma  \vdash  \ottnt{a}  :^{  1  }  \ottnt{A} $ in LDC. Further, if $ \vdash  \ottnt{a}  \leadsto  \ottnt{a'} $ in \textsc{GraD}, then $ \vdash  \ottnt{a}  \leadsto  \ottnt{a'} $ in LDC.
\end{theorem} 

\begin{proof}
By induction on \textsc{GraD} typing judgment.
\end{proof}

%----------------------------------------------------------------------------------------------------

\begin{theorem}[Theorem \ref{DDCT}]
With $\mathcal{L}$ as the parametrizing structure, if $ \Gamma  \vdash  \ottnt{a}  :^{ \ell }  \ottnt{A} $ in $\text{DDC}^{\top}$, then $ \Gamma  \vdash  \ottnt{a}  :^{ \ell }  \ottnt{A} $ in LDC. Further, if $ \vdash  \ottnt{a}  \leadsto  \ottnt{a'} $ in $\text{DDC}^{\top}$, then $ \vdash  \ottnt{a}  \leadsto  \ottnt{a'} $ in LDC.
\end{theorem}

\begin{proof}
By induction on $\text{DDC}^{\top}$ typing judgment.
\end{proof}

%-----------------------------------------------------------------------------------------------------

\section{Derivations for Join and Fork in LDC($ \mathcal{L} $)} \label{ap:der}

\begin{proposition}
$  \emptyset   \vdash  \ottnt{c_{{\mathrm{1}}}}  : \:     T_{ \ell_{{\mathrm{1}}} } \:   T_{ \ell_{{\mathrm{2}}} } \:  \ottnt{A}     \to   T_{  \ell_{{\mathrm{1}}}  \: \sqcup \:  \ell_{{\mathrm{2}}}  } \:  \ottnt{A}   $ where \[ \ottnt{c_{{\mathrm{1}}}} :=  \lambda  \ottmv{x}  .   \eta_{  \ell_{{\mathrm{1}}}  \: \sqcup \:  \ell_{{\mathrm{2}}}  } \:   (   \mathbf{let} \: (  \ottmv{y} ^{ \ell_{{\mathrm{2}}} } , \_  ) \: \mathbf{be} \:   (   \mathbf{let} \: (  \ottmv{z} ^{ \ell_{{\mathrm{1}}} } , \_  ) \: \mathbf{be} \:   \ottmv{x}   \: \mathbf{in} \:   \ottmv{z}    )   \: \mathbf{in} \:   \ottmv{y}    )   , \text{ where }  \eta_{ \ell } \:  \ottnt{a}  \triangleq  (  \ottnt{a} ^{ \ell } ,   \mathbf{unit}   )  \] .
\end{proposition}

\begin{proof}
\begin{enumerate}
\setlength\itemsep{5pt}
\item \label{31} $  \ottmv{x}  :^{  \bot  }   T_{ \ell_{{\mathrm{1}}} } \:   T_{ \ell_{{\mathrm{2}}} } \:  \ottnt{A}     \vdash   \ottmv{x}   :^{ \ell_{{\mathrm{1}}} }   T_{ \ell_{{\mathrm{1}}} } \:   T_{ \ell_{{\mathrm{2}}} } \:  \ottnt{A}   $ [By \rref{ST-VarD,ST-SubRD}]
\item \label{32} $   \ottmv{x}  :^{  \bot  }   T_{ \ell_{{\mathrm{1}}} } \:   T_{ \ell_{{\mathrm{2}}} } \:  \ottnt{A}     ,   \ottmv{z}  :^{ \ell_{{\mathrm{1}}} }   T_{ \ell_{{\mathrm{2}}} } \:  \ottnt{A}     \vdash   \ottmv{z}   :^{ \ell_{{\mathrm{1}}} }   T_{ \ell_{{\mathrm{2}}} } \:  \ottnt{A}  $
\item \label{33} $  \ottmv{x}  :^{  \bot  }   T_{ \ell_{{\mathrm{1}}} } \:   T_{ \ell_{{\mathrm{2}}} } \:  \ottnt{A}     \vdash   \mathbf{let} \: (  \ottmv{z} ^{ \ell_{{\mathrm{1}}} } , \_  ) \: \mathbf{be} \:   \ottmv{x}   \: \mathbf{in} \:   \ottmv{z}    :^{ \ell_{{\mathrm{1}}} }   T_{ \ell_{{\mathrm{2}}} } \:  \ottnt{A}  $ [From (\ref{31}) and (\ref{32}), by \rref{ST-LetPairD}]
\item \label{34} $  \ottmv{x}  :^{  \bot  }   T_{ \ell_{{\mathrm{1}}} } \:   T_{ \ell_{{\mathrm{2}}} } \:  \ottnt{A}     \vdash   \mathbf{let} \: (  \ottmv{z} ^{ \ell_{{\mathrm{1}}} } , \_  ) \: \mathbf{be} \:   \ottmv{x}   \: \mathbf{in} \:   \ottmv{z}    :^{  \ell_{{\mathrm{1}}}  \: \sqcup \:  \ell_{{\mathrm{2}}}  }   T_{ \ell_{{\mathrm{2}}} } \:  \ottnt{A}  $ [From  (\ref{33}), by \rref{ST-SubRD}]
\item \label{35} $   \ottmv{x}  :^{  \bot  }   T_{ \ell_{{\mathrm{1}}} } \:   T_{ \ell_{{\mathrm{2}}} } \:  \ottnt{A}     ,   \ottmv{y}  :^{  \ell_{{\mathrm{1}}}  \: \sqcup \:  \ell_{{\mathrm{2}}}  }  \ottnt{A}    \vdash   \ottmv{y}   :^{  \ell_{{\mathrm{1}}}  \: \sqcup \:  \ell_{{\mathrm{2}}}  }  \ottnt{A} $
\item \label{36} $  \ottmv{x}  :^{  \bot  }   T_{ \ell_{{\mathrm{1}}} } \:   T_{ \ell_{{\mathrm{2}}} } \:  \ottnt{A}     \vdash   \mathbf{let} \: (  \ottmv{y} ^{ \ell_{{\mathrm{2}}} } , \_  ) \: \mathbf{be} \:   (   \mathbf{let} \: (  \ottmv{z} ^{ \ell_{{\mathrm{1}}} } , \_  ) \: \mathbf{be} \:   \ottmv{x}   \: \mathbf{in} \:   \ottmv{z}    )   \: \mathbf{in} \:   \ottmv{y}    :^{  \ell_{{\mathrm{1}}}  \: \sqcup \:  \ell_{{\mathrm{2}}}  }  \ottnt{A} $ [From (\ref{34}) and (\ref{35}), by \rref{ST-LetPairD}]
\item \label{37} $  \ottmv{x}  :^{  \bot  }   T_{ \ell_{{\mathrm{1}}} } \:   T_{ \ell_{{\mathrm{2}}} } \:  \ottnt{A}     \vdash   \eta_{  \ell_{{\mathrm{1}}}  \: \sqcup \:  \ell_{{\mathrm{2}}}  } \:   (   \mathbf{let} \: (  \ottmv{y} ^{ \ell_{{\mathrm{2}}} } , \_  ) \: \mathbf{be} \:   (   \mathbf{let} \: (  \ottmv{z} ^{ \ell_{{\mathrm{1}}} } , \_  ) \: \mathbf{be} \:   \ottmv{x}   \: \mathbf{in} \:   \ottmv{z}    )   \: \mathbf{in} \:   \ottmv{y}    )    :^{  \bot  }   T_{  \ell_{{\mathrm{1}}}  \: \sqcup \:  \ell_{{\mathrm{2}}}  } \:  \ottnt{A}  $
\end{enumerate}
\end{proof}

%-----------------------------------------------------------------------------------------------------

\begin{proposition}
$  \emptyset   \vdash  \ottnt{c_{{\mathrm{2}}}}  : \:     T_{  \ell_{{\mathrm{1}}}  \: \sqcup \:  \ell_{{\mathrm{2}}}  } \:  \ottnt{A}    \to   T_{ \ell_{{\mathrm{1}}} } \:   T_{ \ell_{{\mathrm{2}}} } \:  \ottnt{A}    $ where $$\ottnt{c_{{\mathrm{2}}}} :=  \lambda  \ottmv{x}  .   \eta_{ \ell_{{\mathrm{1}}} } \:   \eta_{ \ell_{{\mathrm{2}}} } \:   (   \mathbf{let} \: (  \ottmv{y} ^{  \ell_{{\mathrm{1}}}  \: \sqcup \:  \ell_{{\mathrm{2}}}  } , \_  ) \: \mathbf{be} \:   \ottmv{x}   \: \mathbf{in} \:   \ottmv{y}    )    , \text{ where }  \eta_{ \ell } \:  \ottnt{a}  \triangleq  (  \ottnt{a} ^{ \ell } ,   \mathbf{unit}   )  $$.
\end{proposition}

\begin{proof}
\begin{enumerate}
\setlength\itemsep{5pt}
\item \label{41} $  \ottmv{x}  :^{  \bot  }   T_{  \ell_{{\mathrm{1}}}  \: \sqcup \:  \ell_{{\mathrm{2}}}  } \:  \ottnt{A}    \vdash   \ottmv{x}   :^{  \ell_{{\mathrm{1}}}  \: \sqcup \:  \ell_{{\mathrm{2}}}  }   T_{  \ell_{{\mathrm{1}}}  \: \sqcup \:  \ell_{{\mathrm{2}}}  } \:  \ottnt{A}  $ [By \rref{ST-VarD,ST-SubRD}]
\item \label{42} $   \ottmv{x}  :^{  \bot  }   T_{  \ell_{{\mathrm{1}}}  \: \sqcup \:  \ell_{{\mathrm{2}}}  } \:  \ottnt{A}    ,   \ottmv{y}  :^{  \ell_{{\mathrm{1}}}  \: \sqcup \:  \ell_{{\mathrm{2}}}  }  \ottnt{A}    \vdash   \ottmv{y}   :^{  \ell_{{\mathrm{1}}}  \: \sqcup \:  \ell_{{\mathrm{2}}}  }  \ottnt{A} $
\item \label{43} $  \ottmv{x}  :^{  \bot  }   T_{  \ell_{{\mathrm{1}}}  \: \sqcup \:  \ell_{{\mathrm{2}}}  } \:  \ottnt{A}    \vdash   \mathbf{let} \: (  \ottmv{y} ^{  \ell_{{\mathrm{1}}}  \: \sqcup \:  \ell_{{\mathrm{2}}}  } , \_  ) \: \mathbf{be} \:   \ottmv{x}   \: \mathbf{in} \:   \ottmv{y}    :^{  \ell_{{\mathrm{1}}}  \: \sqcup \:  \ell_{{\mathrm{2}}}  }  \ottnt{A} $ [From  (\ref{41}) and (\ref{42}), by \rref{ST-LetPairD}]
\item \label{44} $  \ottmv{x}  :^{  \bot  }   T_{  \ell_{{\mathrm{1}}}  \: \sqcup \:  \ell_{{\mathrm{2}}}  } \:  \ottnt{A}    \vdash   \eta_{ \ell_{{\mathrm{2}}} } \:   (   \mathbf{let} \: (  \ottmv{y} ^{  \ell_{{\mathrm{1}}}  \: \sqcup \:  \ell_{{\mathrm{2}}}  } , \_  ) \: \mathbf{be} \:   \ottmv{x}   \: \mathbf{in} \:   \ottmv{y}    )    :^{ \ell_{{\mathrm{1}}} }   T_{ \ell_{{\mathrm{2}}} } \:  \ottnt{A}  $ 
\item \label{45} $  \ottmv{x}  :^{  \bot  }   T_{  \ell_{{\mathrm{1}}}  \: \sqcup \:  \ell_{{\mathrm{2}}}  } \:  \ottnt{A}    \vdash   \eta_{ \ell_{{\mathrm{1}}} } \:   \eta_{ \ell_{{\mathrm{2}}} } \:   (   \mathbf{let} \: (  \ottmv{y} ^{  \ell_{{\mathrm{1}}}  \: \sqcup \:  \ell_{{\mathrm{2}}}  } , \_  ) \: \mathbf{be} \:   \ottmv{x}   \: \mathbf{in} \:   \ottmv{y}    )     :^{  \bot  }   T_{ \ell_{{\mathrm{1}}} } \:   T_{ \ell_{{\mathrm{2}}} } \:  \ottnt{A}   $ 
\end{enumerate}
\end{proof}

%----------------------------------------------------------------------------------------------------

\end{document}